\documentclass[floatfix,aps,prb,superscriptaddress,twocolumn,longbibliography]{revtex4-1}  %
\usepackage{graphicx}  %
\usepackage{dcolumn}   %
\usepackage{bm}        %
\usepackage{amssymb}   %
\usepackage{amsmath}
\usepackage{amsthm}
\theoremstyle{definition}
\newtheorem{definition}{Definition}[section]
\newtheorem{theorem}{Theorem}[section]
\newtheorem{lemma}[theorem]{Lemma}
\newtheorem{fact}{Fact}[section]

\usepackage{multirow}
\usepackage{natbib}
\usepackage{braket}
\usepackage{amsfonts}
\usepackage{dsfont}
\usepackage{hyperref}
\usepackage{xcolor}

\usepackage{mathtools}
\usepackage{subcaption}
\usepackage{paratensor}
\usepackage[export]{adjustbox}
\usepackage{gensymb}

\usepackage{tikz}
\usepackage{tensor}
\usepackage{tikz-network}
\usepackage{paraKDHlattice}
\usepackage{color}
\definecolor{mypurple}{RGB}{255,0,255}

\makeatletter
\newcommand{\customlabel}[2]{%
	\protected@write \@auxout {}{\string \newlabel {#1}{{#2}{\thepage}{#2}{#1}{}} }%
	\hypertarget{#1}{#2}
}
\makeatother

\newcommand{\id}{\mathrm{id}}
\newcommand{\op}{\mathrm{op}}

\newcommand{\C}{\mathbb{C}}
\newcommand{\calC}{\mathcal{C}}
\newcommand{\calM}{\mathcal{M}}

\newcommand{\B}{\mathcal{B}}

\newcommand{\Hom}{\mathrm{Hom}}
\newcommand{\Rep}{\mathrm{Rep}}

\newcommand{\Tr}{\mathrm{Tr}}
\newcommand{\Hil}{\mathcal{H}}

\newcommand{\calS}{\mathcal{S}}

\newcommand{\oA}{o}
\newcommand{\oB}{s}
\newcommand{\calN}{\mathcal{N}}
\newcommand{\calD}{\mathcal{D}}
\newcommand{\Fun}{\mathrm{Fun}}
\newcommand{\Vect}{\mathrm{Vec}}
\newcommand{\Irr}{\mathrm{Irr}}
\newcommand{\GCT}{G}
\newcommand{\CTFG}{H}

\usepackage{array}
\newcolumntype{H}{>{\setbox0=\hbox\bgroup}c<{\egroup}@{}}

\usepackage{pifont}%
\newcommand{\cmark}{\ding{51}}%
\newcommand{\xmark}{\ding{55}}%
\newcommand{\Yes}{{\color{blue}\cmark}}
\newcommand{\No}{{\color{red}\xmark}}

\begin{document}
\title{Secret communication games and a hierarchy of quasiparticle statistics in $3+1D$ topological phases}
	\author{Zhiyuan Wang}
	\affiliation{Max-Planck-Institut f{\"{u}}r Quantenoptik, Hans-Kopfermann-Str. 1, 85748 Garching, Germany}
	\affiliation{Perimeter Institute for Theoretical Physics, 31 Caroline St N, Waterloo, Ontario N2L-2Y5, Canada}
	\date{\today}

\begin{abstract}	
We show that a family of secret communication challenge games~(generalizing Ref.~\onlinecite{wang2024parastatistics})
naturally define a hierarchy of emergent quasiparticle statistics in three-dimensional~(3D) topological phases. %
The winning strategies exploit a special class of the recently proposed $R$-paraparticles~\cite{wang2023para} to allow nonlocal secret communication between the two participating players. 
We first give a  high-level, axiomatic description of %
emergent $R$-paraparticles, %
    and show that any physical system hosting such particles admits a winning strategy. %
We then analyze the games %
using the categorical description of topological phases~(where point-like excitations in 3D are described by symmetric fusion categories~\cite{doplicher1971local,*doplicher1974local,LanKongWen3DAB,LanWen3DEF}), 
and show that only $R$-paraparticles can win the 3D challenge in a noise-robust way, 
and the winning strategy is essentially unique. %
This analysis associates emergent $R$-paraparticles to deconfined %
gauge theories based on %
an exotic class of finite groups~\cite{Howlett1982}. %
Thus, even though this special class of $R$-paraparticles %
are fermions or bosons under the categorical classification~\cite{doplicher1971local,*doplicher1974local}, their exchange statistics can still have nontrivial physical consequences in the presence of appropriate defects, and the $R$-paraparticle language offers a more convenient  description of the winning strategies.  %
Finally, while a subclass of non-Abelian anyons can win the game in 2D, we introduce twisted variants that exclude anyons, thereby singling out $R$-paraparticles in 2D as well. 
Our results establish the secret communication challenge as a versatile diagnostic for both identifying and classifying exotic exchange statistics in topological quantum matter. 
\end{abstract}  

\maketitle
\tableofcontents

\section{Introduction}\label{sec:intro}
Understanding what types of identical particle statistics is possible in nature is a fundamental question in quantum mechanics that has been extensively discussed ever since the birth of the theory, and is still frequently revisited nowadays. A standard textbook argument shows that bosons and fermions are the only two possible types of identical particles. 
Anyons~\cite{Leinaas1977,Wilczek1982Magnetic,Wilczek1982Quantum,Wilczek1990book,Nayak2008NAAnyons,STERN2008204} in two-dimensional~(2D) systems provide an important exception to the boson/fermion dichotomy, %
where anyon braiding realize representations of the braid group instead of the symmetric group. Although anyons are unlikely relevant to elementary particle physics due to dimensionality considerations, they can emerge in topological  phases~\cite{TPorder1,wenTopologicalOrdersEdge1995,kitaev2003fault,levin2005string,kitaev2006anyons,Chen2010LUT,Wen2017Zoo,Barkeshli2019SETclassification} of 2D condensed matter systems, and have promising applications in topological quantum computation~\cite{kitaev2003fault,freedman2003topological,Dennis2002TQM,Nayak2008NAAnyons,wang2010topological,sternTopologicalQuantumComputation2013,lahtinenShortIntroductionTopological2017}. 

In three dimension~(3D), it has long been believed that fermions and bosons are the only two possibilities, corresponding to the trivial and the sign representations of the symmetric group, respectively. %
One may wonder if it is possible to have identical particles that transform under higher dimensional representations of the symmetric group, which, in a sense, generalize non-Abelian anyons to any spatial dimension. %
Indeed, this possibility, known as parastatistics~\cite{Green1952}, has been considered even before the proposal of anyons, and has been extensively investigated by the high energy and mathematical physics community~\cite{Araki1961,greenbergSpinUnitarySpinIndependence1964,Greenberg1965,LANDSHOFF196772,druhl1970parastatistics,Taylor1970b, doplicherFieldsStatisticsNonabelian1972, doplicher1990,tolstoyOnceMoreParastatistics2014,Stoilova2020}. %
Despite being mathematically consistent and physically reasonable, parastatistics was gradually forgotten by the physics community, due to the widespread belief that paraparticles are physically equivalent to ordinary fermions and bosons enriched with internal degrees of freedom, and therefore do not lead to new physics. 
The reason behind this belief~(also called the conventionality argument~\cite{Baker2015Conventionality}, see also the discussion in Refs.~\cite{freedmanProjectiveRibbonPermutation2011,simon2023topological}) is that ordinary fermions and bosons enriched with extra internal degrees of freedom--such as spin, color, or flavor--also realize higher dimensional representations of the symmetric group, and it appears hard to distinguish parastatistics from this trivial case.  %
This belief is further strengthened by the rigorous no-go theorems~\cite{doplicher1971local,*doplicher1974local,Buchholz1982} in algebraic quantum field theory~\cite{halvorson2006algebraic,haagLocalQuantumPhysics1996} and the classification results of 3D topological order~\cite{LanKongWen3DAB,LanWen3DEF}, which shows that topological quasiparticles in 3D gapped phases of matter are classified by symmetric fusion categories~(SFCs). 
All SFCs have the form of either $\mathrm{Rep}(G)$ or $\mathrm{sRep}(G,z)$ for some finite group $G$~\cite{Deligne2002CategoriesTensorielles,TenCat_EGNO}, describing the universal properties of
charged particles coupled to deconfined $G$-gauge fields. 
These theorems have since been widely cited by the physical community as the ``no-go'' theorems for paraparticles, 
whose actual claim is that paraparticles are ``trivial'' rather than impossible. 

Despite the mathematical rigor and generality of these ``no-go'' theorems, %
there is one subtle point that has mostly been overlooked: %
although all these particles in $\mathrm{Rep}(G)$ or $\mathrm{sRep}(G,z)$ are categorically classified as fermions and bosons, 
their seemingly ``trivial'' exchange statistics can still have very nontrivial physical consequences when certain types of defects are present. 
To demonstrate this, Ref.~\onlinecite{wang2024parastatistics} designed a secret communication challenge game, in which winning requires the two participating players to achieve nonlocal communication of a message, using only a sequence of  local operations on two far-separated regions that slowly exchange their positions, without leaving any trace of information outside the two regions. %
Despite being seemingly impossible, Ref.~\onlinecite{wang2024parastatistics} showed that there exists an exotic family of 3+1D topological phases that can pass the challenge in a noise robust way. 
The winning strategy exploits a special class of the recently proposed $R$-paraparticles~\cite{wang2023para} that can be realized in certain deconfined $G$-gauge theories %
based on an exotic class of finite groups~\cite{Howlett1982}. In presence of certain point-like defects, the exchange statistics of these $R$-paraparticles allow non-local secret communication between two players, demonstrating their dramatic difference from our conventional picture of fermions and bosons. 
This gives an explicit physical demonstration of nontrivial $R$-parastatistics~\cite{wang2023para} 
from a quantum information perspective, a 
smoking gun experiment for %
this special class of topological order, and can potentially lead to new quantum technology in secret communication.

Nevertheless, Ref.~\onlinecite{wang2024parastatistics} left several important questions open. First and foremost, %
although Ref.~\onlinecite{wang2024parastatistics} provided a simple physical picture that only emergent $R$-paraparticles can win the 3D version of this challenge game, a complete proof of necessity is still missing. 
Ref.~\onlinecite{wang2024parastatistics} illustrated the winning strategy through an exactly solvable model and briefly hinted a categorical description using fusion diagrams, %
but a systematic way to describe, construct, and classify $R$-parastatistics in this exotic class of deconfined $G$-gauge theories is still lacking. %
In addition, while the 3D version of the game has the best potential to unambiguously single out $R$-paraparticles, the 2D version of the game can also be interesting in detecting topological order and quasiparticle statistics. However, how to operationally distinguish between $R$-paraparticles and non-Abelian anyons using these challenge games is an open question.

This current paper is devoted to a comprehensive study of the secret communication game and its generalizations, answering all above open questions in the process. %
To give a simple physical description of the winning strategies in the most general setting, we introduce a high level, axiomatic description of the universal properties of emergent $R$-paraparticles and show that any physical system hosting such particles can win the challenge. 
We then give a~(physically rigorous) proof~\footnote{Since this paper is intended for physicists, we do not pursue the highest standard of mathematical rigor here. The main gap from being mathematically rigorous 
is discussed at the beginning of Sec.~\ref{sec:categorical_analysis} and Sec.~\ref{sec:categorical_framework}. 
} that only emergent $R$-paraparticles can win the full version of the 3D challenge, 
based on the assumption that point-like quasiparticles in 3D gapped phases are universally described by symmetric fusion categories~(SFCs)~\cite{doplicher1971local,*doplicher1974local,Buchholz1982,LanKongWen3DAB,LanWen3DEF}. 
Using this framework, we classify which systems can succeed in increasingly stringent versions of the challenge game, including generalizations and twisted variants introduced in this paper. %
This analysis eventually lead to a hierarchy of point particle exchange statistics in 3+1D topological phases, %
with the lowest hierarchy being ordinary bosons and the highest hierarchy being what we call ``full-fledged $R$-paraparticles'', defined by their ability to pass different levels of the game, as summarized in Tab.~\ref{tab:protocols_strategies}. 
We then introduce the anti-anyon twists for the 2D version of the challenge,  %
and illustrate how they isolate $R$-paraparticles from non-Abelian anyons in 2D.

Our paper is organized as follow. In Sec.~\ref{sec:parachallenge_variants_twists} we define the basic challenge game, and mention some different variants and extra winning conditions. %
In Sec.~\ref{sec:win} we introduce the axioms of emergent $R$-parastatistics, and describe winning strategies in this framework, where we illustrate  that different types of particle statistics have different capabilities in the game. This describes levels 1-4 in Tab.~\ref{tab:protocols_strategies}. %
Then in Sec.~\ref{sec:WHF} %
we introduce the ``who-entered-first'' challenge, and describe its winning strategy, which singles out ``full-fledged paraparticles'', the highest hierarchy in Tab.~\ref{tab:protocols_strategies}. 
In Secs.~\ref{sec:categorical_analysis} and \ref{sec:ModCatDefect} we present a detailed categorical analysis of possible winning strategies and show that
only emergent $R$-paraparticles~(as defined by our axioms) can win the full version of this challenge game in 3D, establishing the hierarchy on solid grounds. 
In Sec.~\ref{sec:anti-anyon} we present the anti-anyon twist for the 2D version of the game that distinguish $R$-parapartiles from anyons in 2D. 
Finally, in Sec.~\ref{sec:conclusions} %
we conclude our work with an extended discussion on the meaning of $R$-paraparticles in the context of topological order,  how the secret communication games provide new insights into topological phases and long range entanglement, and mention several potential generalizations and future directions. The appendices contain technical details omitted in the main text, where the most important section is App.~\ref{sec:RfromCentralType}, where we give a large number of examples of SFCs that contain $R$-paraparticles. 
	\begin{table*}[t]
		\centering
		\begin{tabular}{|c|c|c|c|c|c|c|H@{\hspace*{-\tabcolsep}}}
				\hline
				$R$-paraparticle & $R=X$& $R=-X$ &  $R^{b'a'}_{ab}=\pm\delta_{aa'}\delta_{bb'}\theta_a\theta^{*}_b$ & $R^{b'a'}_{ab}=\delta_{aa'}\delta_{bb'}\theta_{ab}$  & Eq.~\eqref{eqApp:seth-R} & Eq.~\eqref{eq:RfromUniversalR} & $R=\pm \mathds{1}, m>1$\\
				\hline
				Win in theory &\No & \Yes & \Yes & \Yes &\Yes & \Yes & \Yes \\
				\hline
				Noise robust &\No &\No & \Yes & \Yes & \Yes & \Yes &\Yes \\
				\hline
				Anti-eavesdropping &\No & \No  & \No & \Yes & \Yes & \Yes & \Yes\\ 
				\hline
				Who-entered-first challenge & \No & \No  & \No & \No & \Yes & \Yes &\Yes  \\ %
				\hline
				Identical particle test &\Yes & \Yes & \Yes & \Yes & \Yes   & \No &\Yes \\
				\hline
				Example $G$ & $S_3, D_{4n+2}$ &$Z_2,Q_{4n}$ & $D_8$ & $A_4, Z_m\ltimes Z^{\times m}_2$ &  $D_8\ltimes Z_2^{\times 3}$&  $D_8\ltimes Z_2^{\times 3}$ & -\\
				\hline
				Level & 1 &2  & 3 & 4 & 5 & 5 (mutual)& Beyond SFC \\ 
				\hline
			\end{tabular}
			\caption{\label{tab:protocols_strategies} 
				Hierarchy of particle exchange statistics of  in 3+1D topological phases, defined by the 
				ability of a certain type of $R$-paraparticle~(defined by the axioms in Sec.~\ref{sec:axioms_emergent_para}) to win increasingly stringent versions of the challenge game.  
				Here $X^{b'a'}_{ab}=\delta_{aa'}\delta_{bb'}$, and $R=X$~($R=-X$) describes ordinary bosonic~(fermionic) statistics~(potentially enriched with internal degrees of freedom), considered as trivial types of $R$-paraparticles.  
				We say that a certain type of topological quasiparticle $\psi$ can ``win in theory'' if there exists a strategy using $\psi$ %
				that wins the basic challenge in Sec.~\ref{sec:original_version}, without imposing the extra conditions in Sec.~\ref{sec:extra_win_cond}. The important condition of noise robustness distinguish nontrivial $R$-paraparticles from ordinary fermions and bosons. The who-entered-first challenge defined in Sec.~\ref{sec:WHF} singles out the most nontrivial class of $R$-paraparaticles, which we call ``full-fledged paraparticles''. The identical particle test in Sec.~\ref{sec:IPT} separates mutual parastatistics from self-parastatistics. The row ``Example $G$'' means an example of a group $G$ such that 
				a certain deconfined $G$-gauge theory realizes such a particle.
			}
		\end{table*}

\section{The parastatistics challenge and several generalizations}\label{sec:parachallenge_variants_twists}
In this section we define the protocol of the parastatistics challenge game. %
Specifically, in Sec.~\ref{sec:original_version} we define the basic challenge, including the original version introduced in Ref.~\onlinecite{wang2024parastatistics},  along with   a slightly different variant~(Sec.~\ref{sec:oneptversion}).  %
In Secs.~\ref{sec:IPT} we introduce  
the identical particle test that distinguishes mutual parastatistics from self parastatistics. 
Here we only define the game protocols;
winning strategies using emergent paraparticles will be given in Sec.~\ref{sec:win}, and a detailed categorical analysis will be presented in Sec.~\ref{sec:categorical_analysis}.
\subsection{The basic challenge}\label{sec:original_version}
In the basic version of the parastatistics challenge,
the participants involve two players,  Alice~(A) and Bob~(B), who work as a team against a group of Referees~(R). %
The goal of the players is to send a message to each other using a restricted class of local operations on a common quantum many body system, 
and the Referees' role is to initiate the challenge and monitor the game process to ensure that the players are obeying the rules. %

\subsubsection{Pregame preparations}
Prior to commencing the challenge, Alice and Bob may discuss to agree on an overall strategy. Once their strategy is set, they must:
\begin{enumerate}
	\item %
	Provide a locally-interacting Hamiltonian $\hat{H}$ defined on a two- or three-dimensional lattice with a unique, gapped, and frustration-free ground state $\lvert G\rangle$, and prove these properties rigorously.
	\item %
	Choose the radius $r_0$ of the circular regions (see  Fig.~\ref{fig:game}) and identify two well-separated lattice sites $o$ and $s$ (which may lie on the system’s boundary). Choose an integer $m_0\geq 2$. 
	\item %
	Realize the ground state $\lvert G\rangle$ experimentally on a system of linear size $L \gg r_0$, where $L$ will be chosen by the referees after reviewing items (1) and (2).
\end{enumerate}
By \emph{gapped}, we mean there is a uniform lower bound--independent of system size--on the energy difference between the first excited state and $\ket{G}$. \emph{Frustration-free} means
$\hat{H}=\sum_{i}\hat{h}_i$ 
with each $\hat{h}_i \ge 0$ and $\hat{h}_i \ket{G} = 0$; and the players must exhibit this decomposition in their proof.
We remark that the various requirements on the ground state $\ket{G}$~(unique, gapped, and frustration-free) are imposed mainly to simplify our discussions, and Sec.~\ref{sec:conclusions} we discussion possible ways to relax some  of these requirements. 
We also emphasize that we do not require $\hat{H}$ to be translationally invariant, and we allow defects to be present in the system, provided that the spatial arrangement of the defects satisfy some mild technical conditions that we detail in App.~\ref{app:condition_defect}. Indeed, any winning strategy to this game requires some special topological defects at the points $\oA$ and $\oB$, where one can locally create and measure a single topological quasiparticle.

Once these components are in hand, the Referees will rigorously check the proof and experimentally validate the prepared state $\ket{G}$ by verifying~(through quantum measurements) $\langle G|\hat{h}_i|G\rangle=0$ for all $i$. Alice and Bob are then placed in separate rooms, and the Referees randomly select two numbers $a,b\in\{1,2,\ldots,m_0\}$, and assign $a$ to Alice and $b$ to Bob. The challenge for Alice and Bob is to infer information about their partner's number solely through a restricted set of local operations on the quantum system they prepared, as we detail below.

[It is clear that the difficulty of the challenge game  increases with the integer $m_0$, so Alice and Bob can simply choose $m_0=2$ to make things easier. This $m_0$ is introduced to show that different $R$-paraparticles have a different information transfer capacity~(which is a topological invariant as we show later in Sec.~\ref{sec:reduction_to_TPO}), but we do not dive into that aspect in this paper.]

\subsubsection{General rules during the game}\label{sec:general_rules}
We now introduce some general rules and assumptions that not only apply to the basic challenge but also to all the additional challenges and generalizations that we introduce in later parts of this paper. The  rules are: \\
(1). Whenever a game begins, the state of the physical system is always initialized to be $\ket{\Psi(0)}=\ket{G}$, the ground state prepared by the players;\\
(2). Each player is assigned a circular area of radius $r_0$ in the physical system, within which he or she is allowed to perform arbitrary local unitary operations and measurements. 
(Note that in the 3D version of the game, the circles become spheres of radius $r_0$. In this paper, whenever we say a circle, we mean a circle in the 2D case and a sphere in the 3D case, by default); \\
(3). The circles moves slowly and continuously in time, and the motion is controlled by the Referees. At some initial and final stages of the game, a player's circle may not be present, which simply means that the player has no control over the physical system at the moment;\\
(4). During the game, any direct form of communication between players are forbidden by default~(except in the antiparticle test we introduce later in Sec.~\ref{sec:antiparticle}, where we will explicitly state who are allowed to communicate). Each player is confined in a separate room, with the only access to the outside world being the aforementioned local operations on the physical system, as shown in Fig.~\ref{fig:ABcontrol};\\
(5). Whenever a player's  circle moves, he or she is always obliged to move whatever excitations inside the circle to follow the circle movement, and is not allowed to leave any trace of information behind. The Referees enforce this condition by constantly checking the local ground state condition $\braket{\Psi(t)|\hat{h}_i|\Psi(t)}=0$ everywhere beyond the circle areas, where $\Psi(t)$ is the quantum state of the system at time $t$. If at any moment, the Referees detect an excitation beyond the circle areas, the challenge fails~\footnote{In some game protocols to be presented later, we introduce additional players who play against Alice and Bob, such as 
the scramblers in the anti-anyon twists. In these cases we assume that these opposing players are promised not to leave any excitations behind. Alternatively, we can say that if these opposing players violate any rule, it is their team that lose the game, rather than the team of Alice and Bob}. \\
(6). Right before a player's circle disappears~(he or she will be alerted about this in advance), the player is always obliged to clean up whatever excitation inside the circle, so that the Referee cannot detect an excitation after the circle disappears. 

In the rest of this paper, we will refer to a player's circle by the initial of his/her name, e.g., circle A refers to Alice's circle.

\subsubsection{The original game with two distinguished special points}\label{sec:twoptversion}
The protocol of the original game introduced in Ref.~\onlinecite{wang2024parastatistics} is illustrated in Fig.~\ref{fig:game}. 
When the game begins at $t=0$, circle A appears at the special point $\oA$, while circle B appears at $\oB$, as shown in Fig.~\ref{fig:gamestart}. 
Then the Referees randomly select two paths connecting $\oA$ and $\oB$, as shown in Fig.~\ref{fig:game}, and slowly move the two circles along their respective paths simultaneously, in such a way that the two circles always remain far apart. 
The game ends at $t=T$ in the configuration shown in Fig.~\ref{fig:gameend}, when the two circles complete an exchange of positions. After this, both circles disappear and the Referees perform one last check of the local ground state condition $\braket{\Psi(T)|\hat{h}_i|\Psi(T)}=0$ everywhere. 
If this final check is passed, the Referees will ask Alice about $b$, and ask Bob about $a$, and the players win if they both answer correctly.

The physical intuition behind the game design is the following: under all the restrictive conditions of the game, the only thing the players can do is playing with quasiparticle excitations inside their circle areas, including particle creation, annihilation, movement, and measurement. %
If the system hosts emergent paraparticles, %
the players can exploit their nontrivial exchange statistics to nonlocally transfer information to each other~\cite{wang2024parastatistics}, by holding a paraparticle inside each circle and manipulating its internal state, as we describe in detail in Sec.~\ref{sec:win_2pt}, and in Sec.~\ref{sec:TC_winning_condition} we argue that this is the only possible way to win the game. 

\begin{figure}
	\centering
	\begin{subfigure}[t]{.95\linewidth}
		\centering\includegraphics[width=\linewidth]{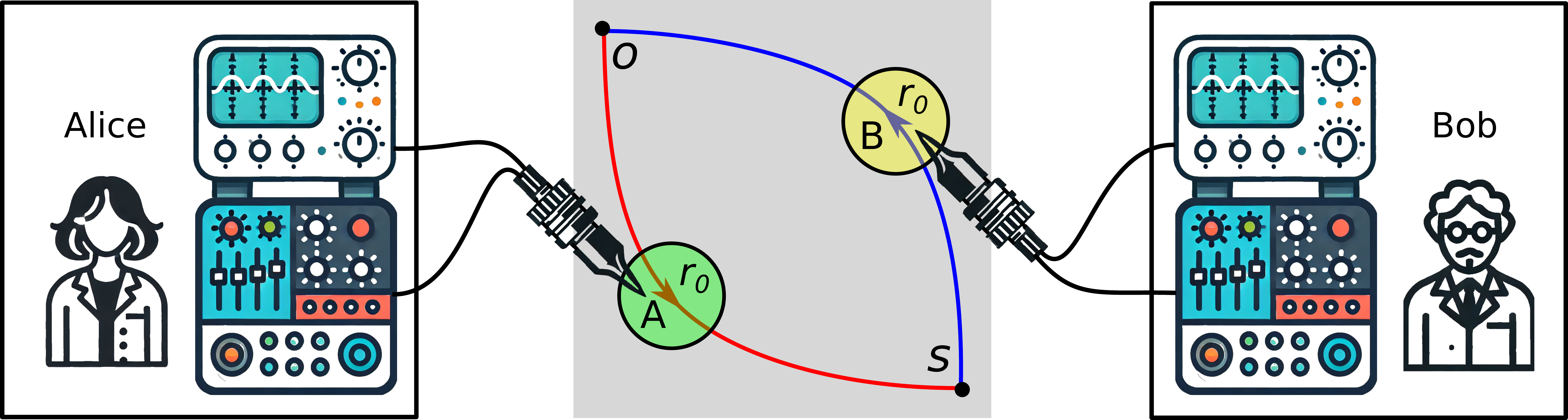}
		\caption{\label{fig:ABcontrol} Local operations controlled by the players}%
\end{subfigure}
\begin{subfigure}[t]{.324\linewidth}
	\centering\includegraphics[width=.96\linewidth]{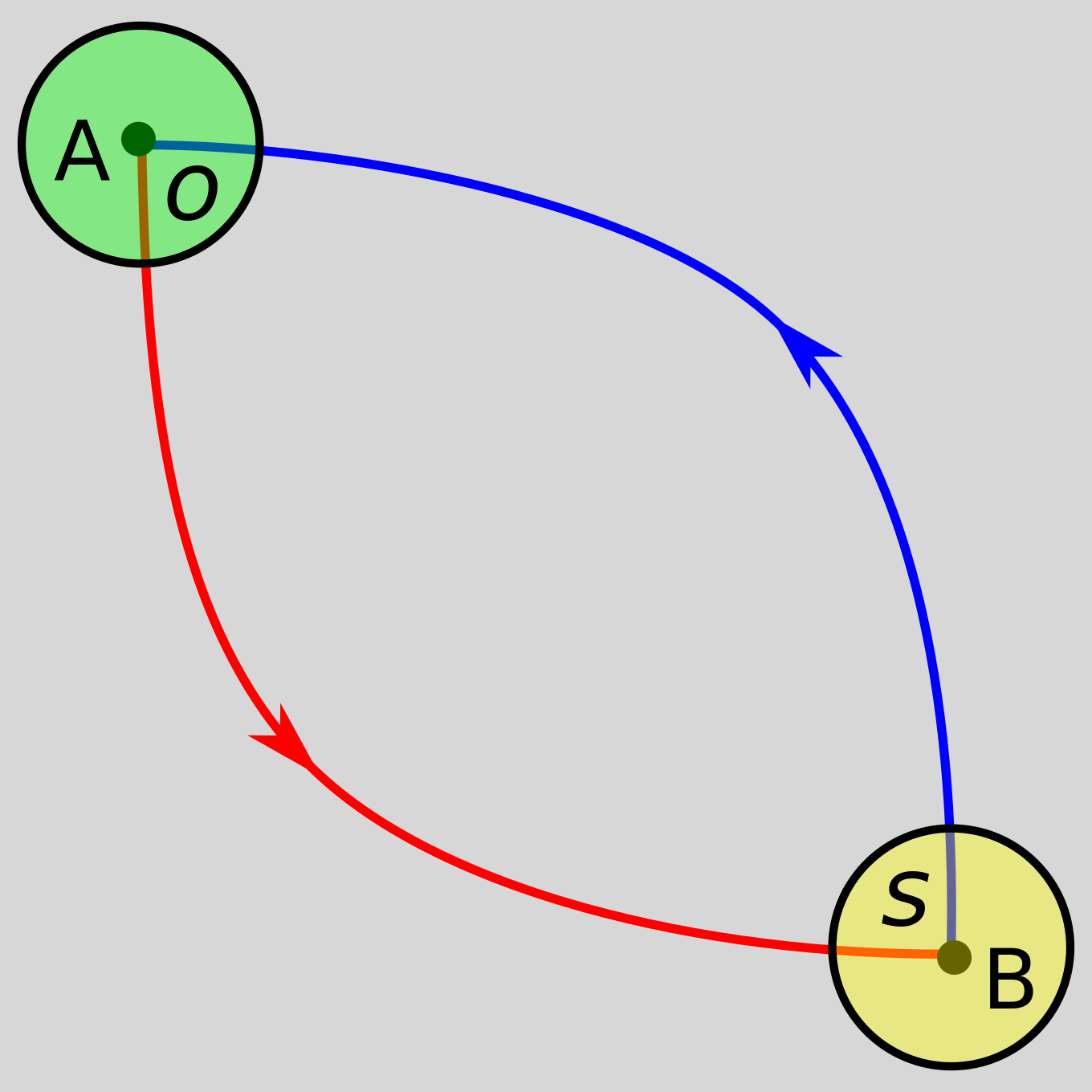} 
	\caption{\label{fig:gamestart} Start $t=0$}
\end{subfigure}
\begin{subfigure}[t]{.324\linewidth}
	\centering\includegraphics[width=.96\linewidth]{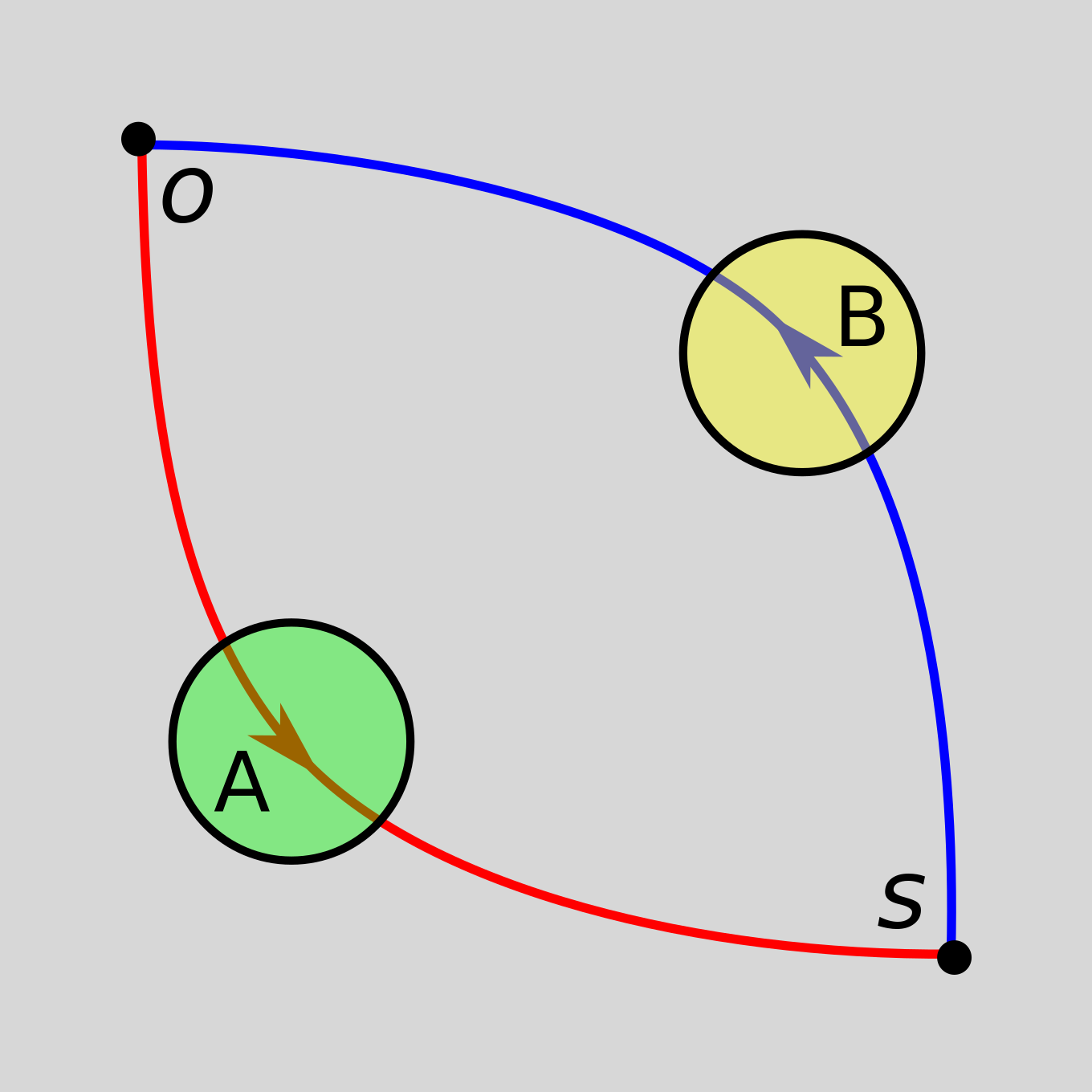} 
	\caption{\label{fig:duringgame} $0<t<T$}
\end{subfigure}
\begin{subfigure}[t]{.324\linewidth}
	\centering\includegraphics[width=.96\linewidth]{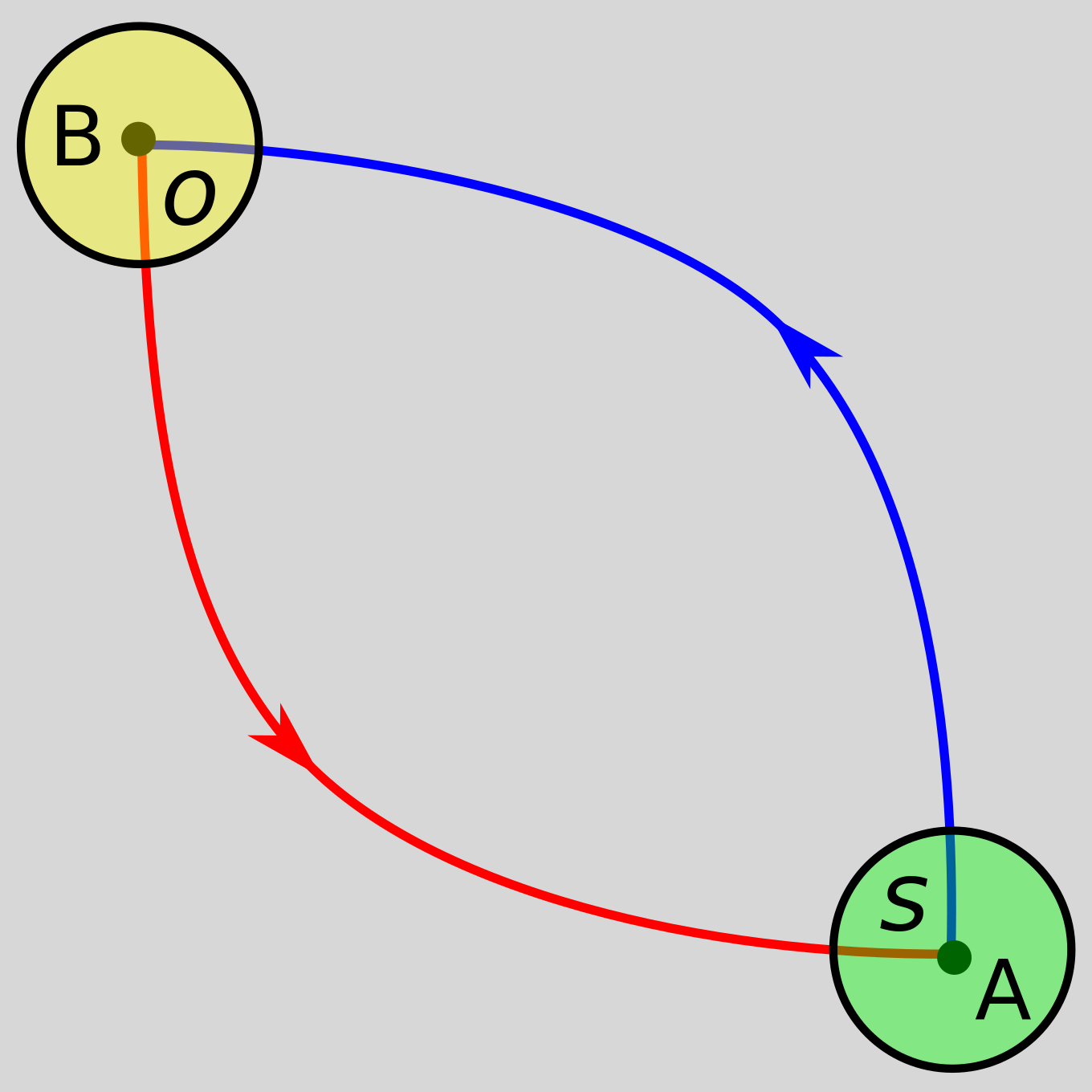} 
	\caption{\label{fig:gameend} End $t=T$}
\end{subfigure}
\caption{\label{fig:game} Illustration of the game process~(figure adapted from Ref.~\onlinecite{wang2024parastatistics}). 
	(a) The two circles have radius $r_0$ defined by the players. Each player is confined in an isolated room, and is only  allowed to do local unitary operations and  measurements within his/her own circle area; (b)-(d) During the game the Referees move the two circles along their respective paths to complete an exchange of positions. The two special points $\oA,\oB$ are chosen by the players, while the two paths are chosen by the Referees. 
} 
\end{figure}

\subsubsection{A variant with one distinguished special point}\label{sec:oneptversion}
We now define a slightly different version of this game in which both players start and end their journey at one common special point, but at different times. This version takes longer to describe, but has the advantage that the winning strategy using paraparticles is conceptually easier to understand. Furthermore, since it requires only one special point in the physical system where paraparticles can be locally created and measured, 
it may be easier to implement and control experimentally. 

In this version of the game, the pregame preparation steps are (almost) the same as the original version, except that the players are only required to choose one special point $\oA$ on the lattice~(allowed to be on the boundary). %
The game process is illustrated in Fig.~\ref{fig:bulk-one-pt}.
The rules and the goal are exactly the same as the original version, the only difference is that, in this version of the game, the spacetime trajectory of Alice's circle is exactly the same as that of Bob's, except being  delayed by $T/4$, where $T$ is the total duration of the game. 
When the game starts at $t=0$, %
circle B appears at $\oA$, and starts slowly moving along the designated path. At $t=T/4$, after Bob has moved a long distance away, circle A appears at $\oA$, and starts moving along the same path at the same speed, always keeping a large distance from circle B.  %
At $t=3T/4$, circle B returns to $\oA$ and disappears shortly after.  %
At $t=T$,  circle A returns to $\oA$ and then disappears. 
In the end, the Referees ask each player about the other player's number, and they win if they both answer correctly.  

\begin{figure}
\begin{subfigure}[t]{.32\linewidth}
	\centering\includegraphics[width=.95\linewidth]{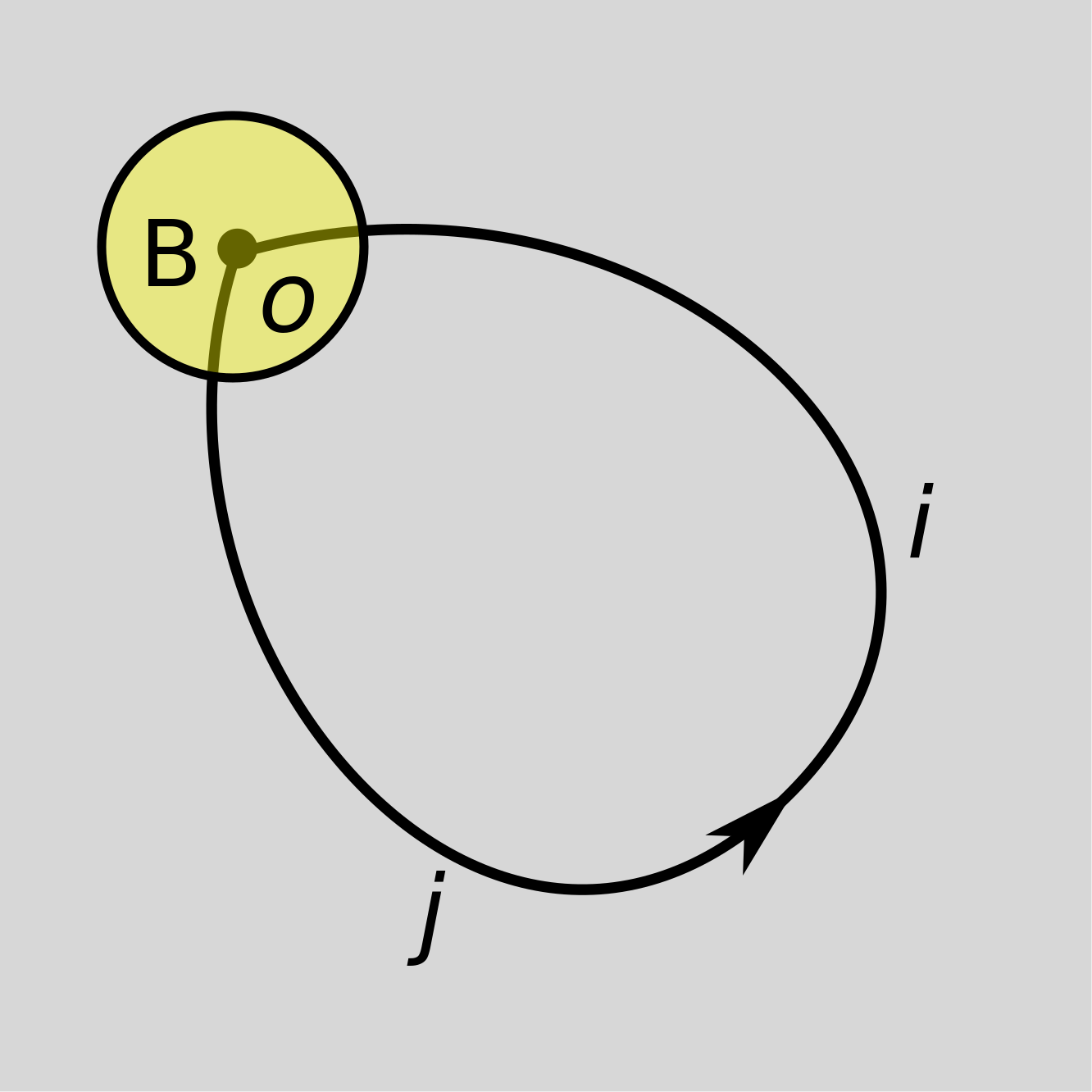} 
	\caption{\label{fig:bulk1pt-1} Start $t=0$}
\end{subfigure}
\begin{subfigure}[t]{.32\linewidth}
	\centering\includegraphics[width=.95\linewidth]{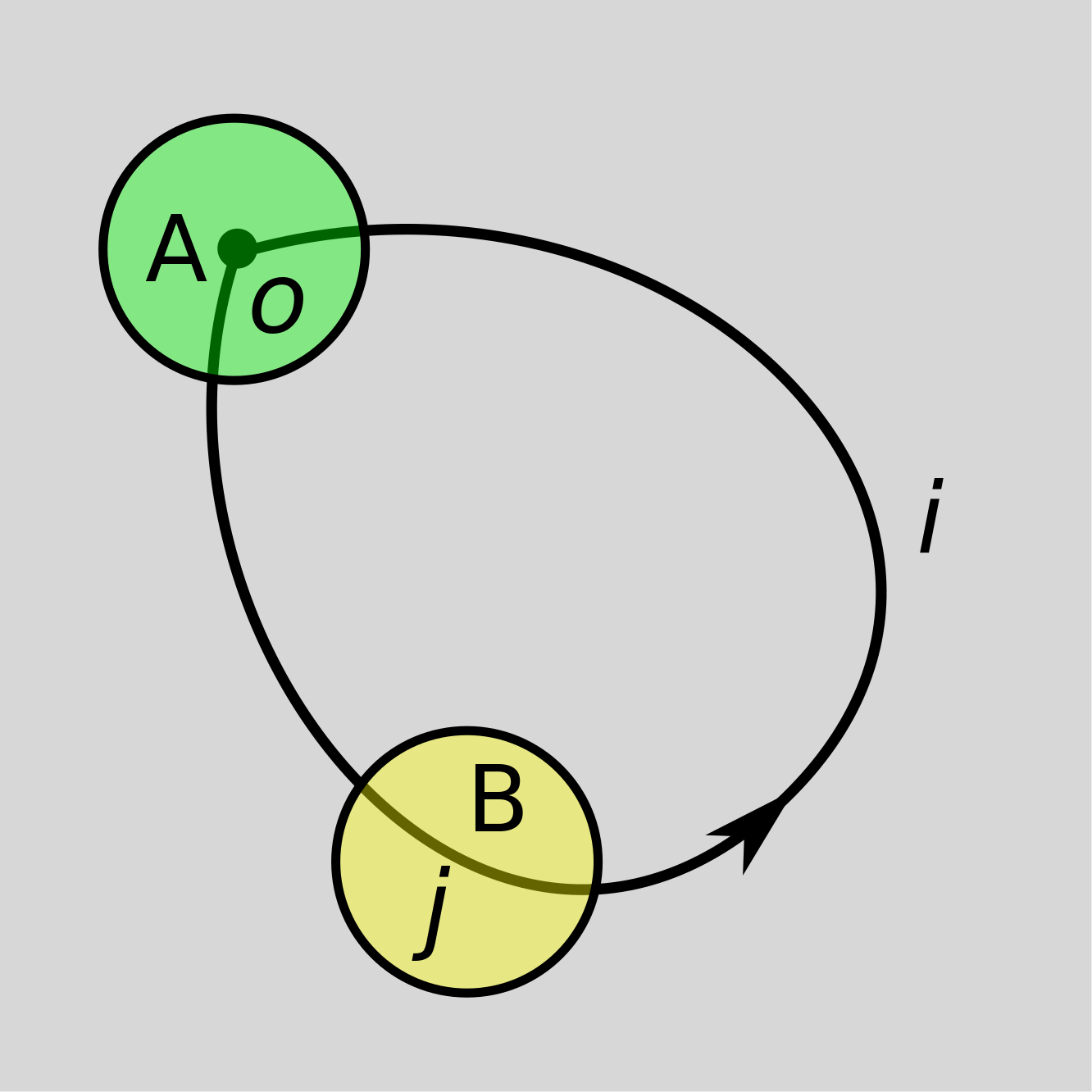} 
	\caption{\label{fig:bulk1pt-2} $t=T/4$}
\end{subfigure}
\begin{subfigure}[t]{.32\linewidth}
	\centering\includegraphics[width=.95\linewidth]{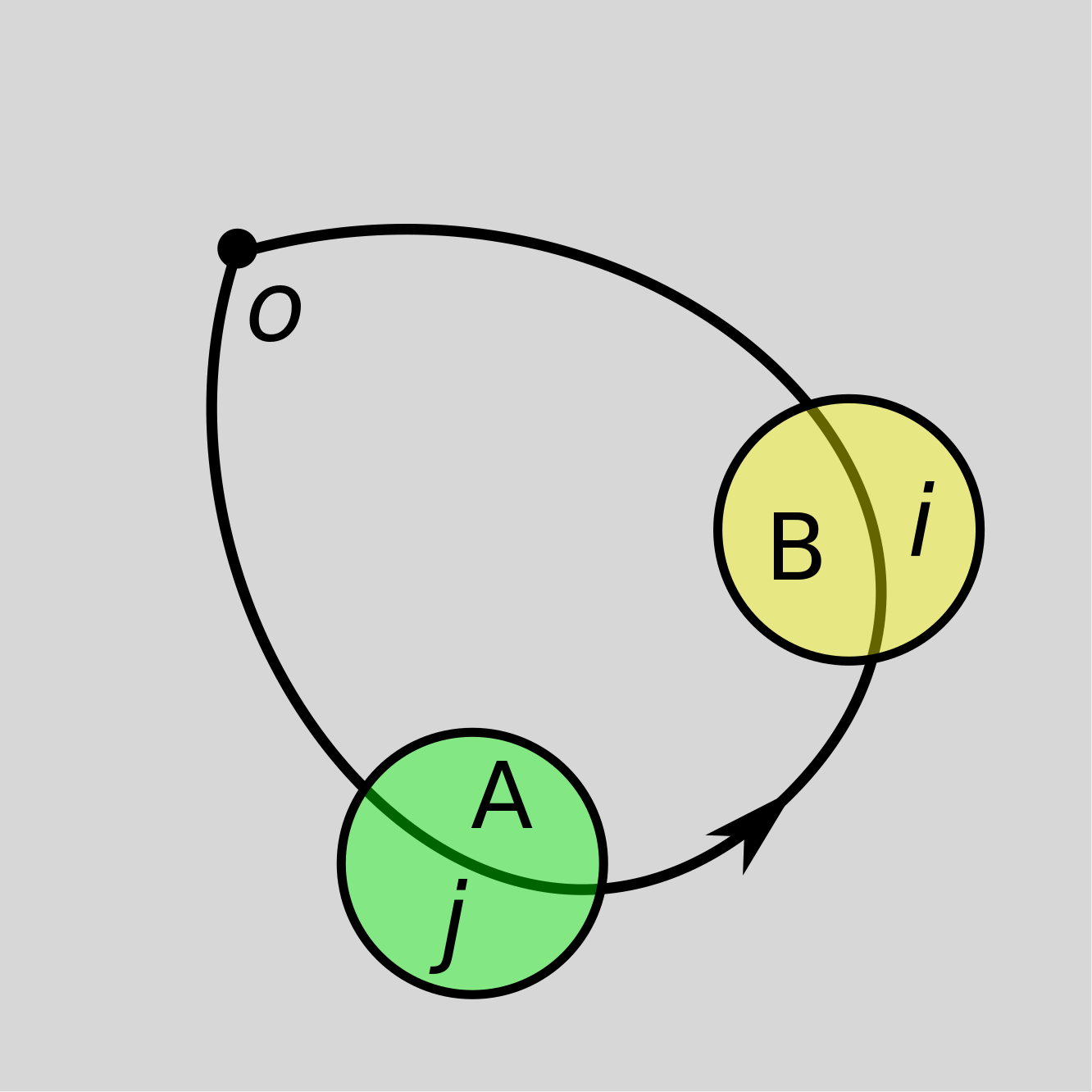} 
	\caption{\label{fig:bulk1pt-3} $t=T/2$}
\end{subfigure}
\begin{subfigure}[t]{.32\linewidth}
	\centering\includegraphics[width=.95\linewidth]{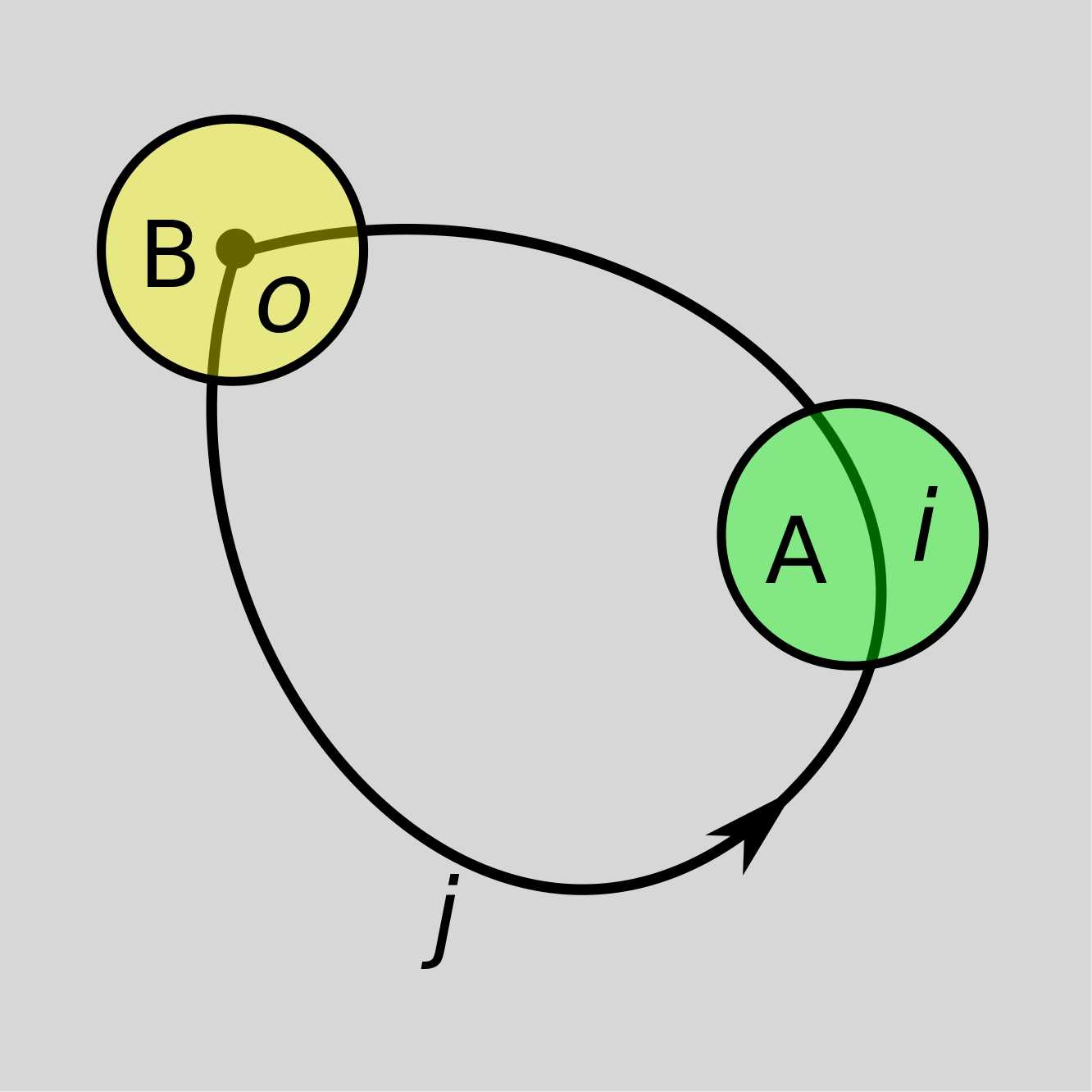} 
	\caption{\label{fig:bulk1pt-4} $t=3T/4$}
\end{subfigure}
\begin{subfigure}[t]{.32\linewidth}
	\centering\includegraphics[width=.95\linewidth]{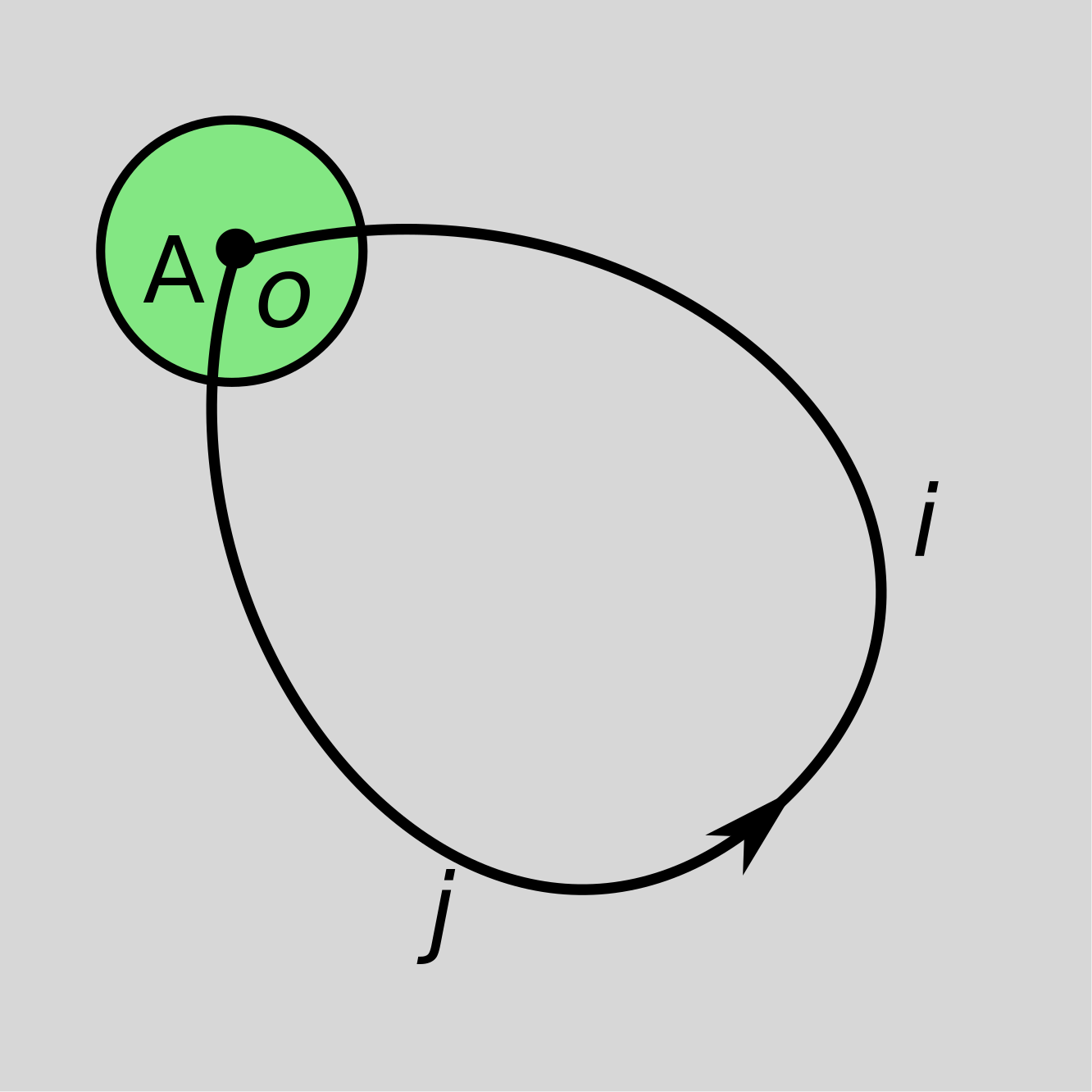} 
	\caption{\label{fig:bulk1pt-5} End $t=T$}
\end{subfigure}
\caption{\label{fig:bulk-one-pt} 
	A variant of the game protocol in which both players start at one distinguished special point, but at different time. Here $i$ and $j$ are two fixed points in the system introduced for illustration purpose. 
}
\end{figure}

\subsection{Extra winning conditions}\label{sec:extra_win_cond}
We now introduce some extra conditions that we impose on the winning strategy. There are two main motivations to introduce these extra constraints. On the theoretical side, these extra conditions establish the hierarchy of nontriviality in Tab.~\ref{tab:protocols_strategies}; on the practical side, the robustness against noise is crucial for actually carrying out the game in realistic experiments. 
\subsubsection{Robustness against noise and eavesdropping}\label{sec:robust_noise_eaves}
For simplicity, here we only consider local noise that happen inside the two circles, which is enough for our theoretical purpose of establishing the hierarchy in  Tab.~\ref{tab:protocols_strategies}, and in App.~\ref{app:noise_robust} we discuss how to treat local noise that happen in other places. 
We can model such local noise by some random local quantum channels inside the two circles, and we assume that such events happen at some constant rate that is not too high. %
We require that any realistic winning strategy should be robust against such local noise. 
This is crucial not only for carrying out the game in experiment, but also for distinguishing paraparticles from ordinary emergent fermions. Without the requirement of noise robustness, topological phases with emergent fermions, such as the toric code~\cite{kitaev2003fault}, admit an unphysical winning strategy, as we show in Sec.~\ref{sec:emergent_fermion}. %
Such a strategy involves creating a superposition of states with different fermion parity, and is therefore vulnerable to local phase noise of the form $e^{i\theta\hat{n}}$.  

A more stringent requirement we add is the robustness against eavesdropping. Specifically, suppose at some time during the game, when both players are in the bulk, an eavesdropper comes in and is allowed to make a finite number~(say, no more than four) of local measurements on the physical system, including in the circle areas. 
Here the eavesdropper is allowed to listen to the pregame discussion of the players so that he has complete knowledge about the physical system $\hat{H}$ and the players' winning strategy. 
In this case, if the eavesdropper cannot obtain any information about the numbers $a$ and $b$ no matter what local observables he measures, then we say that the strategy is robust against eavesdropping.  
Robustness against eavesdropping is needed to separate a rather trivial class of $R$-parastatistics from more non-trivial ones. This rather trivial class is characterized by 
$R^{b'a'}_{ab}=\pm\delta_{aa'}\delta_{bb'}\theta_a\theta^{*}_b$, 
and as we will see in Sec.~\ref{sec:phaseSWAP}, this rather trivial class can give a partial winning strategy that is robust against noise, but not against eavesdropping. Although this rather trivial type of $R$-matrix already shows a difference from ordinary fermions and bosons, they are still way from genuinely nontrivial $R$-matrices that describe full-fledged parastatistics, as the latter provide winning strategy robust against both local noise and eavesdropping.  %

\subsubsection{The identical particle test}\label{sec:IPT}
\begin{figure}
		\centering\includegraphics[width=.35\linewidth]{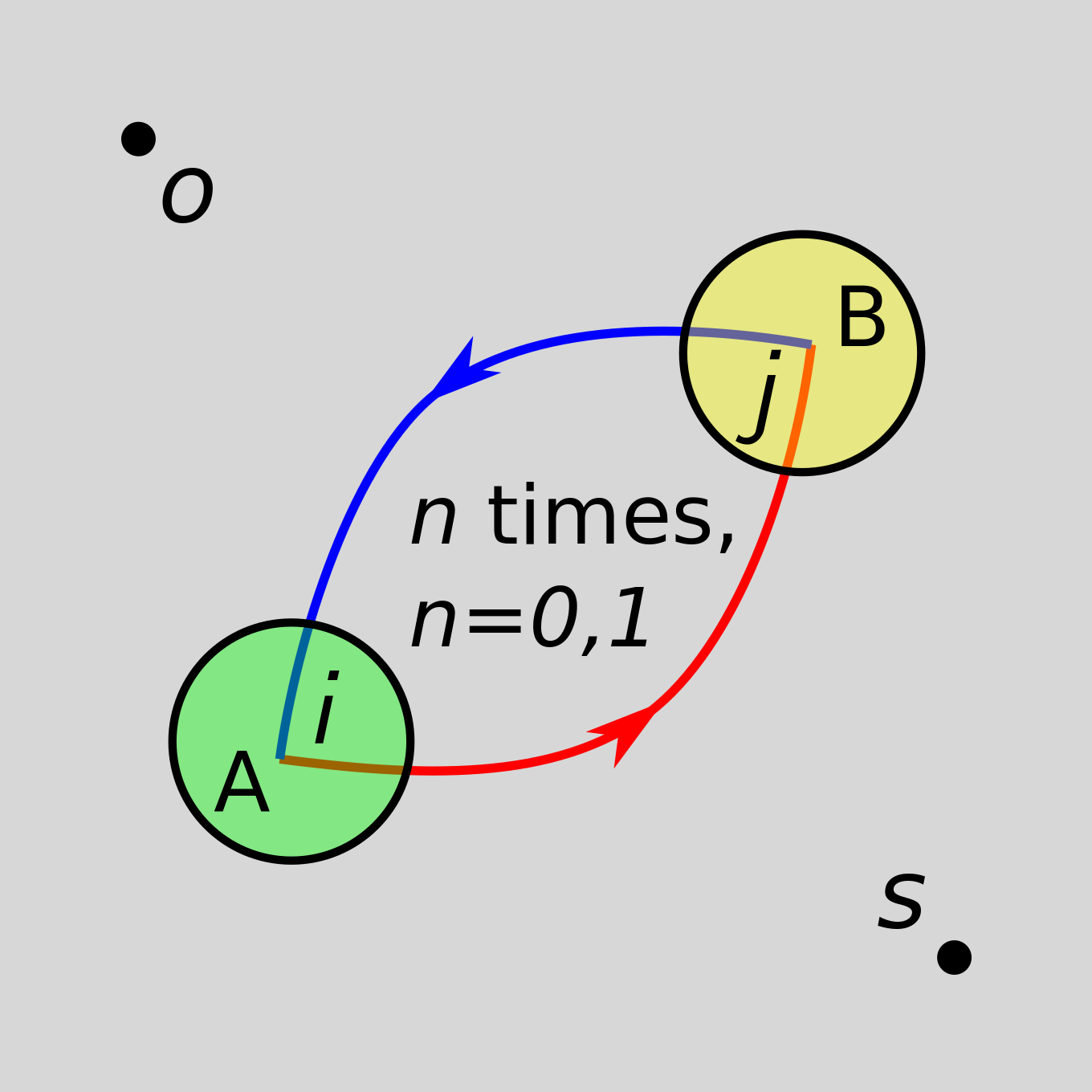} 
		\caption{\label{fig:IPT} The identical particle test 
			checks if the quantum states before and after the exchange are locally indistinguishable. 
		}
\end{figure}
The original version of the challenge game introduced in Sec.~\ref{sec:original_version} does not require the players to use identical particles--indeed, 
it is possible to pass this challenge if Alice and Bob use different types of quasiparticles with non-trivial mutual statistics, as shown in Sec.~\ref{sec:relative_para}.  
Mutual parastatistics is also consistently defined in any spatial dimension and display nontrivial physical behaviors %
compared to ordinary fermions and bosons. %
However, if one wants to single out paraparticles with non-trivial self-statistics, we can introduce an additional identical particle test to enforce the use of identical particles, as we describe below. 

The basic idea of the identical particle test is to spatially exchange the two particles and check if the physical states before and after the exchange are distinguishable by local measurements.  %
To perform this test, we introduce a third player, Charlie, who has knowledge of the Hamiltonian and the ground state.  
At $t=T/2$, when the circles of both Alice and Bob are far away from $\oA$ and $\oB$, as shown in Fig.~\ref{fig:IPT}, %
Charlie comes in, and is allowed to do any local measurements within the two circles~(similar to other players, Charlie sits in a separate room where she can perform the local measurements remotely). %
Then Charlie is asked to temporarily leave the game, and the Referees slowly move the two circles to exchange their positions $n$ times, where $n\in\{0,1\}$, as shown in Fig.~\ref{fig:IPT}. Of course, during this exchange process, Alice and Bob are still obliged to move whatever excitations in the circle areas to follow the circle movements. After this, Charlie comes back, and is allowed to do any local measurements within the two circles. %
Then the Referees ask Charlie the value of $n$. The test can be repeated multiple times, and if Charlie can answer the value of $n$ better than random guessing, then Alice and Bob fail the identical particle test; otherwise, they pass this test. %
We now briefly explain the intuition behind this test, detailed analysis will be given in Sec.~\ref{sec:passingIPT}. 
If Alice and Bob use different types of particles inside their circle areas, then Charlie can simply measure the particle type in each circle when she first comes in and after she comes back, to determine if an exchange has happened. %
By contrast, if Alice and Bob use the same type of particles, they can carefully position the particles in the circles~(e.g. always at the center) such that the local reduced density matrix in the circle at position $i$~(and similarly for $j$) is unchanged before and after the exchange. %
Consequently, Charlie cannot determine if an exchange has happened using any local measurements in the circles. %

In summary, passing the identical particle test %
requires Alice and Bob to use identical particles in their circles, thereby distinguishing between self parastatistics and mutual parastatistics. %

[In essence, this protocol tests the \textit{local indistinguishability} of the two quasiparticles, meaning that the quantum states before and after the exchange cannot be distinguished by any local measurements, but can potentially be distinguished by non-local measurements. The local indistinguishability of $R$-paraparticles is guaranteed by axioms~\ref{Axiom2}-\ref{Axiom4} we introduce later in Sec.~\ref{sec:axioms_emergent_para}.  We note that being locally indistinguishable is the way how $R$-paraparticles avoid a more recent no-go theorem on parastatistics~\cite{mekonnen2025invariance}, which assumes a more restrictive indistinguishability condition~(this will be discussed in more detail in a future work~\cite{wangOnRparaI}).]

\subsubsection{*The antiparticle test}\label{sec:antiparticle}
For some very technical reason we also introduce the antiparticle test, which is a variant of the original challenge in Sec.~\ref{sec:original_version}. 
It is an optional test introduced mainly to simplify our later analysis in Sec.~\ref{sec:TC_winning_condition}, as it turns out to be technically easier to classify the subclass of 3+1D topological phases that can also pass this antiparticle test compared to the most general case. We recommend the readers to skip this section at first reading, as it does not affect any main point of this paper.
\begin{figure}
	\begin{subfigure}[t]{.32\linewidth}
		\centering\includegraphics[width=.85\linewidth]{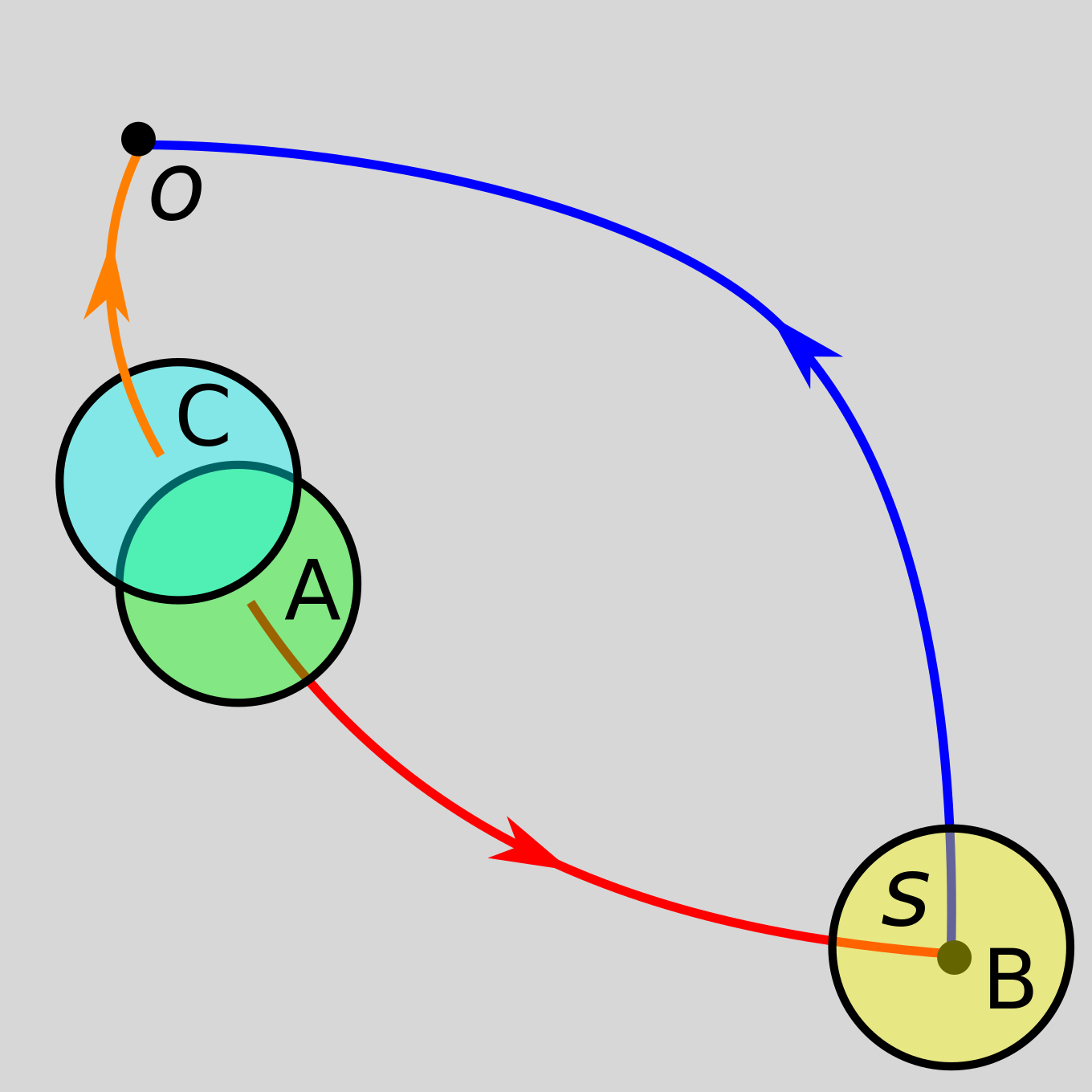} 
		\caption{\label{fig:Exchange-bulk-APtest-1} Start $t=0$}
	\end{subfigure}
	\begin{subfigure}[t]{.32\linewidth}
		\centering\includegraphics[width=.85\linewidth]{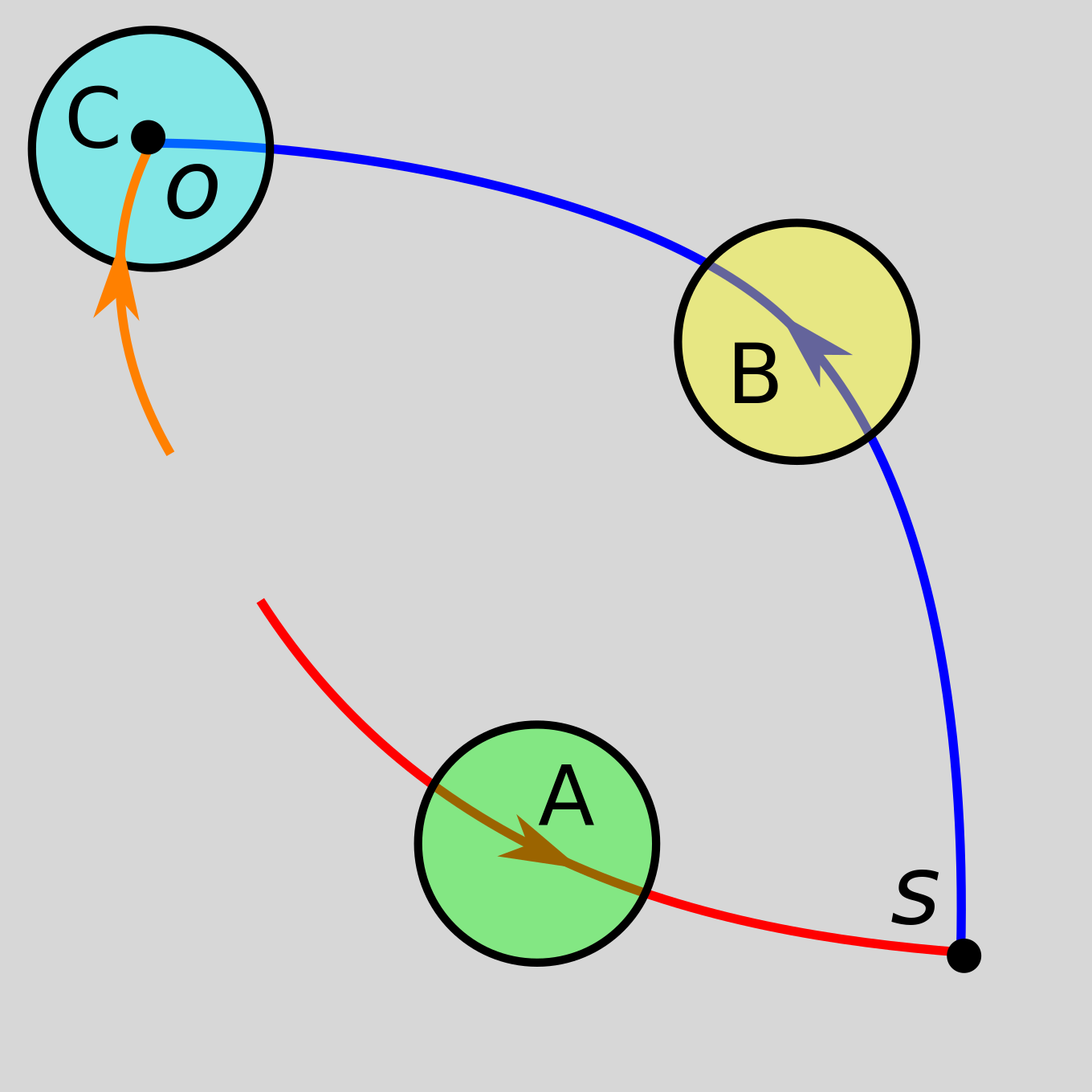} 
		\caption{\label{fig:Exchange-bulk-APtest-2} $0<t=t_1<T$}
	\end{subfigure}
	\begin{subfigure}[t]{.32\linewidth}
		\centering\includegraphics[width=.85\linewidth]{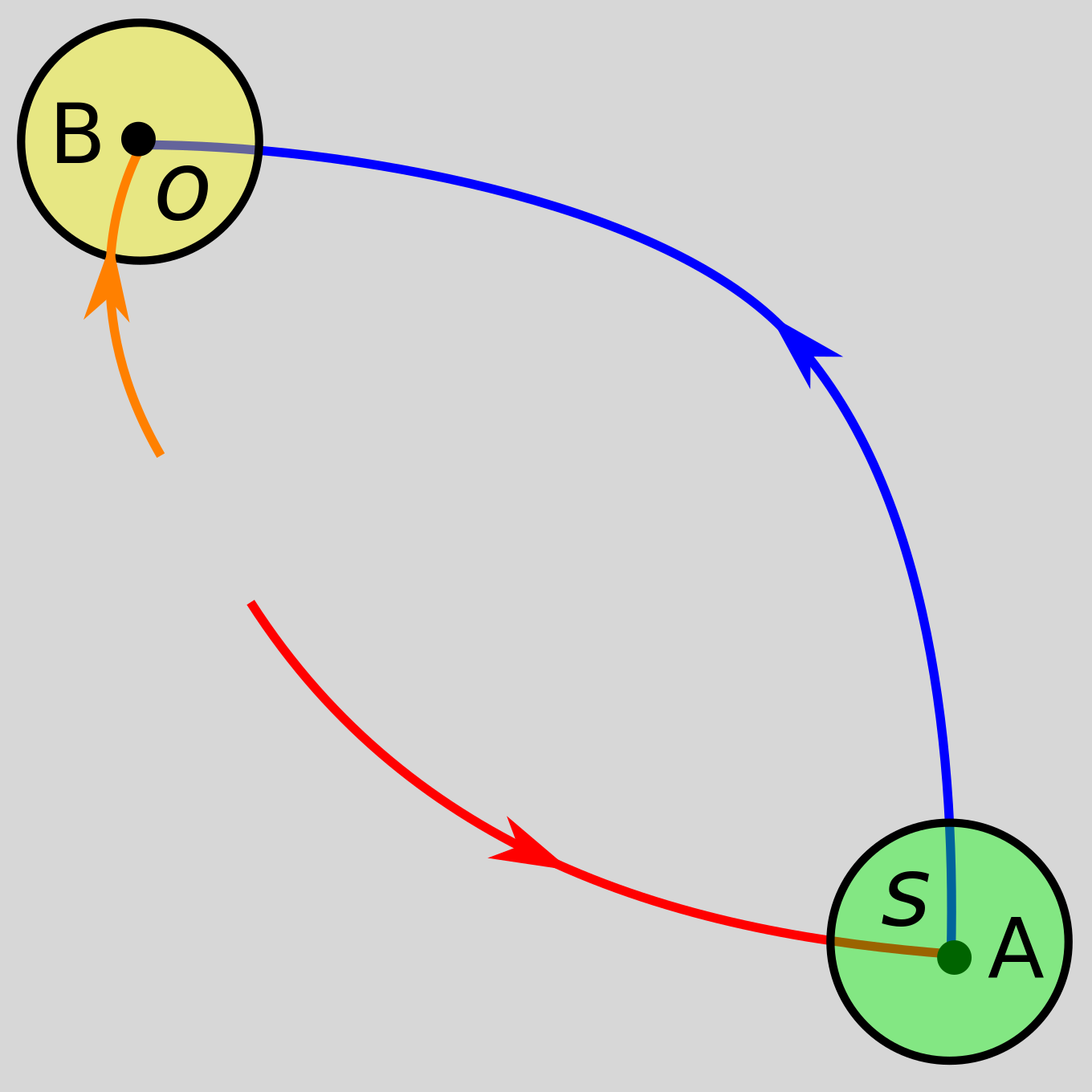} 
		\caption{\label{fig:Exchange-bulk-APtest-3} End $t=T$}
	\end{subfigure}
	\caption{\label{fig:Exchange-bulk-APtest} The antiparticle test.
	}
\end{figure}

The basic set up is very similar to the original challenge, except that:\\
(1). This test requires three players to participate: Alice, Bob, and Carol, who work as a team to win. Throughout the game, classical communication between Alice and Carol is allowed;\\
(2). Before the game begins,  only Bob receives a number $b\in\{1,2,\ldots,m_0\}$ from the Referees, and after the game ends, Alice is asked to report Bob's number, and they win if she answers correctly;\\
(3). At $t=0$, the game begins in the configuration shown in Fig.~\ref{fig:Exchange-bulk-APtest-1}, where the circles A and C have an overlap. Then the three circles start moving along their respective paths shown in the figure; \\
(4). After the circles A and C are fully separated, an identical particle test between circles A and B is performed; \\
(5). At some point $t=t_1$, the circle C arrives at the special point $\oA$, and disappears shortly after, while the circles A and B keeps moving towards their respective destination;\\
(6). At $t=T$, both circles A and B arrive at their respective destination and disappear, and the game ends.\\
The winning strategy to this test will only be presented in Sec.~\ref{sec:win-APtest}, and in Sec.~\ref{sec:win_APtest} we explain why we introduce this test. This test does not affect the hierarchy in Tab.~\ref{tab:protocols_strategies}. %

\section{Winning strategies}\label{sec:win}
In this section we describe winning strategies for the challenge game introduced in Sec.~\ref{sec:parachallenge_variants_twists}, using emergent $R$-paraparticles. %
Specifically, in Sec.~\ref{sec:axioms_emergent_para} we first define emergent $R$-paraparticles through a number of axioms that describe their universal properties. Then in Sec.~\ref{sec:win_original} we present the winning strategy and show that any nontrivial emergent paraparticle satisfying these axioms can win the game. 
In Sec.~\ref{sec:partial_win} we present some partial winning strategies that utilize relatively trivial types of exchange statistics, %
and explain how these strategies fail some of the winning conditions of the game. 

\subsection{Axioms of emergent $R$-parastatistics}\label{sec:axioms_emergent_para}%
We begin by giving a precise definition of emergent $R$-paraparticles in a gapped phase, using a number of axioms that summarize their universal topological properties. 
We call these the axioms of emergent $R$-parastatistics, which turn out to be extremely convenient for describing the winning strategies, as they allow us to forget about the complicated and irrelevant microscopic details and focus on universal properties. 
These axioms were shown~\cite{wang2024parastatistics} to be satisfied by the emergent paraparticles in the 2D exactly solvable model proposed in Ref.~\onlinecite{wang2023para}. More generally, in either 2D or 3D, one can derive these axioms directly from the second quantization formulation of $R$-paraparticles proposed in Ref.~\onlinecite{wang2023para} along with a few extra assumptions about the underlying topological phase that host the paraparticles. Later in Sec.~\ref{sec:FCanalysis_summary}, we will also derive these axioms from the SFC description of paraparticles.

Let $\ket{G}$ be the (unique, gapped) ground state of a 2D or 3D quantum many body system described by a locally interacting Hamiltonian $\hat{H}$. 
Let $\psi$ be a type of point-like quasiparticle above $\ket{G}$. We call $\psi$ an $R$-paraparticle if it satisfies the following axioms:\\ 
\textbf{Axiom \customlabel{Axiom1}{1}. The state space.} We denote an excited state of $\hat{H}$ with $n$ quasiparticles of type $\psi$ as $\ket{G;i_1^{a_1} i_2^{a_2} \ldots i_n^{a_n}}$, 
where $i_1,i_2,\ldots, i_n$ label the particle positions that are assumed to be mutually different~\footnote{
	Since these axioms aim to capture only the universal properties of paraparticles at long distance, they do not address
	what happens when two particles get close, which depends on the non-universal microscopic details of $\hat{H}$. 
}. %
The  mutually independent numbers $a_1,\ldots,a_n\in \{1,2,\ldots,m\}$ label the \textit{internal states} of the paraparticles, and $m\geq 2$ is an integer called the quantum dimension of $\psi$~(the trivial $m=1$ case correspond to ordinary fermions and bosons). 
We use $\Hil_{i_1i_2 \ldots i_n}$ to denote the $m^n$-dimensional subspace of excited states spanned by $\ket{G;i_1^{a_1} i_2^{a_2} \ldots i_n^{a_n}}$, for a given $i_1,\ldots, i_n$.
\\
\textbf{Axiom \customlabel{Axiom2}{2}. Topological degeneracy.} The internal states $a_1,\ldots,a_n$ of the paraparticles cannot be locally accessed or altered using any local unitary operations or measurements when every paraparticle is deep in the bulk, far away from other particles and defects. Formally, this requires that any local operator $\hat{Q}$ satisfies %
\begin{equation}\label{eq:TPdegenerate}
	\!\braket{G;i_1^{b_1} \ldots i_n^{b_n}|\hat{Q}|G;i_1^{a_1} \ldots i_n^{a_n}}=C^{Q}_{i_1\ldots i_n}\!\prod_{j=1}^n \delta_{a_j b_j}+O(e^{-\frac{l}{\xi}}), %
\end{equation}
where $C^O_{i_1\ldots i_n}$ is a constant independent of $a_1,\ldots,a_n$, $\xi$ is the correlation length of $\ket{G}$, and $l$ is the minimal distance between all particles and defects of the system.\\ %
\textbf{Axiom \customlabel{Axiom3}{3}. Particle movements.} %
For any  $k\in \{1,2,\ldots,n\}$, the particle at position $i_k$ can be moved to another position $j_k$ using a unitary operator $\hat{U}_{i_k j_k}$ supported on a path connecting $i_k$ and $j_k$:
\begin{equation}\label{eq:paraparticlemove}
	\hat{U}_{i_k j_k}\ket{G;i_1^{a_1}\ldots i_k^{a_k}\ldots i_n^{a_n}}=\ket{G;i_1^{a_1}\ldots j_k^{a_k} \ldots i_n^{a_n}}.
\end{equation}
 \textbf{Axiom \customlabel{Axiom4}{4}. Exchange statistics.} Consider two sets of position labels $i_1,\ldots,i_n$ and $j_1,\ldots,j_n$ related by a permutation $\tau\in S_n$, i.e., $i_k=j_{\tau(k)}$, for $1\leq k\leq n$. Then we must have $\Hil_{i_1 i_2 \ldots i_n}=\Hil_{j_1 j_2 \ldots j_n}$, since both describe the same subspace of excited states with $n$ identical paraparticles at positions  $\{i_1,\ldots,i_n\}=\{j_1,\ldots,j_n\}$. 
Therefore the two different basis $\{\ket{G;i_1^{a_1}  \ldots i_n^{a_n}}\}$ and $\{\ket{G;j_1^{a_1}  \ldots j_n^{a_n}}\}$ must be related by a unitary transformation. When 
$\tau=(k,k+1)\in S_n$ 
is a transposition that swaps $k$ and $k+1$ for some $k\in\{1,2,\ldots,n-1\}$, this basis transformation is given by 
\begin{equation}\label{eq:exchangestatR-npt}
	\ket{G;\ldots i_k^{a}i_{k+1}^{b}\ldots }=\sum_{a',b'}R^{b'a'}_{ab}\ket{G;\ldots i_{k+1}^{b'}i_k^{a'} \ldots},
\end{equation}
where $\ldots$ collectively denotes other labels that are unaffected by the exchange. %
The basis transformation between $\Hil_{i_1 i_2 \ldots i_n}$ and $\Hil_{j_1 j_2 \ldots j_n}$ for a general permutation $\tau\in S_n$ can be obtained by composing Eq.~\eqref{eq:exchangestatR-npt} for different values of $k$, since the symmetric group $S_n$ is generated by neighboring swaps of the form $(k,k+1)\in S_n$. %
For the consistency of Eq.~\eqref{eq:exchangestatR-npt}, the four-index 
tensor $R^{b'a'}_{ab}=\!\!\begin{tikzpicture}[baseline={([yshift=-.6ex]current bounding box.center)}, scale=0.45]
	\Rmatrix{0}{0}{R}
	\node  at (-1.5*\AL,\AL) {\footnotesize $b' $};
	\node  at (1.7*\AL,\AL) {\footnotesize $a'$};
	\node  at (-1.5*\AL,-\AL) {\footnotesize $a$};
	\node  at (1.5*\AL,-\AL) {\footnotesize $b$};
\end{tikzpicture}$ is required to satisfy the Yang-Baxter equation~\cite{Turaev1988,Majid1990,etingof1999set}
\begin{alignat}{3}\label{eq:YBE}
	\begin{tikzpicture}[baseline={([yshift=-.8ex]current bounding box.center)}, scale=0.5]
		\Rmatrix{0}{\AL}{R}
		\Rmatrix{0}{-\AL}{R}
		\node  at (-\AL,2.5*\AL) {\footnotesize $a$};
		\node  at (\AL,2.5*\AL) {\footnotesize $b$};
		\node  at (-\AL,-2.5*\AL) {\footnotesize $c$};
		\node  at (\AL,-2.5*\AL) {\footnotesize $d$};
	\end{tikzpicture}&=
	\begin{tikzpicture}[baseline={([yshift=-.8ex]current bounding box.center)}, scale=0.5]
		\draw[thick] (-\AL,-2*\AL) -- (-\AL,2*\AL);
		\draw[thick] (\AL,-2*\AL) -- (\AL,2*\AL);
		\node  at (-\AL,2.5*\AL) {\footnotesize $a$};
		\node  at (\AL,2.5*\AL) {\footnotesize $b$};
		\node  at (-\AL,-2.5*\AL) {\footnotesize $c$};
		\node  at (\AL,-2.5*\AL) {\footnotesize $d$};
		\node  at (-1.5*\AL,0*\AL) {\footnotesize $\delta$};
		\node  at (1.5*\AL,0*\AL) {\footnotesize $\delta$};
	\end{tikzpicture},&~~~
	\begin{tikzpicture}[baseline={([yshift=-.8ex]current bounding box.center)}, scale=0.5]
		\Rmatrix{-\AL}{2*\AL}{R}
		\Rmatrix{\AL}{0}{R}
		\Rmatrix{-\AL}{-2*\AL}{R}
		\draw[thick] (-2*\AL,-\AL) -- (-2*\AL,\AL);
		\draw[thick] (2*\AL,\AL) -- (2*\AL,3*\AL);
		\draw[thick] (2*\AL,-\AL) -- (2*\AL,-3*\AL);
		\node  at (-2*\AL,3.5*\AL) {\footnotesize $a$};
		\node  at (0*\AL,3.5*\AL) {\footnotesize $b$};
		\node  at (2*\AL,3.5*\AL) {\footnotesize $c$};
		\node  at (-2*\AL,-3.7*\AL) {\footnotesize $d$};
		\node  at (0*\AL,-3.7*\AL) {\footnotesize $e$};
		\node  at (2*\AL,-3.7*\AL) {\footnotesize $f$};
	\end{tikzpicture}
	&=\begin{tikzpicture}[baseline={([yshift=-.8ex]current bounding box.center)}, scale=0.5]
		\Rmatrix{\AL}{2*\AL}{R}
		\Rmatrix{-\AL}{0}{R}
		\Rmatrix{\AL}{-2*\AL}{R}
		\draw[thick] (2*\AL,-\AL) -- (2*\AL,\AL);
		\draw[thick] (-2*\AL,\AL) -- (-2*\AL,3*\AL);
		\draw[thick] (-2*\AL,-\AL) -- (-2*\AL,-3*\AL);
		\node  at (-2*\AL,3.5*\AL) {\footnotesize $a$};
		\node  at (0*\AL,3.5*\AL) {\footnotesize $b$};
		\node  at (2*\AL,3.5*\AL) {\footnotesize $c$};
		\node  at (-2*\AL,-3.7*\AL) {\footnotesize $d$};
		\node  at (0*\AL,-3.7*\AL) {\footnotesize $e$};
		\node  at (2*\AL,-3.7*\AL) {\footnotesize $f$};
	\end{tikzpicture},\\
	R^2&=\mathds{1},&\quad R_{12}R_{23}R_{12}&=R_{23}R_{12}R_{23},\nonumber
\end{alignat}
which is equivalent to the requirement that $R$ generates a representation of the symmetric group $S_n$ via Eq.~\eqref{eq:exchangestatR-npt}~\cite{kassel2008braid}. In this paper, we also require $R$ to be unitary as a matrix, %
so that Eq.~\eqref{eq:exchangestatR-npt} defines a unitary basis transformation. 
A solution to Eq.~\eqref{eq:YBE} is called an $R$-matrix, and each $R$-matrix defines a type of parastatistics~\cite{wang2023para}.  \\
\textbf{Axiom \customlabel{Axiom5}{5}. Creation and annihilation of paraparticles.}
A single paraparticle $\psi$ can be locally created and annihilated at one of the two special points $\oA$ and $\oB$ in the system.
More precisely, there exists unitary operators $\hat{U}_{\oA,a}$, $\hat{U}'_{\oB,a}$ localized around $\oA,\oB$, respectively, satisfying
\begin{eqnarray}\label{eq:localcreationatcorner}
	\hat{U}_{\oA,a}\ket{G;i_1^{a_1} \ldots i_n^{a_n}}&=&\ket{G;\oA^a i_1^{a_1} \ldots i_n^{a_n}},\nonumber\\
	\hat{U}'_{\oB,a}\ket{G;i_1^{a_1} \ldots i_n^{a_n}}&=&\ket{G;i_1^{a_1} \ldots i_n^{a_n}\oB^a}.
\end{eqnarray}
Since unitary processes are reversible, one can also annihilate paraparticles at $\oA,\oB$ using $\hat{U}^\dagger_{\oA,a}$, $\hat{U}^{\prime\dagger}_{\oB,a}$, respectively.
In general, to make Eq.~\eqref{eq:localcreationatcorner} possible, there must be some special kinds of topological defects located at $\oA$ and $\oB$, which we describe later in Sec.~\ref{sec:ModCatDefect} in the categorical framework. \\
\textbf{Axiom \customlabel{Axiom6}{6}. Measurement of the internal state.} %
The internal state of a paraparticle can be locally measured at $\oA$ or $\oB$:  
there exist observables $\hat{O}_{\oA}$, $\hat{O}'_{\oB}$ localized around $\oA,\oB$, respectively, satisfying
\begin{eqnarray}\label{eq:localmeasurementatcorner}
	\hat{O}_{\oA}\ket{G;i_1^{a_1} \ldots i_n^{a_n}}&=&a_1 \ket{G; i_1^{a_1} \ldots i_n^{a_n}},\text{ if } i_1=\oA,\nonumber\\
	\hat{O}'_{\oB}\ket{G;i_1^{a_1} \ldots i_n^{a_n}}&=&a_n \ket{G;i_1^{a_1} \ldots i_n^{a_n}},\text{ if } i_n=\oB.
\end{eqnarray}
	A caveat is that if we instead have $i_k=\oA$ for some $k>1$~(or $i_k=\oB$ for some $k<n$), then we need to first do a basis transformation in Eq.~\eqref{eq:exchangestatR-npt} %
	to move $i_k$ all the way to the front~(back) before applying Eq.~\eqref{eq:localmeasurementatcorner} to compute %
	$\braket{\hat{O}_{\oA}}$~(or $\braket{\hat{O}'_{\oB}}$). We will see an example soon in Eq.~\eqref{eq:Psi3}, and later in Eq.~\eqref{eq:Psi3-WHF}.  %

We now make a few remarks on these axioms. In principle, %
we also need an assumption that the paraparticles are point-like excitations up to exponentially decaying tails, that is, %
for any local observable $\hat{Q}_j$ at position $j\notin\{i_1,i_2,\ldots,i_n\}$, we have
\begin{equation}\label{eq:particlelocalization}
	\!\braket{G;i_1^{a_1} \ldots i_n^{a_n}|\hat{Q}_j|G;i_1^{a_1} \ldots i_n^{a_n}}=\braket{G|\hat{Q}_j|G}+O(e^{-\frac{l}{\xi}}), 
\end{equation}
where $l$ is the minimal distance between $j$ and $\{i_1,\ldots,i_n\}$. 
This localization property of particle states is crucial for the winning strategy, as the players need to confine the paraparticles within the circle areas, such that having a paraparticle in a circle does not change the value of $\braket{\hat{h}_j}$ outside. 
However, %
we expect that this  property follows from the existence of a spectral gap~\cite{Buchholz1982,hastings2006,nachtergaele2006,haegeman2013elementary}. Furthermore, it often happens in frustration-free Hamiltonians~(e.g., the frustration-free region of the 2D solvable model constructed in Ref.~\onlinecite{wang2023para}) that the exponentially decaying tail in Eq.~\eqref{eq:particlelocalization} vanishes exactly. %

In formulating the axioms above, we have implicitly chosen a basis for the internal space of the paraparticles. If we choose a new basis for the internal space, the $R$-matrix transforms according to 
\begin{equation}\label{eq:basistransformRmat}
R\to R'=(V^\dagger\otimes V^\dagger)R(V\otimes V),
\end{equation}
where $V$ is the unitary transformation between the old and the new basis. Two $R$-matrices that are related by a unitary basis transformation of the form in Eq.~\eqref{eq:basistransformRmat} are called equivalent, denoted by $R'\cong R$, and they describe the same type of parastatistics. [Indeed, it is possible to formulate the axioms in a basis-independent form.] We also emphasize that the local operators $\hat{U}_{\oA,a},\hat{U}'_{\oB,a},\hat{O}_{\oA}$, and $\hat{O}'_{\oB}$ in Axioms 5 and 6 exist for any choice of basis. For example, for any unit vector $\boldsymbol{\xi}$ with components $\{\xi_a\}_{a=1}^m$, there exists a local unitary operator $\hat{U}_{\oA,\boldsymbol{\xi}}$ localized around $\oA$ satisfying
\begin{eqnarray}\label{eq:localcreationatcorner-xi}
	\hat{U}_{\oA,\boldsymbol{\xi}}\ket{G;i_1^{a_1} \ldots i_n^{a_n}}&=&\ket{G;\oA^{\boldsymbol{\xi}} i_1^{a_1} \ldots i_n^{a_n}}\nonumber\\
	&\equiv&\sum_{a=1}^m \xi_a\ket{G;\oA^a i_1^{a_1} \ldots i_n^{a_n}},
\end{eqnarray}
and similarly for  $\hat{U}'_{\oB,\boldsymbol{\xi}}$ at the other special point $\oB$.

A particularly simple example of $R$-matrix that can pass the challenge has quantum dimension $m=4$, with nonzero elements given by
\begin{equation}\label{eqApp:seth-R}
	R^{b'a'}_{ab}=-1, \text{ if }	(b',a')=
	\left(
	\begin{array}{cccc}
		43 & 12 & 24 & 31 \\
		21 & 34 & 42 & 13 \\
		14 & 41 & 33 & 22 \\
		32 & 23 & 11 & 44 \\
	\end{array}
	\right)_{ab},%
\end{equation}
where $43$ is a shorthand for $(4,3)$ and similarly for others. This $R$-matrix will be our primary example for understanding the winning strategies, as it can pass all the challenges in this paper in a simple and perfect way. A few other examples are given in Tab.~\ref{tab:protocols_strategies} and also in App.~\ref{sec:RfromCentralType}.
\subsection{Strategy for the basic challenge}\label{sec:win_original}
Consider a 2D or 3D quantum system described by a Hamiltonian $\hat{H}$ with ground state $\ket{G}$ that satisfy the axioms of emergent parastatistics in Sec.~\ref{sec:axioms_emergent_para}. In the following we show that as long as the $R$-matrix is not of the trivial product form $R^{b'a'}_{ab}=p_{aa'}q_{bb'}$, the system can be used to win the game, and we detail the strategy below. 

\subsubsection{The original version with two special points}\label{sec:win_2pt}
\begin{figure}
	\begin{subfigure}[t]{.32\linewidth}
		\centering\includegraphics[width=.85\linewidth]{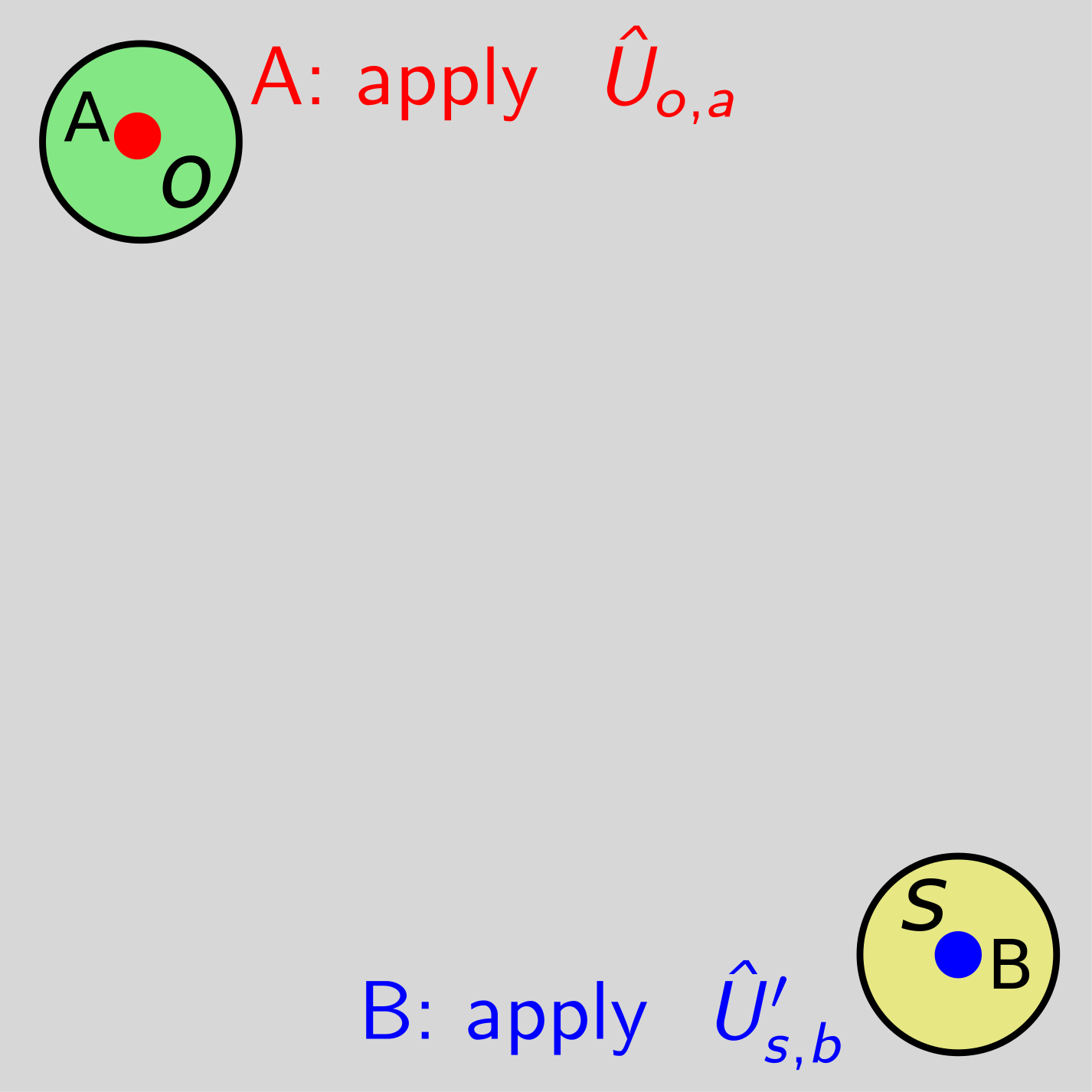} 
		\caption{\label{fig:WSstart} Start $t=0$}
	\end{subfigure}
	\begin{subfigure}[t]{.32\linewidth}
		\centering\includegraphics[width=.85\linewidth]{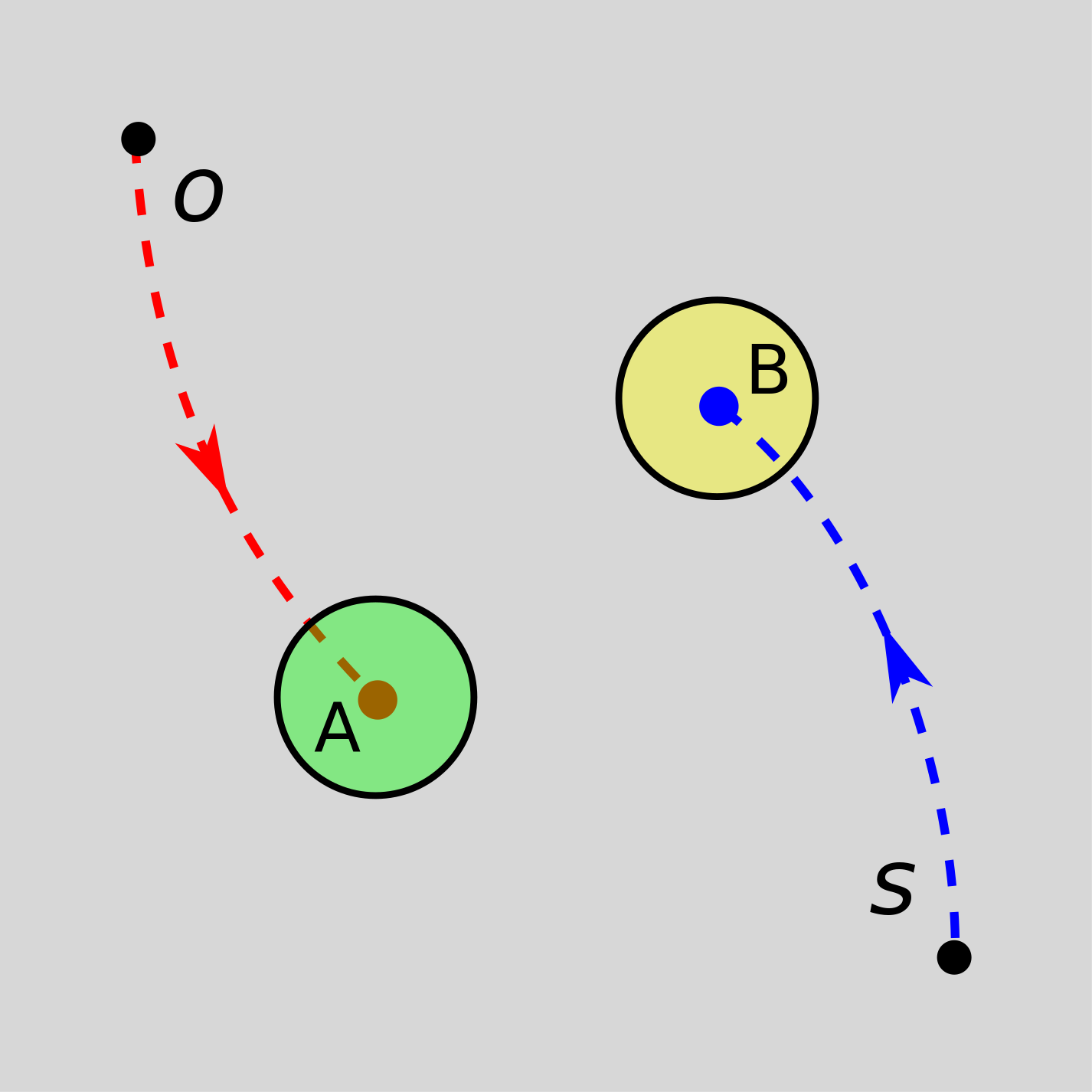} 
		\caption{\label{fig:WSduringgame} $0<t<T$}
	\end{subfigure}
	\begin{subfigure}[t]{.32\linewidth}
		\centering\includegraphics[width=.85\linewidth]{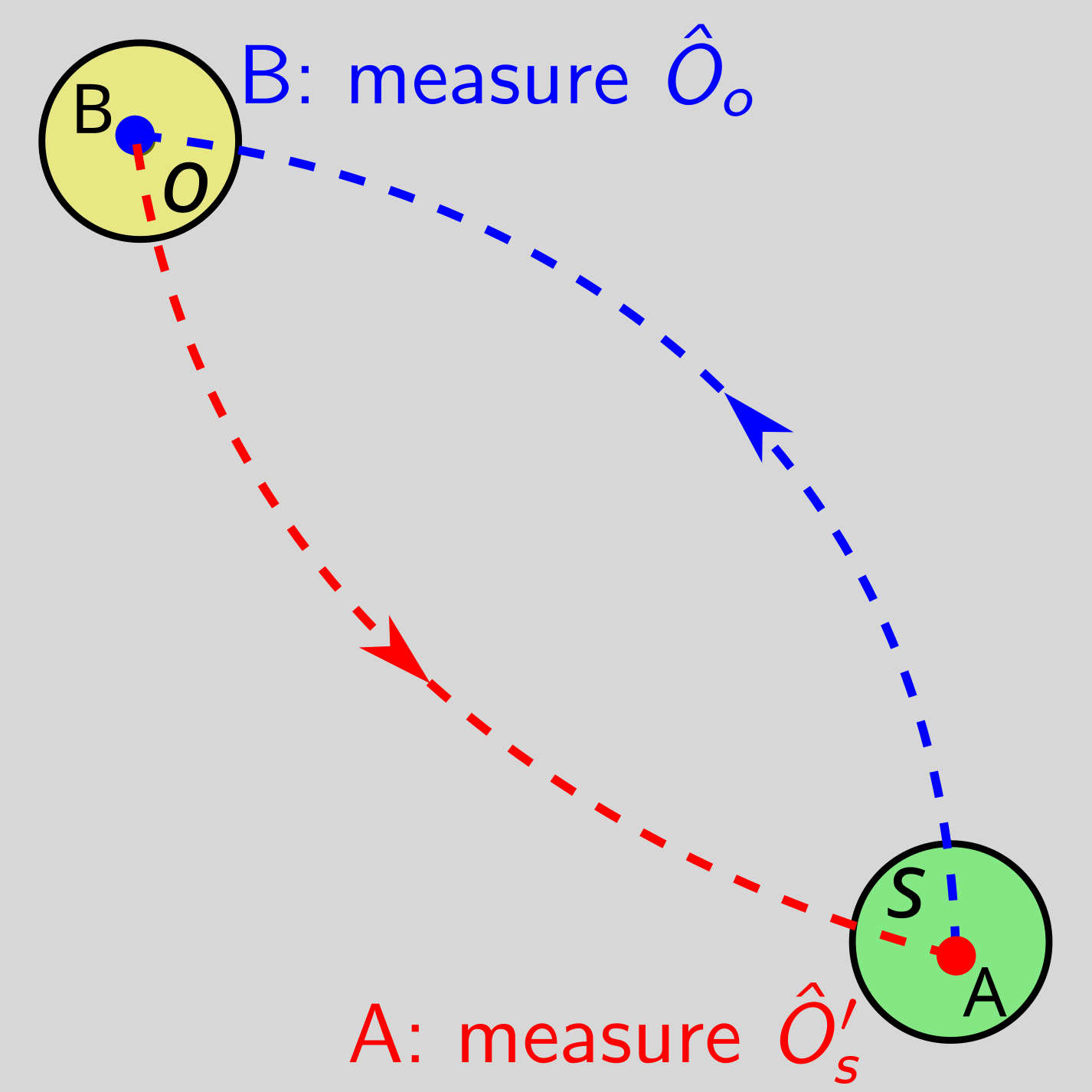} 
		\caption{\label{fig:WSgameend} End $t=T$}
	\end{subfigure}
	\caption{\label{fig:winstrategy} (Figure adapted from Ref.~\cite{wang2024parastatistics}) Sketch of the winning strategy using emergent paraparticles. $\oA$ and $\oB$ are chosen as the two special points in Axioms \ref{Axiom5} and \ref{Axiom6}, %
	where a paraparticle can be locally created and measured. Dashed curves indicate the trajectories traversed by the two circles. %
	}
\end{figure}
We begin with the original version with two special points. 
During the pregame preparation stage, the players submit the Hamiltonian $\hat{H}$, and choose $\oA,\oB$ to be the special points in  Axioms \ref{Axiom5} and \ref{Axiom6}. Then they prepare the ground state $\ket{G}$, and allow the Referees to verify. Once the game begins and the numbers $a,b$ are assigned, the players proceed as follows (see Fig.~\ref{fig:winstrategy}):\\
(a) At $t=0$, Alice applies $\hat{U}_{\oA,a}$ to create a paraparticle with internal state $a$ at $\oA$~(i.e. encodes her number $a$ into the internal state of the particle), and likewise Bob  applies $\hat{U}'_{\oB,b}$ at $\oB$. The global quantum state becomes
$\ket{\Psi(0)}=\ket{G;\oA^a \oB^b}$;\\
(b) For $0<t<T$, Alice and Bob transport their paraparticles along the indicated paths using $\hat{U}_{ij}$ in Eq.~\eqref{eq:paraparticlemove}, strictly following the motion of the circles;\\
(c) When the exchange is complete, the state evolves to
\begin{eqnarray}\label{eq:Psi3}
	\ket{\Psi(T)}=\ket{G;\oB^a \oA^b}=\sum_{a',b'}R^{b'a'}_{ab}\ket{G;\oA^{b'} \oB^{a'}},
\end{eqnarray}
where we use Eq.~\eqref{eq:exchangestatR-npt}. Then each player performs a suitable measurement on the internal state of his or her particle. For example, Alice can 
measure $\hat{O}'_{\oB}$ in Eq.~\eqref{eq:localmeasurementatcorner} to obtain $a'$, and Bob can measure $\hat{O}_{\oA}$ to obtain $b'$, collapsing the state to $\ket{G;\oA^{b'} \oB^{a'}}$. Finally, they annihilate their paraparticles with $\hat{U}^{\prime\dagger}_{\oB,a'}$ and $\hat{U}^{\dagger}_{\oA,b'}$, ensuring that no excitations remain after the game ends.

We now show that as long as $R$ is  nontrivial, the measurement result $a'$ contains information about $b$, and likewise $b'$ contains information about $a$. We begin by considering some concrete examples. First, consider the $R$-matrix in Eq.~\eqref{eqApp:seth-R}. In this case,  the outcomes $a',b'$ are definite, and moreover, knowledge of $b$ and $b'$~(along with the $R$-matrix) uniquely determines $a$, and likewise, knowing $a$ and $a'$ uniquely determines $b$~\footnote{The $R$-matrix here is a perfect tensor: grouping any two indices as inputs and the remaining two as outputs yields a unitary map—an invertible classical gate in our setting. Thus each player can recover the partner’s number by inverting this map.}. For instance, if Alice observes $(a,a')=(2,3)$, she scans the second row for the column where $a'=3$; this is the fourth column, so she infers $b=4$ and $b'=1$. Consequently, with this $R$-matrix, the players can transfer $4$ bits of information in one round, allowing them to win the game with $m_0=4$ with a 100\% success rate. 

For more general $R$-matrices, the players need to carefully decide how they encode and decode information stored in the internal states, including in which basis they apply $\hat{U}_{\oA,a},\hat{U}'_{\oB,a}$ in step (a) and measure $\hat{O}_{\oA}, \hat{O}'_{\oB}$ in step (c)~[see the remark around Eq.~\eqref{eq:localcreationatcorner-xi}].  
For example, consider the $R$-matrix $R^{b'a'}_{ab}=\delta_{aa'}\delta_{bb'}(-1)^{\delta_{ab}}$ with $m=4$. With this $R$-matrix, Bob can transfer $2$ bits of information to Alice in one round. To achieve this, in step (a) Alice always begins with the internal state $\ket{+}=(\ket{1}+\ket{2}+\ket{3}+\ket{4})/2$ by applying $\hat{U}_{\oA,\boldsymbol{\xi}}$ in Eq.~\eqref{eq:localcreationatcorner-xi}, while Bob encodes the 2 bits of information in the standard basis. After the exchange, the internal state evolution is 
\begin{equation}\label{eq:unidirectionalcontrolgate}
	\ket{+}_A\otimes \ket{b}_B\xrightarrow{R}  \ket{b}_B\otimes \ket{\varphi_b}_A,
\end{equation} 
where 
\begin{equation}
	\ket{\varphi_b}=\frac{1}{2}\sum_{b=1}^4 (-1)^{\delta_{ab}}\ket{a}.
\end{equation}
Since $\{\ket{\varphi_b}\}_{b=1}^4$ is an orthonormal basis, Alice can get the value of $b$ by measuring in this basis. 

	We now proceed to the most general case. We prove the contrapositive of our claim: if for a certain $R$-matrix, the above protocol cannot transfer any information between Alice and Bob no matter which input states they use, then $R$ must be of the trivial product form $R^{b'a'}_{ab}=p_{aa'}q_{bb'}$. First, consider any input basis, by requiring that each player's output state does not depend on his partner's input state, $R$ can only be of the form $R^{b'a'}_{ab}=p_{aa'}q_{bb'}\theta_{ab}$. Then, suppose Alice uses the input state $\ket{+}=\sum_a\ket{a}$. Requiring that Alice's output state is independent of $b$, $\theta_{ab}$ must factorize. This proves our main claim.

As long as the players can transfer a nonzero amount of information to each other using the particle $\psi$, their chance of winning is higher than pure guessing. 
In addition, they can use the \textit{multilayer trick} to arbitrarily enhance the chance of winning, by using a physical system obtained by stacking multiple layers of the same system described by $\hat{H}$. 
To be specific, consider the two layer case as an example. 
At the game preparation stage, they can submit the Hamiltonian $\hat{H}_2=\hat{H}\otimes \mathds{1}+\mathds{1}\otimes\hat{H}$, with a unique, gapped ground state $\ket{G}\otimes \ket{G}$. During the game, each player uses the particle $\Psi=\psi\boxtimes \psi$ in each circle area~(here we use $\boxtimes$ instead of $\otimes $ to distinguish from particle fusion), which simply means having a $\psi$ particle in each layer. 
Note that %
$\Psi$ also satisfies the axioms of emergent parastatistics with the $R$-matrix $R^{AB}_{CD}=R^{a_1b_1}_{c_1d_1}R^{a_2b_2}_{c_2d_2}$, where $A=(a_1,a_2)$ is a collective label for the internal states of the two $\psi$, and similarly for $B,C,D$. At the special points $\oA$ and $\oB$, the players can manipulate the internal states of $\psi$ in each layer independently, and the time evolution of different layers are also independent since there are no inter-layer interactions.
Consequently, if a strategy allows them to transfer $k$ bits of information in the single layer case, they can transfer $2k$ bits in the bilayer case %
by using the same strategy independently in each layer. 
Note that if a certain strategy only allows information transfer in a single direction~(e.g., in the example with $R^{b'a'}_{ab}=\delta_{aa'}\delta_{bb'}(-1)^{\delta_{ab}}$ given in  Eq.~\eqref{eq:unidirectionalcontrolgate} above, only Bob can transfer information to Alice), the players can use a bilayer system in which the second layer is obtained from the first layer 
by a $180\degree$ rotation that swaps $\oA$ and $\oB$~(see Fig.~\ref{fig:bilayerDS3} for an example), and in the second layer they use the same strategy but with the role of Alice and Bob swapped, which then allows bidirectional information transfer. 

The above winning strategy is robust against local noise and eavesdropping due to the topological protection of the internal space of emergent paraparticles stated in Axiom \ref{Axiom2}, Eq.~\eqref{eq:TPdegenerate}. In particular, since a paraparticle is a stable topological excitation, local noise that happen inside the circles can always be detected and corrected using local operations, similar to error correction in topological quantum codes~\cite{Dennis2002TQM}. 
The strategy can pass the identical particle test since the quasiparticles used in the two circles are identical~(locally indistinguishable), which follows from Axioms~\ref{Axiom2}-\ref{Axiom4}. 

\subsubsection{The variant with one special point}\label{sec:win_1pt}
The winning strategy for the variant of the game with one special point~(introduced in Sec.~\ref{sec:oneptversion}) is very similar. The pregame preparation is the same as before except that the players only need to choose the one special point to be the $\oA$ in Axioms \ref{Axiom5} and \ref{Axiom6}.  After the game begins, they proceed as follow:\\
(a). At $t=0$, after Bob enters the game, he applies %
$\hat{U}_{\oA,b}$ %
at $\oA$, and the state of the system becomes  $\ket{\Psi(0)}=\hat{U}_{\oA,b}\ket{G}=\ket{G;\oA^b}$;\\
(b). At $t=T/4$, Bob has moved his paraparticle to position $j$, the state of the system becomes $\ket{G;j^b}$.  %
When Alice enters the game, she applies $\hat{U}_{\oA,b}$, %
and the state of the system becomes 
\begin{equation}
	\ket{\Psi(T/4)}=\hat{U}_{\oA,a}\ket{G;j^b}=\ket{G;\oA^a j^b}.%
\end{equation}
(c). Throughout the game, %
the players use $\hat{U}_{ij}$ in Eq.~\eqref{eq:paraparticlemove} to keep the particles inside the circles;\\
(d). At $t=3T/4$, the state of the system is 
\begin{equation}\label{eq:Psi3-1pt}
	\ket{\Psi(3T/4)}=\ket{G;i^a\oA^b}=\sum_{a',b'}R^{b'a'}_{ab}\ket{G;\oA^{b'} i^{a'}},
\end{equation}
where we used Eq.~\eqref{eq:exchangestatR-npt}. 
Then Bob measures $\hat{O}_{\oA}$ to obtain $b'$, and applies $\hat{U}^{\dagger}_{\oA,b'}$ to annihilate his paraparticle.\\
(e). Finally at $t=T$, when Alice returns to $\oA$, she measures %
$\hat{O}_{\oA}$ to obtain $a'$, and applies $\hat{U}^{\dagger}_{\oA,a'}$ to annihilate her paraparticle.\\
The result of each measurement is completely the same as that obtained in the winning strategy for the two-special-point version in Sec.~\ref{sec:win_2pt}.  %
Therefore the players can use the same algorithm there to compute each other's number.

\subsection{Partial winning strategies}\label{sec:partial_win}
The purpose of this section is to demonstrate that the requirement of robustness against noise and eavesdropping and the identical particle test are necessary to separate paraparticles from relatively trivial types of emergent quasiparticle statistics in 3+1D. In Sec.~\ref{sec:emergent_fermion} we present a strategy using emergent fermions that is vulnerable to noise. In Sec.~\ref{sec:phaseSWAP} we present a strategy using trivial paraparticles with $R$-matrices of the form $R^{b'a'}_{ab}=\delta_{aa'}\delta_{bb'}\theta_{a}\theta^*_{b}$, and show that it is vulnerable to eavesdropping. In Sec.~\ref{sec:relative_para} we present a strategy in which Alice  and Bob use different type of particles with nontrivial mutual parastatistics, which fails the identical particle test. 
	\subsubsection{A fragile strategy using emergent fermions}\label{sec:emergent_fermion}
	In this section we describe a partial winning strategy using topological phases that host emergent fermions. We will see that although this strategy can theoretically allow Alice and Bob to send information to each other, it is not robust against noise, and consequently %
	in realistic experimental systems that suffer from environmental noise,  the chance of success is close to random guessing when the system size is large. %
	
	This strategy applies to any topological phase~(in either 2D or 3D) with emergent fermions~\cite{kitaev2003fault,Levin2003Fermions,kitaev2006anyons}~(along with suitable point-like defects), which can be described by a trivial special case of the axioms in Sec.~\ref{sec:axioms_emergent_para} with $m=1, R=-1$. 
	A simple model hosting emergent fermions is the 2D toric code model~\cite{kitaev2003fault}. To satisfy Axioms \ref{Axiom5} and \ref{Axiom6}, we use the hybrid open boundary condition shown in the middle figure in Fig.~\ref{fig:BoundaryDefect}, i.e. the western and southern boundaries are rough, while the eastern and northern boundaries are smooth. With this boundary condition, we can take $\oA$ and $\oB$  to be the upper left and lower right corners, respectively, where emergent fermions can be locally created and measured. To ease our discussion below, it is convenient to construct the operators $\hat{U}_\oA$ and $\hat{U}'_\oB$ in Axiom~\ref{Axiom5} to satisfy  $\hat{U}^2_\oA=\hat{U}'^2_\oB=\mathds{1}$, which can always be done for emergent fermions~(e.g., in the toric code, $\hat{U}_\oA$ and $\hat{U}'_\oB$ are both product of Pauli operators~\cite{Kitaev2012gappedboundary}).  

	To describe the winning strategy, we first describe how Bob can send 1 bit of information to Alice using this model:\\
	(a). At $t=0$, Alice applies the local unitary operator $(1+i\hat{U}_\oA)/\sqrt{2}$ to create a superposition of the ground state and a one-fermion state. Bob's operation depends on the classical bit $n\in\{0,1\}$ he wants to send: if $n=0$, he does nothing, while if $n=1$, he applies  $\hat{U}'_\oB$ to create a fermion. The state of the system becomes
	$\ket{\Psi(0)}=(1+i\hat{U}_\oA)\hat{U}'^n_\oB\ket{G}/\sqrt{2}$;\\
    (b). Throughout the game, Bob and Alice keep whatever excitations in their respective circles to follow the circle movements.  \\
	(c). When the exchange is complete, the state evolves to
	\begin{eqnarray}\label{eq:Psi3-toriccode}
		\ket{\Psi(T)}=\frac{1+i\hat{U}'_\oB (-1)^n}{\sqrt{2}} \hat{U}^n_\oA\ket{G},
	\end{eqnarray}
	where we applied Eq.~\eqref{eq:exchangestatR-npt}, producing the sign factor $(-1)^n$. %
	Alice can determine $(-1)^n$ by measuring the local observable $i\hat{U}'_\oB \hat{P}_\oB$, where $\hat{P}_\oB=1-2\hat{n}_\oB$ is the local fermion parity operator at $\oB$. 
	Finally, the players clean up their circle areas using local operations as before. 
	This allows Bob to send 1 bit of information to Alice. They can then use the multilayer trick mentioned in Sec.~\ref{sec:win_2pt} to win the game. 
	
	[For readers who find the above derivation too abstract and unfamiliar, it may be helpful to use a hand-wavy second quantization derivation, where we use $\hat{\psi}_j$~($\hat{\psi}^\dagger_j$) to denote the emergent fermion annihilation~(creation) operator at site $j$, represented as a string operator connecting $\oA$ and $j$. We have   $\hat{U}_\oA=\hat{c}_{\oA}=\hat{\psi}_\oA+\hat{\psi}^\dagger_\oA$ is the Majorana operator at $\oA$, which happens to be local at $\oA$. %
	The local operator $\hat{U}'_\oB$ satisfies $\hat{U}'_\oB\ket{G}=\hat{c}_{\oB}\ket{G}$, where $\hat{c}_{\oB}=\hat{\psi}_\oB+\hat{\psi}^\dagger_\oB$ is the Majorana operator at $\oB$~(a string operator connecting $\oA$ and $\oB$). 
	Then the quantum state at the end of step (a) is $\ket{\Psi(0)}=(1+i\hat{\psi}^\dagger_{\oA})(\hat{\psi}^\dagger_{\oB})^n\ket{G}/\sqrt{2}$. 
	In step (b), Alice and Bob move their fermions using local unitaries of the form using $e^{i\Delta t (\hat{\psi}^\dagger_{i}\hat{\psi}_{j}+\mathrm{h.c.})}$.
	After the exchange  is complete, the state evolves to
		\begin{eqnarray}\label{eq:Psi3-toriccode-2ndQ}
				\ket{\Psi(T)}&=&(1+i\hat{\psi}^\dagger_{\oB})(\hat{\psi}^\dagger_{\oA})^n\ket{G}/\sqrt{2}\nonumber\\
				&=&(\hat{\psi}^\dagger_{\oA})^n[1+i(-1)^n\hat{\psi}^\dagger_{\oB}]\ket{G}/\sqrt{2}\nonumber\\
				&=&(\hat{U}_{\oA})^n[1+i(-1)^n\hat{U}^{\prime}_{\oB}]\ket{G}/\sqrt{2},
			\end{eqnarray}
	which reproduces Eq.~\eqref{eq:Psi3-toriccode} since $\hat{U}_{\oA}$ and $\hat{U}^{\prime}_{\oB}$ commute. 
	For the relation between the second quantization formulation and the Axioms, and a microscopic realization of the unitary operators $\hat{U}_\oA, \hat{U}'_\oB$, see Sec.~S2 of Ref.~\cite{wang2024parastatistics}.]

\begin{figure}
	\includegraphics[width=.45\linewidth]{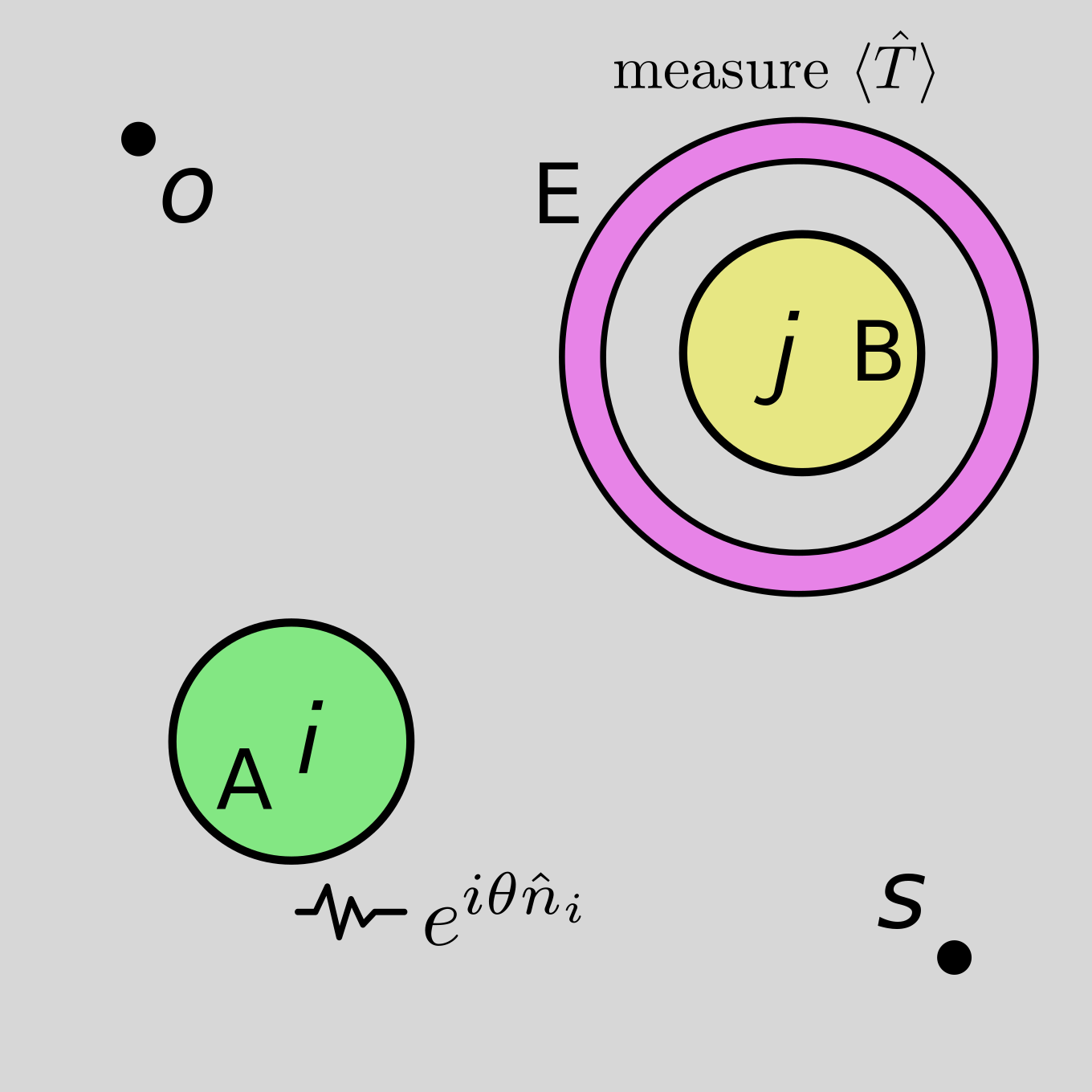} 
	\caption{\label{fig:fermionfragile} The partial winning strategy using emergent fermions is not robust against noise and eavesdropping. For example, a random phase noise $e^{i\theta \hat{n}_i}$ occurring in Alice's circle can change her measurement result. A eavesdropper can steal the information that Bob want to send simply by measuring the local fermion parity within Bob's circle.
	This can be done either by directly measuring in Bob's circle, or by measuring a suitable loop operator in the pink annulus region. %
	}
\end{figure}	
	
	This strategy is not robust against noise. 
	Suppose that at any point during the exchange process, the system suffers a phase noise of the form $e^{i\theta \hat{n}_i}$, where $\hat{n}_i$ %
	is the local fermion number operator at Alice's position, and $\theta$ is a random real number, then %
    after the exchange the sign factor $(-1)^n$ in Eq.~\eqref{eq:Psi3-toriccode} would get replaced by $e^{i\theta} (-1)^n$, which destroys the information. 
    As the system size get larger, the particles get exposed to noise for a longer period of time, and %
    the information stored in the relative phase factor cannot survive in the thermodynamic limit. 
    By contrast, for the paraparticle strategy in Sec.~\ref{sec:win_original}, the information stored in the topologically protected internal space of paraparticles is only vulnerable to local noise when the particles are close to either $\oA$ or $\oB$ due to Eq.~\eqref{eq:TPdegenerate} in Axiom~\ref{Axiom2}, therefore a large amount of information can survive in the thermodynamic limit, assuming the error rate is not too high. 

	This strategy is also vulnerable to eavesdropping: at any point during the game, an eavesdropper Eve can steal the information $n$ simply by measuring the local fermion parity within Bob's circle, as illustrated in Fig.~\ref{fig:fermionfragile}. Indeed, Eve can even steal the information remotely by measuring a suitable loop operator~(e.g., a Pauli string in the toric code) in the pink annulus region
	in Fig.~\ref{fig:fermionfragile}.

\subsubsection{A strategy using paraparticles with $R^{b'a'}_{ab}=\pm\delta_{aa'}\delta_{bb'}\theta_{a}\theta^*_{b}$}\label{sec:phaseSWAP}
	In this section we consider  paraparticles described by the simplest type of $R$-matrices of the form  
	\begin{equation}\label{eq:Rtrivial}
		R^{b'a'}_{ab}=\pm\delta_{aa'}\delta_{bb'}\theta_{a}\theta^*_{b},
	\end{equation}
	where $\theta_a$ is not a constant for $a\in\{1,2,\ldots,m\}$. We will see that in this case we have a partial winning strategy that is robust against noise but vulnerable to eavesdropping. 
	
	For simplicity we consider the case $m=2$ and $\theta_{a}=(-1)^a$ for $a=1,2$; more general cases can be treated in an identical way. We have $R=\pm X(\sigma^z\otimes \sigma^z)$, where $X$ is the swap gate, and $\sigma^z$ is the Pauli matrix.  %
	For convenience, we perform a basis change for the internal space of paraparticles after which $R$ becomes $R=\pm X(\sigma^x\otimes \sigma^x)$. Below we describe a strategy for Bob to send 1 bit of information to Alice.
	
	At $t=0$, Alice applies $\hat{U}_{\oA,1}$ to create a paraparticle with internal state $1$, while Bob applies $(\hat{U}'_{\oB,1})^n$, where $n=0,1$ is the 1 bit of information he wants to send. At $t=T$, Alice measures $\hat{O}'_{\oB}$ and then annihilate  her paraparticle, while Bob simply annihilate his paraparticle. According to the $R$-matrix, Alice will obtain $1$ if $n=0$ and will obtain $2$ if $n=1$.  Therefore Alice knows the value of $n$ from the result of her measurement, and by using the multilayer trick as before, the players can send more information and in both directions. 
	
	Compared to the strategy using emergent fermions in Sec.~\ref{sec:emergent_fermion}, this one does not involve superposition of different particle types, and is therefore robust against noise. However, it is still vulnerable to the same type of eavesdropping as illustrated in Fig.~\ref{fig:fermionfragile}, so this one is still only a partial winning strategy.  

\subsubsection{A strategy using mutual parastatistics}\label{sec:relative_para}
In this section we describe a partial winning strategy using mutual parastatistics. Example of a topological phase with nontrivial mutual parastatistics is given in App.~\ref{app:G64}. In this system, there are different types of topological quasiparticles, described by a SFC $\calC$~(see Sec.~\ref{sec:categorical_framework}). In particular, there are two types of $R$-paraparticles $\psi_\pm$,  both have quantum dimension $m=4$, and their $R$-matrices are given in   %
Eqs.~(\ref{eqApp:seth-R},\ref{eq:seth-R-G64-p},\ref{eqApp:seth-R-G64mutual}). 
In such a system, Axioms~\ref{Axiom1}-\ref{Axiom6} in Sec.~\ref{sec:axioms_emergent_para} generalize in a straightforward way, the main new ingredient is that Eq.~\eqref{eq:exchangestatR-npt} gets modified to
\begin{equation}\label{eq:exchangestatR-mutual}
	\ket{G;\ldots i_k^{a}i_{k+1}^{b}\ldots }=\sum_{a',b'}[R^{(\psi\varphi)}]^{b'a'}_{ab}\ket{G;\ldots i_{k+1}^{b'}i_k^{a'} \ldots},
\end{equation}
if the particle at $i_k$ has type $\psi$ and the particle at $i_{k+1}$ has type $\varphi$, for any $\psi,\varphi\in \calC$. Here $R^{(\psi\varphi)}$ is the mutual $R$-matrix between $\psi$ and $\varphi$, and in the special case $\psi=\varphi$, the mutual $R$-matrix  $R^{(\psi\psi)}$ gives back the self $R$-matrix of $\psi$. 
In this case,  if Alice uses $\psi_+$ and Bob uses $\psi_-$ in the winning strategy in Sec.~\ref{sec:win_2pt}, they can also win the game with $m_0=4$ by exploiting the nontrivial mutual $R$-matrix. 
This strategy is robust against noise and eavesdropping,  but it fails the identical particle test in Sec.~\ref{sec:IPT}. Furthermore, the who-entered first challenge can also be won by exploiting this type of mutual parastatistics, if Alice and Carol use $\psi_+$ while Bob and David use $\psi_-$. Therefore, in this example the mutual parastatistics between $\psi_+$ and $\psi_-$ is as capable as the self parastatistics of either of them except that using mutual parastatistics fails the identical particle test. In principle, there can also exist topological phases where the mutual parastatistics between two particles is more capable than the self parastatistics of any particle, but we do not discuss this kind of examples in this paper.

\section{The who-entered-first challenge}\label{sec:WHF}
The who-entered-first challenge is designed to separate %
$R$-matrices of the swap-type $R^{b'a'}_{ab}=\delta_{aa'}\delta_{bb'}\theta_{ab}$ from %
more non-trivial $R$-matrices such as that defined in %
Eq.~\eqref{eqApp:seth-R}, i.e.,  %
to separate levels 4 and 5 of the hierarchy in Tab.~\ref{tab:protocols_strategies}. Paraparticles that can pass this challenge will be called ``full-fledged paraparticles''.

In the following we first present the game protocol in Sec.~\ref{sec:WEFprotocol}, and then in Sec.~\ref{sec:winWEF} we give a winning strategy using emergent paraparticles defined by the $R$-matrix in Eq.~\eqref{eqApp:seth-R}, and finally in Sec.~\ref{sec:swapRfailWEF} we explain why 
level 4 paraparticles with %
$R^{b'a'}_{ab}=\delta_{aa'}\delta_{bb'}\theta_{ab}$ cannot win this challenge. 
\subsection{Game protocol}\label{sec:WEFprotocol}
\begin{figure}
	\begin{subfigure}[t]{.32\linewidth}
		\centering\includegraphics[width=.95\linewidth]{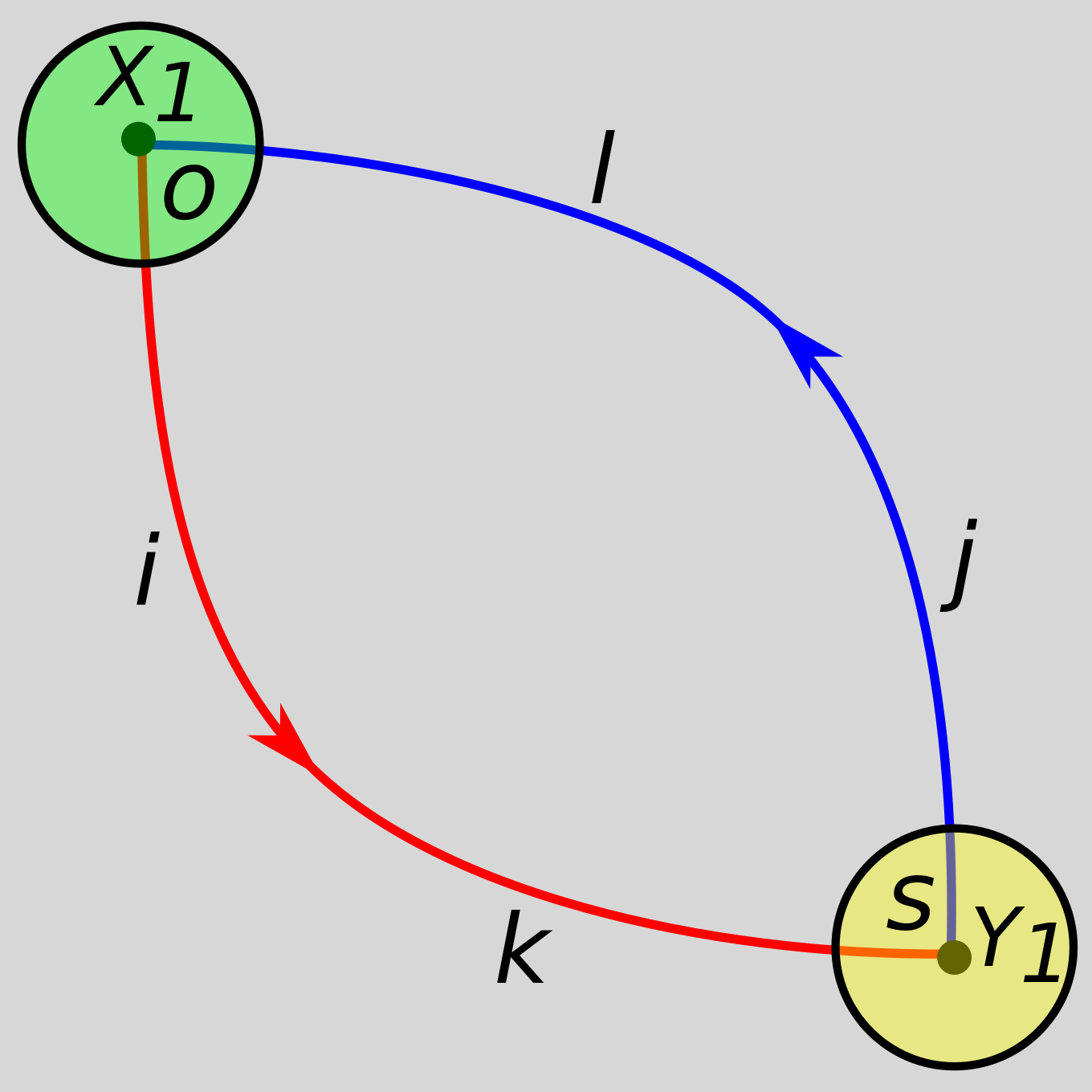} 
		\caption{\label{fig:WHF-0} Start $t=t_0$}
	\end{subfigure}
	\begin{subfigure}[t]{.32\linewidth}
		\centering\includegraphics[width=.95\linewidth]{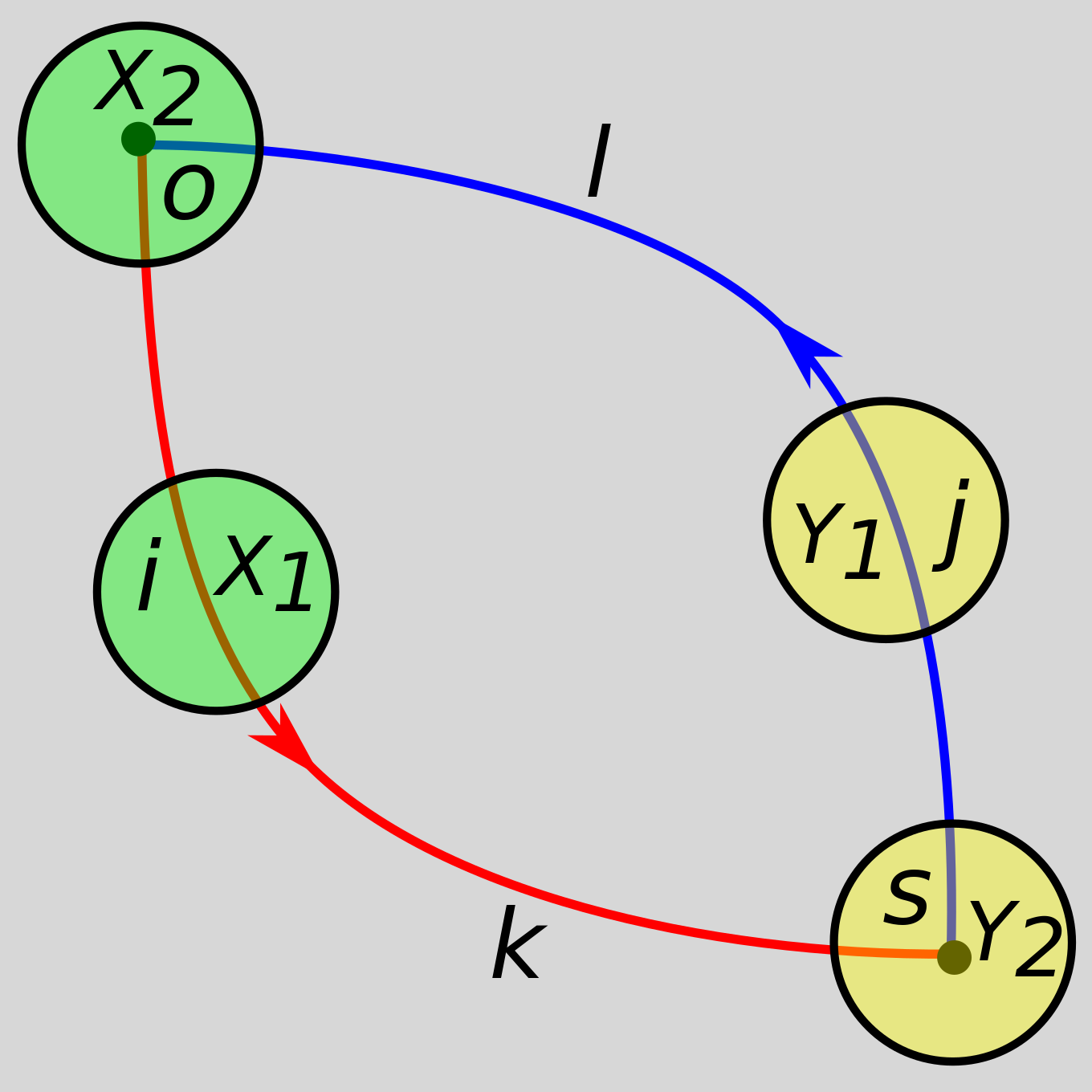} 
		\caption{\label{fig:WHF-1} $t=t_0+T/4$}
	\end{subfigure}
	\begin{subfigure}[t]{.32\linewidth}
		\centering\includegraphics[width=.95\linewidth]{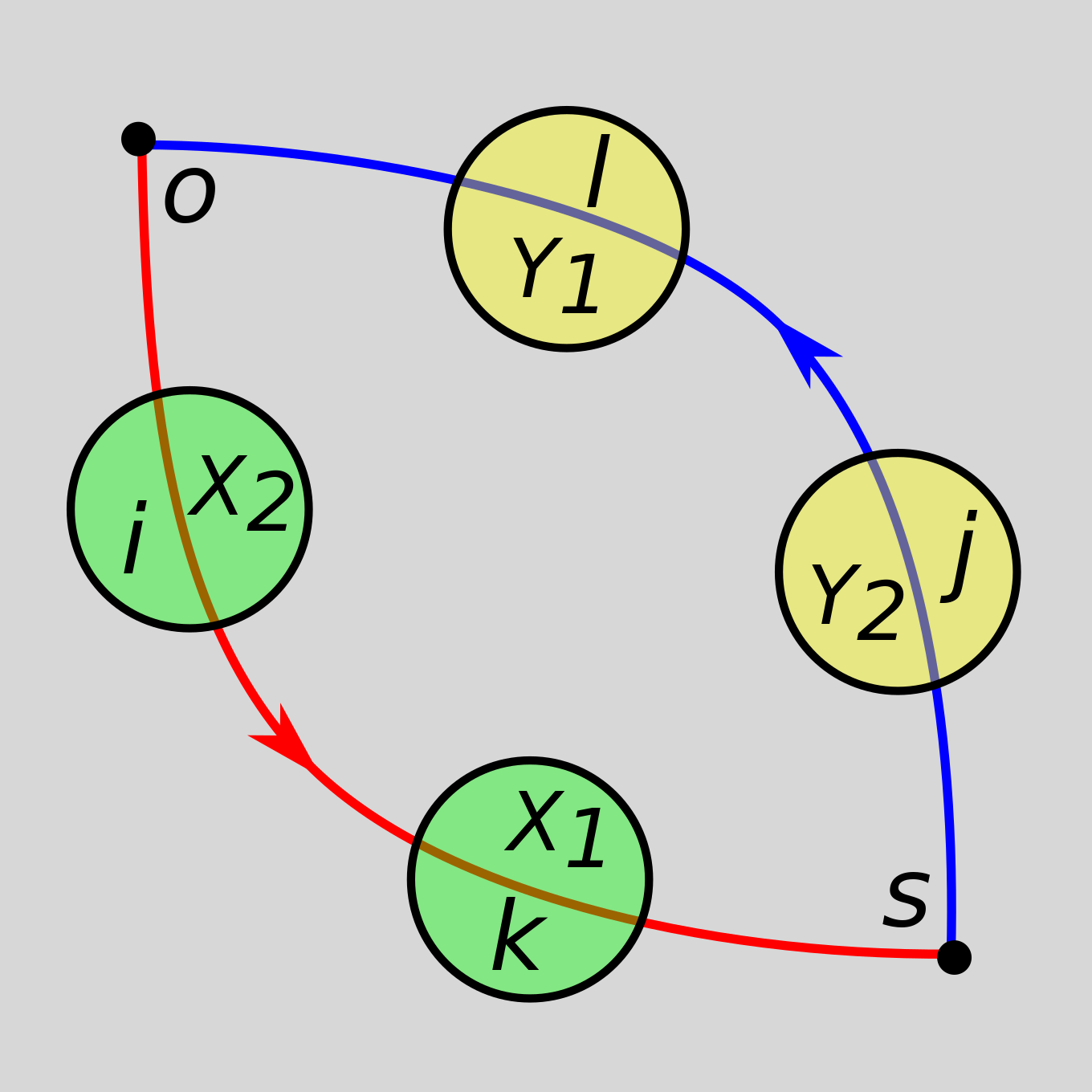} 
		\caption{\label{fig:WHF-2} $t=t_0+T/2$}
	\end{subfigure}
	\begin{subfigure}[t]{.32\linewidth}
		\centering\includegraphics[width=.95\linewidth]{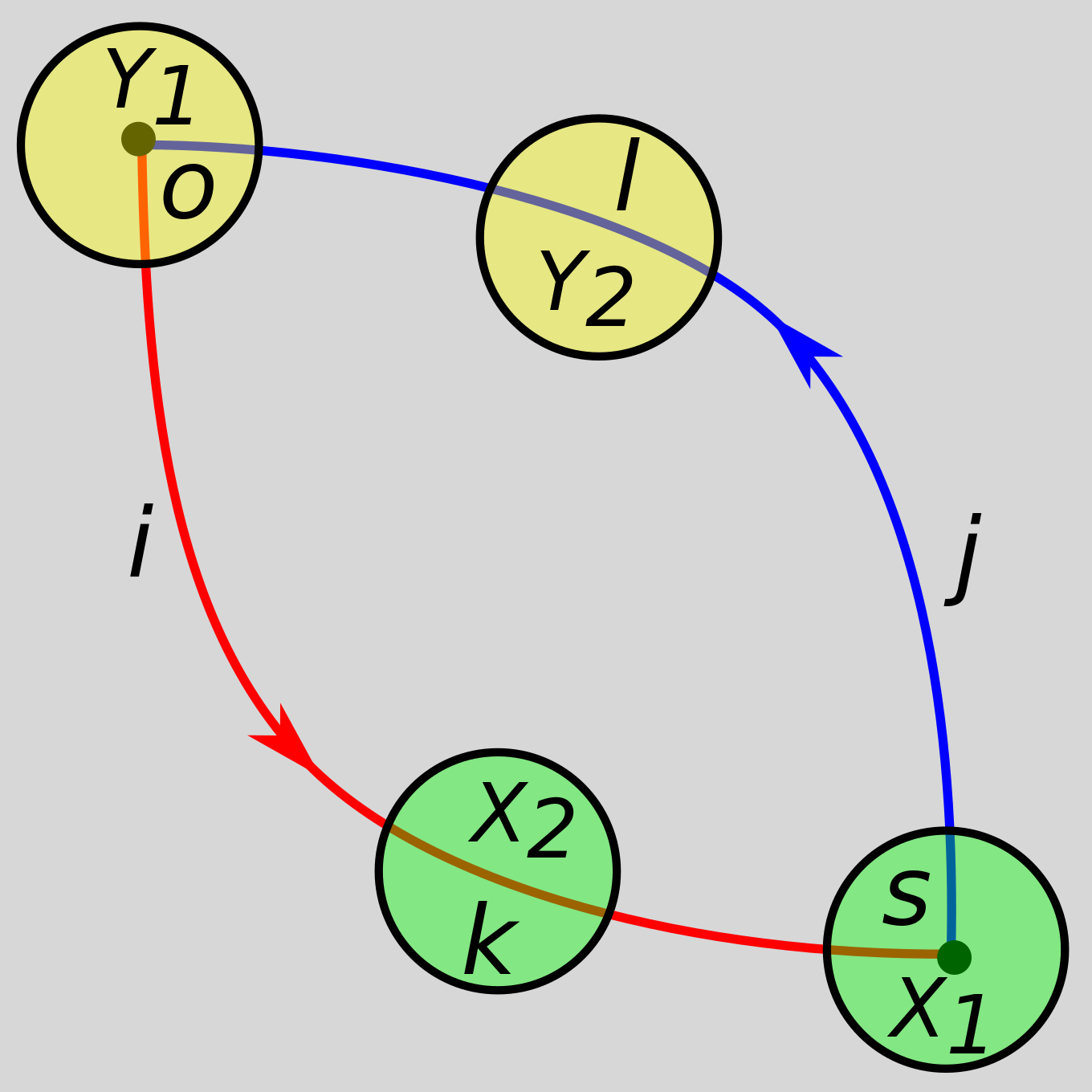} 
		\caption{\label{fig:WHF-3} $t=t_0+3T/4$}
	\end{subfigure}
	\begin{subfigure}[t]{.32\linewidth}
		\centering\includegraphics[width=.95\linewidth]{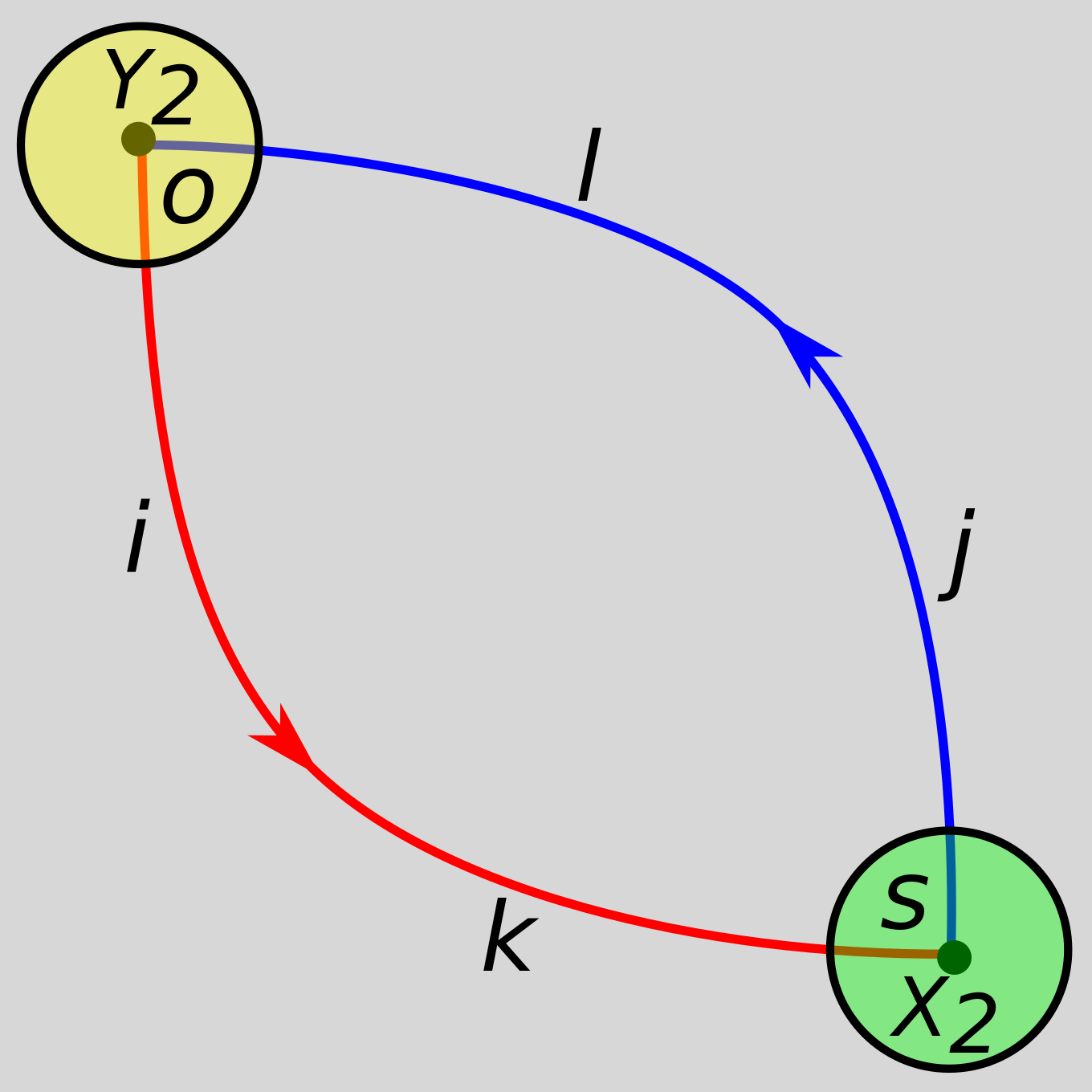} 
		\caption{\label{fig:WHF-4} End $t=t_0+T$}
	\end{subfigure}
	\caption{\label{fig:WHF} 
		Illustrating the game process of the who-entered-first challenge. Here $i,j,k,l$ are four fixed points in the system, introduced for illustration purpose. 
	}
\end{figure}
This challenge is designed as an additional challenge to be performed using the same physical system that the players use to succeed the original challenge. %
It involves four players, Alice, Bob, Carol, and David, and the same group of Referees as before. Before the game starts, the Referees randomly choose $(X_1,Y_1)$ from the set $\{(\text{A},\text{B}),(\text{C},\text{B}),(\text{A},\text{D}),(\text{C},\text{D})\}$. The players do not know the Referees' choice, and at the end of the game, the Referees ask Alice and Carol the value of $Y_1$, and ask Bob and David the value of  $X_1$, and they win if every player answers correctly. In the game Alice and Carol will start from $\oA$, and the value of $X_1$ determine whether Alice~(A) or Carol~(C) will enter the game first, and similarly Bob and David start from $\oB$, and $Y_1$ determines who enter first, hence the name of this challenge.  
To ease our following description we also define $X_2$ and $Y_2$ such that $\{X_1,X_2\}=\{\text{A},\text{C}\}$ and $\{Y_1,Y_2\}=\{\text{B},\text{D}\}$. 

Before the game starts, the four players are allowed to discuss a winning strategy all together. 
After the game starts at $t=0$, the general rules introduced in Sec.~\ref{sec:general_rules} still apply. 
The game process is shown in Fig.~\ref{fig:WHF}:\\
(a) at $t=t_0>0$, %
circle $X_1$ appears at $\oA$, and circle $Y_1$ appears at $\oB$, and both circles start moving along their designated paths, as shown in Fig.~\ref{fig:WHF-0};\\
(b) at $t=t_0+T/4$, after $X_1$ and $Y_1$ have moved a distance away, circles $X_2$ and $Y_2$ appear, as shown in Fig.~\ref{fig:WHF-1}; \\
(d) at $t=t_0+3T/4$, $X_1$ and $Y_1$ arrive at $\oB$ and $\oA$, respectively, and disappear shortly after; \\
(e) finally at $t=t_0+T$, $X_2$ and $Y_2$ arrive at $\oB$ and $\oA$, respectively, and disappear shortly after;\\
Importantly, the initial time $t_0$ and the total duration $T$ are randomly chosen by the Referees, so that each player cannot determine 
if he or she enters the game in the first or the second round by knowing the time of entrance.

\subsection{Winning strategy}\label{sec:winWEF}
We now give a winning strategy using emergent paraparticles defined by the $R$-matrix in Eq.~\eqref{eqApp:seth-R}.
During the pregame discussion, %
the four players agree that regardless of the Referee's choice of $(X_1,Y_1)$~(which of course they do not know), each player creates a paraparticle with a predetermined internal state at the time of entrance, say Alice, Bob, Carol, and David each creates a paraparticle with internal state $a$, $b$, $c$, and $d$, respectively. There are many possible choices for the values of $a,b,c,d$ that can win the game, and one possible choice is  $(a,b,c,d)=(1,2,3,4)$. Below, we analyze the quantum state of the physical system at each stage of the game. We focus on the case in which $(X_1,Y_1)=(\text{A},\text{B})$, the other three possibilities can be analyzed in an identical way. \\
(a). At $t=t_0$: after Alice applies $\hat{U}_{\oA,a}$ and Bob applies $\hat{U}'_{\oB,b}$, we have %
$\ket{\Psi(t_0)}=\ket{G;\oA^a \oB^b}$;\\
(b). At $t=t_0+T/4$:  after Carol applies $\hat{U}_{\oA,c}$ and David applies $\hat{U}'_{\oB,d}$, we have %
\begin{eqnarray}\label{eq:WHFwin-b}
	\ket{\Psi(t_0+T/4)}&=&\hat{U}_{\oA,c}\hat{U}^{\prime}_{\oB,d}\ket{G;i^a j^b}\nonumber\\
	&=&\ket{G;\oA^c i^a j^b \oB^d},
\end{eqnarray}
where $i$ and $j$ are the positions of $X_1$ and $Y_1$ at this time, respectively, as shown in Fig.~\ref{fig:WHF-1}, and in the second line  we use Eq.~\eqref{eq:localcreationatcorner}.\\
(d). At $t=t_0+3T/4$, when Alice and Bob arrive at $\oB$ and $\oA$, respectively, the state of the system evolves to
\begin{eqnarray}\label{eq:Psi3-WHF}
	\ket{\Psi(t_0+3T/4)}&=&\ket{G;k^c \oB^a \oA^b l^d}\\
	&=&\sum_{a',b',c',d'}
	\begin{tikzpicture}[baseline={([yshift=-.8ex]current bounding box.center)}, scale=0.5]
		\Rmatrix{0}{2*\AL}{R}
		\Rmatrix{0}{-2*\AL}{R}
		\Rmatrix{2*\AL}{0}{R}
		\Rmatrix{-2*\AL}{0}{R}
		\node  at (-\AL,-3*\AL) [below=-.04]{\footnotesize $a$};
		\node  at (\AL,-3*\AL) [below=-.04]{\footnotesize $b$};
		\node  at (-3*\AL,-\AL) [below=-.04]{\footnotesize $c$};
		\node  at (3*\AL,-\AL) [below=-.04]{\footnotesize $d$};
		\node  at (\AL,3*\AL) [above=-.04]{\footnotesize $c'$};
		\node  at (-\AL,3*\AL) [above=-.04]{\footnotesize $d'$};
		\node  at (3*\AL,\AL) [above=-.04]{\footnotesize $a'$};
		\node  at (-3*\AL,\AL) [above=-.04]{\footnotesize $b'$};
	\end{tikzpicture}\ket{G;\oA^{b'} l^{d'} k^{c'} \oB^{a'}},\nonumber
\end{eqnarray}
where $k$ and $l$ are the positions of $X_2$ and $Y_2$ at this time, respectively, as shown in Fig.~\ref{fig:WHF-3}, and  in the second line we use Eq.~\eqref{eq:exchangestatR-npt} four times. %
For the $R$-matrix in Eq.~\eqref{eqApp:seth-R}, we can omit the summation in the RHS of Eq.~\eqref{eq:Psi3-WHF}, since the tensor network of the $R$-matrices is only nonzero at one possible value of $(a',b',c',d')$.
Then Bob measures $\hat{O}_{\oA}$ and obtains $b'$, and similarly Alice measures $\hat{O}'_{\oB}$ and obtains $a'$, and the quantum state of the system is still the same as in Eq.~\eqref{eq:Psi3-WHF}. %
After this, Bob and Alice use $\hat{U}_{\oA,b'}$ and $\hat{U}'_{\oB,a'}$, respectively, to annihilate their paraparticles. \\
(e) Finally, at $t=t_0+T$, when Carol and David arrive at $\oB$ and $\oA$, respectively, David measures $\hat{O}_{\oA}$ and obtains $d'$, and similarly Carol measures $\hat{O}'_{\oB}$ and obtains $c'$. Then they use $\hat{U}_{\oA,d'}$ and $\hat{U}'_{\oB,c'}$, respectively, to annihilate their paraparticles.

This completes the analysis for the case $(X_1,Y_1)=(\text{A},\text{B})$. The other three cases can be analyzed in the same way. For example, the analysis for the case $(X_1,Y_1)=(\text{C},\text{B})$ can be obtained by simply applying the substitution $(\text{A},a,a')\leftrightarrow (\text{C},c,c')$ to the above analysis. The measurement results $(a',b',c',d')$ for all the four possible cases are summarized in Tab.~\ref{tab:WHFsolution}. %
Importantly,  Alice can determine %
$Y_1$ from her measurement result $a'$, and similarly for all other players. For example, if she gets $a'=3$ or $4$, then $Y_1=\text{B}$, otherwise $Y_1=\text{D}$. This strategy allows them to win the challenge with 100\% chance of success.

\begin{table}[t]
	\centering
	\begin{tabular}{c|cccc}
		$X_1 Y_1$	  & $a'$ & $b'$ & $c'$ & $d'$ \\%
		\hline
		\text{AB} & 3 & 1 & 2 & 3 \\
		\text{CB} & 4 & 3 & 1 & 2 \\
		\text{AD} & 2 & 1 & 4 & 3 \\
		\text{CD} & 2 & 4 & 4 & 1 \\
	\end{tabular}
	\caption{\label{tab:WHFsolution} A strategy to win the who-entered-first challenge using the paraparticles defined by the set-theoretical $R$-matrix in Eq.~\eqref{eqApp:seth-R}. The initial states are chosen as $\{a,b,c,d\}=\{1,2,3,4\}$. AB means $(X_1,Y_1)=(\text{A},\text{B})$, i.e., Alice and Bob enter first, and similarly for other combinations. $a',b',c'$, and $d'$ are the measurement results of A, B, C, and D, respectively. }
\end{table}

\subsection{Why swap-type $R$-matrices cannot win}\label{sec:swapRfailWEF}
We now explain why swap-type $R$-matrices of the form $R^{b'a'}_{ab}=\delta_{aa'}\delta_{bb'}\theta_{ab}$ %
cannot win the who-entered-first challenge. 
Let $\rho_A$ be the internal state of the paraparticle that Alice creates at $\oA$, which is allowed to be a mixed state in general, and $\rho_B,~\rho_C$, and $\rho_D$ are defined similarly for Bob, Carol, and David. Let $\rho'_A$ be the reduced density matrix describing the output internal state that Alice measures before her circle disappears at $\oB$, and similarly for $\rho'_B,~\rho'_C$, and $\rho'_D$. From the RHS of Eq.~\eqref{eq:Psi3-WHF} we can express $\rho'_A$ as~[for the case $(X_1,Y_1)=(\text{A},\text{B})$]
\begin{equation}\label{eq:rhopA-WEF}
	\rho'_A=\adjincludegraphics[height=12ex,valign=c]{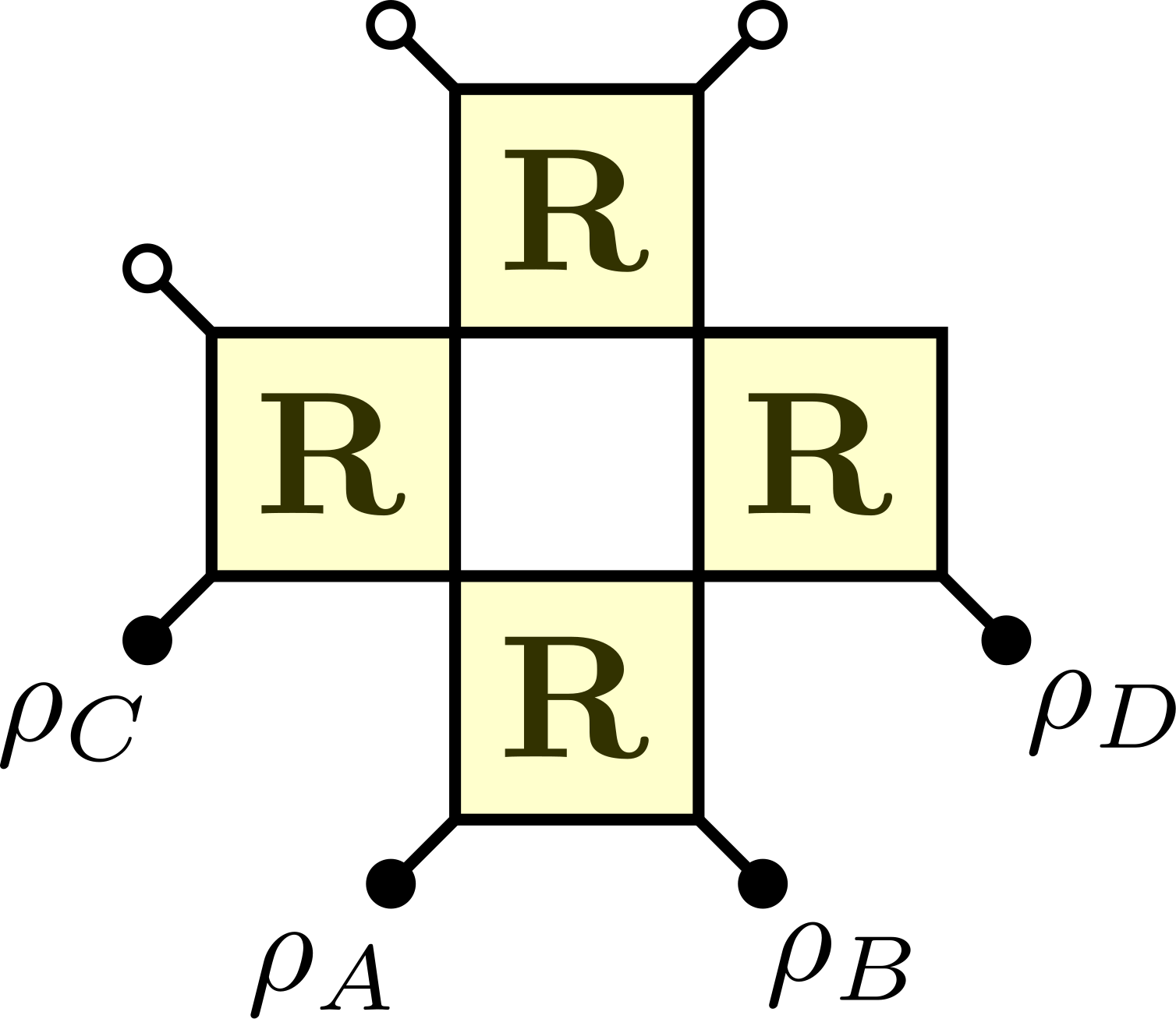},
\end{equation}
where we use the folded picture 
$\adjincludegraphics[height=4ex,valign=c]{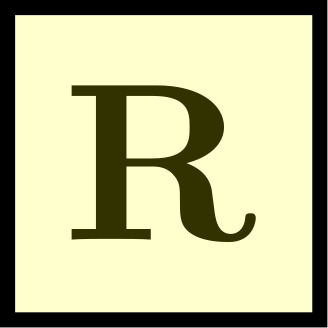}=
\adjincludegraphics[height=5ex,valign=c]{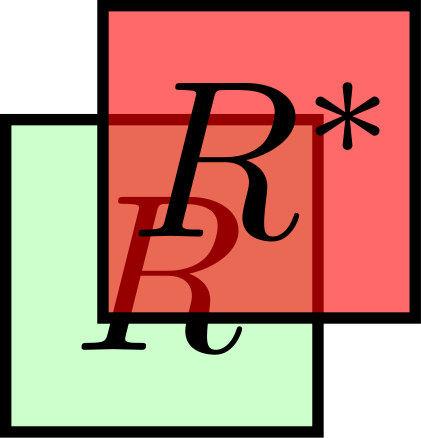}$ to describe the evolution of density matrices in the Heisenberg picture, and we use a dot $\adjincludegraphics[height=2ex,valign=c]{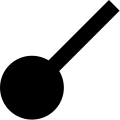}=
\adjincludegraphics[height=3ex,valign=c]{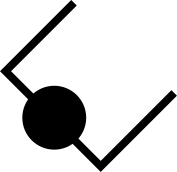}$ to represent a density matrix, and an open circle 
$\adjincludegraphics[height=2ex,valign=c]{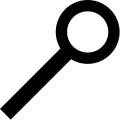}=
\adjincludegraphics[height=3ex,valign=c]{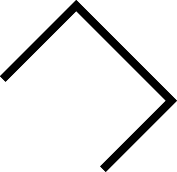}$ represents the identity matrix $\mathds{1}$~(the partial trace is obtained by contracting with $\mathds{1}$). %
Here comes the key observation: a swap-type $R$-matrix satisfies
\begin{eqnarray}\label{eq:diagonalcommute}
\adjincludegraphics[height=8ex,valign=c]{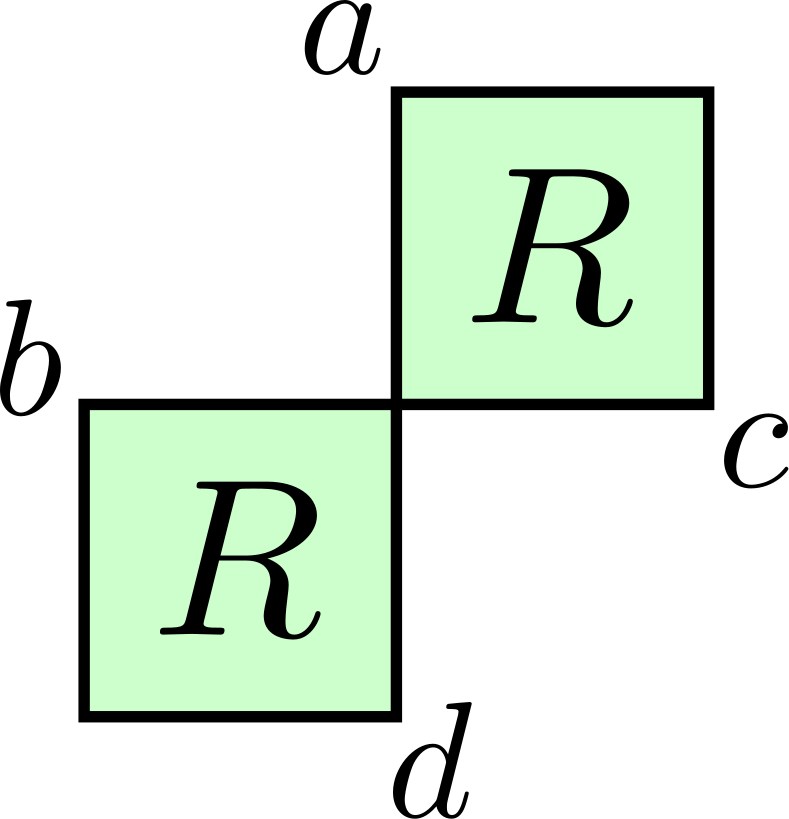}=
\adjincludegraphics[height=8ex,valign=c]{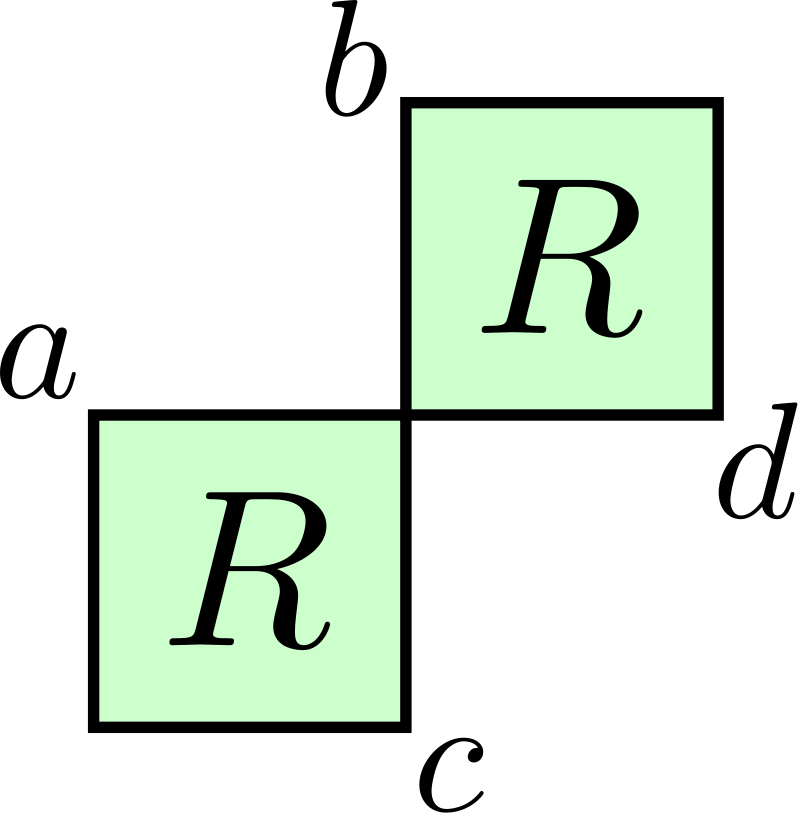}~,\quad
	\adjincludegraphics[height=8ex,valign=c]{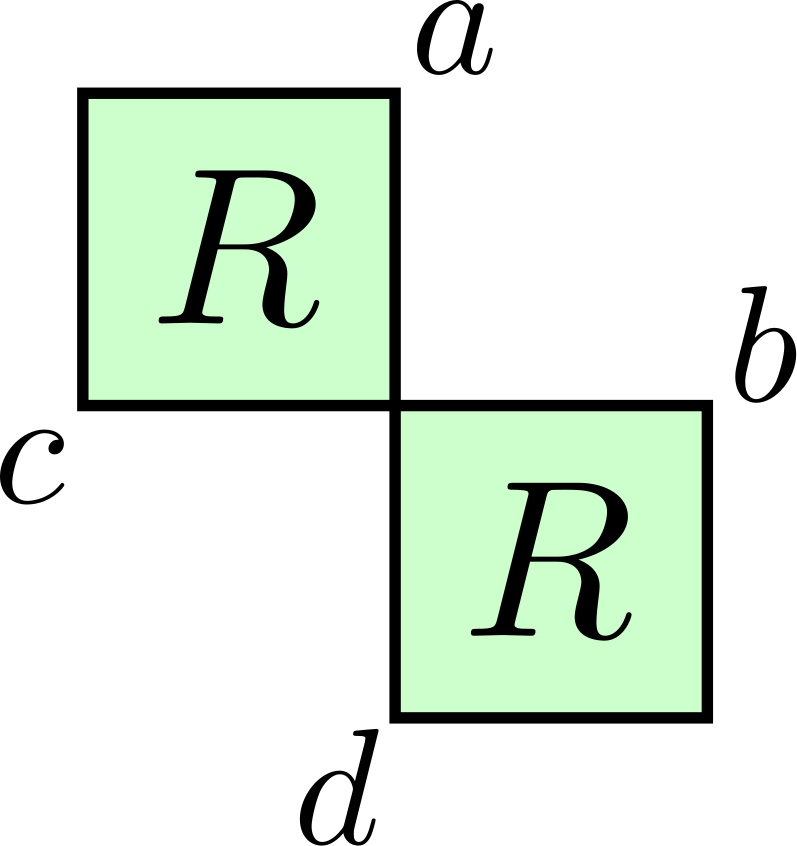}=
	\adjincludegraphics[height=8ex,valign=c]{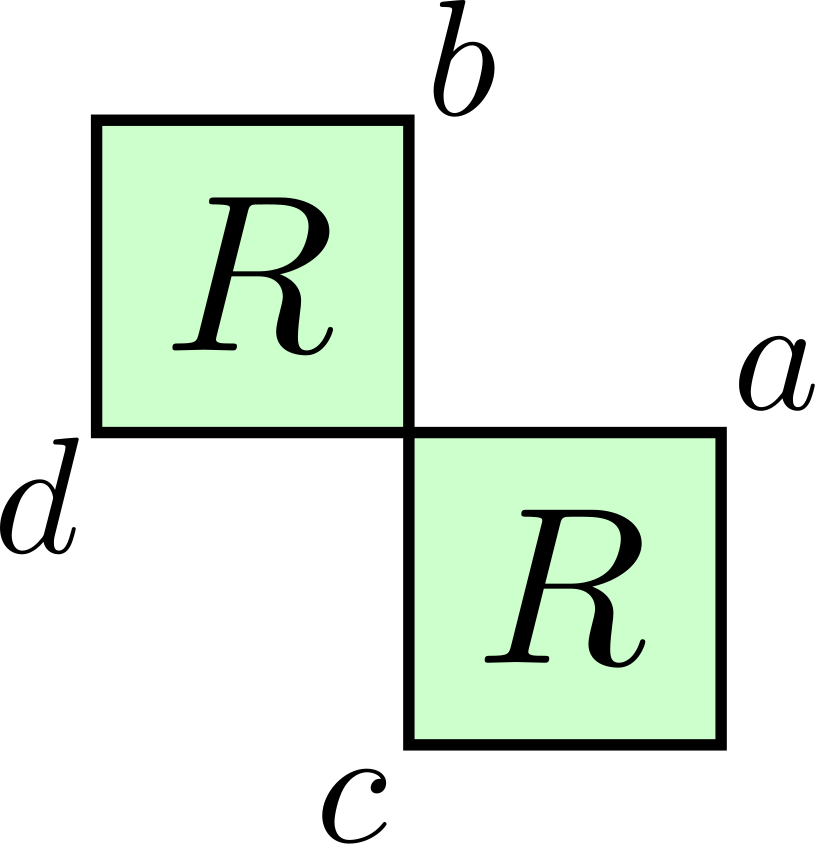}~,
\end{eqnarray}
leading to
\begin{eqnarray}
	&&\adjincludegraphics[height=12ex,valign=c]{Figures/antiswap-2.png}=
	\adjincludegraphics[height=12ex,valign=c]{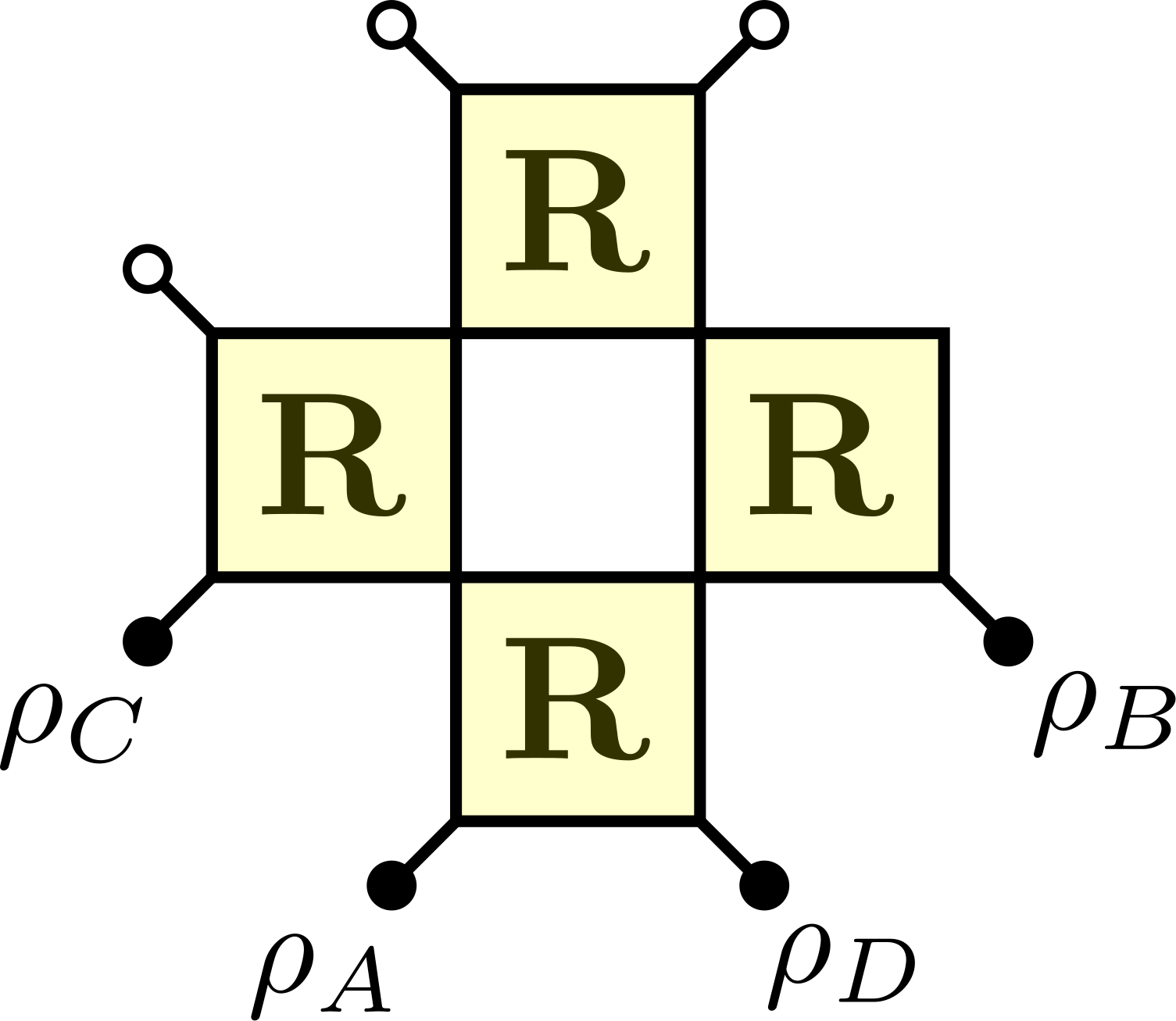}\nonumber\\
	&=&\adjincludegraphics[height=12ex,valign=c]{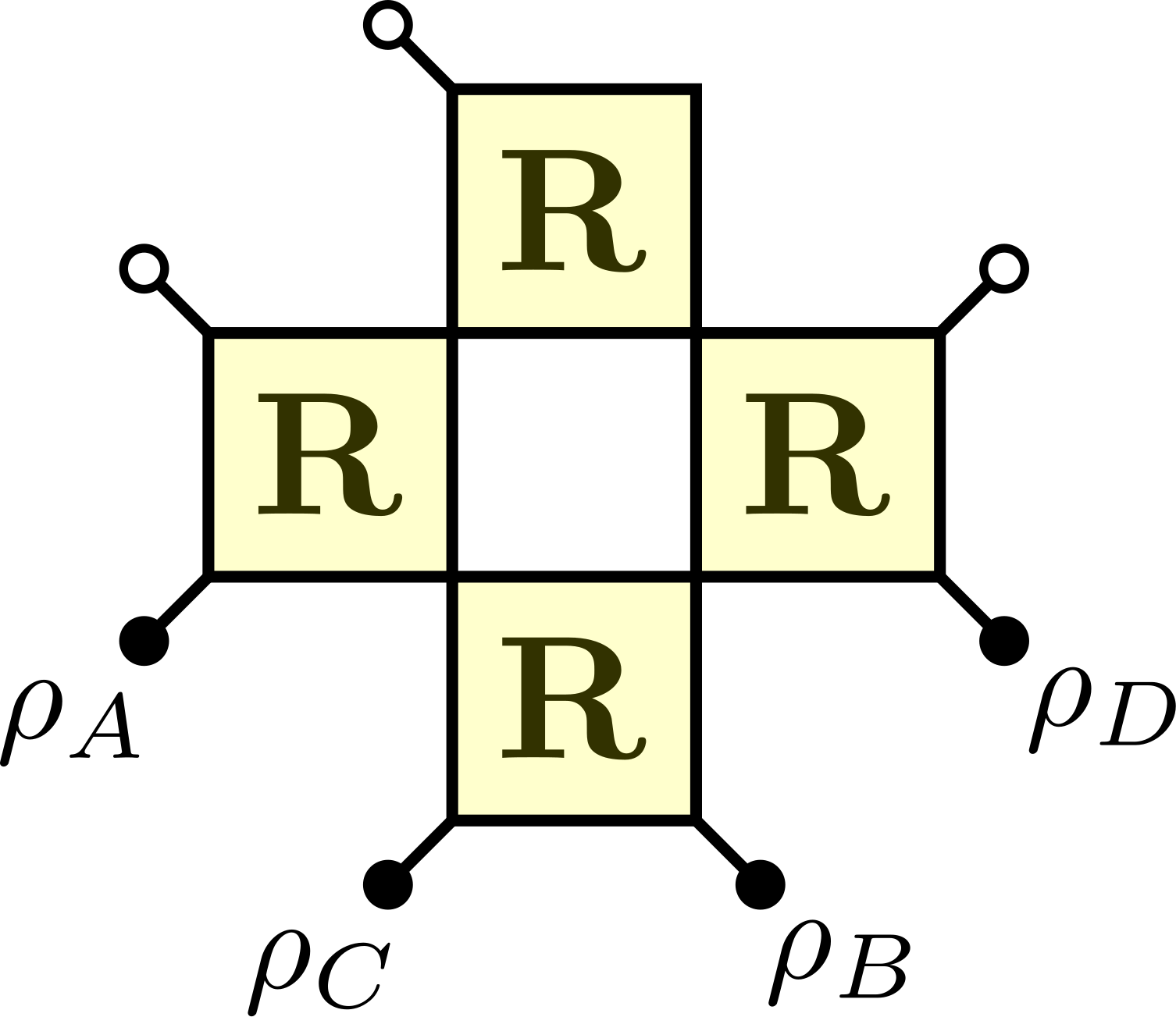}=
	\adjincludegraphics[height=12ex,valign=c]{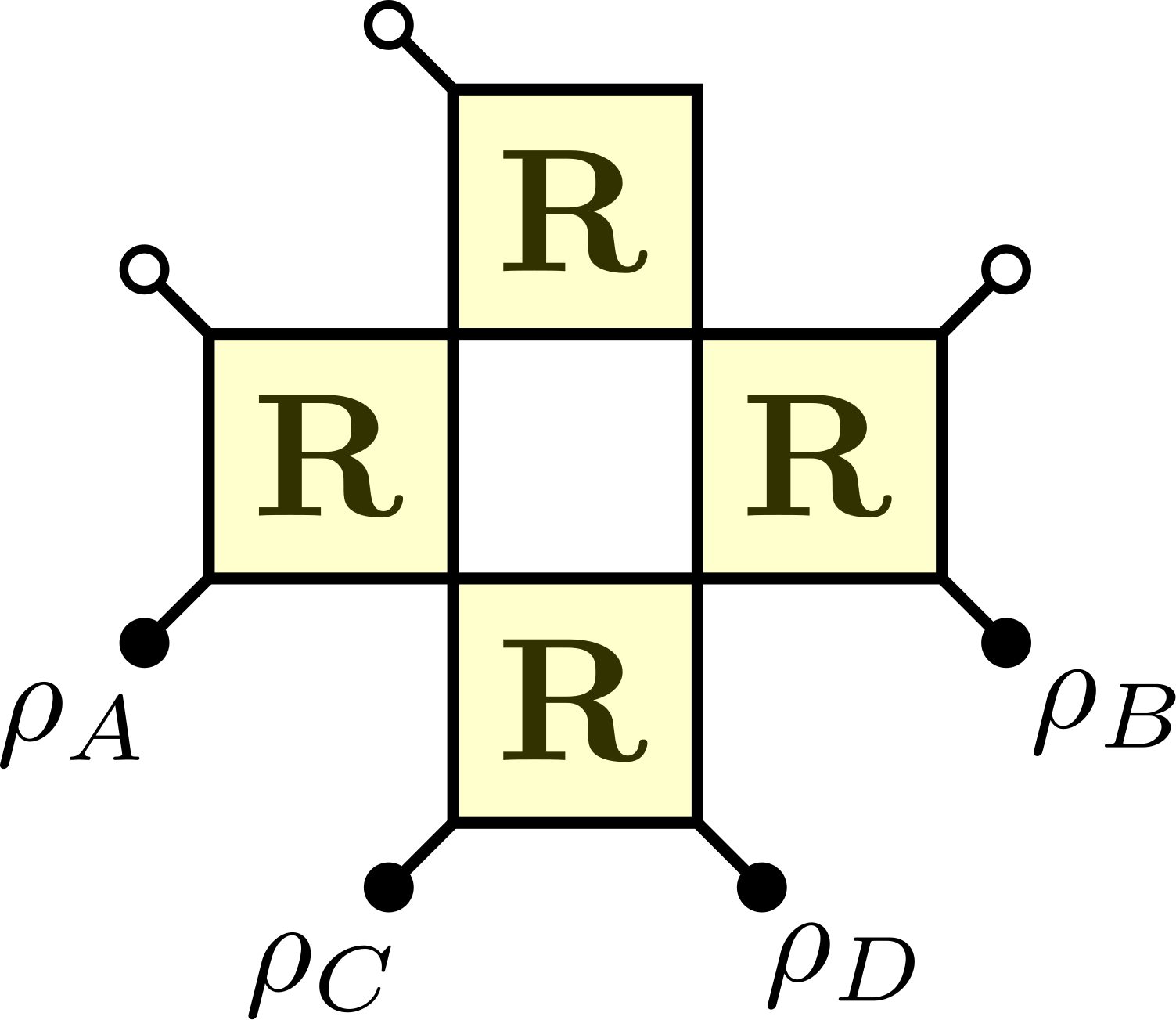},
\end{eqnarray}
which means that $\rho'_A$ is the same for all the four possible cases of $(X_1,Y_1)$. %
Therefore, Alice cannot obtain any information about $Y_1\in\{\text{B},\text{D}\}$ by measuring $\rho'_A$. The analysis for other players is identical, leading to the conclusion that with a swap-type $R$-matrix, the players cannot have a strategy to win the who-entered-first challenge any better than pure guessing.

\section{Categorical analysis of the challenge game}\label{sec:categorical_analysis}
In this section we present a more detailed analysis of the challenge game and its winning strategy using the framework of tensor category theory. The ultimate goal is to figure out which family of gapped phase of matter can be used to win this challenge. Rigorously performing this analysis in full generality is mathematically hard, since we first need a complete, rigorous classification of all gapped phases of matter in 3D~(and 2D), %
which is still %
a hard open question in mathematical physics~(if one assumes no more than locality and a spectral gap). However, there is now a combination of convincing physical arguments along with rigorous mathematical analysis~(based on some additional natural physical assumptions) showing 
that the universal properties of point-like excitations in 3D~(2D) topological phases are described by symmetric~(braided) fusion categories~\cite{kitaev2006anyons,kong2014braided,LanKongWen3DAB,LanWen3DEF,Johnson-Freyd2022,kong2022invitation}~\footnote{Note: for the 2D case, topological quasiparticles in a bosonic~(fermionic) topological phase are described by (super-)~modular tensor categories~\cite{lanModularExtensionsUnitary2017}, a special class of braided fusion categories satisfying the braiding non-degeneracy condition~[(super-)~modularity]. This (super-)~modularity condition is not  important for our discussions in this paper, so we do not emphasize it here.
}.
In this paper, we perform our analysis within this tensor categorical framework. For simplicity, in this section we focus on the case 
in which the Hamiltonian $\hat{H}$ is defined on a topologically trivial manifold~(e.g. on a sphere) and is translationally invariant everywhere except at the two special points $\oA$ and $\oB$. 
The more general case where $\oA$ and $\oB$ lie on some higher-dimensional defects can be analyzed in a similar way, but requires the additional tools of module categories, %
which will be treated in Sec.~\ref{sec:ModCatDefect}.

Specifically, in Sec.~\ref{sec:reduction_to_TPO} we present a simple argument that reduces the problem to universal topological properties of the system, in Sec.~\ref{sec:categorical_framework} we briefly review the tensor categorical description of point-like quasiparticles in  3D~(or 2D) gapped phases of matter~\cite{kitaev2006anyons,kong2014braided,LanKongWen3DAB,LanWen3DEF}. In Sec.~\ref{sec:win_strategy_diagrammatic} we apply this categorical framework to give an alternative description of the winning strategies presented in Secs.~\ref{sec:win} and \ref{sec:WHF} %
using fusion diagrams. Finally in Sec.~\ref{sec:TC_winning_condition} we derive the necessary condition for winning the challenge within the categorical framework, and show that the special class of SFCs defined in Sec.~\ref{sec:win_strategy_diagrammatic}~(the SFCs that contain nontrivial $R$-paraparticles) are the only subclass that can pass the full version of the challenge in 3+1D.

\subsection{Reduction of the problem to universal topological properties}\label{sec:reduction_to_TPO}
Our goal here is to show that whether a physical system can pass the challenge depends only on its universal topological properties, which allows us to forget about microscopic details and justifies our categorical analysis. To this end, we first prove the important fact that the set of physical systems that can be used to win any version of the game~(including the original version and all its variants and generalizations) is invariant under local unitary transformations~(LUTs)~\cite{Chen2010LUT}. To begin, suppose that for some physical system described by a locally-interacting Hamiltonian $\hat{H}$, we have a winning strategy with a success rate larger than pure guessing in the thermodynamic limit. Let $\ket{G}$ be the unique, gapped, and frustration-free ground state of $\hat{H}$. We can describe the winning strategy using the sequence of physical operations performed by the players: $\hat{U}_A(t),\hat{O}_A(t), \hat{U}_B(t),\hat{O}_B(t)$ for $0\leq t\leq T$, which means that Alice performs  the local unitary operation  $\hat{U}_A(t)$ and measurement $\hat{O}_A(t)$ at time $t$, and similarly for Bob. Now let $\hat{U}$ be any LUT that can be represented as a finite depth unitary circuit~\footnote{For simplicity, here we only consider the subclass of finite depth LUTs. But in principle we can relax the frustration-free condition in the game requirement~(see App.~\ref{sec:relax_FF}) and generalize these arguments  to the full class of LUTs.}. It is then clear that we also have a winning strategy for the physical system described by the transformed Hamiltonian $\hat{H}'=\hat{U}\hat{H}\hat{U}^\dagger$. Indeed, since any finite depth unitary circuit maps local operators into local operators, $\hat{H}'$ is also locally-interacting, and its ground state $\ket{G'}=\hat{U}\ket{G}$ also qualifies the requirements of the game as it is also unique, gapped, and frustration-free. It follows that the players can simply use the transformed operations $\hat{U}\hat{U}_A(t)\hat{U}^\dagger,\hat{U}\hat{O}_A(t)\hat{U}^\dagger, \hat{U}\hat{U}_B(t)\hat{U}^\dagger,\hat{U}\hat{O}_B(t)\hat{U}^\dagger$ to win the game~\footnote{A subtlety here is that these transformed operations are generally supported on a larger area due to the evolution by $\hat{U}$, so the players may need to choose a larger $r_0$ for the radius of the two circles~(or spheres in 3D) in order to accommodate these local operations. %
It is for this reason that in the game protocol we allow the players to choose the radius $r_0$. 
}, since with this choice the time evolution of the system is isomorphic to the original case without applying the LUT $\hat{U}$. 

The above fact implies that whether a physical system can pass the challenge depends only on its underlying topological order. If a physical system can win the challenge, then so do any other systems in the same topological phase, as any two states in the same topological phase are related by a LUT~\cite{hastings2005quasiadiabatic1,Chen2010LUT,LPPL}.  
A more formal argument can be made as follow~(see Fig.~\ref{fig:winnability_function}).  Let $\calS$ denote the set of gapped ground states of locally interacting Hamiltonians  on  a given lattice quantum system. %
Local unitary transformations define an equivalence relation $\sim$ between states in $\calS$: two states $\ket{G_1},\ket{G_2}\in\calS$ are called equivalent $\ket{G_1}\sim \ket{G_2}$ if there exists a local unitary transformation $\hat{U}$ such that $\ket{G_1}=\hat{U} \ket{G_2}$. A topological order is an equivalence class of states in $\calS/\sim$. Let $\pi:\calS \to \calS/\sim$ be the canonical projection, i.e., for any $\ket{G}\in\calS$, $\pi(\ket{G})$ denotes its topological order. 
Now let $w:\calS\to \{0,1\}$ denote the winnability function indicating whether a given ground state $\ket{G} $ can win the challenge, i.e., $w(\ket{G})=1$ means that a winning strategy exists for the state $\ket{G}$. 
Since $w(\ket{G})=w(\hat{U}\ket{G})$, for any LUT $\hat{U}$, $w$ must be a class function, i.e., $w:\calS\to \{0,1\}$ must factor through $\calS/\sim$ as $w=\bar{w}\circ\pi$, where $\bar{w}:S/\sim\to \{0,1\}$ indicates whether a given topological order can win the challenge, as depicted in Fig.~\ref{fig:winnability_function}. 

\begin{figure}
	\includegraphics[width=.95\linewidth]{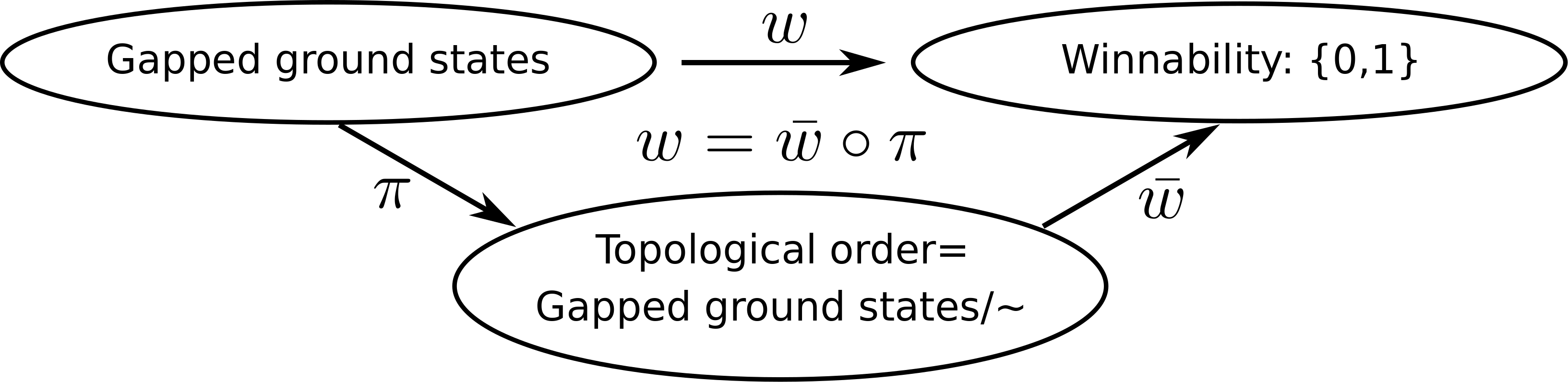} 
	\caption{\label{fig:winnability_function} A formal argument that whether a physical system can win the challenge game depends only on its underlying topological order. %
	}
\end{figure}

\subsection{Basic assumption: the tensor categorical description of a gapped phase}\label{sec:categorical_framework}
In this section we briefly outline our assumption that the universal properties of point-like excitations in a 3D~(2D) topological phases are \textit{described by} a symmetric~(braided) fusion category, which we denote by $\calC$. 
A rigorous formulation of this assumption  requires significant effort--indeed, the hardest part is to specify what the two words ``described by'' mean in precise mathematical language, i.e., what precisely do we mean by saying that a microscopic lattice quantum system is ``described by'' some  topological quantum field theory~(or some tensor category) in the long distance limit?  To our knowledge, there are two ways to formulate this: 
one is to use a lattice version of algebraic quantum field theory~\cite{halvorson2006algebraic,haagLocalQuantumPhysics1996,naaijkens2017quantum}, 
another is to use the idea of information convex set--a central concept in the entanglement bootstrap program~\cite{Shi2020,Shi2020Verlinde,Kim2021DomainWall,yang2025topological}. 
Here we use ideas from both formalisms, but we focus on illustrating the physics without pursuing the highest standard of mathematical rigor. 
(Actually, a fully rigorous formulation   
may still be lacking~\cite{kitaev2024almost}.) 
The basics of symmetric~(braided) fusion category and topological order can be found in standard textbooks and review papers~\cite{simon2023topological,kong2022invitation,preskill1999lecture,wenQuantumFieldTheory2007,kong2014braided}, here we will put a special emphasis on locality and the connection between microscopic degrees of freedoms and the macroscopic description. %

\paragraph{Quasiparticles and their types} 
A quasiparticle in a condensed matter system is a point-like excitation above its ground state $\ket{G}$. More precisely, let $\ket{\Psi}$ be an excited state of the Hamiltonian $\hat{H}$ of the system. If all local properties of $\ket{\Psi}$ are the same as that of $\ket{G}$~(this can be formulated rigorously by comparing their reduced local density matrices) except in a finite region centered around a position $x$, then we say $\ket{\Psi}$ has a point-like excitation at $x$. We always assume that a quasiparticle can be moved from one position to another by applying a sequence of local unitary operations on a path connecting the two positions~\footnote{In our challenge game, it is clear that whatever excitations the players have in their respective circles must be movable by local unitary operations, since the players are always obliged to move them to follow the circle movements. However, the special topological excitations $\sigma$ and $\bar{\sigma}$ sitting at $\oA$ and $\oB$ do not need to be mobile, since they remain at the same position throughout the game. In 3D, it is known that there are phases of matter hosting quasiparticles with restricted mobility, called fractons~\cite{haah2011local,vijay2016fracton1,nandkishore2019fractons}. 
We do not consider fracton phases in this work, but we expect that even in a 3D fracton phase, the category of all point-like topological quasiparticles with unrestricted mobility still forms an SFC $\calC$, and the fusion between fractons and mobile quasiparticles is still described by a $\calC$-module category $\calM$, so the more general module category analysis of the winning strategy presented in Sec.~\ref{sec:ModCatDefect} should still be valid.  
}. 

A topological phase may host different types of quasiparticles with distinct physical properties. The particle type, also known as topological charge or superselection sector, is a label of a quasiparticle that specifies its universal properties in topological processes such as fusion and braiding.  
A quasiparticle may have some internal degrees of freedom, leading to a degeneracy in the corresponding excited state of the system. If this degeneracy is stable against any local perturbation that do not close the spectral gap~(when the quasiparticle is far from boundary and other quasiparticles), then we say that the quasiparticle is simple.
Therefore by definition, a simple particle type cannot be changed by any kind of local physical processes, including local unitary operations and measurements.  
In the tensor category description of topological phases, %
it is always assumed that there is a finite set of possible simple particle types, which we denote by $\Irr(\calC)$. There is a special particle type $I\in\Irr(\calC)$, called vacuum, which physically describe quasiparticles that can be created by local operators. For each simple particle type $\beta\in\Irr(\calC)$, there is a unique antiparticle type, denoted by $\bar{\beta}$, such that it is physically possible to create a pair of quasiparticles $\beta\bar{\beta}$ at any position using local unitary operations~(we will elaborate this more later).  
If a quasiparticle is not simple, we call it composite, and a composite quasiparticle can always be decomposed as a direct sum~(superposition) of simple quasiparticles. The type of any simple quasiparticle can be measured locally, that is, there exists a local observable $\hat{T}$ at the vicinity of the quasiparticle being measured,  satisfying
\begin{equation}\label{eq:measure_type}
	\hat{T}\Ket{\adjincludegraphics[height=4ex,valign=c]{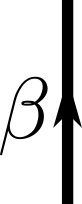}}=\beta\Ket{\adjincludegraphics[height=4ex,valign=c]{Figures/FCbasics/FS-beta.png}}, \quad \forall \beta\in \Irr(\calC).
\end{equation} 
Here we adopt a diagrammatic representation in which a quasiparticle is depicted as an arrow pointing upward~(the time direction), with label $\beta$ indicating its type. 
\paragraph{Fusion and splitting}
Two topological quasiparticles can be fused into another type of quasiparticle, which is generally a composite type. %
The process of fusion can be done either by physically moving two quasiparticles to the same location, or simply by zooming out and viewing the two quasiparticles as a single excitation in the region containing both of them.  
The result of fusion is abstractly described by the fusion rules of the theory
\begin{equation}\label{eq:FCrules}
	\sigma\times \psi=\sum_{\beta\in \Irr(\calC)} N_{\sigma\psi}^\beta \beta,\quad \forall \sigma, \psi\in \Irr(\calC),
\end{equation}
where $\{N_{\sigma\psi}^\beta\}$ are non-negative integers called fusion multiplicities.  Intuitively, Eq.~\eqref{eq:FCrules} means that there are $N_{\sigma\psi}^\beta$  linearly independent ways to fuse $\sigma$ and $\psi$ into $\beta$.
More precisely, we have
\begin{equation}\label{eq:fusion_basis}
	\Ket{\adjincludegraphics[height=4ex,valign=c]{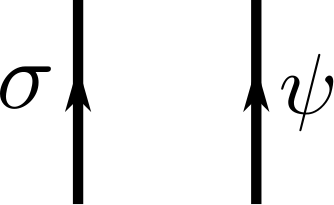}}=\sum_{\beta}\sum_{a=1}^{N_{\sigma\psi}^\beta} \Ket{\adjincludegraphics[height=7ex,valign=c]{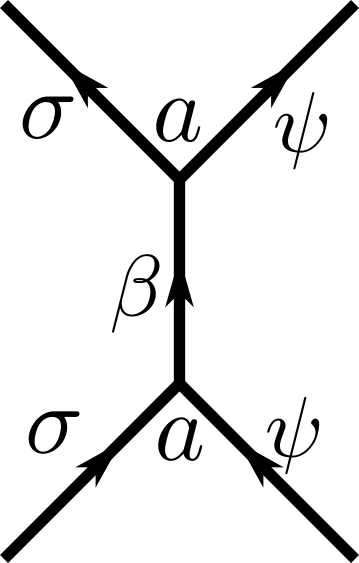}},
\end{equation} 
where we use $\{\Ket{\adjincludegraphics[height=5ex,valign=c]{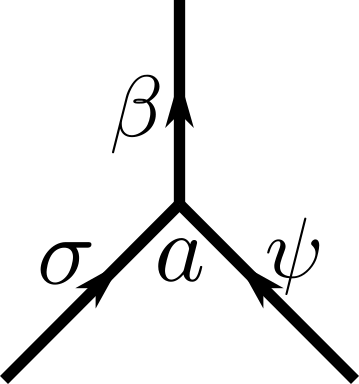}}|1\leq a\leq N_{\sigma\psi}^\beta\}$ to denote a basis for the fusion space $V^\beta_{\sigma\psi}$, and the trivalent vertex indicates the physical process of fusion. 
Note that in this paper we use physical normalization of fusion diagrams
\begin{equation}
\Ket{\adjincludegraphics[height=7ex,valign=c]{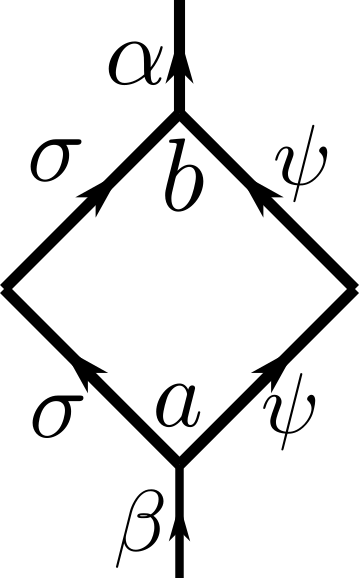}}=
\delta_{ab}\delta_{\alpha\beta}\Ket{\adjincludegraphics[height=4ex,valign=c]{Figures/FCbasics/FS-beta.png}}
\end{equation}
instead of the isotopic normalization~\cite{simon2023topological}.  %
Both the resulting particle type $\beta$ and the fusion multiplicity index $a$ can be measured locally, i.e., there exists observables $\hat{T},\hat{O}$ supported on a region containing $\sigma$ and $\psi$,  %
satisfying
\begin{eqnarray}\label{eq:measure_fusion}
	\hat{T}\Ket{\adjincludegraphics[height=7ex,valign=c]{Figures/FCbasics/FS-completeness-2.png}}=
	\beta\Ket{\adjincludegraphics[height=7ex,valign=c]{Figures/FCbasics/FS-completeness-2.png}},\quad %
	\hat{O}\Ket{\adjincludegraphics[height=7ex,valign=c]{Figures/FCbasics/FS-completeness-2.png}}=
	a\Ket{\adjincludegraphics[height=7ex,valign=c]{Figures/FCbasics/FS-completeness-2.png}}.
\end{eqnarray}
Note that if we measure $\hat{T}$ and $\hat{O}$ in the state in Eq.~\eqref{eq:fusion_basis} and obtain $\beta$ and $a$, then after the measurement the state will collapse to the eigenstate in Eq.~\eqref{eq:measure_fusion},  
according to the measurement axiom of quantum mechanics.

To precisely define the connection between the \textit{microscopic} quantum many body system and its \textit{macroscopic} tensor categorical description, we need to specify how the topological degrees of freedoms are encoded in the quantum state of the microscopic lattice system, i.e., how the fusion space $V^\beta_{\sigma\psi}$ is embedded in the Hilbert space of the lattice system. We formulate this connection using the concept of
information convex set introduced in Ref.~\onlinecite{Shi2020}.  
Let $K$ be a simply connected finite region of the system. Let $\ket{\Psi}$ be a quantum state with two far separated particles of type $\sigma$ and $\psi$ in region $K$~(both are far from the boundary of $K$), such that the total fusion channel in region $K$ is equal to $\beta$, and suppose all other particles in $\ket{\Psi}$ are far from the region $K$. Let $\rho=\ket{\Psi}\bra{\Psi}$, and let $\rho_K=\Tr[\rho]$  be the reduced density matrix in region $K$. 
The information convex set $\Gamma^\beta_{\sigma\psi}(K)$ is defined as the convex set of all density matrices $\rho'_K$ supported on $K$ that satisfies\\
(1). $\rho'_K$ is locally indistinguishable from the reference state $\rho_K$ in any $O(1)$-sized subregion of $K$;\\
(2). the total fusion channel of $\rho'_K$ in region $K$ is also equal to $\beta$.\\
For a mathematically precise definition of the information convex set, we refer to Ref.~\onlinecite{Shi2020}~\footnote{Although Ref.~\onlinecite{Shi2020} only defines information convex set in regions with no excitations~(or more precisely, for a reference state that satisfies their axioms A0 and A1 exactly everywhere), it is straightforward to extend their definitions to our case.}. 
A key assumption we make is the following equality, which is analogous to Theorem 4.5 of Ref.~\onlinecite{Shi2020}:
\begin{equation}\label{eq:structureinfoconvex}
\Gamma^\beta_{\sigma\psi}(K)\cong \calS(V^\beta_{\sigma\psi}),
\end{equation}
where $\calS(V^\sigma_{\sigma\psi})$ denotes the convex set of all density matrices defined on the fusion space $V^\sigma_{\sigma\psi}$, and $\cong$ here means that there exists an isomorphism between the two sets that preserve the geometry of the two convex sets~(see the isomorphism theorems in Ref.~\onlinecite{Shi2020} for the precise definition). %
Eq.~\eqref{eq:pair_creation} specifies how the fusion space $V^\beta_{\sigma\psi}$ is encoded in the microscopic degrees of freedom, and it is included as a part of our assumption that the gapped phase of matter in question is \textit{described by} the tensor category $\calC$. This assumption will be crucial for our analysis in Sec.~\ref{sec:TC_winning_condition}. 

The reverse process of fusion is called \textit{splitting}, which can be done via local unitary operations. More precisely, suppose we have a quantum state $\Ket{\adjincludegraphics[height=4ex,valign=c]{Figures/FCbasics/FS-beta.png}}$ with a simple quasiparticle of type $\beta$ at some position $x$, and suppose $N_{\sigma\psi}^\beta>0$ for some $\sigma,\psi\in \Irr(\calC)$. Then there exists a unitary operator $\hat{U}_a^{\beta,\sigma\psi}$ localized around $x$ satisfying 
\begin{equation}\label{eq:U_spltting}
	\hat{U}_a^{\beta,\sigma\psi}\Ket{\adjincludegraphics[height=4ex,valign=c]{Figures/FCbasics/FS-beta.png}}= 
	\Ket{\adjincludegraphics[height=7ex,valign=c]{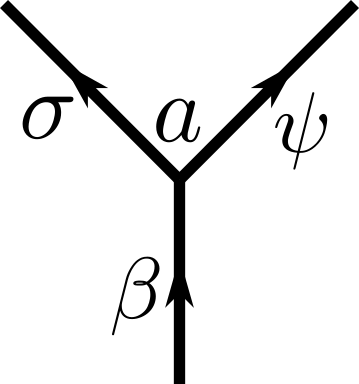}}.
\end{equation} 
In particular, taking $\beta=I$ in Eq.~\eqref{eq:U_spltting} forces $\psi=\bar{\sigma}$, and since in any unitary fusion category $N_{\sigma\bar{\sigma}}^I=1$, we can omit the fusion multiplicity label $a$ and obtain
\begin{equation}\label{eq:pair_creation}
	\hat{U}^{I,\sigma\bar{\sigma}}\Ket{G}=\Ket{\adjincludegraphics[height=7ex,valign=c]{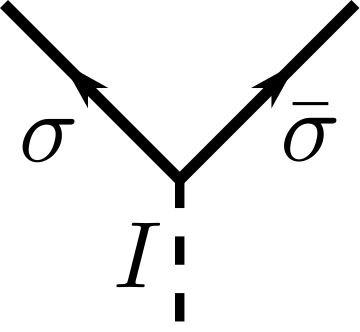}},
\end{equation}
where we use a dashed line to represent the vacuum type. Eq.~\eqref{eq:pair_creation} means that one can create a particle-antiparticle pair at any location using local unitary operations. 

\paragraph{The state space and change of basis}
In general, a quasiparticle $\beta$ can be split into quasiparticles $\alpha,\gamma,\psi$ %
in different ways, leading to a topologically degenerate~(assuming $\alpha,\gamma,\psi$ are far apart) %
space of excitations $V^{\alpha\gamma\psi}_{\beta}$.
This space is spanned by basis states of the form
\begin{equation}\label{eq:BFCstatespace}
\ket{\mu;a,b}=	\Ket{\adjincludegraphics[height=9ex,valign=c]{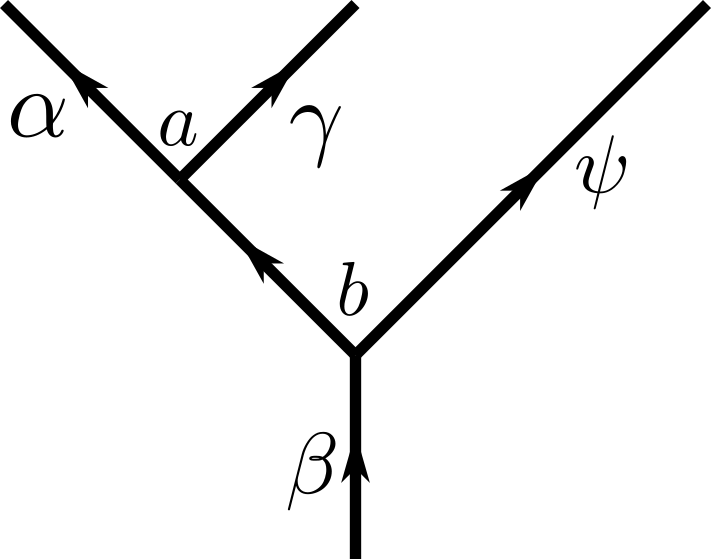}},
\end{equation} 
where the labels $\mu,a,b$ run over all possible allowed values, but the tree structure on the RHS is fixed. 
Note that one can in principle construct a different basis using a different tree structure that split $\beta$ into $\alpha,\gamma,\psi$,
and these two basis are related by a unitary transformation~(also called the $F$-move of fusion diagrams)
\begin{equation}\label{eq:F-move}
	\Ket{\adjincludegraphics[height=9ex,valign=c]{Figures/FCbasics/F-move-L-2.png}}=\sum_{\nu,c,d} [F^{\alpha\gamma\psi}_{\beta}]^{\mu ab}_{\nu cd} \Ket{\adjincludegraphics[height=9ex,valign=c]{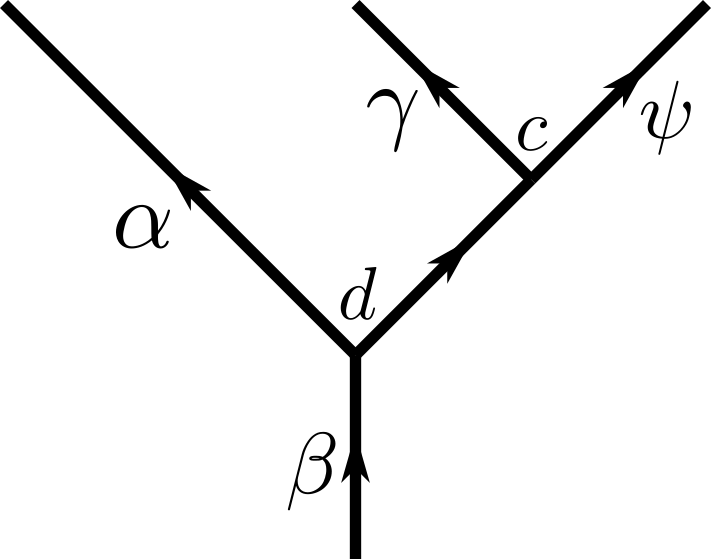}},
\end{equation} 
where $F^{\alpha\gamma\psi}_{\beta}$ is required to be unitary as a matrix~(we only consider unitary fusion categories). Furthermore, $F^{\alpha\gamma\psi}_{\beta}$ is also required to satisfy the so-called pentagon equation~\cite{MooreSeiberg1989,kitaev2006anyons} to guarantee the consistency of Eq.~\eqref{eq:F-move}. 

More generally, the space $V^{\alpha_1\alpha_2\ldots\alpha_n}_{\beta}$ with $n$ quasiparticles of types $\alpha_1,\alpha_2,\ldots,\alpha_n$ and total fusion channel $\beta$ is spanned by basis states of the form
\begin{equation}\label{eq:BFCstatespace-npt}
	\ket{\beta_2,\ldots,\beta_{n-1};a_1,\ldots,a_{n-2}}=\Ket{\adjincludegraphics[height=12ex,valign=c]{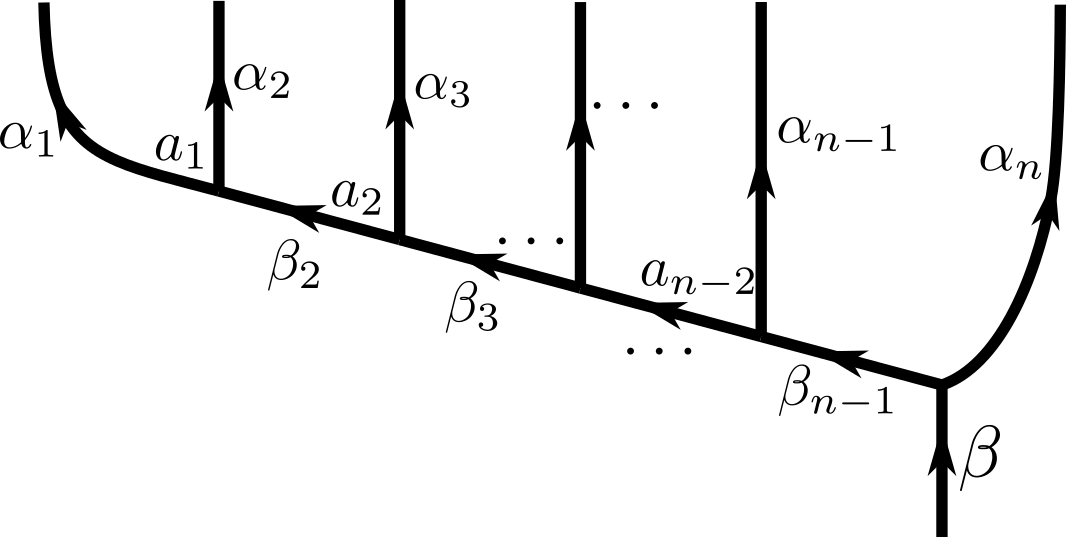}}.
\end{equation}
where the labels $\beta_2,\ldots,\beta_{n-1}$ and $a_1,\ldots,a_{n-2}$ run over all possible allowed values, but the tree structure on the RHS is fixed. 
\paragraph{Locality}
We emphasize that Eqs.~(\ref{eq:fusion_basis}-\ref{eq:F-move}) can all be applied locally to any subpart of a large fusion diagram such as Eq.~\eqref{eq:BFCstatespace}. For example, measuring the particle type operator $\hat{T}$ at position $x_3$ gives
\begin{equation}\label{eq:locality-measure}
	\hat{T}_{x_3}\Ket{\adjincludegraphics[height=9ex,valign=c]{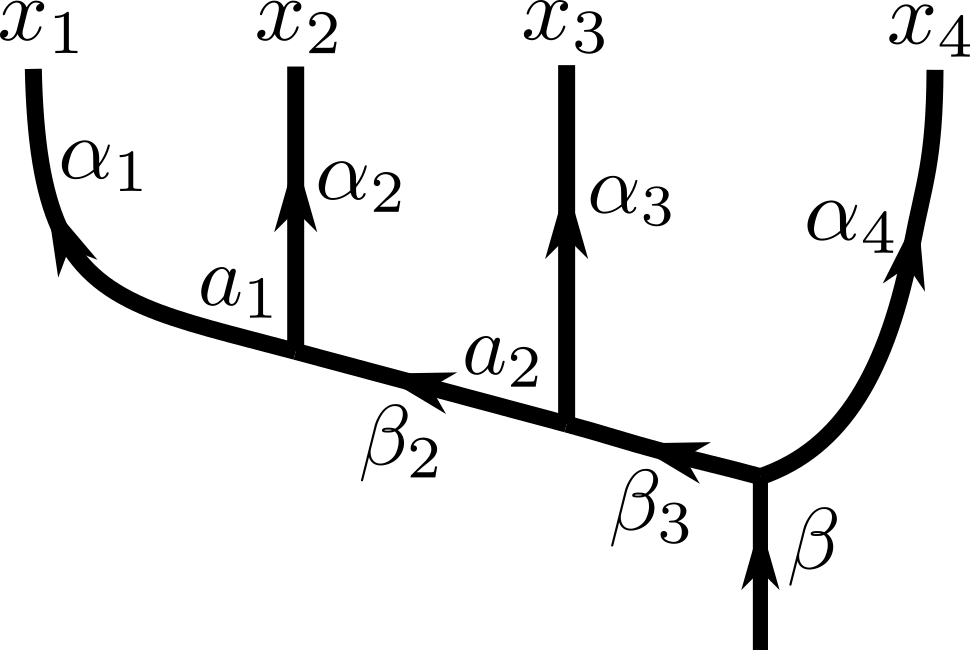}}=
	\alpha_3\Ket{\adjincludegraphics[height=9ex,valign=c]{Figures/FCbasics/BFClocality-1-2.png}}, 
\end{equation}
where we use $x_1,\ldots,x_4$ to explicitly label the positions of the quasiparticles. Similarly, applying the splitting operator $\hat{U}_{x_3,a}^{\alpha_3,\beta\gamma}$ at $x_3$ leads to
\begin{equation}\label{eq:locality-splitting}
	\hat{U}_{x_3,a}^{\alpha_3,\beta\gamma}\Ket{\adjincludegraphics[height=9ex,valign=c]{Figures/FCbasics/BFClocality-1-2.png}}=
	\Ket{\adjincludegraphics[height=9ex,valign=c]{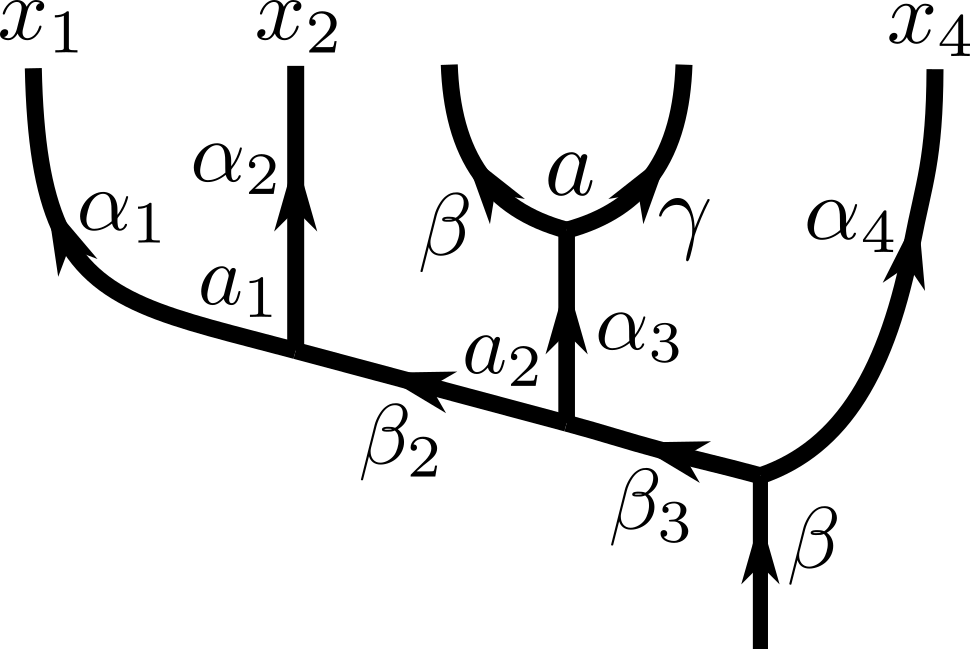}}. 
\end{equation}
The $F$-move in Eq.~\eqref{eq:F-move} can also be applied locally. In this way, two different basis for the space $V^{\alpha_1\alpha_2\ldots\alpha_n}_{\beta}$ constructed from two different fusion trees can be related by a unitary transformation composed of a sequence of $F$-moves. 

\paragraph{Braiding/exchange} 
When we spatially braid~(exchange) two topological quasiparticles in 2+1D~(3+1D) spacetime, the topologically degenerate subspace of excited states in Eq.~\eqref{eq:BFCstatespace-npt} generally undergo a nontrivial unitary evolution governed by the $R$-symbol. The $R$-symbol is defined as follow
\begin{equation}\label{eq:R-move}
	\adjincludegraphics[height=8ex,valign=c]{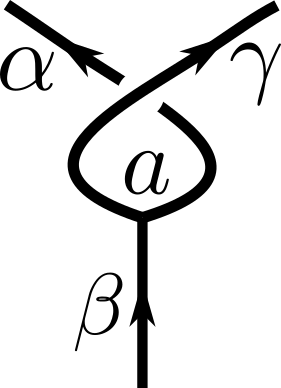}=\sum_{b} [R^{\gamma\alpha}_{\beta}]_{ab}\adjincludegraphics[height=8ex,valign=c]{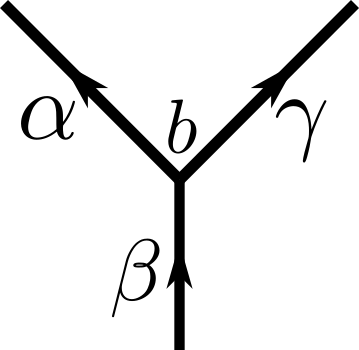},
\end{equation} 
where the LHS represents a counterclockwise braid between $\gamma$ and $\alpha$, and $R^{\gamma\alpha}_{\beta}$ is required to be unitary as a matrix. In addition, $R^{\gamma\alpha}_{\beta}$ needs to satisfy the hexagon equation to guarantee the consistency between fusion and braiding.  The unitary transformation associate to a clockwise braid can be obtained by inverting Eq.~\eqref{eq:R-move}:
\begin{equation}\label{eq:R-move-CW}
	\adjincludegraphics[height=8ex,valign=c]{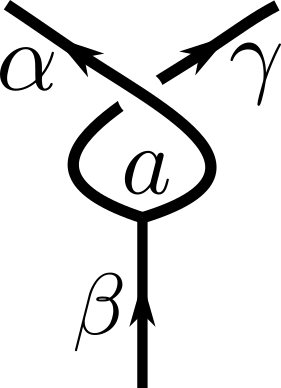}=\sum_{b} [(R^{\alpha\gamma}_{\beta})^{-1}]_{ab}\adjincludegraphics[height=8ex,valign=c]{Figures/FCbasics/braidingBFC-R-2.png}.
\end{equation} 
In 3+1D spacetime, there is no physical distinction between a CW and a CCW exchange, therefore we should impose the condition
\begin{equation}\label{def:SFCvsBFC}
	\sum_{b} [R^{\gamma\alpha}_{\beta}]_{ab} [R^{\alpha\gamma}_{\beta}]_{bc}=\delta_{ac}. 
\end{equation} 
A braided fusion category satisfying Eq.~\eqref{def:SFCvsBFC} is called a SFC, which describes the universal properties of point-like excitations in 3+1D gapped phases. 

We warn the readers not to confuse the $R$-symbol and the $R$-matrix: the former $[R^{\gamma\alpha}_{\beta}]_{ab}$ has five indices, three of which $\gamma,\alpha,\beta$ are particle type labels, while the latter $R^{b'a'}_{ab}$ has four indices. Their precise relation will be given in Sec.~\ref{sec:win1pt-SFC}, specifically in Eq.~\eqref{eq:RmatfromFRmove}. 

\paragraph{Time evolution and transition amplitude}
Consider a typical physical process involving motion of topological quasiparticles, as 
depicted in Fig.~\ref{fig:ta-3D}. Here we assume that all the particles are far separated, and move sufficiently slowly so that the adiabatic theorem holds. 
To compute the transition amplitude of such a process, we project the spacetime trajectories of the particles onto a 2D plane to obtain the fusion diagram in Fig.~\ref{fig:ta}. We then complete the top and bottom part of the diagram by fusing all the particles involved in this process, as shown in Fig.~\ref{fig:ta}, which corresponds to choosing a suitable fusion basis for the topologically degenerate state space of the initial and final particle configurations. Finally we evaluate this diagram using the diagrammatic rules presented in this section, the details can be found in standard textbooks~\cite{simon2023topological}. 
This gives the topological part of the transition amplitude. In general, the transition amplitude also involves a non-universal dynamical phase factor which depends on the microscopic details of the system. Such a dynamical phase factor does not affect our analysis of the winning strategies, therefore we ignore this non-universal part in this paper.
\begin{figure}
	\begin{subfigure}[t]{.68\linewidth}
		\centering\includegraphics[width=.85\linewidth]{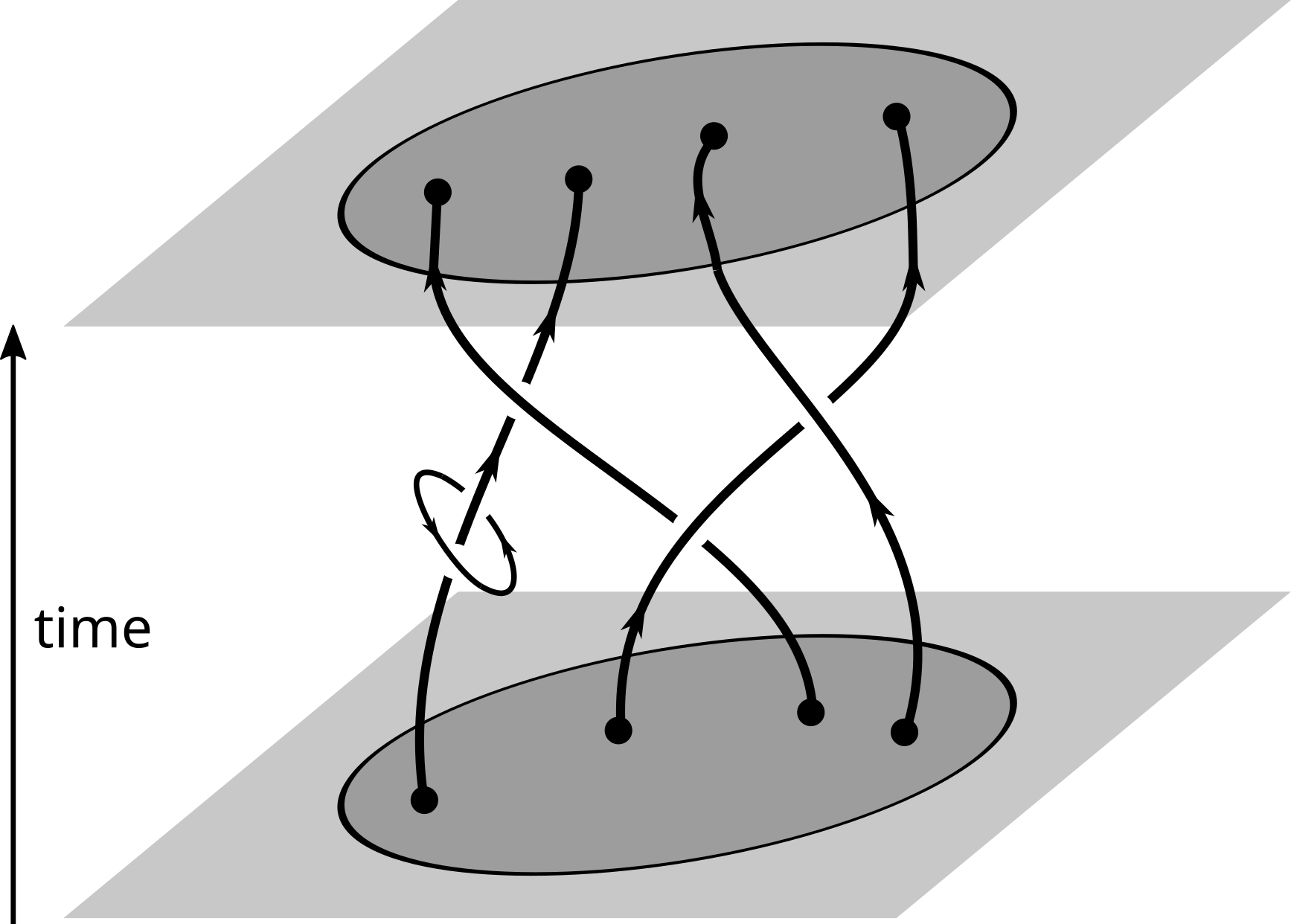} 
		\caption{\label{fig:ta-3D}}
	\end{subfigure}
	\begin{subfigure}[t]{.30\linewidth}
		\centering\includegraphics[width=.85\linewidth]{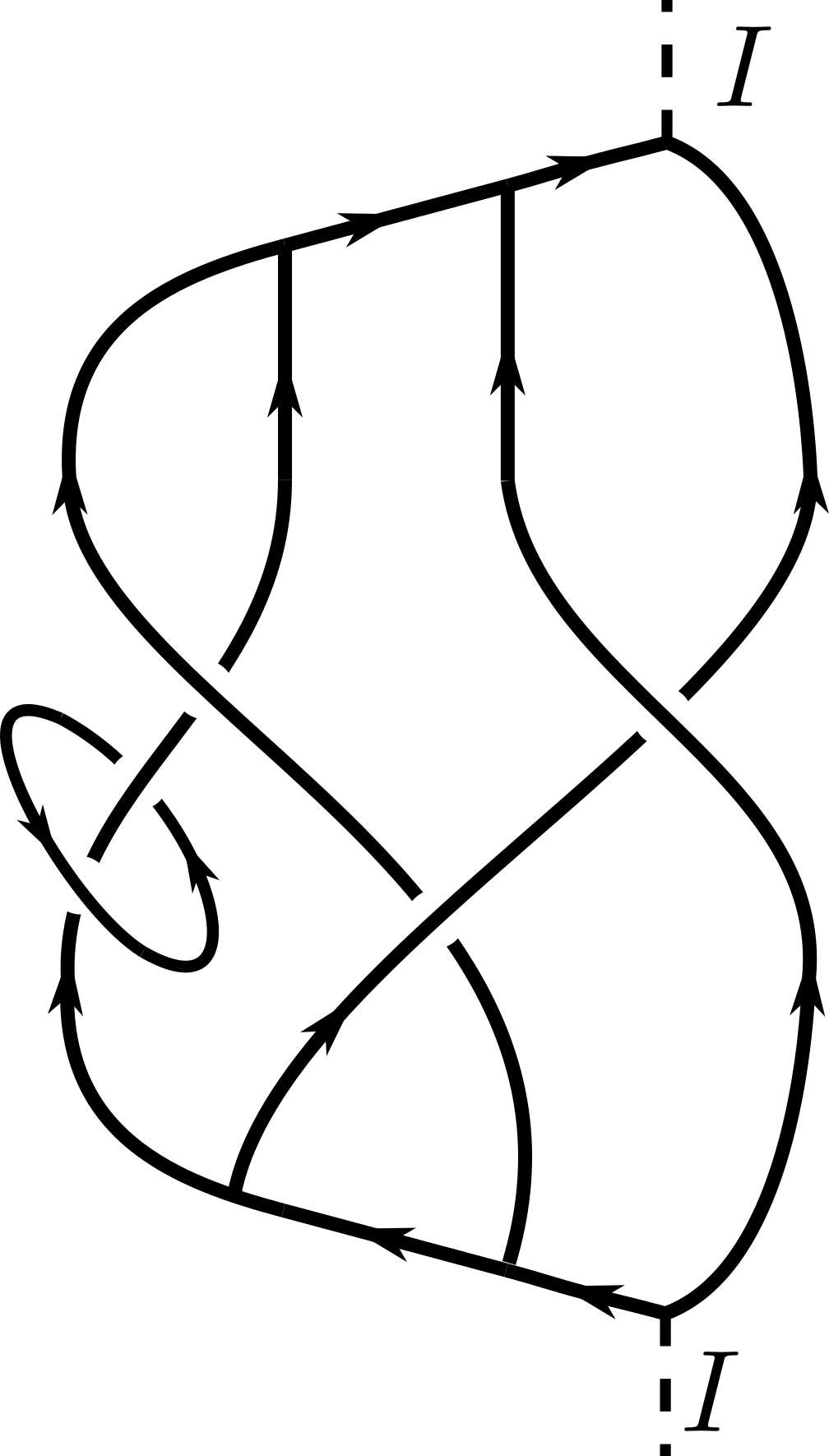} 
		\caption{\label{fig:ta}}
	\end{subfigure}
	\caption{\label{fig:transition_amplitude} 
		Computing transition amplitude by evaluating fusion diagrams. 
		(a) A typical quantum dynamical process involving adiabatic motion of quasiparticles, where we assume all particles to be far separated; (b) The fusion diagram obtained by projecting the spacetime trajectories to a 2D plane, and completing the top and bottom of the diagram by fusing all the particles involved in this process~(here we assume that the total fusion channel is the vacuum $I$, the general case can be treated in an identical way). Evaluation of this diagram gives the topological part of the transition amplitude. 
	}
\end{figure}

\subsection{Diagrammatic representation of the winning strategies}\label{sec:win_strategy_diagrammatic}
As a first application of the tensor categorical framework, in this section we provide an alternative description of the winning strategies to the challenge games %
using fusion diagrams. Compared to our previous description~(Sec.~\ref{sec:win_original}) using the axioms of emergent parastatistics, this categorical description is more general as it can be applied to strategies using non-Abelian anyons~(for the 2D version of the game) as well, and it allows us to visualize the entire game process in a single space-time diagram.  %

In the current section we only give an abstract description of the special class of SFCs that can pass the challenge, and concrete examples of this type of SFCs will be given in App.~\ref{sec:RfromCentralType}. In the next section we will show that this class of SFCs is essentially the only class that can pass the full version of the challenge in 3+1D. The generalization of this diagrammatic analysis to non-Abelian-anyon-based winning strategies for the 2D version of the game will be given in Sec.~\ref{sec:anti-anyon}. %
\subsubsection{The basic challenge with one special point}\label{sec:win1pt-SFC}
We begin by categorically describing the winning strategy for the version with one special point analyzed in Sec.~\ref{sec:win_1pt}, as this turns out to be the simplest case. 
Consider a SFC $\calC$ with a fusion rule of the following form%
\begin{equation}\label{eq:sigmapsifusion-0}
	\sigma \times \psi  =m~\sigma,
\end{equation}
where $m>1$ is the fusion multiplicity, and $\sigma,\psi\in \calC$ are simple objects.  
In the diagrammatic description, such a fusion rule %
leads to a fusion vertex of the form $\adjincludegraphics[height=5ex,valign=c]{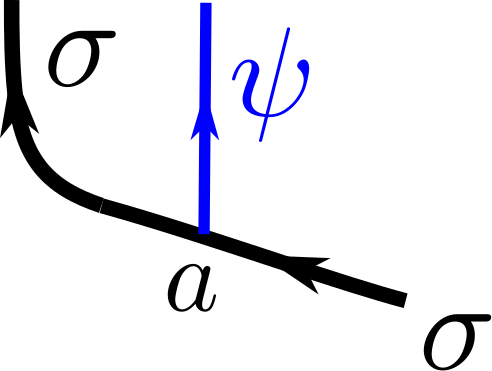}$, where $a=1,2,\ldots,m$ labels basis vectors in the fusion space $V^{\sigma\psi}_{\sigma}$, which we previously called the internal space of the paraparticles. %
Physically, this vertex describes a splitting process that can be implemented by a unitary operation $\hat{U}_a^{\sigma,\sigma\psi}$ in Eq.~\eqref{eq:U_spltting}, which is localized around $\sigma$. 
During the pregame preparation stage, the players submit a physical system whose ground state $\ket{G}$ already has a $\sigma$ particle at some point $\oA$, which they choose as the special point to begin and end their journey. Such a system can still satisfy all the requirements of the game, which we discuss later in Sec.~\ref{eq:BFC-groundstate}. 
Importantly, with this configuration the players can locally create a single paraparticle $\psi$ at $\oA$, using the local operator $\hat{U}_a^{\sigma,\sigma\psi}$ in Eq.~\eqref{eq:U_spltting}, which realizes 
the unitary operator $\hat{U}_{\oA,a}$ in Axiom~\ref{Axiom5}. %

We claim that the winning strategy  in Sec.~\ref{sec:win_1pt} is described by the following fusion diagram:
\begin{equation}\label{def:RfromSFC}
	\adjincludegraphics[height=15ex,valign=c]{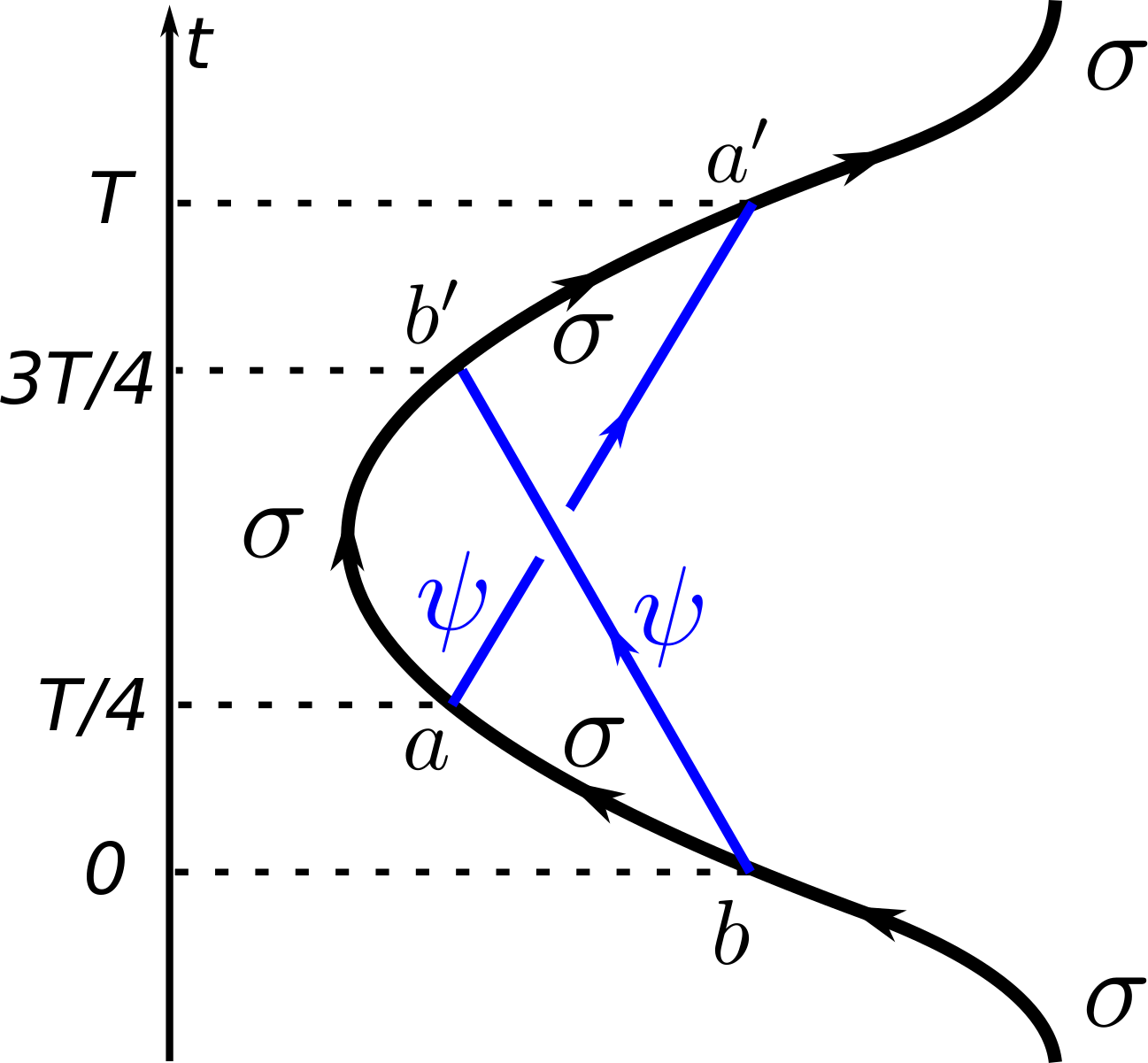}
	=R^{b'a'}_{ab}\quad\adjincludegraphics[height=12ex,valign=c]{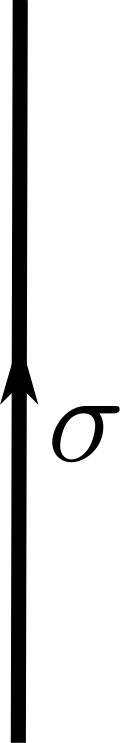}.
\end{equation}
Here the LHS describes the time evolution of the physical system. 
The trivalent vertex at $t=0$ means that Bob applies the local operator $\hat{U}_b^{\sigma,\sigma\psi}$  to create a paraparticle at $\oA$ and encode his number $b$ in the fusion space~(which we previously called the internal space). When his circle moves away from $\oA$, he keeps the paraparticle $\psi$ in the circle while leaving $\sigma$ unchanged at $\oA$.
Importantly,  the $\sigma$ particle at $\oA$ is considered to be a part of the background~(since it is already present in the ground state $\ket{G}$ submitted by the players), so leaving the $\sigma$ at $\oA$ does not violate the rules of the game~(as this does not locally leave any information behind). 
The trivalent vertex at $t=3T/4$  represents a  measurement in the fusion space $V^{\sigma\psi}_{\sigma}$ performed by Bob, using the local operator $\hat{O}$ in Eq.~\eqref{eq:measure_fusion}~(which realizes the observable $\hat{O}_{\oA}$ in Axiom~\ref{Axiom6}). The other two trivalent vertices at $t=T/4$ and $t=T$ represents the corresponding operations performed by Alice. Note that throughout the entire game the $\sigma$ particle sits at $\oA$ without ever changing its state.

The RHS of Eq.~\eqref{def:RfromSFC} evaluates the fusion diagram using te diagrammatic rules in Sec.~\ref{sec:categorical_framework}. 
The transition amplitude $R^{b'a'}_{ab}$ is exactly the $R$-matrix of the paraparticle $\psi$, and in App.~\ref{app:deriveYBE} we show that $R^{b'a'}_{ab}$ defined by Eq.~\eqref{def:RfromSFC} always satisfies the YBE. 
We can derive an explicit expression for the $R$-matrix by applying the $F$-move in Eq.~\eqref{eq:F-move} and the $R$-move in Eq.~\eqref{eq:R-move}:
\begin{equation}\label{eq:RmatfromFRmove}
	R^{b'a'}_{ab}=\sum_{\beta,e,f,g}[F^{\sigma\psi\psi}_{\sigma}]^{\sigma c d }_{\beta e f} [R^{\psi\psi}_\beta]_{eg} [F^{\sigma\psi\psi*}_{\sigma}]^{\sigma a b}_{\beta g f}.
\end{equation}
However, we will not actually use this complicated expression in this paper, as in App.~\ref{sec:RfromCentralType} we present a more convenient way to compute the $R$-matrix directly from group-theoretical data.  %

\subsubsection{The basic challenge with two special points}\label{sec:win2pt-SFC}
It is straightforward to show~(see App.~\ref{app:proofdualFrule}) that  if $\calC$ has a fusion rule of the form  in Eq.~\eqref{eq:sigmapsifusion-0},  $\calC$  must also have a fusion rule
\begin{equation}\label{eq:sigmapsifusion-1}
	\psi\times \bar{\sigma}=m~\bar{\sigma},
\end{equation}
meaning that $\psi$ can also be locally created in the vicinity of $\bar{\sigma}$. The winning strategy for the challenge game with two special points $\oA$ and $\oB$ is described by the following diagram
\begin{equation}\label{eq:SFCdescriptiongame}
	\adjincludegraphics[height=14ex,valign=c]{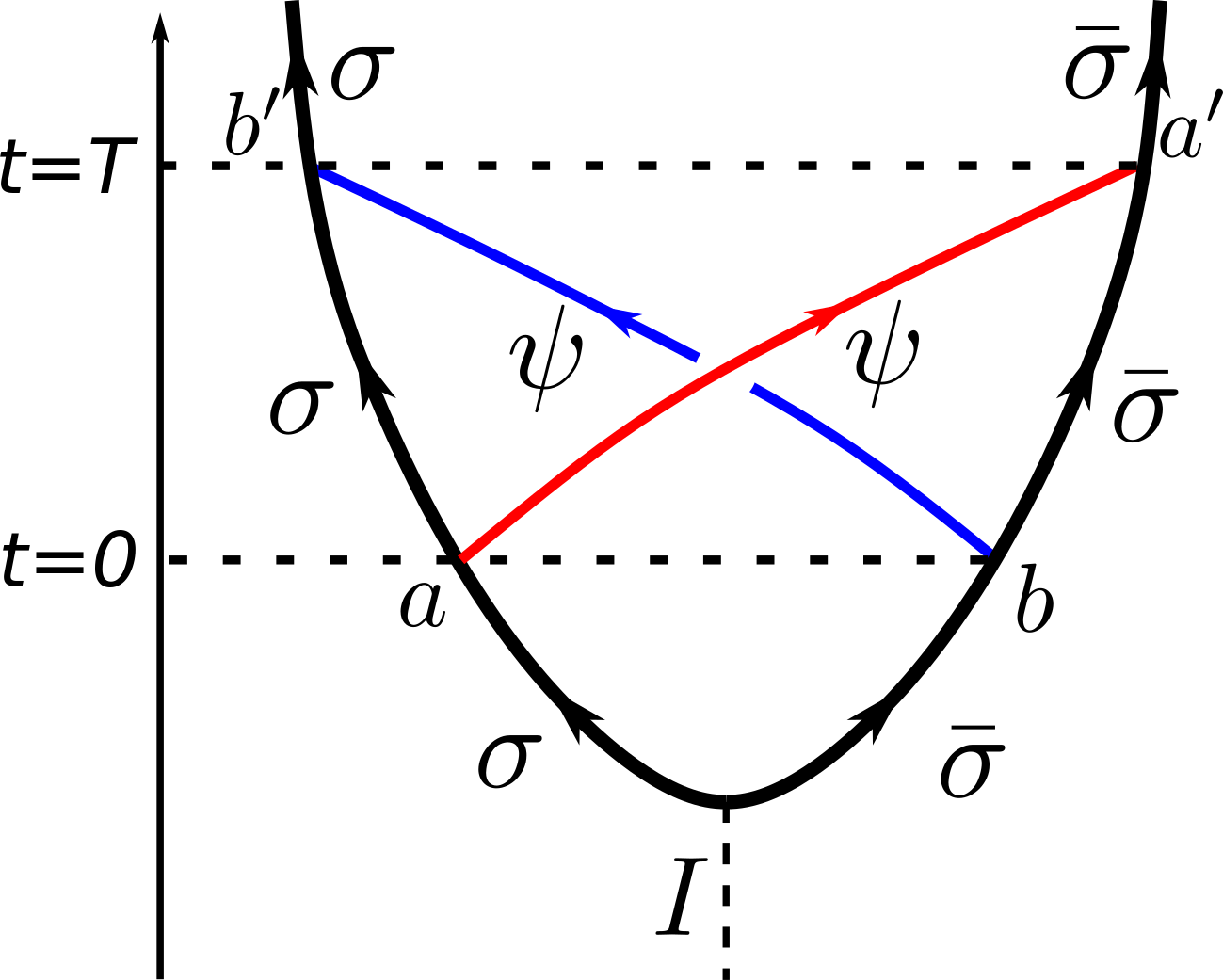}=
	R^{b'a'}_{ab}~\adjincludegraphics[height=14ex,valign=c]{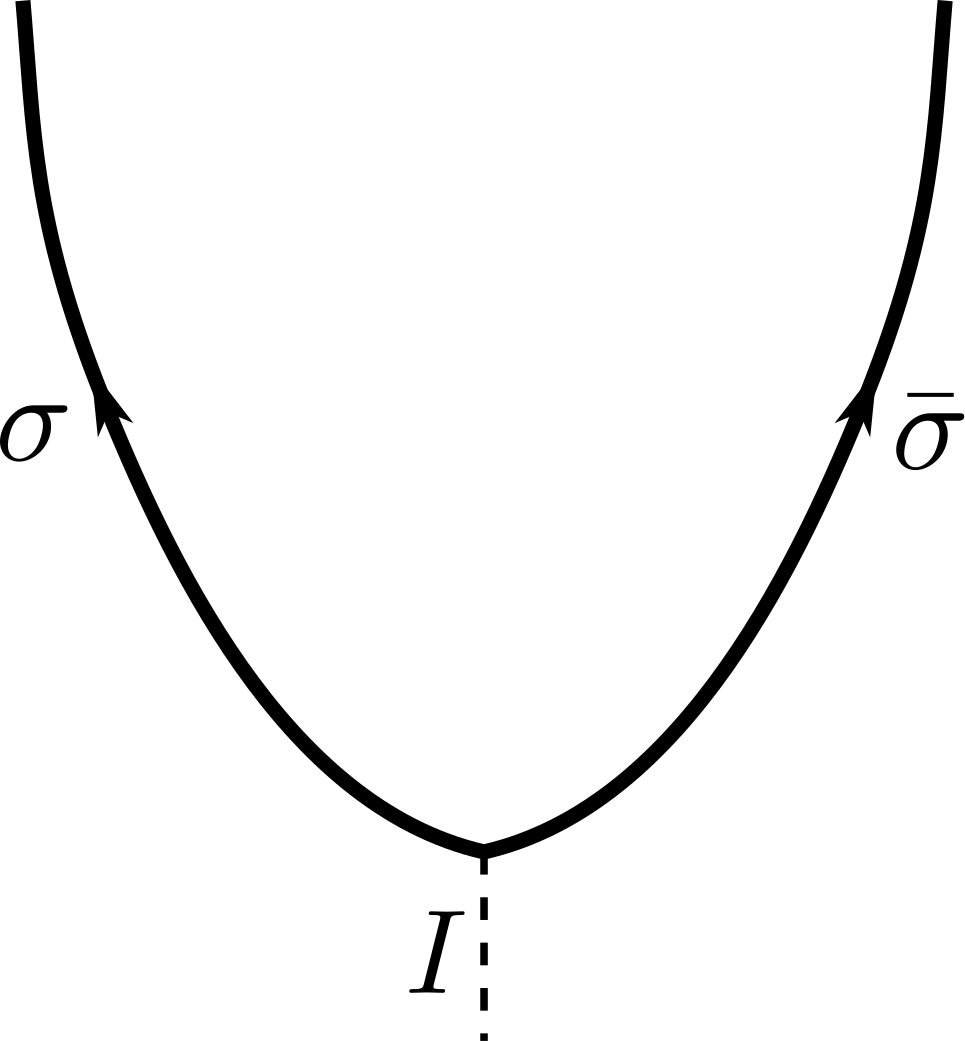},
\end{equation}
where, as before, the LHS describes the evolution of the system and the RHS evaluates the fusion diagram. Not surprisingly, exactly the same $R$-matrix appears in the RHS as in Eq.~\eqref{def:RfromSFC}. Indeed, it is straightforward to show that Eq.~\eqref{eq:SFCdescriptiongame} is equivalent to Eq.~\eqref{def:RfromSFC}, since we have
\begin{equation}\label{eq:SFC2pt-1pt}
	\adjincludegraphics[height=14ex,valign=c]{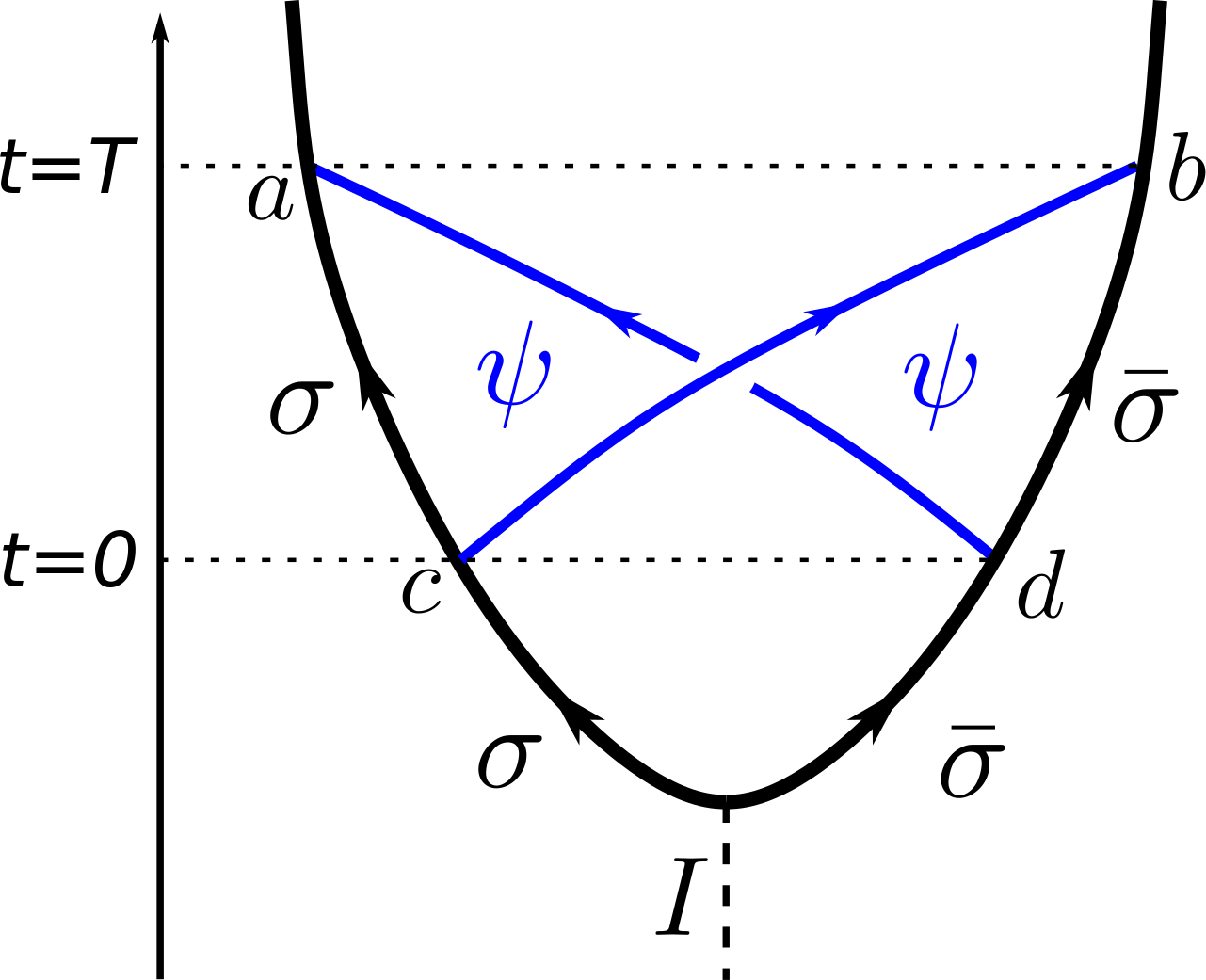}=
	\sum_{e,f}\adjincludegraphics[height=14ex,valign=c]{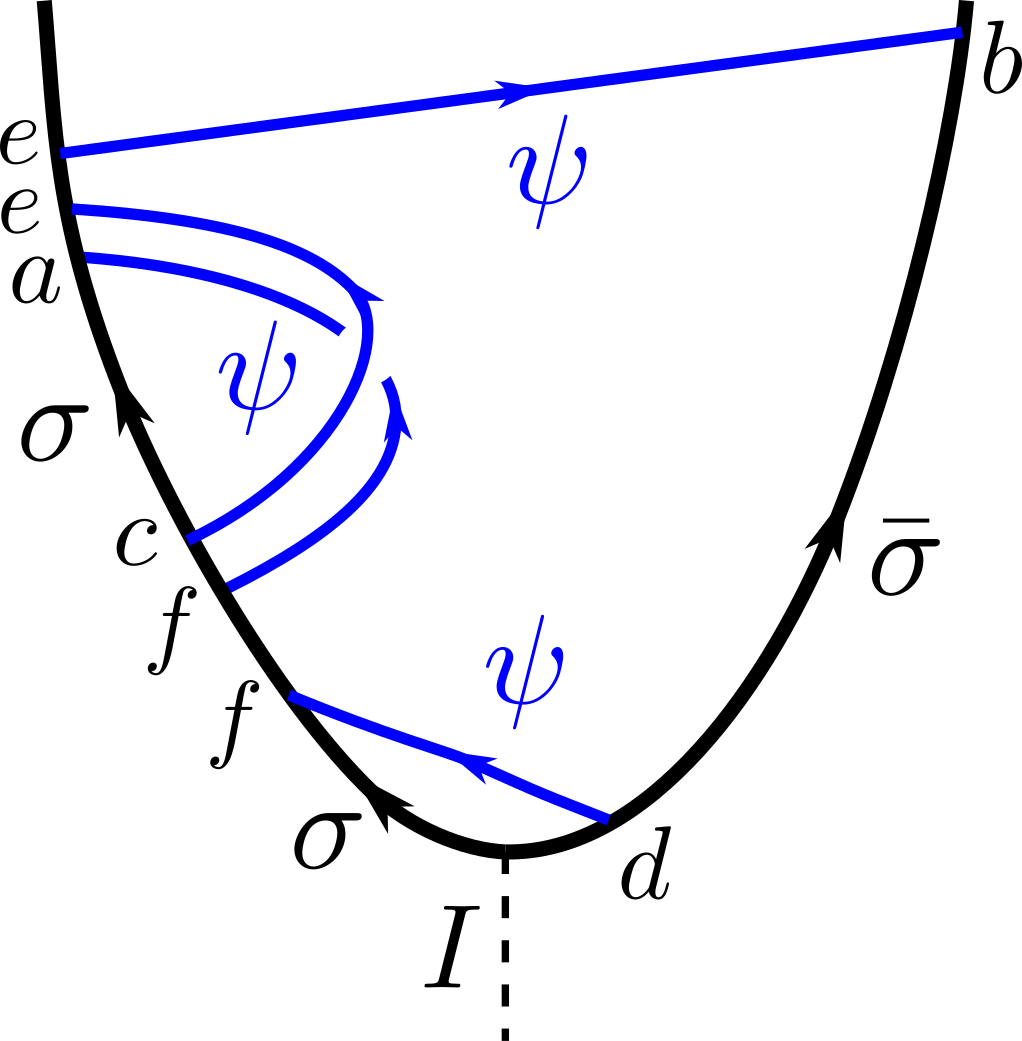}=\adjincludegraphics[height=14ex,valign=c]{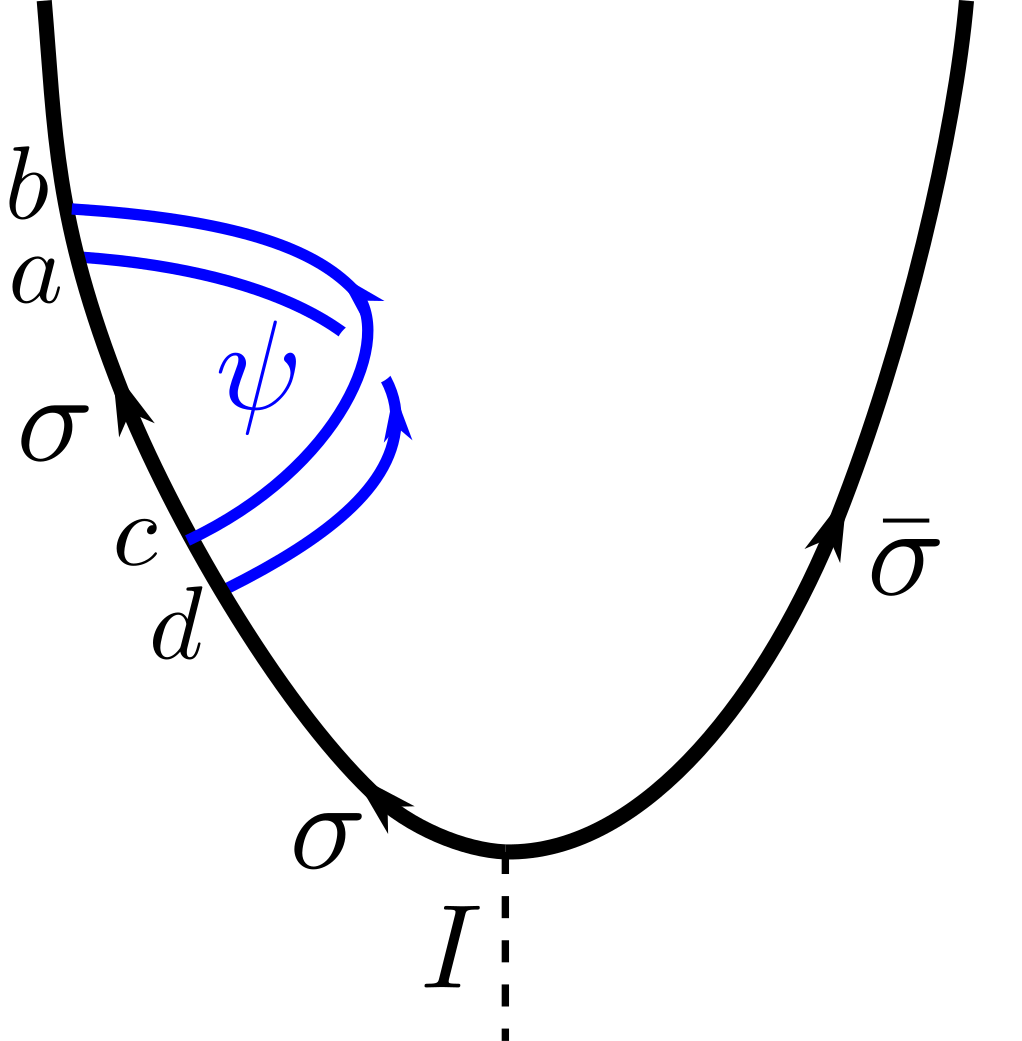},
\end{equation}
where in the last step we used the following convention
\begin{equation}\label{eq:SFC-longstring}
	\adjincludegraphics[height=14ex,valign=c]{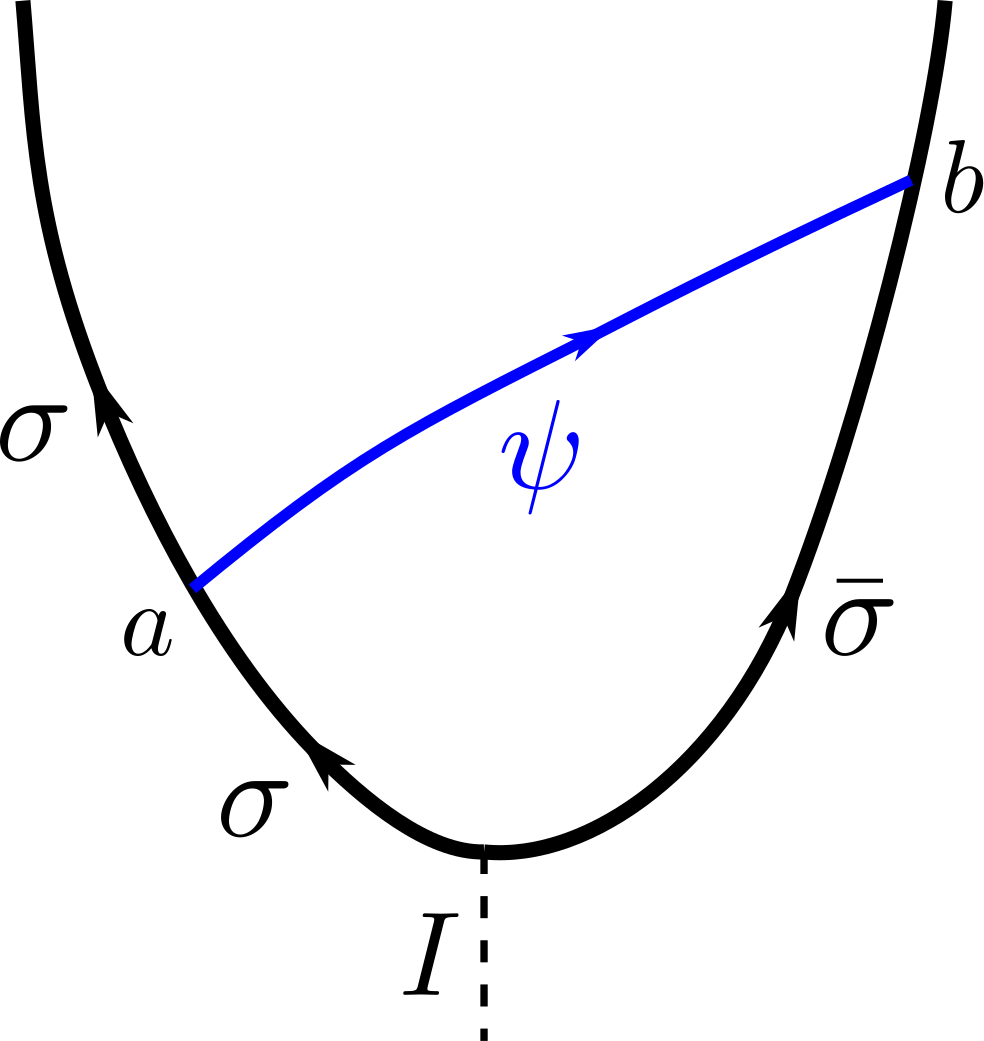}=
	\adjincludegraphics[height=14ex,valign=c]{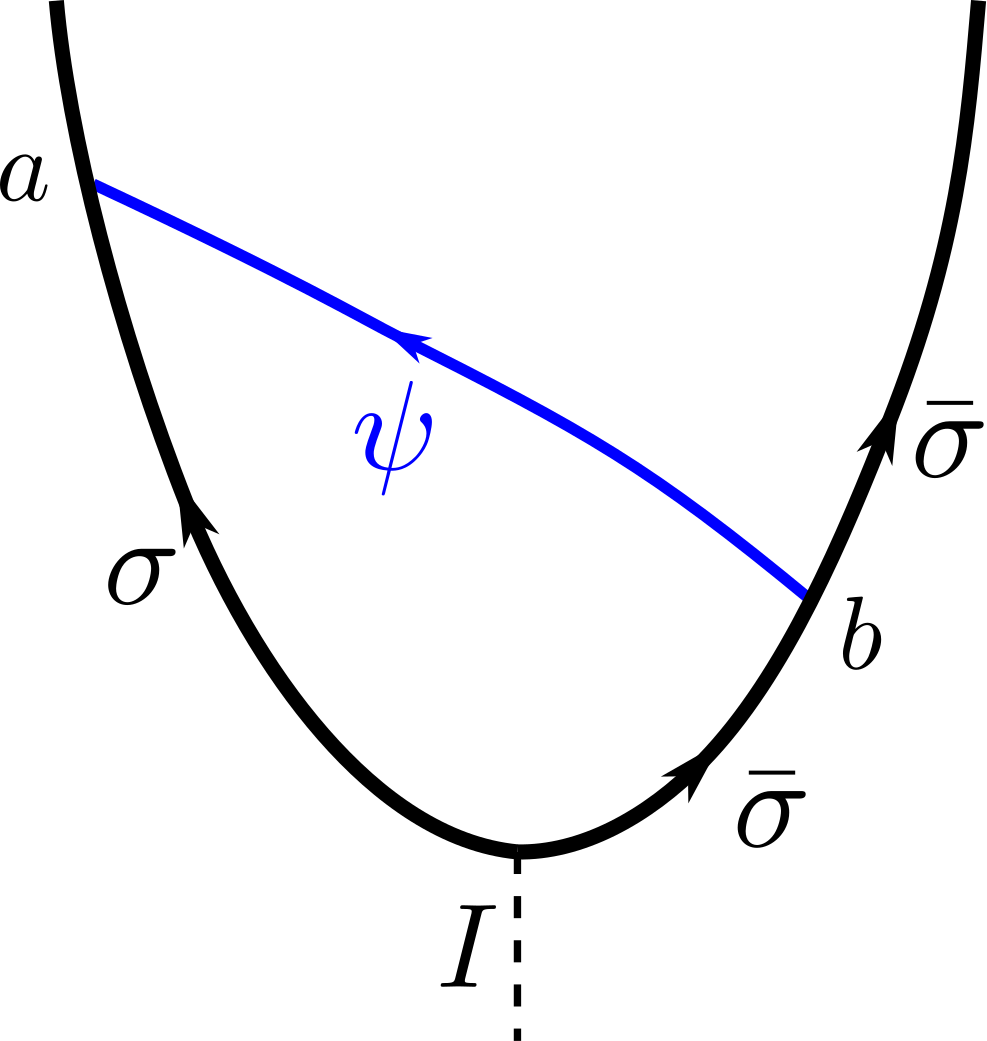}=
	\delta_{ab}\adjincludegraphics[height=14ex,valign=c]{Figures/SFC/SFCR-ssb-RHS.png},
\end{equation}
which can always be guaranteed by choosing a suitable basis for the fusion space of the paraparticle $\psi$ with $\bar{\sigma}$.

\subsubsection{The who-entered-first challenge}
In a similar fashion, the winning strategy for the who-entered-first challenge~(given in Sec.~\ref{sec:winWEF}) is described by the following diagram
\begin{equation}\label{eq:SFCdescriptionWHF}
	\adjincludegraphics[height=18ex,valign=c]{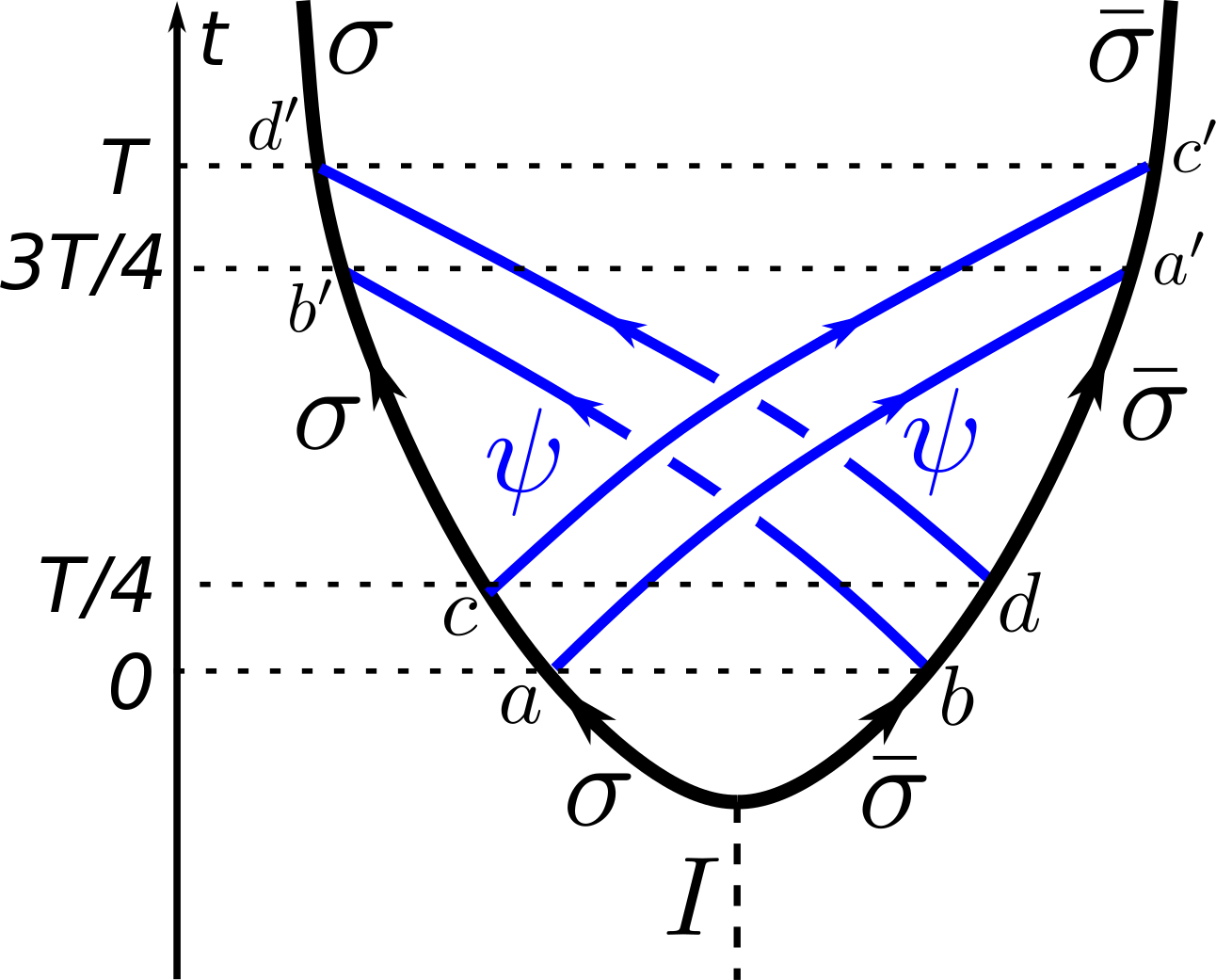}=
	\begin{tikzpicture}[baseline={([yshift=-.8ex]current bounding box.center)}, scale=0.5]
		\Rmatrix{0}{2*\AL}{R}
		\Rmatrix{0}{-2*\AL}{R}
		\Rmatrix{2*\AL}{0}{R}
		\Rmatrix{-2*\AL}{0}{R}
		\node  at (-\AL,-3*\AL) [below=-.04]{\footnotesize $a$};
		\node  at (\AL,-3*\AL) [below=-.04]{\footnotesize $b$};
		\node  at (-3*\AL,-\AL) [below=-.04]{\footnotesize $c$};
		\node  at (3*\AL,-\AL) [below=-.04]{\footnotesize $d$};
		\node  at (\AL,3*\AL) [above=-.04]{\footnotesize $c'$};
		\node  at (-\AL,3*\AL) [above=-.04]{\footnotesize $d'$};
		\node  at (3*\AL,\AL) [above=-.04]{\footnotesize $a'$};
		\node  at (-3*\AL,\AL) [above=-.04]{\footnotesize $b'$};
	\end{tikzpicture}\adjincludegraphics[height=14ex,valign=c]{Figures/SFC/SFCR-ssb-RHS.png},
\end{equation}
where the LHS describes the physical process shown in Fig.~\ref{fig:WHF}, and the RHS is the evaluation of the fusion diagram. 
Here we use a derivation similar to that in Eq.~\eqref{eq:SFC2pt-1pt} and the categorical definition of the $R$-matrix in Eq.~\eqref{def:RfromSFC}. This leads to the same tensor network of $R$-matrices as in Eq.~\eqref{eq:Psi3-WHF}. 

\subsubsection{*Winning strategy for the antiparticle test}\label{sec:win-APtest}
We now describe the winning strategy for the antiparticle test introduced in Sec.~\ref{sec:antiparticle}. 
Since this test is a small twist to the basic challenge, most part of the winning strategy is the same as given in Sec.~\ref{sec:win_2pt}, so in the following we focus the difference.\\ %
(1). At $t=0$, Bob creates a paraparticle $\psi$ in his circle and stores his number in its internal state as before; Alice creates a pair of $\psi$ and $\bar{\psi}$ in the overlapping region of circles A and C, where $\bar{\psi}$ is the antiparticle of $\psi$. 
For the rest of the game, Alice keeps $\psi$ in circle A, while Carol keeps $\bar{\psi}$ in circle C;\\
(2). At $t=t_1$, as shown in Fig.~\ref{fig:Exchange-bulk-APtest-2}, Carol measures the internal state of her paraparticle %
and reports the result $c$ to Alice. Then she %
annihilates her paraparticle and leaves the game;\\
(3). The rest of the strategy is essentially the same as in Sec.~\ref{sec:win_2pt}.
In the end Alice measures the internal state of her particle at $\oB$ and obtains $a'$, and then clean up. \\
The fusion diagram for the entire process is
\begin{equation}\label{eq:SFC-APtest}
	\adjincludegraphics[height=14ex,valign=c]{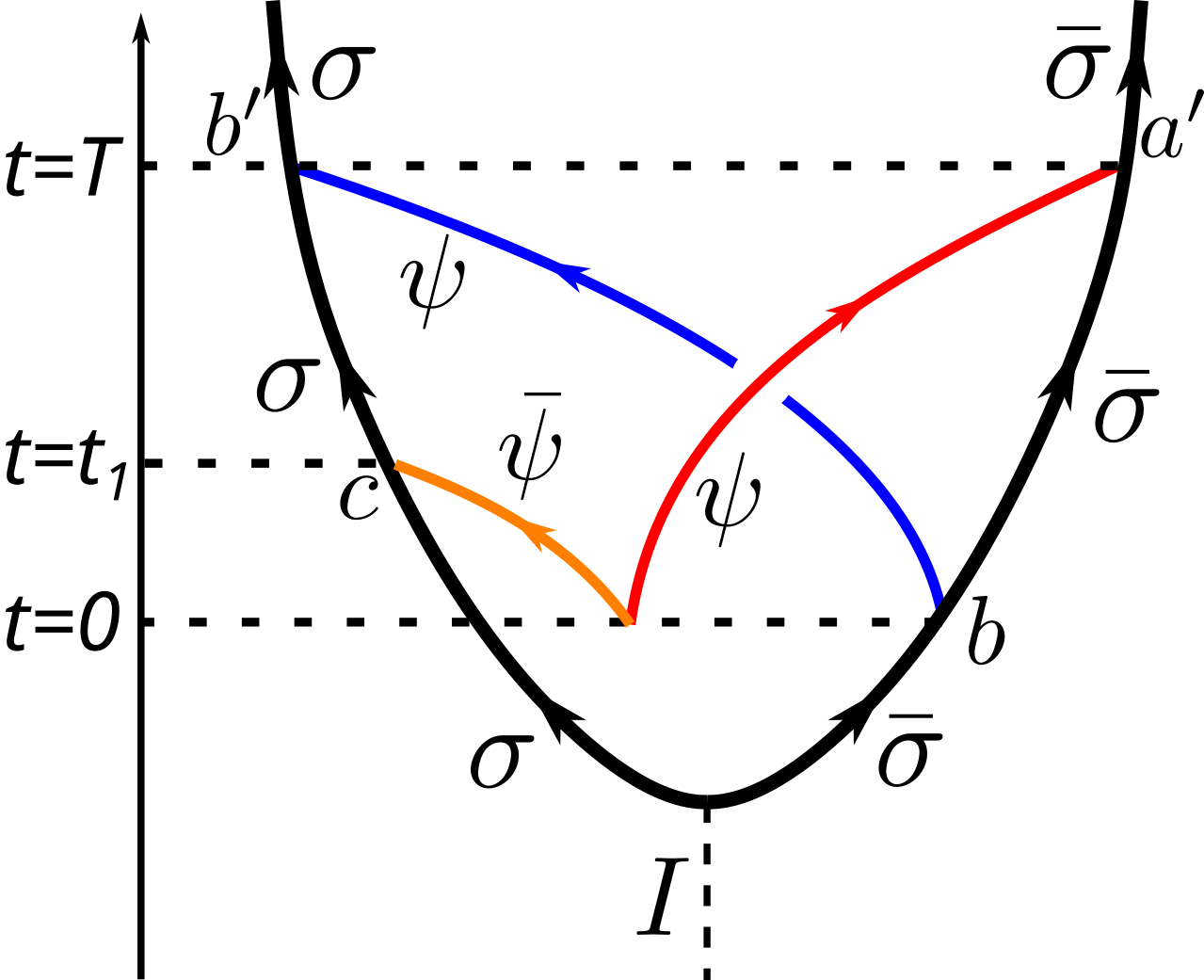}=
	\sum_a\kappa_{ca} R^{b'a'}_{ab}~\adjincludegraphics[height=14ex,valign=c]{Figures/SFC/SFCR-ssb-RHS.png},
\end{equation}
where %
we used the following $F$-move:
\begin{equation}
	\adjincludegraphics[height=8ex,valign=c]{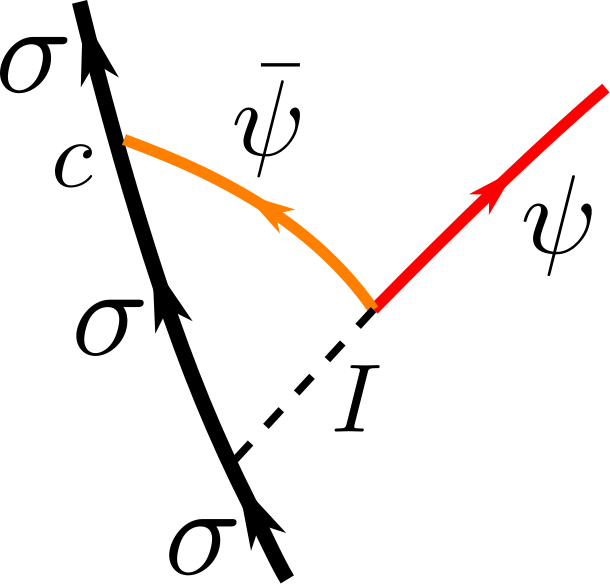}=
	\sum_c [F^{\sigma\bar{\psi}\psi}_\sigma]^*_{\sigma ac;I}~\adjincludegraphics[height=8ex,valign=c]{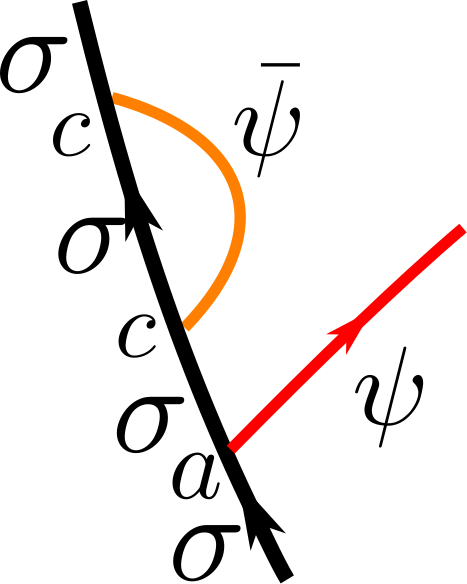}
	\equiv\sum_c \kappa_{ac}\adjincludegraphics[height=8ex,valign=c]{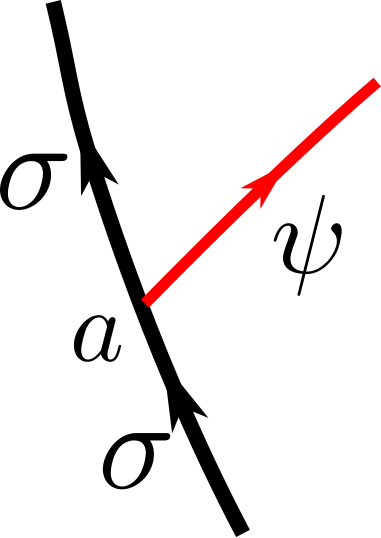}.
\end{equation}
Since the players know $\kappa$ and $R$ beforehand, as long as $R$ is nontrivial, Alice can obtain information about $b$ using $a'$ and $c$, the algorithm is similar to the one for the original challenge. With this strategy, all nontrivial $R$-paraparticles of hierarchies 4 and 5 can win this challenge in a noise-robust way. For hierarchy 3, we can combine the above strategy and the one in Sec.~\ref{sec:phaseSWAP}, to obtain a strategy for this challenge that is robust against local noise but not robust against eavesdropping. For hierarchy 2~(emergent fermions), we also have a fragile strategy for this challenge similar to the one in Sec.~\ref{sec:emergent_fermion}. 
So this test does not affect the main pattern in Tab.~\ref{tab:protocols_strategies}. We will explain why we introduce this test in Sec.~\ref{sec:win_APtest}. 

We mention that it is also possible to perform this analysis using a small extension to the axioms in Sec.~\ref{sec:axioms_emergent_para} involving pair creation of $R$-paraparticles; however,  we will not present these extended axioms in this paper. 

\subsection{Deriving the necessary condition for success}\label{sec:TC_winning_condition}
We assume that point-like excitations in the  3D~(2D) condensed matter system prepared by the players are described by a symmetric~(braided) fusion category $\mathcal{C}$, and derive necessary conditions on $\mathcal{C}$ for the players to be able to win the game. 

\subsubsection{The ground state $\ket{G}$}\label{eq:BFC-groundstate}
We begin by analyzing the ground state  $\ket{G}$ prepared by the players. %
In this section we consider Hamiltonian  of the form $\hat{H}=\hat{H}_0+\hat{h}_{\oA}+\hat{h}'_{\oB}$, where $\hat{H}_0$ is translationally invariant~(defined on a topologically trivial manifold), and  $\hat{h}_{\oA}$, $\hat{h}'_{\oB}$ are local Hermitian operators localized around $\oA$, $\oB$, respectively. Due to the frustration-free and the uniqueness condition on the ground state $\ket{G}$ of $\hat{H}$, %
$\ket{G}$ is locally isomorphic to the ground state $\ket{G_0}$ of $\hat{H}_0$ everywhere away from $\oA$, $\oB$. At the special points  $\oA$ and $\oB$, $\ket{G}$ may potentially have some point-like excitations~\footnote{Here $\ket{G}$ is considered as an excited state of $\hat{H}_0$ above the ground state $\ket{G_0}$. } that are trapped by the local potentials $\hat{h}_{\oA}$, $\hat{h}'_{\oB}$, respectively. Let $\sigma$ and $\sigma'$ be the types of the quasiparticles at $\oA$ and  $\oB$, respectively. In the following we show that both $\sigma$ and $\sigma'$ must be simple such that $\sigma'=\bar{\sigma}$, using the fact that $\ket{G}$ is the unique gapped ground state of the locally interacting Hamiltonian $\hat{H}$. 

Let $\hat{T}_{\oA}$ and $\hat{T}_{\oB}$ be the local observables that measure the particle types at positions $\oA$ and $\oB$, respectively, as defined by Eq.~\eqref{eq:measure_type}. $\hat{T}_{\oA}$ and $\hat{T}_{\oB}$ have the following decomposition
\begin{equation}
	\hat{T}_{\oA}=\sum_{\beta\in\Irr(\calC)} \beta \hat{\Pi}^\beta_\oA,\quad \hat{T}_{\oB}=\sum_{\beta\in\Irr(\calC)} \beta \hat{\Pi}^\beta_\oB,
\end{equation}
where $\hat{\Pi}^\beta_\oA$ is a local projection operator for the particle type $\beta$ at position $\oA$, and similarly for $\hat{\Pi}^\beta_\oB$. The global superselection rule of $\hat{H}_0$ on a closed manifold without boundary requires that all topological excitations in the system must fuse into vacuum, which enforces the following equality for the state $\ket{G}$
\begin{equation}\label{eq:global_superselection}
	\hat{\Pi}^\beta_\oA \hat{\Pi}^\varphi_\oB\ket{G}=0,%
\end{equation}
for all simple types $\beta,\varphi \in\Irr(\calC)$, unless $\beta=\bar{\varphi}$. Since $\ket{G}$ is the unique gapped ground state of the locally interacting Hamiltonian $\hat{H}$, the exponential clustering theorem~\cite{hastings2006,nachtergaele2006} claims that for any local observables $\hat{O}_{\oA}$ and $\hat{O}'_{\oB}$, we have
\begin{equation}
	\braket{G|\hat{O}_{\oA}\hat{O}'_{\oB}|G}=\braket{G|\hat{O}_{\oA}|G}\braket{G|\hat{O}'_{\oB}|G}+O(e^{-d_{\oA\oB}/\xi}),
\end{equation}
where $d_{\oA\oB}$ is the distance between $\oA$ and $\oB$. Now we take $\hat{O}_{\oA}=\hat{\Pi}^\beta_\oA, \hat{O}_{\oB}=\hat{\Pi}^\varphi_\oB$, for some simple types $\beta,\varphi \in\Irr(\calC)$,  and take $d_{\oA\oB}$ to be sufficiently large so that we can ignore the exponentially small correction. We have 
\begin{equation}\label{eq:expclusteringpipi}
	\braket{G|\hat{\Pi}^\beta_\oA\hat{\Pi}^\varphi_\oB|G}=\braket{G|\hat{\Pi}^\beta_\oA|G}\braket{G|\hat{\Pi}^\varphi_\oB|G}.
\end{equation}
The expectation value $\braket{G|\hat{\Pi}^\beta_\oA|G}$ is nonzero whenever $\beta\in\sigma$, i.e. the simple type $\beta$ appears at least once in the decomposition of $\sigma$, and similarly for $\braket{G|\hat{\Pi}^\varphi_\oB|G}$. 
If either $\sigma$ or $\sigma'$ is not simple, we can find $\beta\in\sigma$ and $\varphi\in\sigma'$ such that $\beta\neq\bar{\varphi}$. For such a pair $(\beta,\varphi)$, the RHS of Eq.~\eqref{eq:expclusteringpipi} is nonzero while the LHS is zero, a contradiction. 
In summary, in the ground state $\ket{G}$ submitted by the players, the particle types $\sigma$ and $\sigma'$ at the two special points $\oA$ and $\oB$ must both be simple, with $\sigma'=\bar{\sigma}$, such that they fuse into the vacuum.  Such a state can be prepared from the translationally invariant state $\ket{G_0}$~(which has no quasiparticles anywhere) by creating a pair $\sigma \bar{\sigma}$ at some point using Eq.~\eqref{eq:pair_creation} and then move $\sigma$ to $\oA$ and $\bar{\sigma}$ to $\oB$. 
\subsubsection{The necessary condition for bidirectional secret communication}\label{sec:deriv_nece_cond}
\begin{figure}
	\begin{subfigure}[t]{.48\linewidth}
		\centering\includegraphics[width=.7\linewidth]{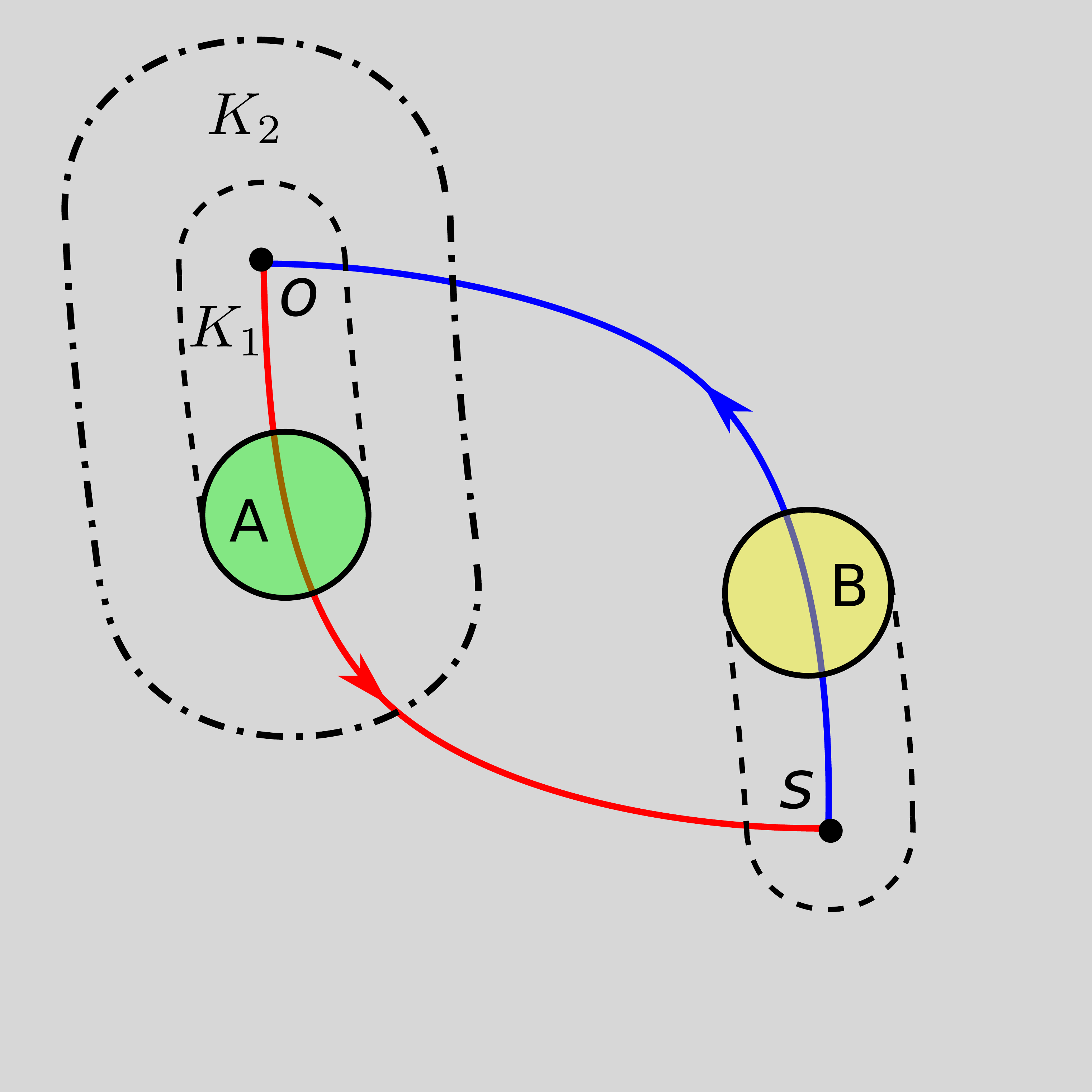} 
		\caption{\label{fig:ExchangeBootstrap-T-4} $t=t_1$}%
	\end{subfigure}
	\begin{subfigure}[t]{.48\linewidth}
		\centering\includegraphics[width=.7\linewidth]{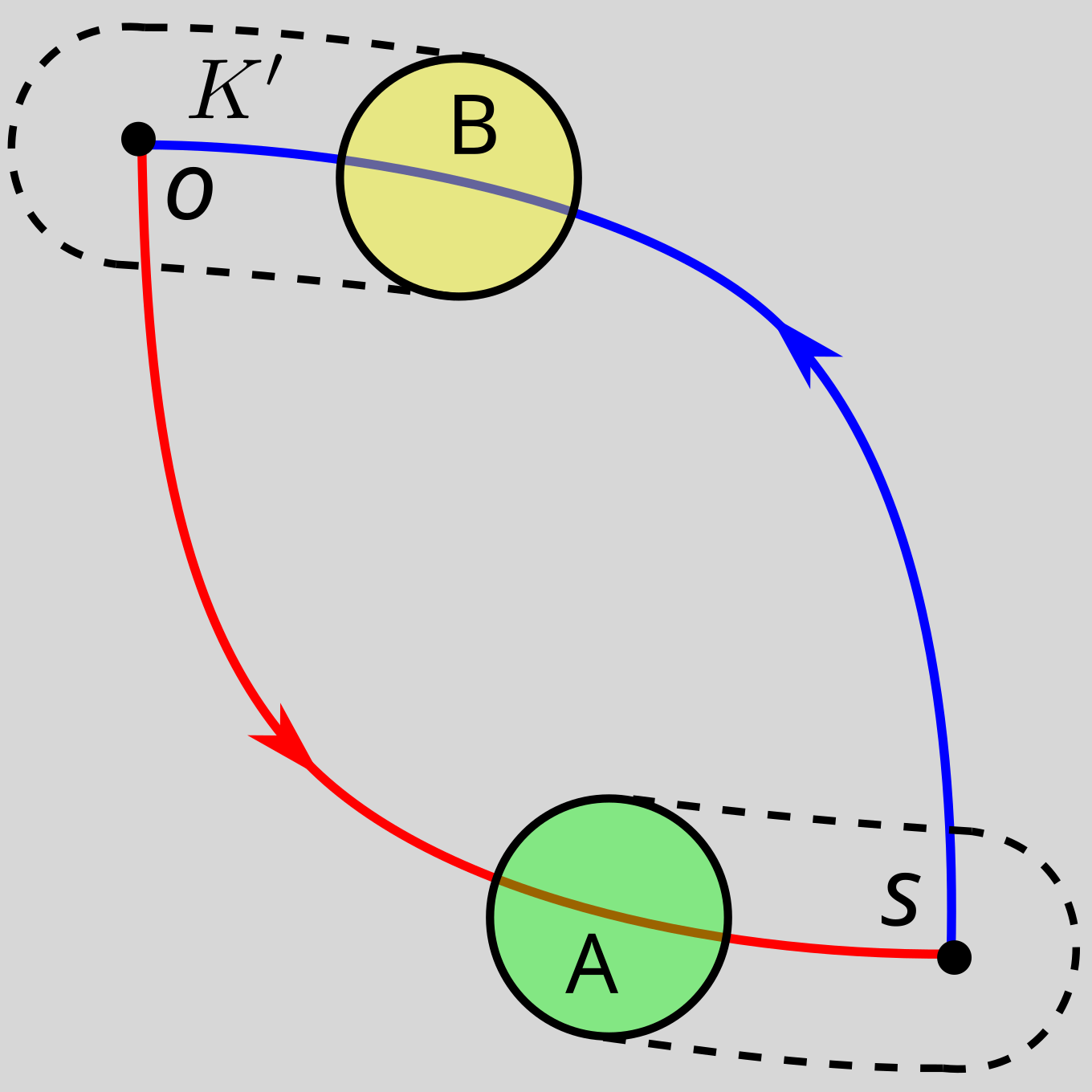} 
		\caption{\label{fig:ExchangeBootstrap-3T-4} $t=t_2$}
	\end{subfigure}
	\caption{\label{fig:ExchangeBootstrap} Analyzing non-local information encoding in the game process. %
		(a) definition of regions $K_1$ and $K_2$ at $t=t_1$; (b) definition of region $K'$ at $t=t_2$.}
\end{figure}
We first derive the necessary conditions for Alice to be able to send some information to Bob under the game rules, which is necessary for winning, and later we swap their roles to obtain the necessary condition for bidirectional communication. 
To this end, we analyze the quantum state of the physical system at each stage of the game process, taking into account all possible allowed actions of the players.  %
Before the game starts, the state of the system is initialized to be the ground state $\ket{G}$ prepared by the players, which has simple quasiparticles $\sigma$ and $\bar{\sigma}$ at $\oA$ and $\oB$, respectively, as we argued above. %
Let $\hat{U}_A(t)$ and $\hat{U}_B(t)$ be the local unitary operations performed by Alice and Bob at time $t\in[0,T]$, which act inside circles A and B, respectively.  
Let $\rho(t)$ be the quantum state of the physical system at time $t$, which can be a mixed state in general, and  
let $\rho_A(t)=\mathrm{Tr}_{\overline{A(t)}}[\rho(t)]$ be the reduced density matrix of the system inside circle A at time $t$, where $A(t)$ is the set of points inside circle A and $\overline{A(t)}$ is its complement, and $\rho_B(t)$ is defined similarly.
Therefore, all the influence Alice can make on the physical system is through the set of local operations $\{U_A(t)\}_{0\leq t\leq T}$, which is subject to the constraint that no excitations are left beyond the circle areas at any time, and all the information Bob can extract from the system are contained in the set $\{\rho_B(t)\}_{0\leq t\leq T}$. To make it possible for Alice to send any information to Bob, the set  $\{\rho_B(t)\}_{0\leq t\leq T}$ must have a nontrivial dependence on the set $\{U_A(t)\}_{0\leq t\leq T}$. 
Let $\vec{r}_A(t)$ and $\vec{r}_B(t)$ be the centers of circles A and B at time $t$, respectively. It is straightforward to prove the following facts using locality:\\
\begin{fact}\label{fact:locality}
Consider two time points $t_1,t_2$~(satisfying $0<t_1<t_2<T$) when both circles are far away from $\oA$ and $\oB$, e.g., $t_1=T/4$ and $t_2=3T/4$, as shown in Fig.~\ref{fig:ExchangeBootstrap}. Then we have:\\ 
	(1) For any $t\in[0,T]$, $\rho_B(t)$ does not depend on $\{U_A(t)\}_{t_1\leq t\leq T}$; \\%
	(2) For any $t\in[0,t_2]$, $\rho_B(t)$ does not depend on $\{U_A(t)\}_{0\leq t\leq T}$. %
\end{fact}
In other words, any local operation Alice performs after $t=t_1$ cannot influence Bob's measurements, and any local measurement Bob performs before $t_2$ does not depend on Alice's operations. %
Therefore, if Alice can send her number $a$ to Bob by encoding it in $\{U_A(t)\}_{0\leq t\leq T}$, then $a$ must already be encoded in $\{U_A(t)\}_{0\leq t\leq t_1}$, and 
$\{\rho_B(t)\}_{t_2\leq t\leq T}$ must have a nontrivial dependence on $\{U_A(t)\}_{0\leq t\leq t_1}$, such that Bob can decode the number $a$ from $\{\rho_B(t)\}_{t_2\leq t\leq T}$ using his knowledge about the physical system along with knowledge about all the local operations he has done $\{U_B(t)\}_{0\leq t\leq T}$. 
In summary, at $t=t_1$, the number $a$ must already be encoded in the physical system. The key question is precisely where this information is stored. 

Let us analyze in detail the quantum state of the system at $t=t_1$. As shown in Fig.~\ref{fig:ExchangeBootstrap-T-4}, let $K_1=\cup_{0\leq t\leq t_1}A(t)$ be the region~(enclosed by the dashed curve) traversed by circle A from $t=0$ up to $t=t_1$. This is the region where $\{U_A(t)\}_{0\leq t\leq t_1}$ act on. Let $K_2$ be the annulus region enclosing $K_1$~(bounded by the dashed and the dot-dashed curves), and the width of $K_2$~(denoted by  $d_{K_2}$) is taken to be much larger than the correlation length $\xi$. Let $K_{12}=K_1\cup K_2$, and let $\overline{K_{12}}$ be its complement.
In the following we will argue that at $t=t_1$, %
the number $a$ must be encoded in the reduced density matrix $\rho_{K_{12}}$ of the region $K_{12}$~\footnote{It may be tempting to say that $a$ must be stored in the region $K_1$ given that $\{U_A(t)\}_{0\leq t\leq t_1}$ only acts on $K_1$. However, this is not true in general if the initial state has strong entanglement between $K_1$ and its complement. For example, one can encode a classical bit in the GHZ state using a local phase gate, but such information cannot be locally decoded from the reduced density matrix of any proper subsystem.}.

For gapped ground states, it is generally expected that $I(K_1, \overline{K_{12}})\leq C e^{-d_{K_2}/\xi}$ due to the exponential clustering theorem~\cite{hastings2006,nachtergaele2006}. Although a rigorous proof of the exponential decay of mutual information in gapped ground states is not yet known in the literature, we assume it to be true in this paper. We now use the important lemma that follows from Ref.~\onlinecite{petz2003recovery}: %
\begin{lemma}\label{thm:petz}
	Consider a quantum system consisting of three non-overlapping subsystems $A,B$ and $C$. Let $\rho$ and $\rho'$ be two mixed states of this system satisfying $\rho_{AB}=\rho'_{AB},\rho_{BC}=\rho'_{BC}$, and $I(A:C|B)_{\rho}=I(A:C|B)_{\rho'}=0$, where
	\begin{equation}
		I(A:C|B)_{\rho}=S(\rho_{AB})+S(\rho_{BC})-S(\rho_{ABC})-S(\rho_{B}),
	\end{equation}
	is the condition mutual information. Then $\rho=\rho'$. 
\end{lemma}
At $t=0$, the initial state of the system  $\rho(0)=\ket{G}\bra{G}$ is a pure state, so we have $I(K_1, \overline{K_{12}}|K_{2})_{\rho(0)}=I(K_1, \overline{K_{12}})_{\rho(0)}\approx 0$ for $d_{K_2}\gg \xi$. Since 
$\{U_A(t)\}_{0\leq t\leq t_1}$ only acts on the region $K_1$, we have $I(K_1, \overline{K_{12}}|K_{2})_{\rho(t)}=I(K_1, \overline{K_{12}}|K_{2})_{\rho(0)}\approx 0$~[more generally, if $\{U_A(t)\}_{0\leq t\leq t_1}$ are allowed to be local quantum channels acting on $K_1$, we still have $I(K_1, \overline{K_{12}}|K_{2})_{\rho(t)}\leq I(K_1, \overline{K_{12}}|K_{2})_{\rho(0)}\approx 0$ due to the data processing inequality~\cite{nielsenQuantumComputationQuantum2010}]. 
Now consider how the reduced density matrices $\rho_{K_{12}}(t),\rho_{\overline{K_{1}}}(t)$, and $\rho(t)$ changes as the encoded number $a$ varies~(here $\overline{K_{1}}$ is the complement of $K_1$). It is clear that $\rho_{\overline{K_{1}}}(t)$ is independent of $a$, since $\{U_A(t)\}_{0\leq t\leq t_1}$ only acts on $K_1$. If $\rho_{K_{12}}(t)$ is also independent of $a$, then $\rho(t)$ must also be independent of $a$ due to Lemma~\ref{thm:petz}. Therefore, given that $a$ is encoded in $\rho(t)$, it must be  encoded in $\rho_{K_{12}}(t)$. 
In the following we extend this line of thought and show that $a$ can only be stored in the topological degree of freedom in $\rho_{K_{12}}(t)$. 

Let $\psi_1(t)$ and $\psi_2(t)$ be the total fusion channel of all quasiparticles in circle A and B at time $t$, respectively. The requirement to pass the identical particle test requires that both  $\psi_1(t)$ and $\psi_2(t)$ must be simple, and $\psi_1(t)=\psi_2(t)$, see Sec.~\ref{sec:passingIPT}. 
It is then clear that during the time interval $t\in[t_1,t_2]$, $\psi_1(t)$ must be constant in time~(which we denote by $\psi\in\calC$), since for $t\in[t_1,t_2]$, circle A is away from $\sigma$ and $\bar{\sigma}$, and local operations inside circle A cannot change the total fusion channel. 

At time $t=t_1$, the total fusion channel in the region $K_{12}$ must be equal to $\sigma$, since $\sigma$ was the total fusion channel in $K_{12}$ before the game begins, and 
the local operations $\{U_A(t)\}_{0\leq t\leq t_1}$
performed  inside this region  cannot change its total fusion channel. This implies that $\calC$ must have a fusion rule of the form 
\begin{equation}\label{eq:sigmapsifusion-extra1}
	\sigma \times \psi=m~\sigma+\beta,
\end{equation}
where $m\geq 1$ is the fusion multiplicity and $\beta$ denotes the sum of all other particle types in the RHS not equal to $\sigma$. Later we will see in Sec.~\ref{sec:win_APtest} that passing the antiparticle test requires $\beta=0$. This leads to the fusion rule in Eq.~\eqref{eq:sigmapsifusion-0}. The proof in App.~\ref{app:proofdualFrule} then shows that 
$\calC$ must also have a fusion  rule in Eq.~\eqref{eq:sigmapsifusion-1}. 

We now argue that the number $a$ can only be encoded in the fusion space $V^\sigma_{\sigma\psi}$. Indeed, consider how 
$\rho_{K_{12}}(t)$ changes when $a$ varies. It is clear that the reduced density matrix on any $O(1)$-sized subregion of $K_{12}$ %
cannot change when $a$ is varied--otherwise, an eavesdropper can obtain information about $a$ via local measurement.  %
This means that as the number $a$ varies, the reduced density matrix $\rho_{K_{12}}(t)$ for different values of $a$
must belong to the same 
information convex set~\cite{Shi2020}~[defined above Eq.~\eqref{eq:structureinfoconvex}] $\Sigma_{\rho(t)}(K_{12})$ on the region $K_{12}$. 
Given that there exists topological quasiparticles $\sigma,\psi$ in the region $K_{12}$ with total fusion channel $\sigma$, Eq.~\eqref{eq:structureinfoconvex} now becomes
\begin{equation}\label{eq:structureinfoconvexK12}
	\Sigma_{\rho(t)}(K_{12})\cong \calS(V^\sigma_{\sigma\psi}),%
\end{equation}
Therefore, given that the number $a$ is encoded in $\rho_{K_{12}}(t)\in \Sigma_{\rho(t)}(K_{12})$,  Eq.~\eqref{eq:structureinfoconvex} implies that $a$ can only be stored in the fusion space $V^\sigma_{\sigma\psi}$. 

We now consider the quantum state of the system at $t=t_2$, as shown in Fig.~\ref{fig:ExchangeBootstrap-3T-4}. At this point, Bob still has no knowledge about the number $a$, due to Fact.~\ref{fact:locality}. Therefore, if Bob  eventually knows the value of $a$ after the game ends, he must be able to gain useful information about $a$ through his measurement during $t\in[t_2,T]$. %
It is straightforward to see that any useful information Bob can obtain during   $t\in[t_2,T]$ must already be stored in the local reduced density matrix $\rho_{K'}(t)$, where $K'$ is the region enclosed by the dashed curve in Fig.~\ref{fig:ExchangeBootstrap-3T-4}. 
More precisely, %
when the set $\{U_B(t)\}_{0\leq t\leq t_2}$ is fixed and the value of $a$ varies, $\rho_{K'}(t_2)$ must have a nontrivial dependence on $a$. 
Following a similar argument as before, we conclude that at $t=t_2$, the useful information about $a$ that Bob can eventually learn must be stored in the fusion space $V_{\sigma\psi}^\sigma$. %

In summary, the above analysis shows that any winning strategy to the game must have the following two essential steps~(note that the winning condition of the game requires bidirectional communication, and the necessary condition for Bob to be able to send information to Alice is obtained by swapping the roles of Alice and Bob in the above analysis):\\
(1). At early game $t\in[0,t_1]$, Alice encodes $a$ into the fusion space $V_{\sigma\psi}^\sigma$, while Bob  encodes $b$ into the fusion space $V_{\psi\bar{\sigma}}^{\bar{\sigma}}$;\\
(2). At late game $t\in[t_2,T]$, Alice decodes information from the fusion space $V_{\psi\bar{\sigma}}^{\bar{\sigma}}$, while Bob decodes information from the fusion space $V_{\sigma\psi}^\sigma$. \\
The space-time diagram of the whole process is given in the LHS of  Eq.~\eqref{eq:SFCdescriptiongame}. We can now use a similar argument given in Sec.~\ref{sec:win_2pt} to show that the strategy can succeed if and only if the $R$-matrix in the RHS of Eq.~\eqref{eq:SFCdescriptiongame} is nontrivial. This shows that the winning strategy given in Sec.~\ref{sec:win_2pt} and Sec.~\ref{sec:win2pt-SFC} is essentially unique, up to a potentially different encoding and decoding algorithm, and up to adding or removing pointless operations during the game~(e.g. creating and annihilating local excitations in the circle for no purpose, or making a measurement at $t=0$ that gives no useful information). %

\subsubsection{Passing the identical particle test}\label{sec:passingIPT}
Let $\psi_1$ and $\psi_2$ be the total fusion channel of all quasiparticles in circle A and B at some time $t$, respectively. 
Suppose that an identical particle test is performed at this time $t$.
We show that in order to guarantee that the players can pass the identical particle test, both $\psi_1$ and $\psi_2$ must be simple such that $\psi_1=\psi_2$. %
The reason is that when Charlie first enters the game in the configuration shown in Fig.~\ref{fig:IPT}, she can 
measure the total fusion channel in each circle~[by measuring the particle type operator $\hat{T}$ in Eq.~\eqref{eq:measure_type}], 
and obtain simple particle types $\beta\in\psi_1$ and $\varphi\in\psi_2$, respectively.  %
After this the quantum  state  collapse into $\hat{\Pi}^\beta_i \hat{\Pi}^\varphi_j\ket{\Psi}$. %
If either $\psi_1$ or $\psi_2$ is not simple, then there is always a nonzero chance that $\beta\neq \varphi$. In this case, after Charlie temporarily leaves the game, whatever local operations Alice and Bob apply  cannot change the total fusion channel $\beta$, $\varphi$ within each circle area, and when Charlie comes back, she can simply measure %
$\hat{T}$ in each circle again to determine if an exchange has happened. Therefore to guarantee that the players can pass the identical particle test, %
both $\psi_1$ and $\psi_2$ must be simple such that $\psi_1=\psi_2$.

\subsubsection{Passing the antiparticle test}\label{sec:win_APtest}
We now show that passing the antiparticle test requires %
$\beta=0$ in Eq.~\eqref{eq:sigmapsifusion-extra1}.
The analysis for the ground state is the same as before and we do not repeat it here. 
Since the antiparticle test is performed between circles A and B after A and C separates, according to our previous analysis in Sec.~\ref{sec:passingIPT} the particle types in the circles A and B have to be the same simple type which we denote by $\psi$. Since the particles in the circles A and C were initially created by applying a sequence of local unitaries in the bulk of the ground state,
they must fuse into vacuum. Therefore the particle type in circle C must be $\bar{\psi}$, the antiparticle of $\psi$.  %
Here comes the important point: at time $t=t_1$, before Carol's circle disappears, she should be able to fuse $\bar{\psi}$ into $\sigma$ using local operations, and she should be able to do this without changing the particle type $\sigma$ at the point $\oA$, as the Referees will check local ground state condition at $\oA$ after her circle disappears. This requires 
$\adjincludegraphics[height=6ex,valign=c]{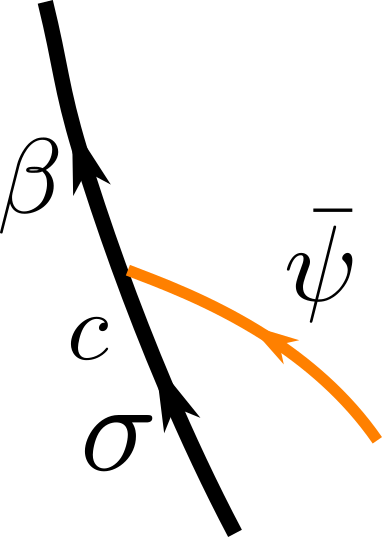}=0$ for any  $\beta\neq \sigma$, which is equivalent to $N_{\sigma\bar{\psi}}^\beta=0$. Therefore, we must have $\sigma \times \bar{\psi}=m~\sigma$. This is the main simplification resulting from this antiparticle test. Although the original challenge also requires the players to locally annihilate their paraparticles (without changing the defect types $\sigma$ and $\bar{\sigma}$) before their circles disappear, it does not strictly require $N_{\sigma\bar{\psi}}^\beta=0$ for all $\beta\neq \sigma$. For example, we can imagine an SFC with fusion rules 
\begin{equation}\label{eq:sigmapsifusion-extra}
	\sigma \times \psi=m~\sigma+\beta,\quad   \psi\times \bar{\sigma}=m~\bar{\sigma}+\bar{\beta},
\end{equation}
with $\beta\neq \sigma$, satisfying
\begin{equation}\label{eq:SFC-extra}
	\adjincludegraphics[height=15ex,valign=c]{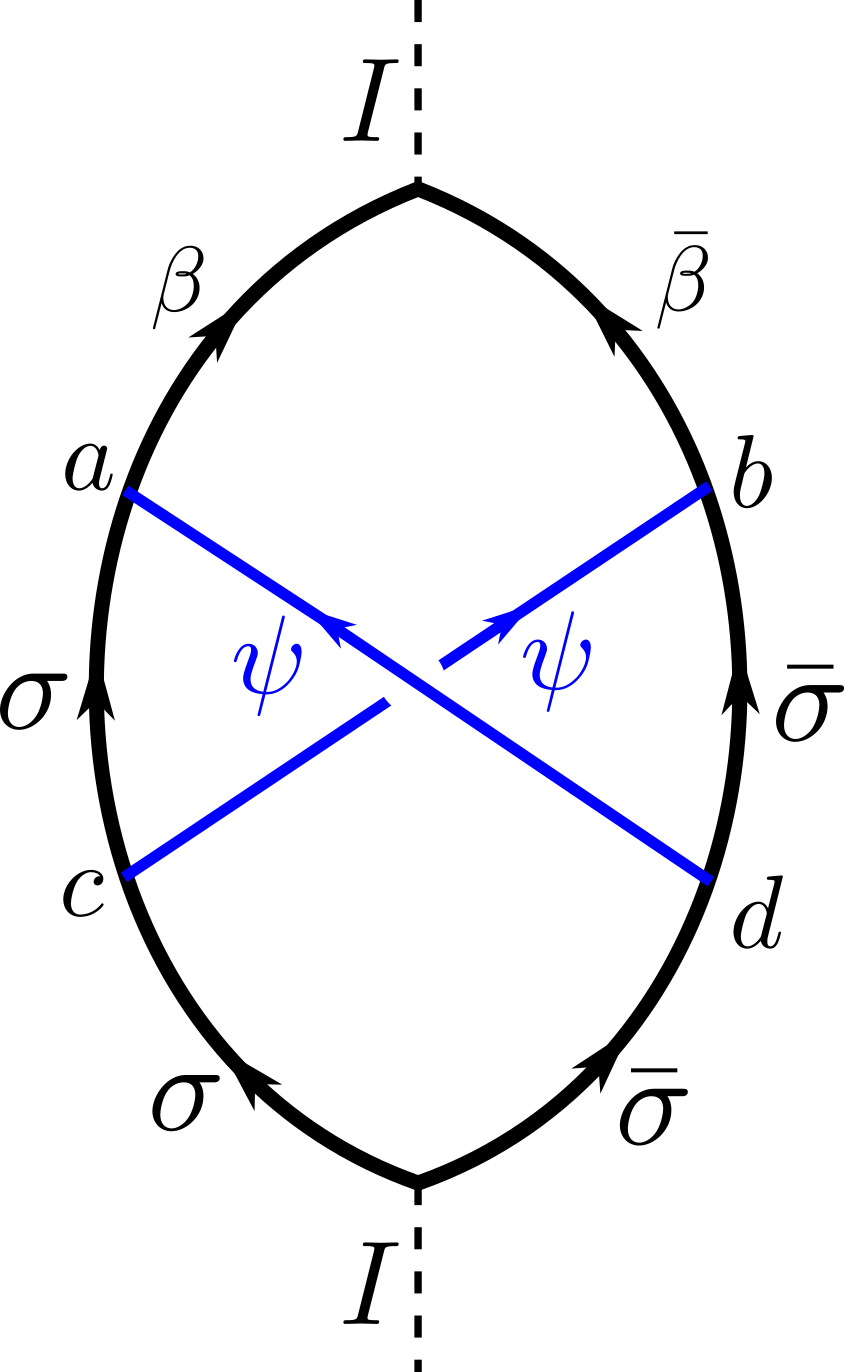}=0. %
\end{equation}
Eq.~\eqref{eq:SFC-extra} means that after the players fuse their paraparticles into the defects $\sigma,\bar{\sigma}$, there is zero possibility that the defect types $\sigma,\bar{\sigma}$ make a transition into $\beta,\bar{\beta}$. In this case, the players can simply pretend that the particle type %
$\beta$ does not exist, and the winning strategy is described by the same fusion diagram as in Eq.~\eqref{eq:SFCdescriptiongame}. The problem is that the fusion rule in Eq.~\eqref{eq:sigmapsifusion-extra} is more complicated than Eq.~\eqref{eq:sigmapsifusion-0}, and we do not yet have a systematic way to construct and classify SFCs with such fusion rules satisfying Eq.~\eqref{eq:SFC-extra}, so for simplicity in this paper we focus on the simpler class of SFCs with fusion rule in Eq.~\eqref{eq:sigmapsifusion-0}, and we enforce this simplification by adding the antiparticle test. %

\subsubsection{Winning the who-entered-first challenge}\label{sec:condition_winWEF}
\begin{figure}%
	\centering\includegraphics[width=.4\linewidth]{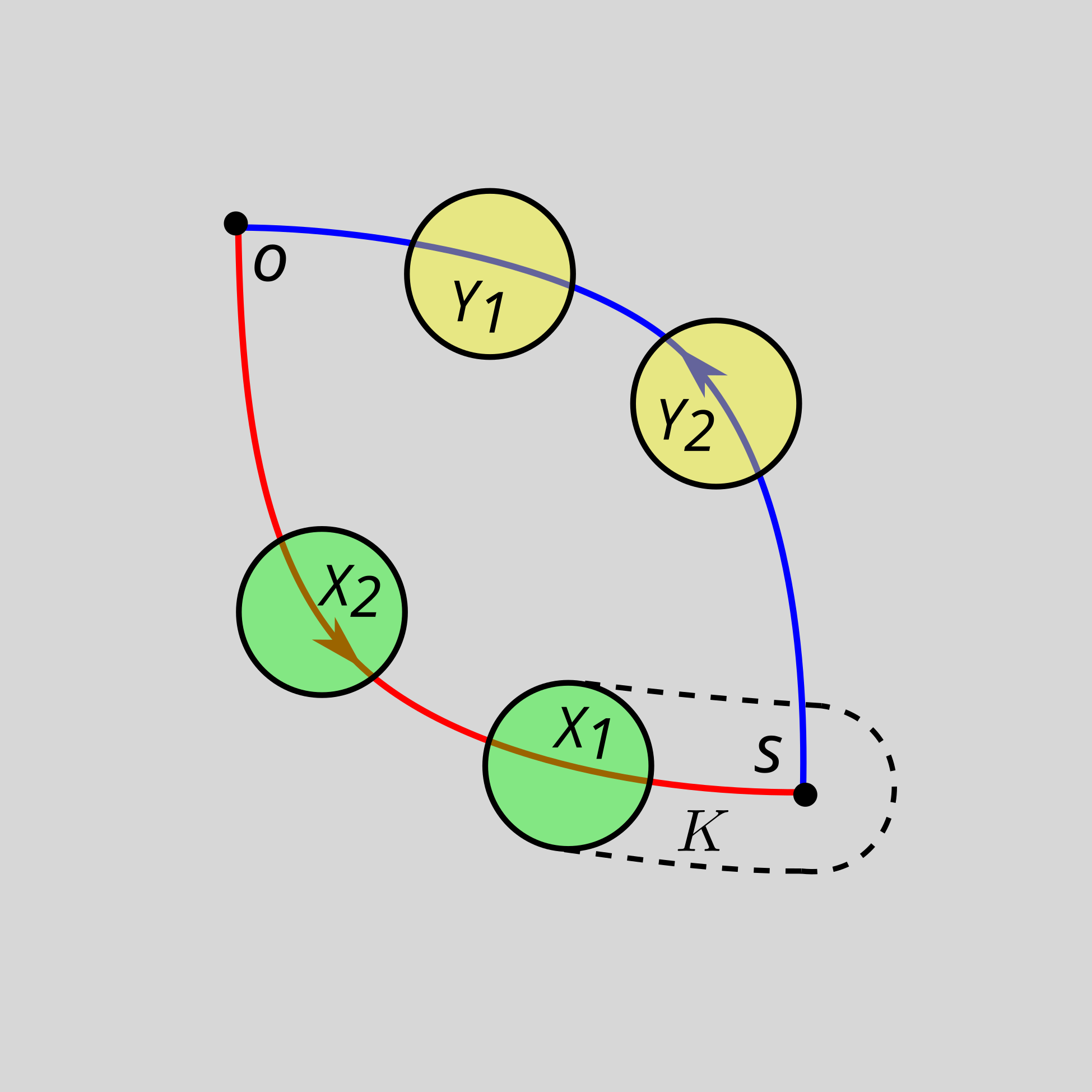} 
	\caption{\label{fig:WEFBootstrap} The who-entered-first challenge at $t=t_2<t_0+3T/4$. %
	}
\end{figure} 
We now extend the analysis of Sec.~\ref{sec:deriv_nece_cond} to the who-entered-first challenge, and show that swap-type $R$-matrices cannot win the challenge. %
First, it is clear that at the time when a player enters the game, he or she has no knowledge about the value of $(X_1,Y_1)$, since the local reduced density matrix near $\oA$ or $\oB$ is equal to the ground state value. Let $t_{2}<t_0+3T/4$ be a time when the circle $X_1$  is separated from the blue path~(the path of circles $Y_1$ and $Y_2$) by at least $2r_0$~(e.g., we can take $t_{2}=t_0+T/2$), as shown in Fig.~\ref{fig:WEFBootstrap}. It is clear that for $0\leq t\leq t_2$, the reduced density matrix in the circle $X_1$, denoted by $\rho_{X_1}(t)$, does not depend on $\{\hat{U}_Z(t')|t_0\leq t'\leq t, Z\in \{X_2,Y_1,Y_2\}\}$, similar to Fact.~\ref{fact:locality} we used before. By an inductive argument, one can show that before $t=t_2$, the player $X_1$ has no knowledge about $Y_1$. 
Therefore, if the player $X_1$ eventually knows $Y_1$, the reduced density matrix in the region $K$ shown in  Fig.~\ref{fig:WEFBootstrap}~(which is the region $X_1$ operates on during $t\in[t_2,t_0+3T/4]$), denoted by $\rho_K(t_2)$, must have a nontrivial dependence on $Y_1$.  Using a similar argument in Sec.~\ref{sec:deriv_nece_cond} [the argument around Eq.~\eqref{eq:structureinfoconvexK12}], this information about $Y_1$ must be stored in the fusion space $V_{\psi\bar{\sigma}}^{\bar{\sigma}}$.   %

The space time diagram of the game process is given in 
Eq.~\eqref{eq:SFCdescriptionWHF}. Then the local reduced density matrix of the system in the fusion space $V_{\psi\bar{\sigma}}^{\bar{\sigma}}$~\footnote{Here ``the local reduced density matrix of the system in the fusion space $V_{\psi\bar{\sigma}}^{\bar{\sigma}}$'' means $\Phi[\rho_K(t_2)]$, where $\Phi$ is the isomorphism from the information convex set in the region $K$ to the convex set $\calS (V_{\psi\bar{\sigma}}^{\bar{\sigma}})$.} is given by an expression similar to $\rho'_A$ in Eq.~\eqref{eq:rhopA-WEF}. It is clear that the values of $b,d$ cannot depend on $Y_1$, since  $b,d$ are determined by the action of Bob and David at their time of entrance, when both players have no knowledge about $Y_1$. From here we can use the analysis given in Sec.~\ref{sec:swapRfailWEF}, which shows a necessary condition to win this challenge is that $R$ does not satisfy Eq.~\eqref{eq:diagonalcommute}. This completes the proof that swap-type $R$-matrices cannot win the who-entered-first challenge. It also shows that the winning strategy given in Sec.~\ref{sec:winWEF} is essentially unique. 

\subsubsection{Summary: winning strategy  is essentially unique}\label{sec:FCanalysis_summary}
In summary, we have shown that in order to have a winning strategy for the 3+1D version of the challenge game that \\
(1) satisfies all the requirements stated in  Sec.~\ref{sec:parachallenge_variants_twists}, 
including the robustness against noise and eavesdropping,\\
(2) pass the identical particle test and the antiparticle challenge, and in addition assuming that\\
(3) the topological defects $\sigma$ and $\bar{\sigma}$ at the special points $\oA$ and $\oB$ are also mobile topological quasiparticles, \\
then the SFC $\calC$ must have a fusion rule of the form in  Eq.~\eqref{eq:sigmapsifusion-0}, such that the $R$-matrix defined by Eq.~\eqref{def:RfromSFC} or Eq.~\eqref{eq:SFCdescriptiongame} is nontrivial, and the winning strategy given in Sec.~\ref{sec:win_2pt} and Sec.~\ref{sec:win2pt-SFC} is essentially unique. If one additionally wants to win the who-entered-first challenge, then the $R$-matrix must not be of the swap-type.  Examples of this type of SFCs are given in App.~\ref{sec:RfromCentralType}.

\subsubsection{Deriving the axioms of emergent parastatistics from the SFC formulation}
We now briefly show that the quasiparticle $\psi$ in this type of SFC naturally satisfies the axioms of emergent parastatistics in Sec.~\ref{sec:axioms_emergent_para}. This is based on the identification of $n$-particle states in the two formalisms 
\begin{equation}\label{eq:nptstatecorr}
	\ket{G;i_1^{a_1} i_2^{a_2} \ldots i_n^{a_n}}=\Ket{\adjincludegraphics[height=20ex,valign=c]{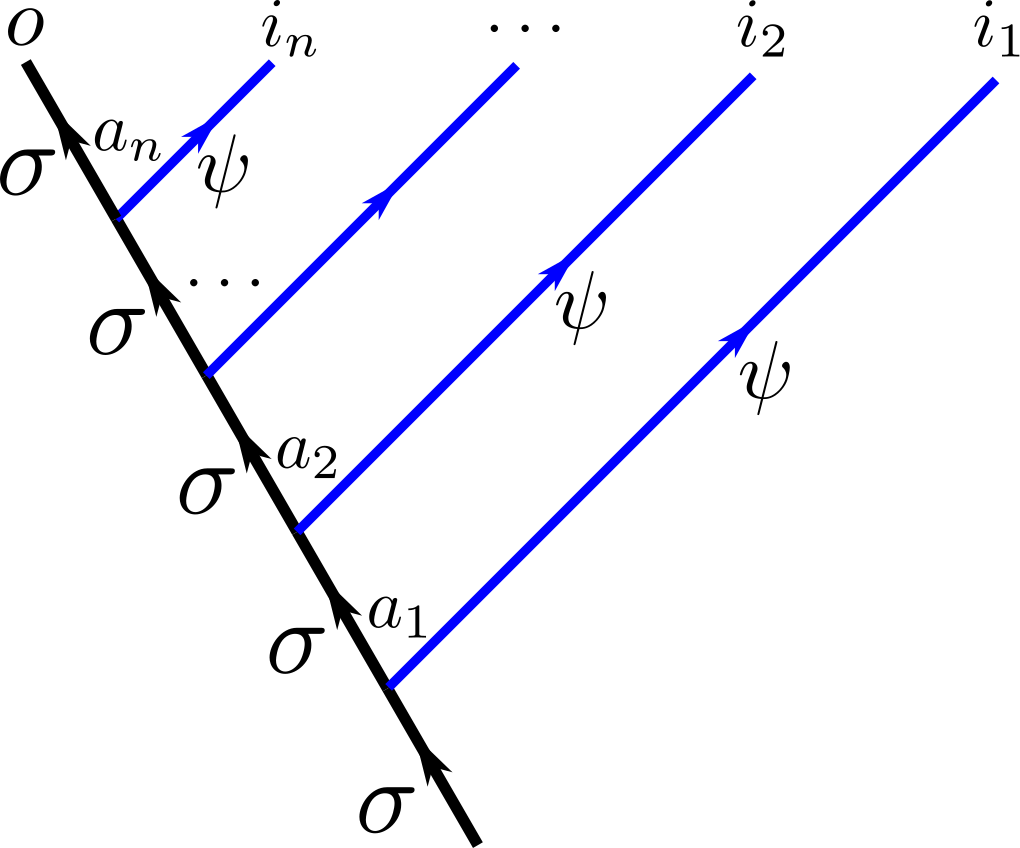}}.
\end{equation}
Indeed, the RHS of Eq.~\eqref{eq:nptstatecorr} describes the quantum state obtained by first creating a $\psi$ near $\sigma$ using $\hat{U}_{\oA,a_1}\equiv\hat{U}^{\sigma,\sigma\psi}_{\oA,a_1}$ and move it to position $i_1$, and then  repeat this process up tp $a_n$ and $i_n$, %
and this sequence of operations exactly produces the state in the LHS according to Axioms~\ref{Axiom3} and \ref{Axiom5}. %
In particular, the internal state of a paraparticle is identified with the fusion space $V_{\sigma}^{\sigma \psi}$. 
Then all the Axioms~\ref{Axiom1}-\ref{Axiom6} can be derived from the  properties of SFC and topological order, for example, Axiom~\ref{Axiom4} follows from the identity
\begin{equation}\label{eq:SFCR-pullthrough}
	\adjincludegraphics[height=11ex,valign=c]{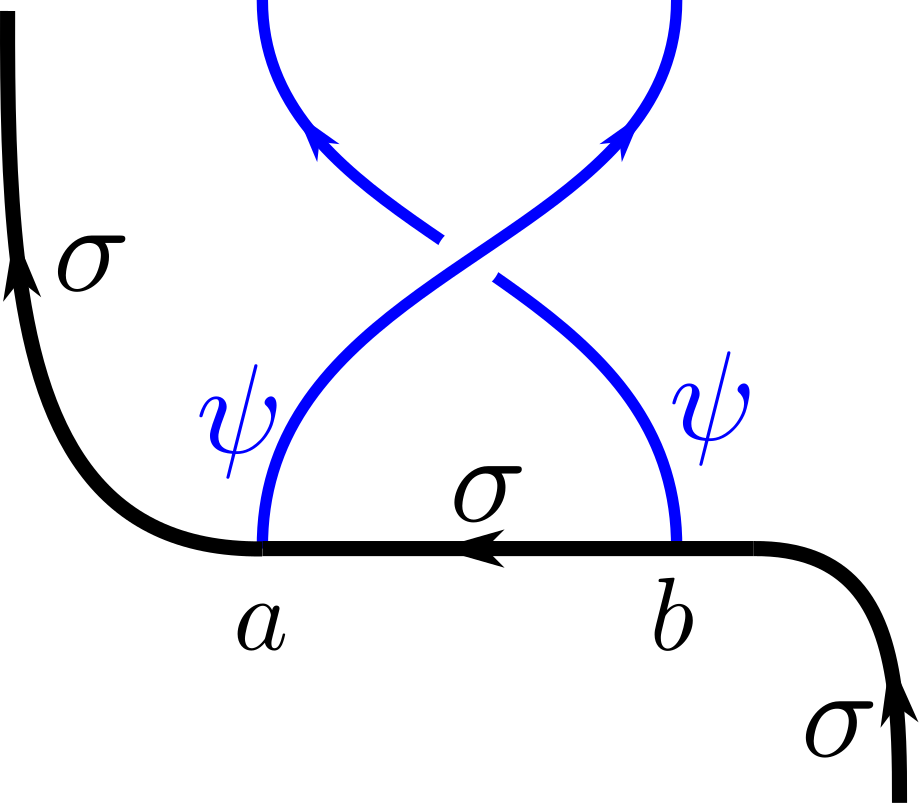}=
	\adjincludegraphics[height=11ex,valign=c]{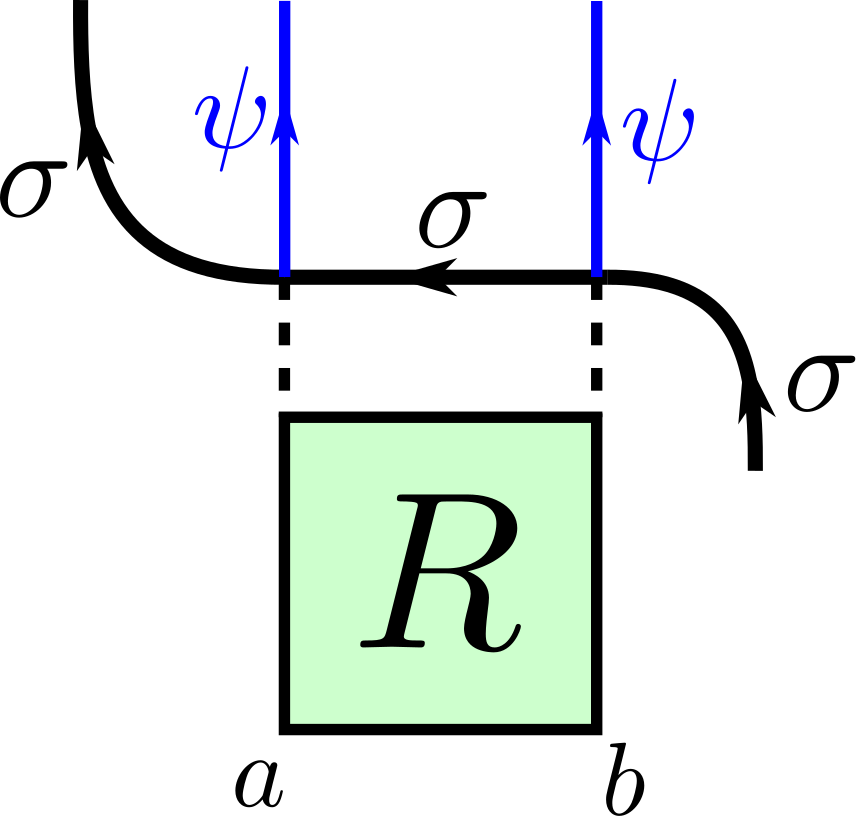},
\end{equation}
which is equivalent to Eq.~\eqref{def:RfromSFC}. The unitary operators in Eq.~\eqref{eq:localcreationatcorner} of Axiom~\ref{Axiom5} correspond to
\begin{equation}
	\hat{U}_{\oA,a}\equiv\hat{U}^{\sigma,\sigma\psi}_{\oA,a},\quad \hat{U}'_{\oB,b}\equiv\hat{U}^{\bar{\sigma},\psi\bar{\sigma}}_{\oB,b}
\end{equation}
where the operators in the RHS were first introduced in Eq.~\eqref{eq:U_spltting}. The observables in Eq.~\eqref{eq:localmeasurementatcorner} of Axiom~\ref{Axiom6} correspond to the observable $\hat{O}$ in Eq.~\eqref{eq:measure_fusion}. 

Therefore, under the assumptions made in this section, only nontrivial $R$-paraparticles can pass the full version of the challenge in 3+1D. 
In the next section we generalize our above analysis to the more general case where $\sigma$ and $\bar{\sigma}$ are arbitrary point-like topological defects in the system.

\section{The general module category description for the point-like defect $\sigma$}\label{sec:ModCatDefect} %
Up to now, in our analysis for the winning strategy, we have been assuming that the defects $\sigma$ and $\bar{\sigma}$ at the special points $\oA$ and $\oB$ are some special topological quasiparticles in the system. However, since the game protocol does not require the defects $\sigma$ and $\bar{\sigma}$ to be mobile~(as they do not change position throughout the game), they can in principle be more general point-like topological defects that are not necessarily quasiparticles. For example, in Ref.~\onlinecite{wang2024parastatistics}, we described a winning strategy where $\sigma$ is the defect lying at the intersection between two different types of gapped boundaries of a 2D topological phase.

To systematically describe winning strategies in this more general case, it is important to understand the fusion and splitting between the paraparticle $\psi$ and the defect $\sigma$. Such a process is described by a module category $\mathcal{M}$~(with $\sigma\in\mathcal{M}$) over the symmetric fusion category $\calC$~(with $\psi\in\mathcal{C}$)~\cite{Kitaev2012gappedboundary,kongAnyonCondensationTensor2014}. In the following we first briefly review the module category description of point-like defects in a topological phase, and then analyze winning strategies in this framework.
\subsection{Point-like defects and module categories}\label{sec:ModCatrecap}
Let $\calC$ be a unitary braided or symmetric fusion category describing point-like excitations~(topological quasiparticles) of a 2D or 3D topological phase. A point-like topological defect in such a phase is always equipped with the structure of a module category $\calM$ over $\calC$, which physically describes the fusion between quasiparticles and the defect. An object of the category $\calM$ describes a possible state of the defect $\calM$, which is also called a topological charge or superselection sector of the defect. Such a topological charge cannot be changed by any kind of local operations near the defect, and the only way to change it is by fusing a bulk quasiparticle with the defect. Similar to the fusion of bulk quasiparticles, the fusion of quasiparticles in $\calC$ and the point-like defect also has a well-defined set of fusion rules 
\begin{equation}\label{eq:FCrules-modcat}
	\sigma\times \psi=\sum_{\beta\in\calM} \tilde{N}_{\sigma\psi}^\beta \beta,
\end{equation}
for $\sigma \in \calM$ and $\psi\in\calC$. Physically, Eq.~\eqref{eq:FCrules-modcat} means that fusing the defect with a quasiparticle will generally transform the defect into a different state, and there are $\tilde{N}_{\sigma\psi}^\beta$ linearly independent ways to fuse $\sigma$ and $\psi$ into $\beta$.  %
Eqs.~(\ref{eq:fusion_basis}-\ref{eq:U_spltting}) still hold for the module category case, in which the left most legs are understood as objects of the module category $\calM$.  
In particular, in the fusion process $\adjincludegraphics[height=5ex,valign=c]{Figures/FCbasics/FS-F-2.png}$, 
both the resulting defect state $\beta$ and the index $a$ can be locally measured, and the splitting process $\adjincludegraphics[height=5ex,valign=c]{Figures/FCbasics/FS-S-2.png}$ can be achieved by a local unitary operation. 

In a module category we also have an $F$-move of fusion diagrams analogous to Eq.~\eqref{eq:F-move}:
\begin{equation}\label{eq:F-move-MC}
	\Ket{\adjincludegraphics[height=9ex,valign=c]{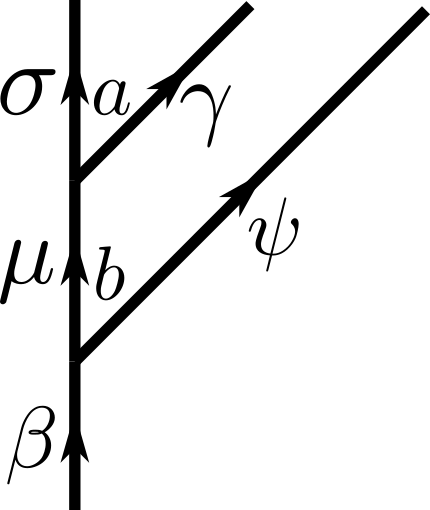}}=\sum_{\nu,c,d} [\tilde{F}^{\sigma\gamma\psi}_{\beta}]^{\mu ab}_{\nu cd} \Ket{\adjincludegraphics[height=9ex,valign=c]{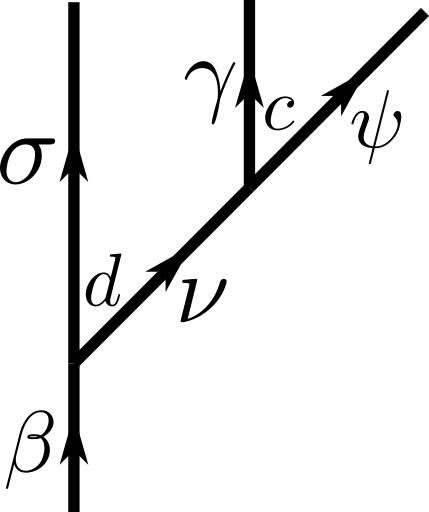}},
\end{equation} 
where $\sigma,\mu,\beta$ are simple objects of $\calM$, while $\gamma,\psi,\nu$ are simple objects of $\calC$. Here $\tilde{F}^{\sigma\gamma\psi}_{\beta}$ is a unitary matrix, called the $\tilde{F}$-symbol of the module category $\calM$, which is analogous to the $F$-symbol of fusion categories. The $\tilde{F}$ symbol of $\calM$ and the $F$-symbol of $\calC$ also need to satisfy a pentagon equation~\cite{Kitaev2012gappedboundary} to guarantee the consistency of Eqs.~(\ref{eq:F-move},\ref{eq:F-move-MC}). 
Eqs.~(\ref{eq:BFCstatespace-npt}-\ref{eq:locality-splitting}) also hold for the module category case, where all the labels on the leftmost legs~(e.g., $\alpha_1,\beta$ and $\beta_2,\ldots,\beta_n$) are now understood as simple objects of $\calM$.

It should be immediately clear that $\calM=\calC$ is always a module category over $\calC$, where the fusion between $\calM$ and $\calC$ is simply defined as fusion in $\calC$, and the $\tilde{F}$-symbol is given by the $F$-symbol of $\calC$. The physical meaning is that in a topological phase $\calC$, having no defect at a certain point can be viewed as a trivial special case of having a defect, where the states of the defect is simply given by the quasiparticle types in $\calC$. Therefore, our conclusions in this section should include Sec.~\ref{sec:categorical_analysis} as a special case. 
\subsection{Module category description of the winning strategy}
We now generalize the categorical description of the winning strategies given in Sec.~\ref{sec:categorical_analysis} to the more general case where the defect at the special point $\oA$ is described by a module category $\calM$ over the braided or symmetric fusion category $\calC$ of the bulk phase. As before, we need a fusion rule of the form Eq.~\eqref{eq:sigmapsifusion-0}, where  $\sigma \in \calM$ is a state of the defect and $\psi\in\calC$ is the paraparticle. 
The winning strategies for this more general case can still be described by exactly the same spacetime diagrams in  Sec.~\ref{sec:win_strategy_diagrammatic},
the only difference is in the physical meaning of the trivalent vertices. 
For example, we still use Eq.~\eqref{def:RfromSFC} to define the $R$-matrix, where all the trivalent vertices in the diagram now describe the fusion and splitting processes between the paraparticle $\psi$ and the defect $\sigma$. Eq.~\eqref{eq:RmatfromFRmove} is still valid if we replace the $F$-symbol by the $\tilde{F}$-symbol of the module category $\calM$. Notice that as before, fusion and splitting never change the state $\sigma$ of the defect, therefore the Referees cannot locally detect any difference at the special point $\oA$.  
It should then be clear that the strategy described in Sec.~\ref{sec:win1pt-SFC} for the version with one special point still works in this more general case. 

We now move on to understand the strategy with two special points~(described in Sec.~\ref{sec:win2pt-SFC}) in the module category case. 
To generalize the fusion diagram in Eq.~\eqref{eq:SFCdescriptiongame} to the module category case, we need to understand the physical meaning of the ``dual object'' of $\sigma$, denoted by $\bar{\sigma}$, 
and the bottom vertex describing the splitting process $I\to\sigma\bar{\sigma}$. Compared to fusion categories, in a module category $\calM$, we generally do not have a natural notion of antiparticles or a distinguished unit object $I$~(a vacuum). For $\sigma\in\calM$, its dual object $\bar{\sigma}$ should be understood as an object in the dual module category %
$\bar{\calM}$ describing the defect at the other special point $\oB$~\footnote{Technically, the module category defined in Sec.~\ref{sec:ModCatrecap} is more precisely known as a right module category in the math literature~\cite{TenCat_EGNO}. In general, if $\calM$ is a right module category over the fusion category $\calC$, then the dual $\bar{\calM}$ is naturally a left $\calC$-module category. A left $\calC$-module category  is analogous to a right $\calC$-module category, %
	but in which an object $\psi\in\calC$ fuse with $\bar{\sigma}\in \calM$ from the left, denoted by $\psi\times \bar{\sigma}$, and the fusion rule and the $\tilde{F}$-move are defined in a similar way. 
	If $\calC$ is symmetric~(or braided in the 2D case), then one can use the braiding structure of $\calC$ to turn a right $\calC$-module category into a left $\calC$-module category~\cite{kong2022invitation,TenCat_EGNO} and vice versa. Since in this paper, $\calC$ is always symmetric~(or braided in the 2D case), we do not need to distinguish left and right $\calC$-module categories.}. 
To understand the meaning of the vertex $I\to \sigma\bar{\sigma}$, we need to first understand the fusion between objects in $\calM$ and $\bar{\calM}$, whose physical meaning depends on the how the defect $\calM$ is realized. We explain this through specific examples below.  %

A simple way to realize point-like defects in a topological phase is to have a line-like defect~(such as a domain wall in 2D or a string excitation in 3D) with two end points,  as shown in Fig.~\ref{fig:LineDefect}. Here $\calM$~($\bar{\calM}$) describes the point-like defect at left~(right) end point. The fusion $\sigma\times\bar{\sigma}$ can be understood as the physical process of continuously shrinking the line defect, as shown in Fig.~\ref{fig:LineDefect}. Eventually, the line defect shrinks to a point and becomes a quasiparticle $q\in\calC$ in the bulk. This whole process defines the fusion %
\begin{equation}
	\sigma\times\bar{\sigma} =\sum_{q\in\calC}N_{\sigma\bar{\sigma}}^q q.
\end{equation}

More generally, objects in $\calM$ and $\bar{\calM}$ can fuse into objects in a fusion category $\calD$ potentially different from $\calC$. For example, in Fig.~\ref{fig:BoundaryDefect}, the point-like defects  $\sigma\in\calM$ and $\bar{\sigma}\in\bar{\calM}$ 
are realized as the intersection points between two different types of gapped boundaries, whose boundary excitations are described by fusion categories $\mathcal{B}$ and $\mathcal{D}$, respectively. In this case, if we continuously shrink the boundary $\mathcal{B}$ up to a point, the defects $\sigma\in\calM$ and $\bar{\sigma}\in\bar{\calM}$ eventually fuse into a point-like excitation $q$ on the boundary $\calD$. The unit object $I\in \calD$ simply refers to the state with no excitation on the boundary $\calD$, i.e., the unique gapped ground state with open  boundary $\calD$. 

We now explain the meaning of the bottom vertex $I\to\sigma\bar{\sigma}$  in  Eq.~\eqref{eq:SFCdescriptiongame}. As before, splitting is the  reverse process of fusion. 
The bottom vertex $I\to\sigma\bar{\sigma}$ simply describes the way how the players prepare the initial state of the system that they submit to the Referees. For example, to create the configuration in Fig.~\ref{fig:LineDefect}, they can begin with the translationally invariant ground state of the system and then create a tiny line defect somewhere in the bulk, and then continuously stretching it to a segment~(i.e., the reverse process of fusion). The resulting state is also the unique gapped ground state of a local Hamiltonian provided that the line defect is gappable. 
With all the discussions above, it is now straightforward to see that the winning strategy with two special points still works in the module category case, and all the derivations in Sec.~\ref{sec:win2pt-SFC} are still valid. 
We conclude these discussions with a categorical definition of paraparticles:
\begin{definition}\label{def:cat_paraparticles}
	A simple object $\psi$ of a symmetric fusion category $\calC$ is called an $R$-paraparticle if $\calC$ has a module category $\calM$ with a defect fusion rule of the form in Eq.~\eqref{eq:sigmapsifusion-0}, such that the $R$-matrix given by Eq.~\eqref{def:RfromSFC} is nontrivial. 
\end{definition}
It is clear that our categorical analysis in Sec.~\ref{sec:TC_winning_condition} can also be straightforwardly generalized to the module category case, provided that the defect configuration in $\ket{G}$ satisfies some general requirements stated in Sec.~\ref{app:condition_defect}, leading us to the conclusion that only $R$-paraparticles can win the full version of the challenge game in 3+1D. 
\begin{figure}
	\begin{subfigure}{1\linewidth}
		\begin{tikzpicture}
			\node[midway] (RoughSmooth) {\includegraphics[width=.2\linewidth]{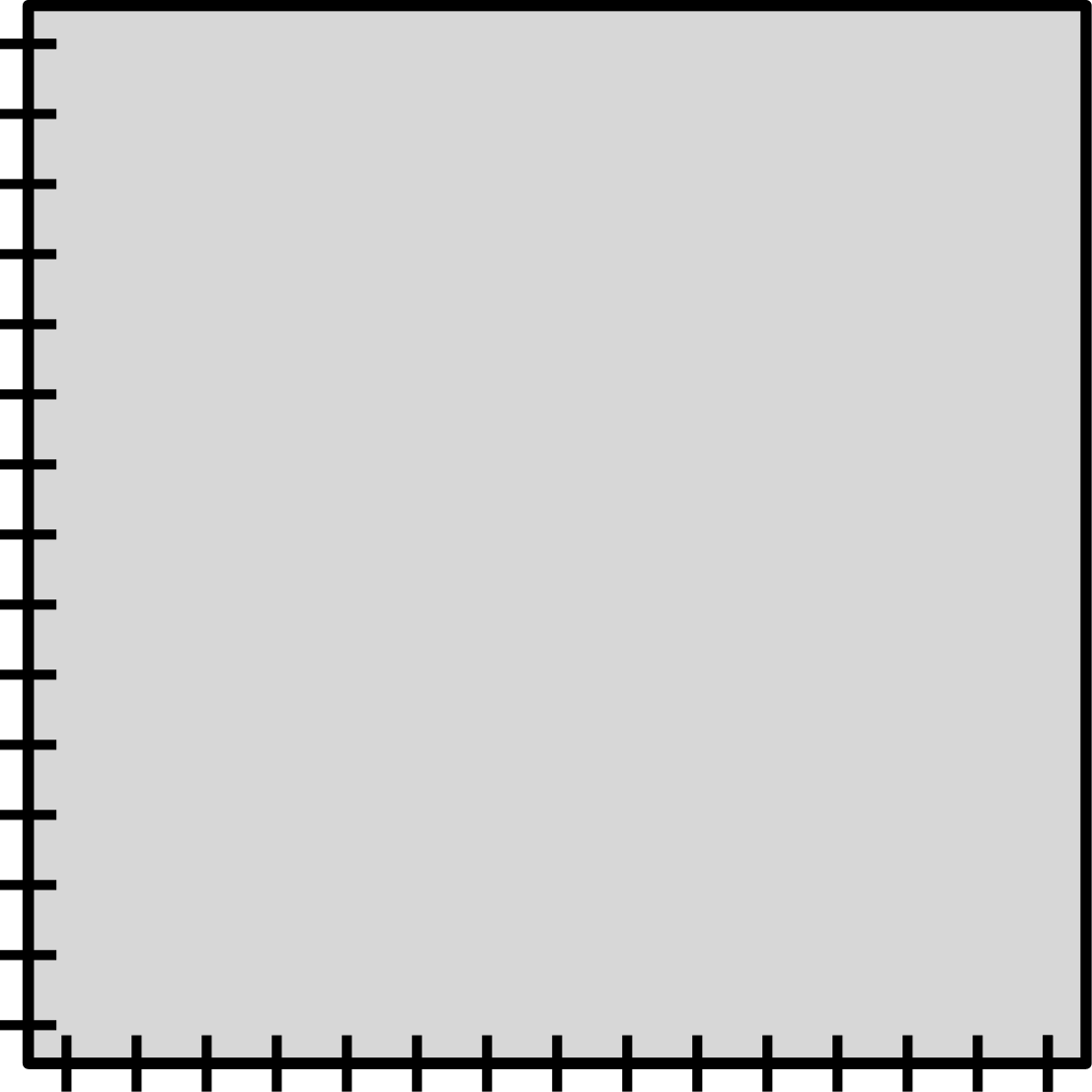}};
			\node[midway] {$\calC$};
			\node[below left=-.15cm and -.8cm of RoughSmooth] {$\mathcal{B}$};
			\node[above right=-.15cm and -.8cm of RoughSmooth] {$\mathcal{D}$};
			\node[above left=-.2cm and -.8cm of RoughSmooth] {$\sigma\in\calM$};
			\node[below right=-.2cm and -.8cm of RoughSmooth] {$\bar{\sigma}\in\calM^{\mathrm{op}}$};
			
			\node[right= .8cm of RoughSmooth] (AlmostRough) {\includegraphics[width=.2\linewidth]{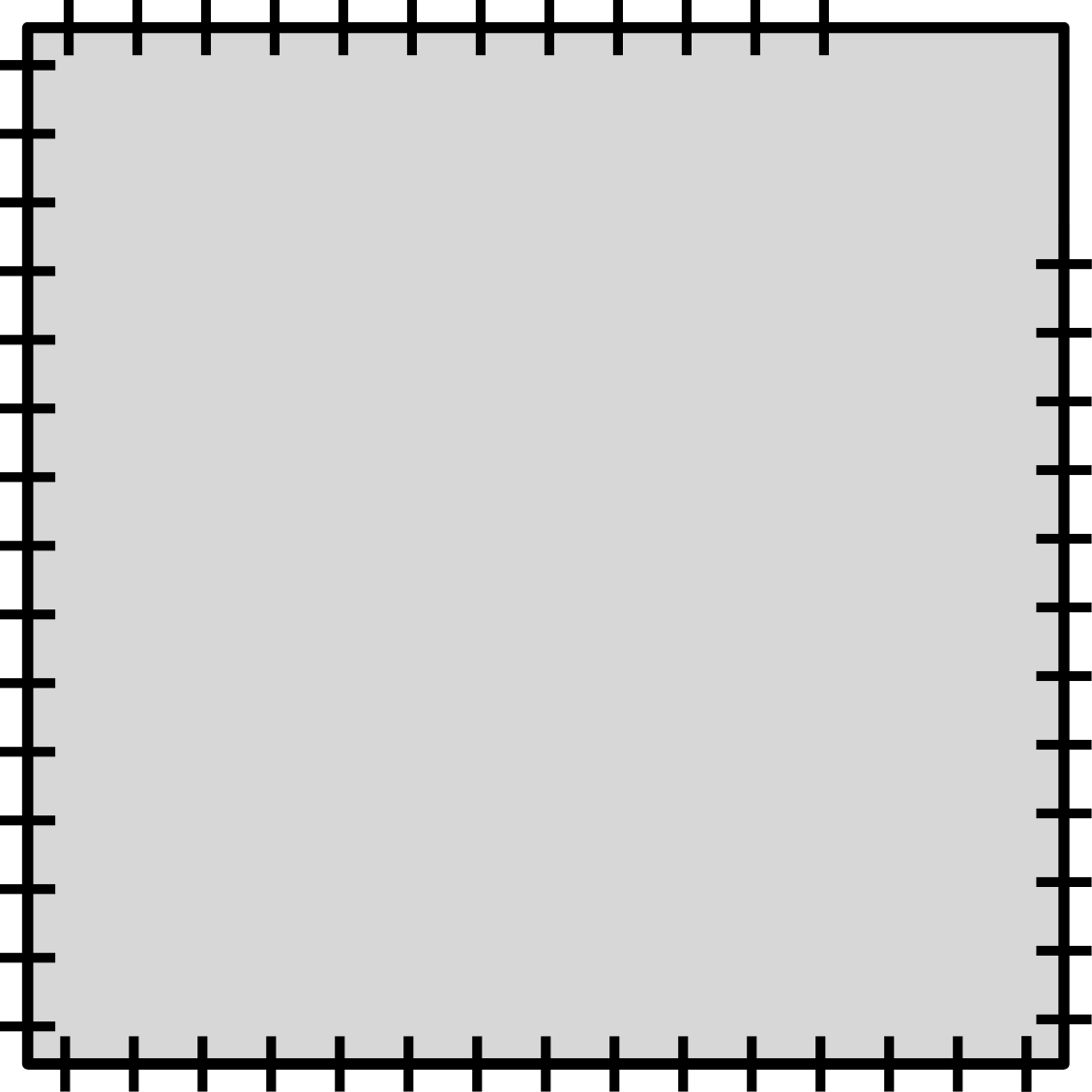}};
			\node[above right=-.8cm and -.2cm of AlmostRough] {$\sigma$};
			\node[above right=-.2cm and -.8cm of AlmostRough] {$\bar{\sigma}$};
			
			\node[below= .8cm of AlmostRough] (Rough) {\includegraphics[width=.2\linewidth]{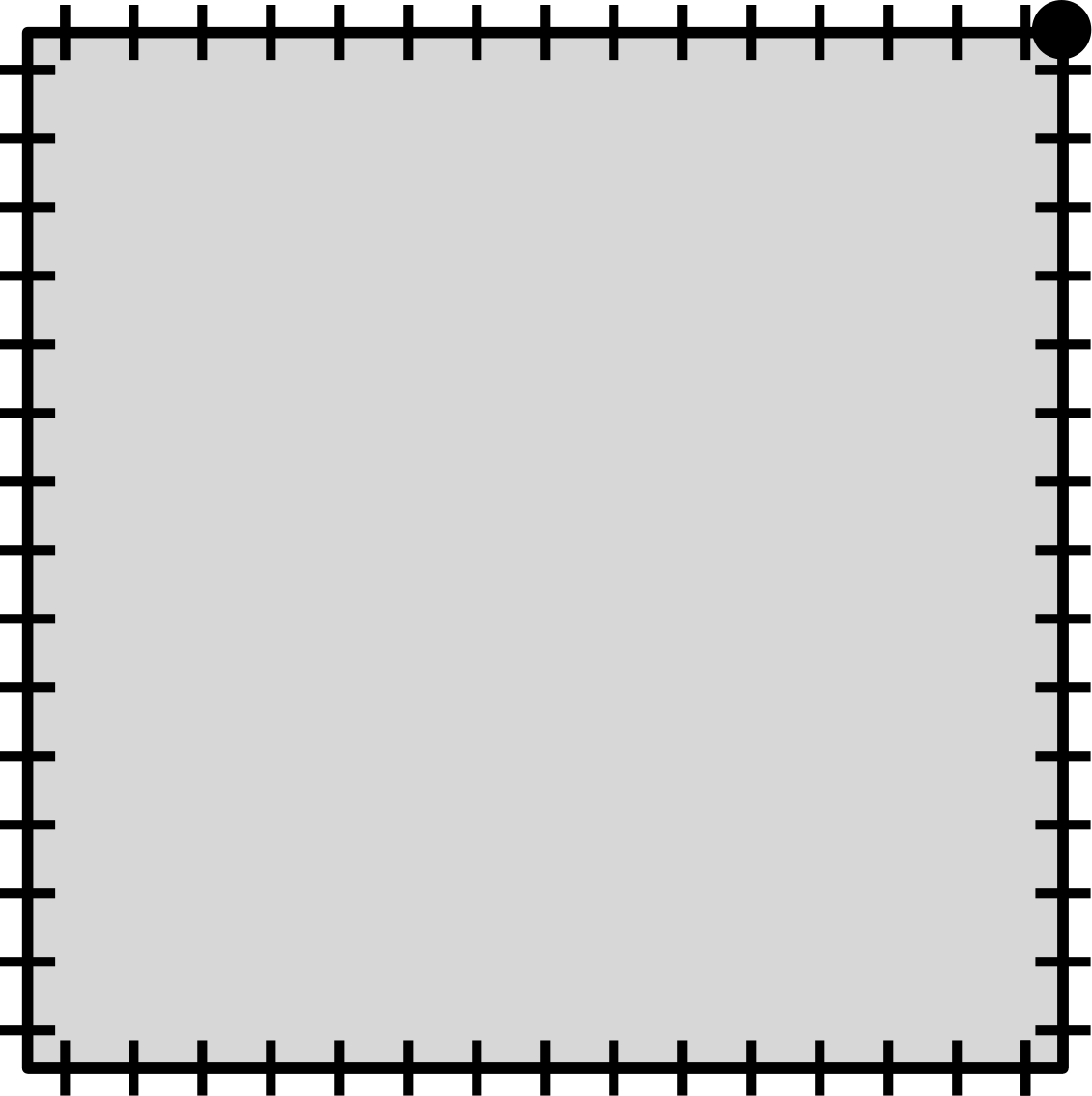}};
			\node[above right=-.2cm and -.7cm of Rough] {$\beta\in\mathcal{B}$};
			
			\node[left= .8cm of RoughSmooth] (AlmostSmooth) {\includegraphics[width=.2\linewidth]{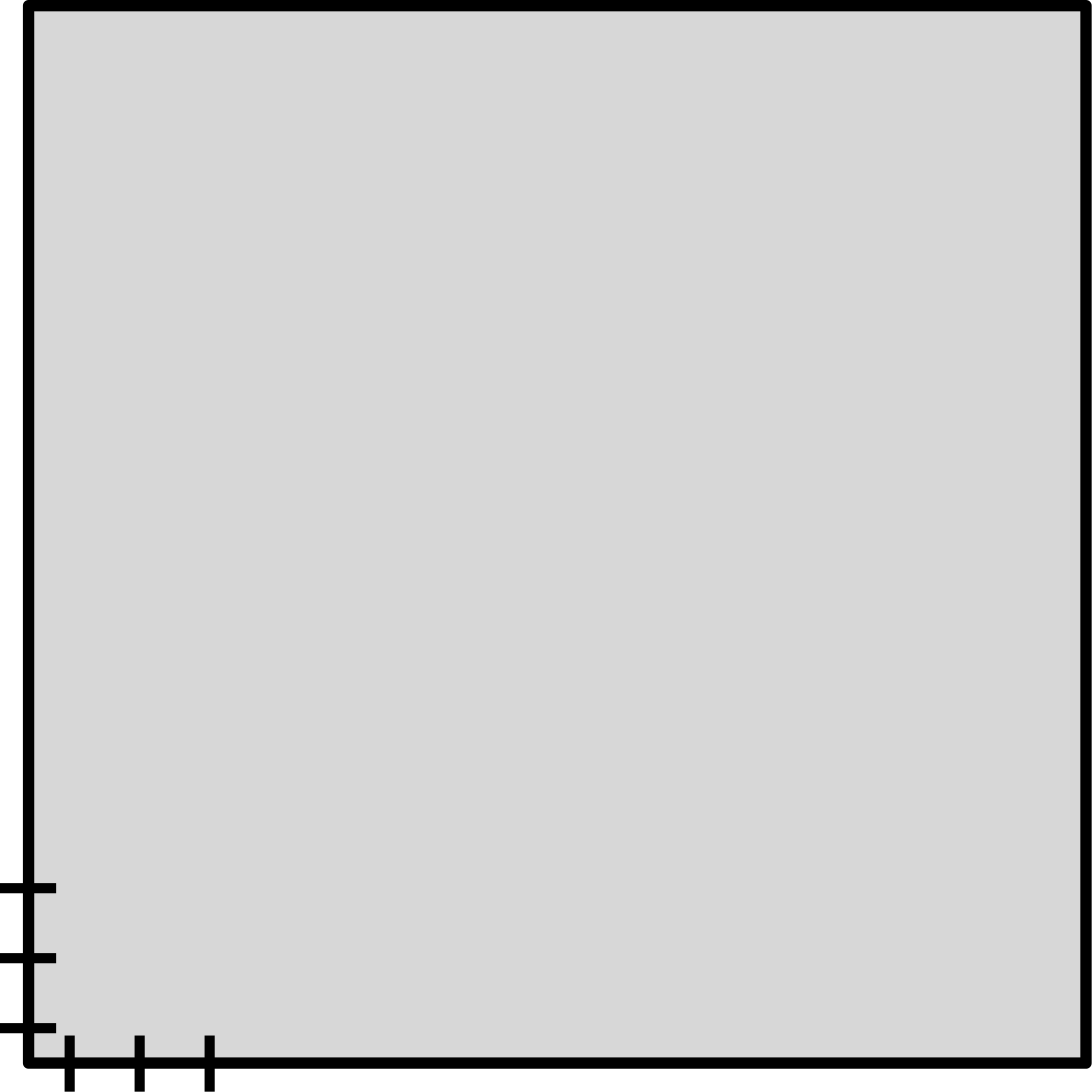}};
			\node[below left=-.7cm and -.2cm of AlmostSmooth] {$\sigma$};
			\node[below left=-.2cm and -.7cm of AlmostSmooth] {$\bar{\sigma}$};
			
			\node[below = .8cm of AlmostSmooth] (Smooth) {\includegraphics[width=.2\linewidth]{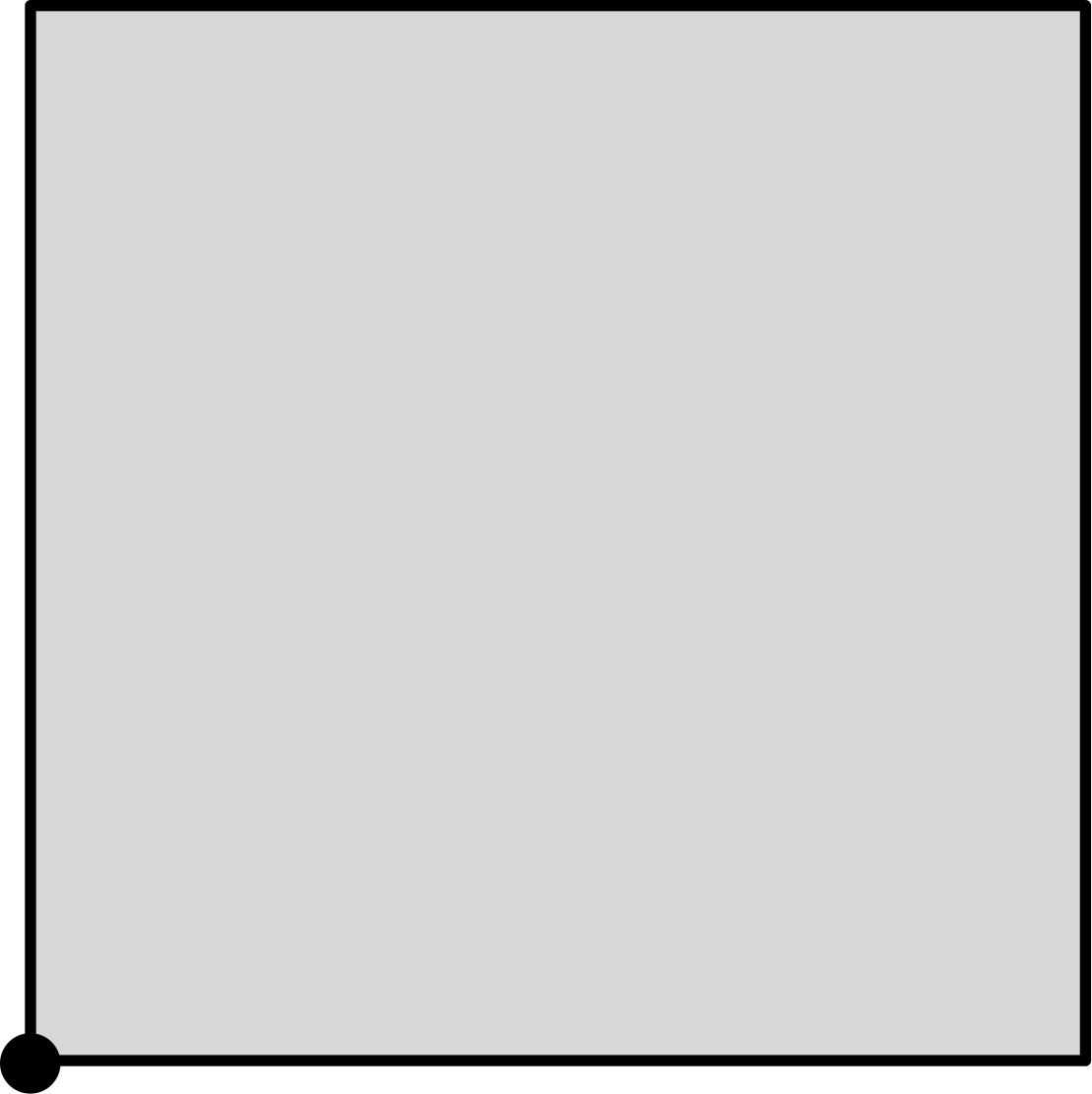}};
			\node[below left=-.2cm and -.7cm of Smooth] {$\beta\in\mathcal{D}$};
			
			\draw[->,thick] ($(RoughSmooth.west)-(0.1,0)$) -- ($(AlmostSmooth.east)+(0.1,0)$); 
			\draw[->,thick] ($(RoughSmooth.east)+(0.1,0)$) -- ($(AlmostRough.west)-(0.1,0)$); 
			\draw[->,thick] ($(AlmostSmooth.south)-(0,0.1)$) -- ($(Smooth.north)+(0,0.1)$); 
			\draw[->,thick] ($(AlmostRough.south)-(0,0.1)$) -- ($(Rough.north)+(0,0.1)$);

		\end{tikzpicture}
		\caption{\label{fig:BoundaryDefect} Defects on the boundary.}
	\end{subfigure}
	\begin{subfigure}[t]{\linewidth}
		\begin{tikzpicture}
			\node[midway] (DefectLine) {\includegraphics[height=1ex]{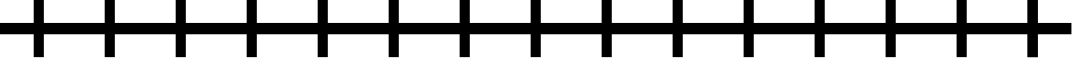}};
			\node[above left=-.2cm and -.8cm of DefectLine] {$\sigma\in\calM$};
			\node[above right=-.2cm and -.7cm of DefectLine] {$\bar{\sigma}\in\calM^\op$};

			\node[right= 1.6cm of DefectLine] (DefectLineShrinked) {\includegraphics[height=1ex]{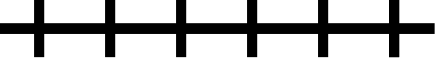}};
			\node[above left=-.2cm and -.3cm of DefectLineShrinked] {$\sigma$};
			\node[above right=-.2cm and -.3cm of DefectLineShrinked] {$\bar{\sigma}$};

			\node[right= 1.6cm of DefectLineShrinked] (quasiparticle) {\includegraphics[height=.7ex]{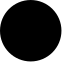}};
			
			\draw[->,thick] ($(DefectLine.east)+(0.3,0)$) -- ($(DefectLineShrinked.west)-(0.3,0)$); 
			\draw[->,thick] ($(DefectLineShrinked.east)+(0.3,0)$) -- ($(quasiparticle.west)-(0.3,0)$);
		\end{tikzpicture}
		\caption{\label{fig:LineDefect} A line defect in the bulk. }
	\end{subfigure}
	\caption{\label{fig:DefectRealization} Illustrating the fusion of point like defects. (a) The two point-like defects at the intersections between two different types of gapped boundaries can be fused into a point-like boundary excitation, by continuously shrinking one of the boundaries. (b) Similarly,  the two defects at the end points of a line-like defect can be fused into a bulk excitation, by shrinking the line segment to a point. }
\end{figure}
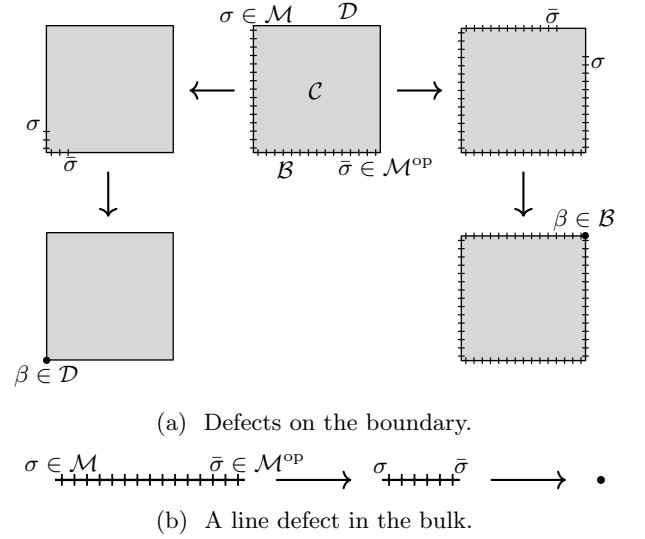

\subsection{Black defects and their classification}\label{sec:blackdefect-realization}
In this section we focus on an important subclass of the module categories mentioned in Definition~\ref{def:cat_paraparticles}, where we give an explicit description of their categorical data, and provide a partial classification along with several examples. We will also briefly mention their physical realization in lattice models.   

We begin with the following definition:
\begin{definition}\label{def:blackdefect}
	Let $\calC$ be a braided or symmetric fusion category describing point-like quasiparticles in a topological phase. A point-like defect in this topological phase is called a black defect if it is described by a module category $\calM$ with exactly one simple object, which we denote by $\sigma\in\calM$.
\end{definition}
It is clear that if $\calM$ describes a black defect, then for any $\psi\in\calC$ we have a fusion rule $\sigma\times \psi=m \sigma$, %
where $m=d_\psi$ is the quantum dimension of $\psi$. The reason is that since $\calM$ has only one simple object $\sigma$, the object  $\sigma\times \psi$ can be nothing but a direct sum of $m$ copies of $\sigma$, and $m=d_\psi$ follows from the associativity of fusion. The physical picture is that a black defect can absorb any particle in $\calC$ without ever changing its state $\sigma$, i.e. all particles in $\calC$ can be locally created and annihilated here. %

Notice that since the fusion multiplicities in Eq.~\eqref{eq:FCrules-modcat} must always be integers, if $\calC$ admits a black defect $\calM$, then all objects of $\calC$ must have integer quantum dimensions. Therefore, fusion categories containing objects with non-integer quantum dimensions, such as the Ising and the Fibonacci fusion categories, cannot admit any black defect.

In the following we give 
more concrete examples of black defects %
and mention their physical realization in lattice systems. 

\subsubsection{Black defects for $\Rep(G)$}\label{sec:blackdefectRep(G)}
We now classify black defects for $\calC=\Rep(G)$ or $\mathrm{sRep}(G,z)$ describing point-like quasiparticles in a 3D topological phase. Indeed, for a given finite group $G$, the classification of module categories over $\Rep(G)$ and $\mathrm{sRep}(G,z)$ are the same, as the fusion structure of the two are the same~(they only differ in the braiding), so we focus on the case $\calC=\Rep(G)$. 

For a finite group $G$, Ref.~\onlinecite{Ostrik2003} proved that indecomposable module categories over $\Rep(G)$ are classified by a pair $(H,\omega)$, where $H$ is a subgroup of $G$ and 
where $\omega\in \mathcal{H}^2(H,U(1))$ is a 2-cocycle of $H$ satisfying $\omega(1,g)=\omega(g,1)=1$ and
\begin{equation}\label{eq:def_projRep}
	\omega(f,g)\omega(fg,h)=\omega(g,h)\omega(f,gh), ~~\forall f,g,h\in H.
\end{equation}
Given such a pair $(H,\omega)$, a module category $\calM$ is constructed as the category of projective representations of $H$ with the cocycle $\omega$, denoted by $\calM=\Rep^\omega(H)$. More explicitly, an object of $\calM$ is a map $\sigma: H\to M_d(\C)$ satisfying
\begin{equation}
	\sigma(g)\sigma(h)=\omega(g,h)\sigma(gh), \quad \forall g,h\in H,
\end{equation}
and $\sigma(1)=\mathds{1}$, where $1\in G$ is the group unit, and $M_d(\C)$ is the algebra of $d\times d$ complex valued matrices. 
Morphisms of $\calM$ are defined as intertwiners between projective representations.  
Fusion between $\sigma\in\calM$ and $\psi\in \calC$ is defined by the tensor product representation
\begin{equation}
	(\sigma\otimes \psi)(g):=\sigma(g)\otimes \psi(g), \forall g\in H.
\end{equation}
Notice that $\sigma\otimes \psi$ is still a projective representation of $H$ with the same cocycle $\omega$, i.e. $\sigma\otimes \psi\in \calM$, giving $\calM$ the structure of a  module category over $\calC$.  

To describe black defects, we now impose the extra condition that $\calM=\Rep^\omega(H)$ has exactly one simple object, which we call $\sigma$. This is equivalent to saying that $H$ has exactly one irreducible projective representation corresponding to the cocycle $\omega$~[Eq.~\eqref{eq:def_projRep}]. Such a 2-cocycle $\omega$ is called non-degenerate. A central type factor group~(CTFG)~\cite{Howlett1982,Shahriari1991CTFG,etingof2000semisimpleTHAclassification} is a finite group that admits a non-degenerate 2-cocycle. 
In this situation, for any $\psi\in\calC$, since $\sigma\otimes \psi\in\calM$ and the semisimple category $\calM$ has only one simple object $\sigma$,  we must have $\sigma\otimes \psi=d_\psi \sigma$, as we argued before. %
Several examples of such groups are given in App.~\ref{sec:RfromCentralType}, where we show how to compute the $R$-matrix explicitly. 

Central type factor groups are the essential mathematical structure that underlies nontrivial $R$-paraparticles in  topological phases. Indeed, any type of nontrivial $R$-paraparticle $\psi$ that can pass the challenge must be equipped with the structure of a CTFG. This can be seen as follows. Let $\psi\in\calC$ be a non-trivial $R$-paraparticle satisfying the condition in Definition~\ref{def:cat_paraparticles}, and let $\calC_\psi$ be the symmetric fusion subcategory of $\calC$ generated by $\psi$~(the smallest fusion subcategory of $\calC$ containing $\psi$). Since any object $\varphi$ of $\calC_\psi$ is contained in $\psi^{\otimes n}$ for some integer $n$, and fusing $\sigma$ with $\psi$ does not change $\sigma$~[due to the fusion rule Eq.~\eqref{eq:sigmapsifusion-0}], fusing $\sigma$ with $\varphi$ cannot change $\sigma$ either, therefore we must have $\sigma\times \varphi=d_\varphi \sigma~\forall \varphi\in\calC_\psi$. Now let $\calM_\sigma$ be the linear subcategory of $\calM$ that has only one simple object $\sigma$. Then $\calM_\sigma$ is a module category over $\calC_\psi$ with exactly one simple object. According to the classification results for SFCs and their module categories, we must have $\calC=\mathrm{sRep}(G,z)$ and $\calM_\sigma=\Rep^\omega(H)$, where $H\leq G$ is a CTFG and $\omega$ is a non-degenerate 2-cocycle of $H$.  
In a future work we present an algorithm that constructs the underlying CTFG directly from a given $R$-matrix. %

We finally mention that given a non-Abelian CTFG $H$, we can construct an SFC $\calC$ that satisfies the winning condition in Sec.~\ref{sec:FCanalysis_summary}, where $\sigma\in\calC$ is a quasiparticle in the bulk. 
This is based on the close connection between CTFGs and the so-called groups of central type~\cite{Howlett1982}, 
which we present in App.~\ref{sec:RfromCentralType}, where several examples are given.   %

\subsection{The relative meaning of the $R$-matrix}\label{sec:weakequivalence}
We note that the categorical definition of $R$-matrix given in Eq.~\eqref{def:RfromSFC} depends not only on the type of paraparticle $\psi$, but also on the black defect $\sigma$ at which $\psi$ can be locally created and measured. This is why in several occasions we used the terminology ``the $R$-matrix of $\psi$ with respect to a black defect $\sigma$''. 
In particular, it is possible that the same type of particle $\psi$ in a topological phase can be described by different $R$-matrices with respect to different defects, and in App.~\ref{app:ex-diffR} we give an explicit example, where a particle has a nontrivial $R$ with respect to defect $\sigma$ but has a trivial $R=-X$ with respect to defect $\sigma'$. This means that the $R$-matrix encodes more information about the statistics of the particle $\psi$--it also encodes some information about the defect $\sigma$ at which $\psi$ can be locally created and measured. The implication of this on the theoretical foundation of $R$-parastatistics and the relation to no-go theorems will be discussed more extensively in a future paper~\cite{wangOnRparaI}.

Importantly, in Definition~\ref{def:cat_paraparticles}, we call $\psi$ a nontrivial $R$-paraparticle as long as there exists a defect $\sigma$ with respect to which $R$ is nontrivial. If $\psi$ has a trivial $R$-matrix with respect to any possible black defect in a certain topological phase, we call it trivial~(a fermion or a boson), an example is the boson in $\Rep(S_3)$ with quantum dimension $d=2$. A trivial particle cannot pass the challenge game in a noise-robust way no matter what kind of defect we put at $\oA$ and $\oB$, while a nontrivial paraparticle can, provided that we choose the defect $\sigma$ carefully. 

The fact that a nontrivial $R$-paraparticle can also be described by a trivial $R$-matrix is a manifestation of the fact that it is categorically a fermion or a boson in its own right. This means that a system of $R$-paraparticles will be physically indistinguishable from ordinary fermions or bosons if we are only allowed to do physical operations in the bulk where no defect is present. We can only see the nontrivial effect of parastatistics~(e.g., winning the challenge) if we bring some of these particles close to a nontrivial defect $\sigma$ with respect to which $R$ is nontrivial.  

Here is a slightly different perspective on the nontriviality of $R$-parastatistics in topological phases.  
The braiding of a SFC $\calC$ is trivial in its own right, %
but this trivial braiding can become nontrivial when it is projected onto the fusion space between particles in $\calC$ and a defect $\sigma\in\calM$, and the secret communication challenge is an explicit demonstration of this nontriviality at the physical level. This summarizes the relative meaning of $R$-parastatistics in the context of topological order.

\subsection{A hierarchy of nontrivial exchange statistics in 3D topological phases}
We have seen in Tab.~\ref{tab:protocols_strategies} that there is a hierarchy of nontrivial $R$-paraparticles depending on which version of the challenge game a certain type of $R$-paraparticle can pass~(i.e., what extra requirement we put on the winning condition). 
It is then natural to ask if we can translate this hierarchy of nontriviality conditions on $R$ into conditions on the underlying group $G$ of the SFC $\calC=\mathrm{sRep}(G,z)$, and classify finite groups according to which level of the Tab.~\ref{tab:protocols_strategies} the particles in $\calC$ can reach. Unfortunately, we do not yet have a complete answer to this, as such a classification seems to involve some hard problems in group theory. Below we present some general facts that we know. 
In this section we will always require passing the identical particle test and the antiparticle test to simplify our classification scheme. For convenience, we define the parastatistics level of a group $G$~[denoted by $l(G)$] to be the maximal level~(see Tab.~\ref{tab:protocols_strategies}) any particle $\psi\in \calC=\mathrm{sRep}(G,z)$ can reach with respect to any possible black defect $\sigma$ over $\calC$, and for any central element $z\in G$ satisfying $z^2=1$.   

\subsubsection{Level 1 and 2~(trivial): only ordinary fermions and bosons}
At levels 1 and 2, any $\psi\in\mathrm{sRep}(G,z)$ is either an ordinary boson or fermion, in that it has a trivial $R$-matrix $R=\theta X$ with respect to any possible black defect in the system, where $\theta=\psi(z)=\pm 1$.  
Such particles cannot pass the challenge game in a noise-robust way. 
Here is a simple sufficient condition for $G$ to belong to one of these two levels: if $G$ is Abelian, or $G$ is non-Abelian but $G$ does not have any nontrivial CTFG subgroup $H$. In the former case, any simple $\psi\in\mathrm{sRep}(G,z)$ has quantum dimension $m=1$, implying $R=\psi(z)=\pm 1$.  In the latter case, the only possible black defect is the module category $\calM=\Vect$ corresponding to the trivial forgetful functor $\mathrm{Forg}: \mathrm{sRep}(G,z)\to\Vect$, which produces the trivial $R$-matrix $R=\theta X$. 
Examples of such non-Abelian groups include $G=S_3$ or $G=Q_4$~(the quaternion group of order 6). 
Such a group $G$ can reach at most level 2, depending on whether $G$ has a central element of order 2. 

\subsubsection{Level 3 and 4: minimal extensions of fermions and bosons}
At levels 3 and 4, some particles in $\mathrm{sRep}(G,z)$ have nontrivial $R$-matrices of the swap-type with respect to certain black defects. We call such particles minimal extensions of fermions and bosons. 
According to our argument above, in order to reach at least level 3, $G$ must be non-Abelian and has a non-trivial CTFG subgroup $H$. 
Fact.~\ref{fact:AbCTFGSWAP} in App.~\ref{app:AbCTFG} implies that if all CTFG subgroups of $G$ are Abelian, then $G$ can reach at most level 4. 
An example of level 3 is the dihedral group $D_8$, see App.~\ref{app:AbCTFGD8}. An example of level 4 is the alternating group $A_4$, see App.~\ref{app:CTFGA4}. 

\subsubsection{Level 5: full-fledged paraparticles}
At level 5, there exists $\psi\in\mathrm{sRep}(G,z)$ that can pass both the basic challenge and the who-entered-first challenge in a fully robust way, and is therefore called a full-fledged paraparticle. In order to reach level 5, 
$G$ must have a non-Abelian CTFG subgroup; however, this is not a sufficient condition. Examples of level 5 include $A_4\times Z_3$ and $D_8\ltimes Z_2^{\times 3}$~(given in Apps.~\ref{sec:A4Z3} and \ref{app:G64}, respectively), both are CTFGs themselves.

\section{The anti-anyon twist in 2D}\label{sec:anti-anyon}
In Sec.~\ref{sec:TC_winning_condition}  we have shown that only nontrivial paraparticles can win the full version of the challenge game in 3+1D. In 2+1D, however, there exists a special class of non-Abelian anyons that can also win the game. Indeed, consider a 2+1D topological phase described by some modular tensor category $\calC$, and suppose that there is a simple particle type $\psi\in \calC$ with a fusion rule of the form in Eq.~\eqref{eq:sigmapsifusion-0}, where $\sigma$ is an arbitrary point-like defect. Then it is clear that the categorical description of the winning strategy presented in Sec.~\ref{sec:win_strategy_diagrammatic}, in particular the fusion diagram in Eq.~\eqref{eq:SFCdescriptiongame}, still applies. The only difference here is that the $R$-matrix in Eq.~\eqref{eq:SFCdescriptiongame} is no longer involutive~(i.e. $R^2\neq \mathds{1}$), and one has to distinguish upper and lower crossing between the two $\psi$ lines, depending on whether the braid between Alice and Bob is counterclockwise or clockwise; but in any case, the winning strategy works as long as $R^{b'a'}_{ab}$ does not factorize as $p_{aa'}q_{bb'}$. Such a winning strategy also satisfies all the requirements in Sec.~\ref{sec:parachallenge_variants_twists}.

In order to single out paraparticles from more general non-Abelian anyons, in this section we introduce some extra twists~(which we call ``anti-anyon twists'') to the game protocol that prevent non-Abelian anyons from winning. We present the twisted game protocols in Sec.~\ref{sec:AAtwists}, and in Sec.~\ref{sec:AAtwist_examples} we give several examples of non-Abelian anyons that can pass the original challenge but get blocked by these anti-anyon twists. Unfortunately, we will see that there still exists a very special type of non-Abelian anyons that can pass all these anti-anyon twists. In Sec.~\ref{sec:interference_test} we describe an interference experiment that physically distinguish paraparticles from this special type of non-Abelian anyons. 
\subsection{Twisting the game protocols with full braids}\label{sec:AAtwists}
Our guiding principle for designing the anti-anyon twists is to exploit the fundamental distinction between paraparticles and non-Abelian anyons according to their definitions: since paraparticles satisfy $R^2=\mathds{1}$, a full braid between two paraparticles does not change the quantum state of the whole system, while a full braid between two anyons leads to a nontrivial unitary evolution. Therefore, if we twist the game protocol by adding some random full braids between the players' circles, we can scramble anyon-based strategies without complicating the paraparticle-based strategies.
Below in Sec.~\ref{sec:simpletwist} we first illustrate this idea by adding a simple twist to the original two players challenge, which can already block a class of non-Abelian anyons. Then in Sec.~\ref{sec:immunity} we introduce a more powerful twist involving more players that play as scramblers. The effectiveness of these anti-anyons twists will be analyzed in Sec.~\ref{sec:AAtwist_examples} through specific examples. 
\begin{figure}
	\begin{subfigure}[t]{.45\linewidth}
		\centering\includegraphics[width=.8\linewidth]{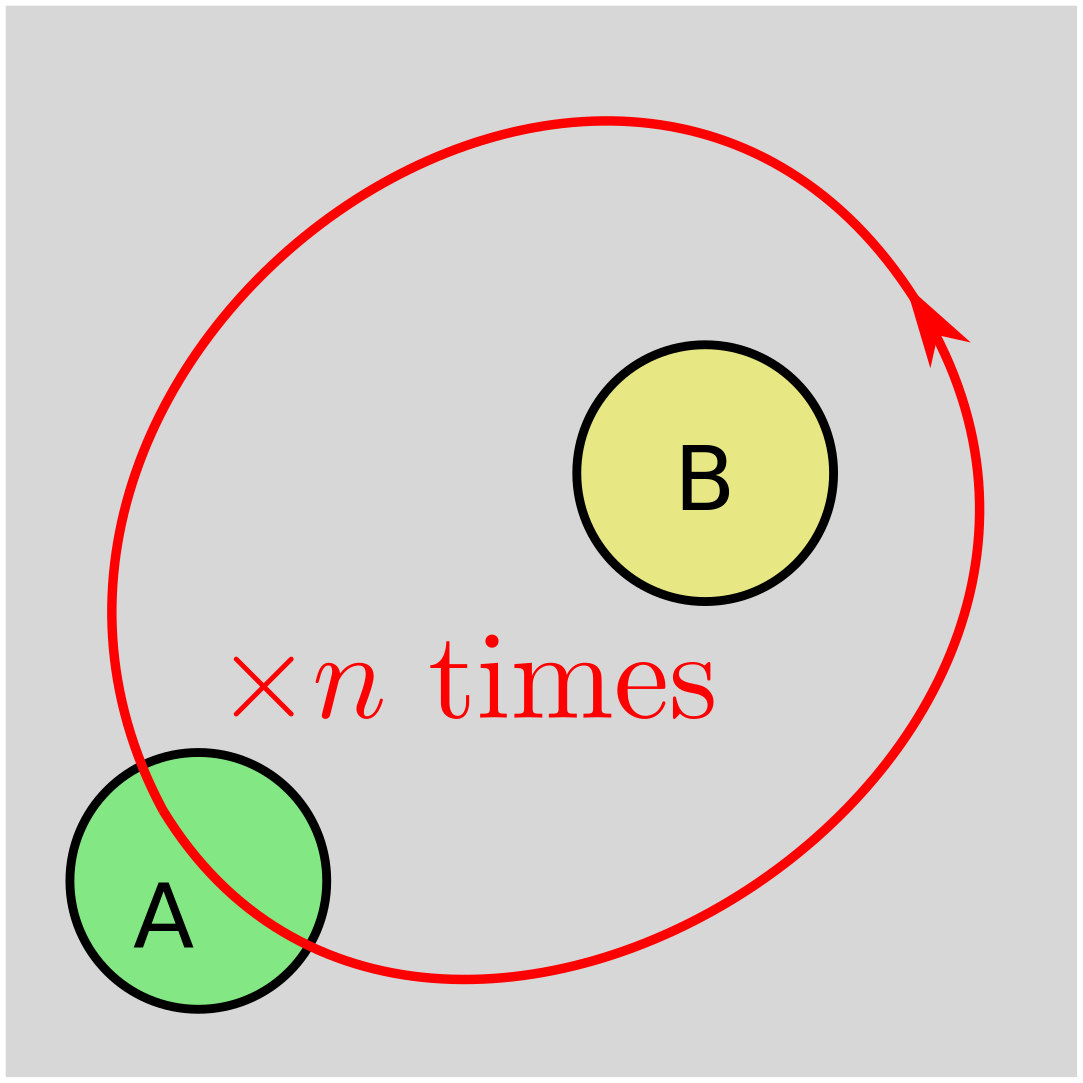}
		\caption{\label{fig:AAtwist-protocol}}
	\end{subfigure}
	\begin{subfigure}[t]{.48\linewidth}
		\centering\includegraphics[width=.7\linewidth]{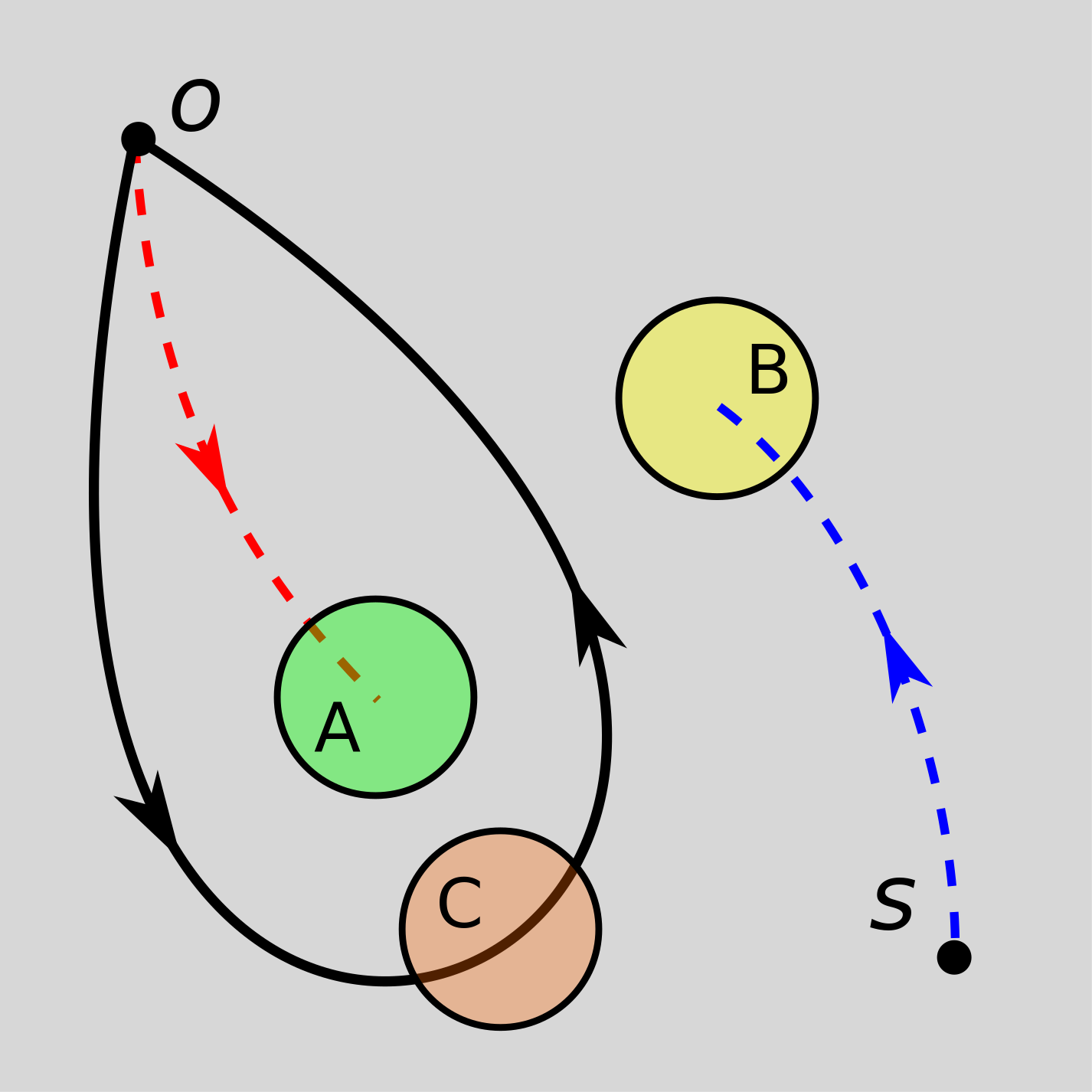} 
		\caption{\label{fig:Immunity-Game} }
	\end{subfigure}
	\caption{\label{fig:AAtwist} Illustrating game protocols for the anti-anyon twists in 2D. (a) The simple twist; (b) Adding a full braid scrambler.  }
\end{figure}
\subsubsection{Twisting the basic challenge}\label{sec:simpletwist}
A simple twist can be performed somewhere in the middle of the game~(say at $t=T/2$), when both circles are deep in the bulk, as shown in Fig.~\ref{fig:AAtwist-protocol}. At this point the Referees temporarily halt the movement of circle B, and braid the circle A around circle B in counterclockwise direction for $n$ times, where $n\in\mathbb{Z}$ is chosen randomly.   
Importantly, throughout this process Alice still needs to confine the excitations inside her circle, and both players do not know the value of $n$~\footnote{One may worry that Alice may be able to calculate $n$ by carefully analyzing her circle movements in this process. To prevent this, we can add another round of the twist in which the role of Alice and Bob are exchanged, i.e., this time  the circle B braids around circle A in counterclockwise direction for $k$ times. In this way the two circles complete $n+k$ full braids. But Alice does not know $k$, while Bob does not know $n$, so neither of them know $n+k$.}. 
After this, the game proceeds as before. 

We now briefly analyze how this simple twist affects the winning strategy presented in Secs.~\ref{sec:win_original} and  \ref{sec:win_strategy_diagrammatic}. 
Taking into account the added $n$ full braids, the time evolution of the whole system is now described by the space-time diagram 
\begin{equation}\label{eq:SFCdescription-AAtwist}
	\adjincludegraphics[height=14ex,valign=c]{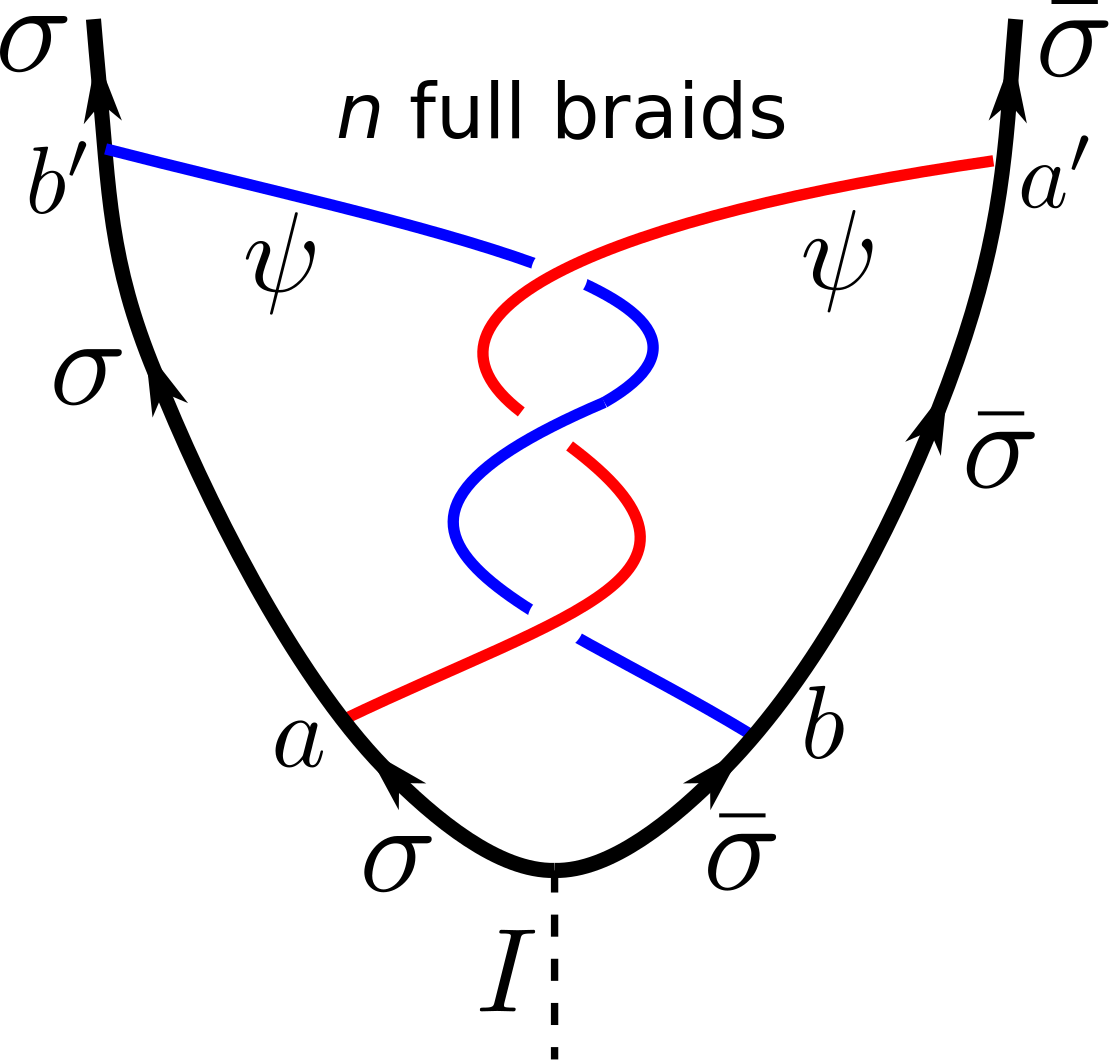}
	=[R^{2n+1}]^{b'a'}_{ab}~\adjincludegraphics[height=14ex,valign=c]{Figures/SFC/SFCR-ssb-RHS.png},
\end{equation}
where the RHS is obtained using a derivation similar to Eq.~\eqref{eq:SFC2pt-1pt}. 
At $t=T$, the density matrix for the internal states of the two particles is
\begin{eqnarray}\label{eq:reducedDM_B}
	\rho_{BA}(T)&=&\sum_n p_n R^{2n+1}(\rho_A\otimes\rho_B)R^{-2n-1},
\end{eqnarray}
where  $R^{2n+1}$ is considered as a linear map from $V_A\otimes V_B$ to $V_B\otimes V_A$, and $V_A$~($V_B$) is the internal space~(the fusion space between $\sigma$ and $\psi$) of Alice's~(Bob's) particle, and $p_n$ is the probability distribution for $n$ decided by the Referees. Here $\rho_A$ and $\rho_B$ are the initial states of the two quasiparticles prepared by Alice and Bob, respectively, which we allow to be mixed states. 
The reduced density matrix for each particle's final state is
\begin{equation}
	\rho'_A=\Tr_B[\rho_{BA}(T)],\quad \rho'_B=\Tr_A[\rho_{BA}(T)],
\end{equation}
and $\Tr_A$ means the partial trace over $V_A$.
In Sec.~\ref{sec:simpletwistexample} we give an example where, by choosing a suitable probability distribution $\{p_n\}$, we can make $\rho_A'$ independent of $\rho_B$, and $\rho_B'$ independent of $\rho_A$, thereby making both players fail the game.  %

\subsubsection{Immunity to  full-braid scramblers}\label{sec:immunity}
We now design a more powerful anti-anyon twist by adding more players that act as full-braid scramblers. The twisted game protocol is depicted in Fig.~\ref{fig:Immunity-Game}. 
Similar to the simple twist, this twist is also performed when both circles are deep in the bulk, say at $t=T/2$. At this point a new player, Charlie, who plays against Alice and Bob, enters the game. Charlie's circle starts from $\oA$, then completes $n$ full braids around circle A~(Fig.~\ref{fig:Immunity-Game} shows the $n=1$ case), and finally returns back to $\oA$ and disappears. Here, Charlie is required to use the same type of particle $\psi$ in her circle~(more precisely, the total fusion channel of circle C is required to be the same as that of circle A). %
The integer $n$ is chosen by Charlie in advance, and both Alice and Bob do not know the value of $n$. This describes one round of the scrambling twist. Charlie is given the right to demand as many rounds of this full-braid scrambling twist
as she wants, and she can choose the value of $n$ independently in each round. [More generally, we can allow more scramblers to be simultaneously present in the game, and we allow a scrambler to start and end their journey at $\oB$ as well~(but a scrambler cannot start at $\oA$ and end at $\oB$ or vice versa).]  After all the scrambling twists are finished, the game proceeds as before in Sec.~\ref{sec:twoptversion}. 

The  above process can be described by the following space-time diagram~(for the $n=1$ case)
\begin{equation}\label{eq:SFCdescription-immunityFBS}
	\adjincludegraphics[height=14ex,valign=c]{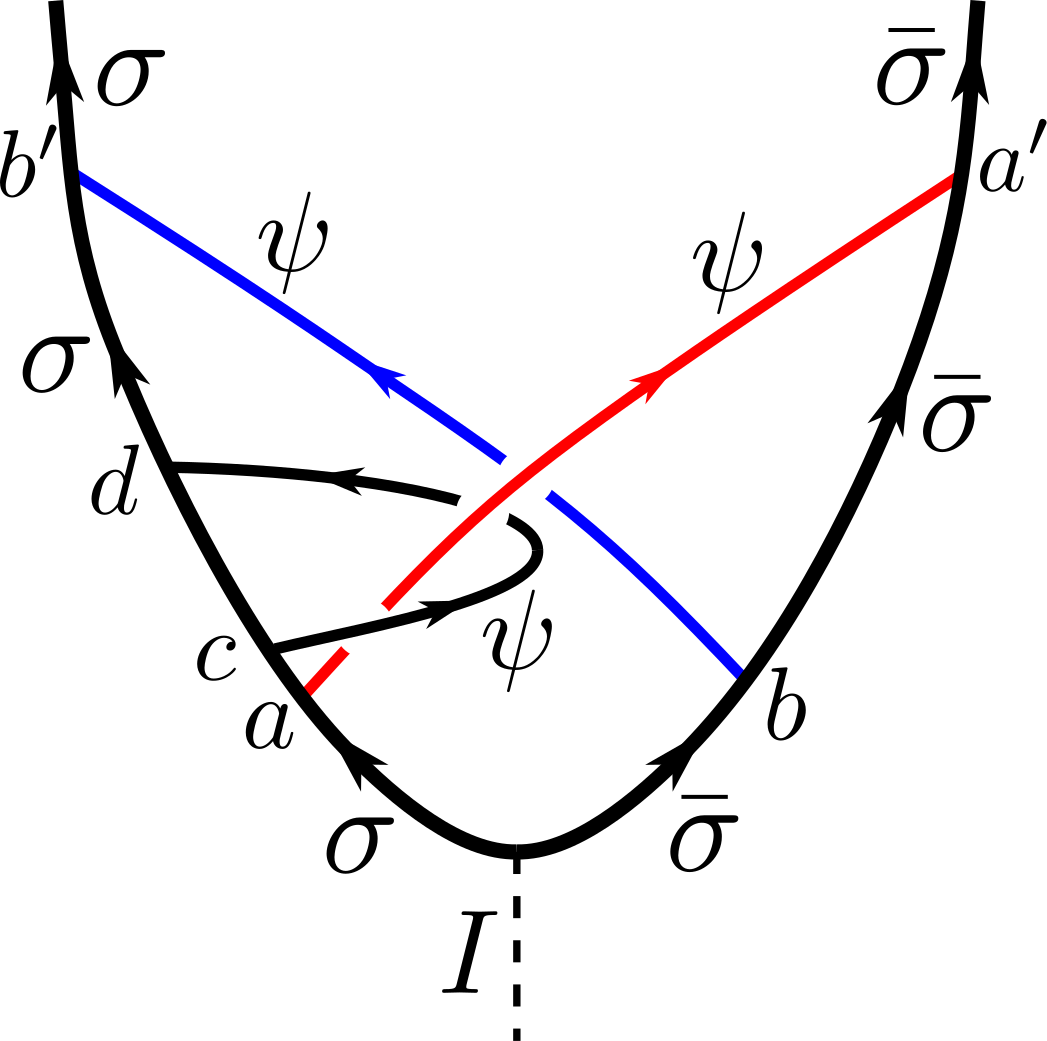} =
	\begin{tikzpicture}[baseline={([yshift=-.8ex]current bounding box.center)}, scale=0.5]
		\Rmatrix{0}{2*\AL}{R}
		\Rmatrix{-2*\AL}{0}{R^2}
		\node  at (-\AL,-1*\AL) [below=-.04]{\footnotesize $a$};
		\node  at (\AL,1*\AL) [below=-.04]{\footnotesize $b$};
		\node  at (-3*\AL,-\AL) [below=-.04]{\footnotesize $c$};
		\node  at (\AL,3*\AL) [above=-.04]{\footnotesize $a'$};
		\node  at (-\AL,3*\AL) [above=-.04]{\footnotesize $b'$};
		\node  at (-3*\AL,\AL) [above=-.04]{\footnotesize $d$};
	\end{tikzpicture}
	\adjincludegraphics[height=14ex,valign=c]{Figures/SFC/SFCR-ssb-RHS.png},
\end{equation}
where the black curve depicts Charlie's trajectory, and the RHS is derived using the methods introduced in Sec.~\ref{sec:win_strategy_diagrammatic}. The effect of this twist on the winning strategy will be analyzed through examples in Sec.~\ref{sec:FBSexample}.

\subsection{Examples}\label{sec:AAtwist_examples}
In the following we consider several examples of non-Abelian anyons and understand how the anti-anyon twists introduced in Sec.~\ref{sec:AAtwists} prevent them from winning. The non-Abelian anyons we consider here appear in Kitaev's quantum double models~\cite{kitaev2003fault} based on some finite non-Abelian group $G$, or more formally, they appear as simple objects in the modular tensor category $\Rep[D(G)]$, the category of finite dimensional representations of Drinfeld's quantum double $D(G)$. 
As a qualified physical system for the challenge game, we consider the quantum double model on a 2D lattice with the hybrid open boundary shown in the middle figure of Fig.~\ref{fig:BoundaryDefect}. With this hybrid boundary, the players can choose the special points $\oA$ and $\oB$ to be the upper left and lower right corners, which form black defects $\sigma\in \calM$ and $\bar{\sigma}\in\calM^\op$, respectively, as we mentioned in Sec.~\ref{sec:blackdefectQD}. For any simple particle type $\psi\in \Rep[D(G)]$, we have the defect fusion rules
Eqs.~(\ref{eq:sigmapsifusion-0},\ref{eq:sigmapsifusion-1}), and Eq.~\eqref{eq:SFCdescriptiongame} gives the $R$-matrix, which is generally non-involutive for anyons.  
All the categorical data and detailed computations can be found in the accompanying Mathematica code~\cite{MmaCode-SFC}. 
Note that in this section, whenever we say that a certain type of anyon can or cannot pass a certain game scenario, we implicitly mean that it can or cannot pass with the specific choice of defects $\sigma$ and $\bar{\sigma}$ we mentioned above. It is in principle possible that a certain type of anyon cannot pass a certain challenge with one type  of defect $\sigma$ but can pass with another type of defect $\sigma'$. After all, our goal here is to help understand how the anti-anyon twists work, not to give a comprehensive analysis about the capability of a specific type of anyon. 
\subsubsection{Anyons that pass the original challenge but get blocked by the simple twist}\label{sec:simpletwistexample}
Consider the quantum double model $D(G)$ with $G=S_3$, the symmetric  group of order 6. This model has a non-Abelian anyon with quantum dimension $m=2$ and topological twist factor $\theta=e^{4\pi i/3}$. The $R$-matrix is
\begin{equation}
	R^{b'a'}_{ab}=\delta_{aa'}\delta_{bb'}\omega^{ab},
\end{equation} 
where $\omega=e^{2\pi i/3}$, and $a,b,c,d\in\{1,2\}$.
Since the phase factor $\omega^{ab}$ does not factorize as $\theta_a\phi_b$ for some $\theta$ and $\phi$, the players can win the game according to our discussion in Sec.~\ref{sec:win_2pt}. However, notice that since $R^3$ is simply the swap gate, the Referees can prevent the players from winning by choosing the probability distribution in Eq.~\eqref{eq:reducedDM_B} to be $p_n=\delta_{n,1}$, as this leads to $\rho_A'=\rho_A$ and $\rho_B'=\rho_B$. 

\subsubsection{Anyons that pass the simple twist but get blocked by full-braid scramblers}\label{sec:FBSexample}
Consider again the quantum double model $D(G)$ with $G=S_3$. This time we take $\psi$ to be the non-Abelian anyon with quantum dimension $m=3$ and topological twist factor $\theta=+1$. The $R$-matrix is
\begin{equation}\label{eq:DS3anyon-7}
	R=X\sum_{a=1}^3 P_a\otimes Q_a,
\end{equation}
where $P_a$ is the projector to $\ket{a}$, and $Q_a$ swaps the two states $\ket{(a\pm 1) \mathrm{mod}3}$ while leaves $\ket{a}$ invariant~(here the range of $(x~\mathrm{mod}~3)$ is taken to be in $\{1,2,3\}$). For example, $R\ket{1}\otimes \ket{2}=\ket{3}\otimes \ket{1}$, and  $R\ket{3}\otimes \ket{3}=\ket{3}\otimes \ket{3}$. In other words, $R$ applies the gate $Q_a$ to Bob's qutrit controlled by Alice's qutrit.  Notice that we have $R^3=\mathds{1}$ and $R^5=R^{-1}=XRX$.

With this type of $R$-matrix, even after adding the simple twist protocol in Sec.~\ref{sec:simpletwist}, the players can still transfer a nonzero amount of information to each other no matter how the Referees choose the probability distribution $\{p_n\}$ in Eq.~\eqref{eq:reducedDM_B}--indeed, by using two layers of this system, the players can have a winning strategy with a 100\% success rate. The detailed algorithm is quite technical and is given in App.~\ref{app:winD(S3)}. 

We now show that this type of non-Abelian anyon is completely blocked by the full-braid scrambler twist introduced in Sec.~\ref{sec:immunity}. Here, having one scrambler~(Charlie) with four rounds of the scrambling twists is enough to completely destroy the winning strategy. A possible scrambling strategy is as follow: in the first round, Charlie creates a $\psi$ with internal state $\ket{c}$ with $c=1$ at $\oA$, and completes $n=1$ full braid around circle A, then come back to $\oA$ and measure the internal state of her particle, and obtain $d$. According to the $R$-matrix in Eq.~\eqref{eq:DS3anyon-7}, we have $R^2=XRX$, and after Charlie's measurement, the internal state of Alice's particle collapse to $\ket{a}$ with $a\in\{1,2,3\}$, and Charlie can calculate $a$ from $c,d$ and $R$. In the second round, Charlie creates a $\psi$ with a mixed internal state $\rho_C=\mathds{1}/3$, and completes $n=2$ full braids around circle A, and go back to $\oA$. After this round, the internal state of Alice's particle evolves to the mixed state $\mathds{1}/3$, i.e. any information initially stored by Alice in the internal state of her particle is completely erased by the full-braid scrambler. Using two more rounds of full-braid scrambling, Charlie can similarly erase the information stored in Bob's particle, thereby completely destroying the winning strategy.

\subsubsection{Anyons that are immune to full braid scramblers}
Although the full-braid scrambler twist is much more powerful and can prevent many types of non-Abelian anyons from winning, there is still one special class of non-Abelian anyons that are stubbornly immune to it. This class of non-Abelian anyons have a nontrivial $R$-matrix~(i.e., not of the trivial product form) that satisfies $R^2=e^{i\theta}\mathds{1}$, where $e^{i\theta}\neq 1$ is a phase factor. For example, the fusion product of a nontrivial paraparticle and an Abelian anyon will generally have such an $R$-matrix. 
This type of $R$-matrix is completely immune to full-braid scrambling, as the global phase factor $e^{i\theta}$ will never show up in measurement. To some extent, we can also call this type of non-Abelian anyons ``paraparticles'', as they realize a representation of $S_n$ up to a phase~\footnote{Note that we do not call this a projective representation since the associated 2-cocycle is trivial.}. 
Perhaps the only way to separate this type of non-Abelian anyons from paraparticle is to do interference experiments, which we describe in the next section.  

\subsection{Interference tests}\label{sec:interference_test}
We now describe an interference test~(based on anyon interferometry experiments~\cite{Wen1997FQHinterferometer,Camino2005Realization,Stern2006Proposed,BondersonKitaev2006,bondersonInterferometryNonAbelianAnyons2008,Nayak2009Interferometric,Feldman_2021}) that can in principle completely distinguish anyons and paraparticles in 2+1D systems. However, this experiment requires the players to have a dynamics Hamiltonian under the evolution of which the particles can propagate in space, and it is cumbersome to precisely formulate all the requirements for such a Hamiltonian. In the following we only explain the basic idea without worrying about technical details. 

The players are required to submit a locally interacting Hamiltonian $\hat{H}(J)$ depending on a tuning parameter $J$ that satisfies: \\
(1). The players can win the  challenge game with one special point $\oA$, using a topological excitation $\psi$ of $\hat{H}(0)$;\\ %
(2). At some $J\neq 0$, where $\hat{H}(J)$ is still in the same phase as $\hat{H}(0)$, the spectrum of the particle $\psi$ has nontrivial dispersion relation and  therefore $\psi$ can
 propagate in space under the evolution of $\hat{H}(J)$; \\
(3). As shown in Fig.~\ref{fig:Interference-highlevel}, the system $\hat{H}(J)$ contains a hard wall with two slits, which the particle $\psi$ cannot penetrate except through the slits. The wall is located at a distance away from the special point $\oA$;\\
(4). At a point between the double slits, there is a point $C$~(the central red dot in Fig.~\ref{fig:Interference-highlevel}) at which $\psi$ can stay without dispersing away under the evolution of $\hat{H}(J)$. Further more, we require that the point $C$ is isolated from the rest of the system, in the sense that any $\psi$ created from $\oA$  can never get into contact~(interaction) with a $
\psi$ stationed at $C$ under the evolution of $\hat{H}(J)$;\\ %

As a concrete example, if $\psi$ is a paraparticle, then such a system $\hat{H}(J)$ can be explicitly realized using the exactly solvable model constructed in Refs.~\onlinecite{wang2023para,wang2024parastatistics}, in which $J$ is the tunneling constant of the emergent free paraparticle. As shown in Fig.~\ref{fig:Interference-lowlevel}, 
the double-slit wall in Fig.~\ref{fig:Interference-highlevel} can be realized by removing all the paraparticle tunneling terms~(three-body terms in the original spin model) on the dashed links, and the isolated point $C$ is at the central red dot. 

The players are also required to engineer a time dependent Floquet driving potential $\hat{V}_\oA(t)$ localized at $\oA$ that keeps creating particles of type $\psi$ at $\oA$ and shooting them towards the double slit. The experiment is done in two separate rounds. In the first round, the isolated point $C$ is unoccupied, while in the second round, the Referees put a particle $\psi$ at $C$. In each round, the Referees measure the particle density distribution on a line parallel to the double slit wall, as shown in Fig.~\ref{fig:Interference-2DESM}, and they gather a sufficient amount of data to obtain a stable interference pattern. The winning condition is that the interference patterns obtained in the two rounds are the same, no matter how one choose the relevant length parameters of this geometry. That is, one should not be able to decide if $C$ is occupied by a $\psi$ by looking at the interference pattern on the screen.  %

This experiment can in principle completely distinguish anyons and paraparticles in 2D, since anyons have $R^2\neq \mathds{1}$, leading to a~(generally non-Abelian) phase difference between the two interfering paths, which changes the resulting interference pattern, analogous to the Aharonov-Bohm effect. For example, for the aforementioned type of anyons with $R^2=e^{i\theta}\mathds{1}$, an anyon at the isolated point $C$ shifts the interference pattern by an amount 
\begin{equation}
	\Delta y = -\frac{\theta \lambda L}{2\pi d},	
\end{equation}
where $\lambda$ is the wavelength of the particle, $d$ is the distance between the slits, and $L$ is the distance between the double slit wall and the measurement screen, as shown in Fig.~\ref{fig:Interference-highlevel}. 

\begin{figure}
	\begin{subfigure}[t]{\linewidth}
		\includegraphics[width=.65\linewidth]{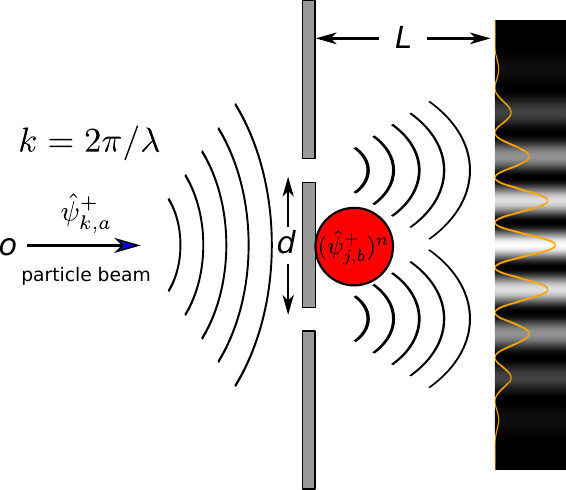} 
		\caption{\label{fig:Interference-highlevel}}
	\end{subfigure}
	\begin{subfigure}[t]{\linewidth}
		\includegraphics[width=.75\linewidth]{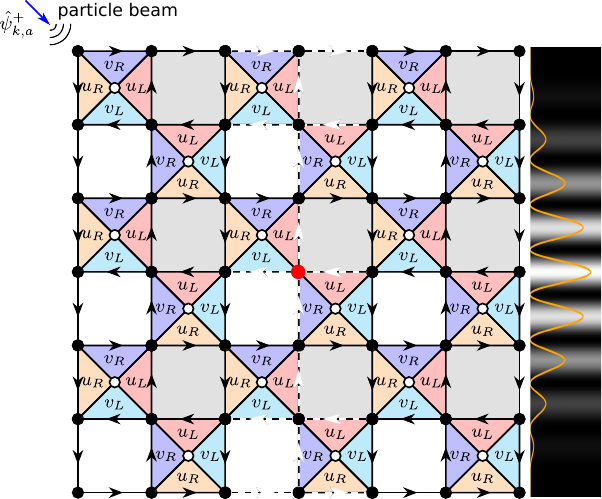} 
		\caption{\label{fig:Interference-lowlevel}}
	\end{subfigure}
	\caption{\label{fig:Interference-2DESM} An interference experiment that can in principle distinguish anyons and paraparticles. (a) A schematic view of the experiment at high level; (b) a possible low level implementation on a lattice model proposed in Ref.~\onlinecite{wang2023para}.
	}
\end{figure}

\section{Conclusions}\label{sec:conclusions}
In summary, we have demonstrated that our proposed family of secret communication challenge games naturally define a hierarchy of particle exchange statistics in 3+1D topological phases, summarized in Tab.~\ref{tab:protocols_strategies}. 
We introduced the axioms of emergent $R$-parastatistics, %
and showed that they provide a natural and general description of the winning strategies, which exploits the nontrivial exchange statistics of $R$-paraparticles to achieve non-local secret communication between the two players.  
We then performed a detailed analysis of all possible winning strategies to the game using the tensor categorical description of topological quasiparticles in 3D gapped phases, and showed that only emergent $R$-paraparticles~(as defined by our axioms) can win the full version of the game in 3D, and our proposed winning strategy is essentially unique. This establishes the hierarchy on solid grounds, and provides an operational definition of emergent $R$-paraparticles, as those quasiparticles that can be used to win the challenge game in a noise-robust way.  We also proposed the anti-anyon twists and showed how they exclude anyons in 2D.

Our analysis 
clarifies the meaning  of $R$-parastatistics in the context of topological order, deconfined gauge theories, and tensor category theory.
In short, the class of $R$-paraparticles involved in this paper are realized as charged particles in  deconfined non-Abelian gauge theories based on an exotic class of finite groups~(intimately related to groups of central type~\cite{Howlett1982}). 
Although all charged particles %
in deconfined gauge theories %
were previously classified as fermions or bosons, 
they can still lead to nontrivial physical consequences when interplay with certain types of defects.  %
The ability of $R$-paraparticles to pass the challenge game demonstrates their dramatic difference from our conventional picture of fermions and bosons. 
In the categorical description, such exotic class of deconfined gauge theories admit a special type of point-like defect $\sigma$ that has a fusion rule $\sigma\times\psi=m\sigma$ with the paraparticle $\psi$, and exchanging two paraparticles induces a nontrivial unitary transformation~(characterized by the $R$-matrix) on the particle-defect fusion space, which the winning strategies exploit. 
Such an exotic exchange behavior cannot be captured by the standard first or second quantization frameworks for conventional fermions and bosons, but are instead naturally captured by our axioms of emergent $R$-parastatistics. Whether one should call them $R$-paraparticles or fermions/bosons coupled to $G$-gauge fields is a matter of definition and viewpoint, however, we argue that a good definition of a fundamental physical concept should be understandable by the \textit{majority} of physicists, not just a small community of theoretical/mathematical physicists who know an extensive amount of non-Abelian gauge theories, representation theory, or even tensor categories.
Our axiomatic description of emergent $R$-paraparticles %
are simple enough to be understandable with basic knowledge of quantum mechanics. %
In comparison, to determine whether a certain 3D topological phase~(or a deconfined $G$-gauge theory) %
can pass the challenge requires a serious amount of representation theory calculations, while in our formulation the $R$-matrix alone makes winnability crystal clear.  %
We emphasize that the theory of $R$-paraparticles~(including the second quantization theory~\cite{wang2023para} and the axioms proposed in this paper) is not just a reformulation of some special types of charged particles in some exotic class of deconfined gauge theories. As we mentioned before, $R$-paraparticles realized in deconfined gauge theories~(described by the aforementioned class  SFCs) only correspond to a subclass of $R$-paraparticles proposed in Ref.~\onlinecite{wang2023para}, and there in principle exists a more exotic class of $R$-paraparticles that are beyond the description of SFCs~(and deconfined gauge theories). As a simple example, for the $R$-matrix $R=-\mathds{1}_{m\times m}$, both the second quantization theory~\cite{wang2023para} and the axioms in Sec.~\ref{sec:axioms_emergent_para} are still consistently defined, and the corresponding $R$-paraparticle can in principle win the challenge game if one assumes that they can be realized in a 2D or 3D physical system satisfying the axioms. Although these beyond SFC $R$-paraparticles are unlikely realizable in gapped phases of matter~\footnote{A small exception is that these beyond-SFC $R$-paraparticles can be realized in a family of 1D solvable quantum spin chains~\cite{BSundar2019,wang2023para}, which can be either gapped or gapless depending on the model parameters, and they satisfy all our axioms except Axiom.~\ref{Axiom2}~(because there is no topological order in 1D).}~(according to the conjecture that point-quasiparticles in 3D topological phases are described by SFCs~\cite{LanKongWen3DAB,LanWen3DEF}), their realizability in much-less-understood gapless phases remains a possibility. 

The secret communication challenge games also advance our understanding of topological phases in general. 
It allows us to distinguish  different 3D topological phases using a physical process that only involves exchanging point-like excitations, without involving any string-like excitations that are much harder to manipulate experimentally. 
It provides a reliable way to detect topological order and long-range quantum entanglement: 
it is straightforward to show that %
a product ground state cannot pass the challenge, therefore, combined with our arguments in Sec.~\ref{sec:reduction_to_TPO} that winnability is invariant under local unitary transformations, experimentally passing the challenge is a demonstration of topological order. For this reason, experimentally winning the game even in 2D systems using anyons can also be interesting~\footnote{Here we mention a possibility that is reasonably achievable using current experimental technology. The quantum double model based on the solvable groups $D_8$ and $A_4$
	can win at least levels 3 and 4 in Tab.~\ref{tab:protocols_strategies}, respectively, provided that one carefully design a hybrid boundary condition that we describe in App.~\ref{sec:blackdefectgappedboundary}. 
	Ground state of these quantum double models can be efficiently prepared using the recently developed topological quantum state preparation protocol involving measurements and feedforward~\cite{HierarchyTPO_LOCC,Iqbal2024}~(indeed, $D_8$ is already realized in Ref.~\onlinecite{Iqbal2024}, which is called $D_4$ in there due to a different naming convention. However, the system realized in Ref.~\onlinecite{Iqbal2024} cannot win our challenge due to the uniform boundary condition used there). It would be interesting to actually play the challenge games on these platforms.
}. 
The different versions of the games provide us a versatile way to define topological invariants %
of the ground state. For example, %
the maximal amount of information the players can transfer in one round of the game subject to a certain winning condition is an invariant under local unitary transformations of the system's ground state. 

We now compare our secret communication games to some other applications of topological phases and particle statistics in the literature. 
First, we mention that permutational quantum computation~(PQC) proposed in Ref.~\onlinecite{Jordan2010PQC} may have some philosophical similarity to our secret communication games, in that it also exploits the nontrivial unitary transformation of the multiparticle fusion space induced by the exchange operaion in an SFC~\footnote{Note that in the original PQC scheme proposed in Ref.~\onlinecite{Jordan2010PQC}, PQC was illustrated using ordinary spin-1/2 particles, such as ordinary electrons, where quantum information is stored in the spin degree of freedom, and consequently there is no topological protection. However, it is also possible to perform PQC using~(indeed, this was already mentioned in Ref.~\onlinecite{Jordan2010PQC}) topological quasiparticles in 3+1D topological phases, analogous to topological quantum computation~\cite{kitaev2003fault,Nayak2008NAAnyons} in 2D using non-Abelian anyons, where quantum information stored in the fusion space is topologically protected. }. 
However, it is difficult to use PQC to establish a hierarchy of 3+1D topological phases and point particle statistics, similar to what we achieved in this paper. For example, one may try to distinguish 3+1D topological phases by comparing their PQC power. The difficulty, though, is that it is generally extremely hard to prove rigorous separation between different computational complexity classes. Although PQC can solve certain problems that appear classically hard~\cite{Jordan2010PQC}, there is no formal proof that PQC offers quantum advantage. Indeed, it is not even known if PQC can efficiently simulate a classical computer~\cite{Jordan2010PQC}. Furthermore, %
computational power is often defined as the asymptotic time complexity of a certain algorithm when the problem size is large, and in practice, 
it requires a large number of quasiparticles and %
exchange operations to perform meaningful computations, while our challenge game only requires a single exchange operation between two quasiparticles to demonstrate the nontriviality of a topological phase and quasiparticle statistics. 

We emphasize that the ability to perform secret communication in the scenario we study is strictly more demanding than secret sharing %
using topological phases~\cite{Kato2016Information,fiedlerJonesIndexSecret2017,yang2025topological}. It is known that every topological phase with a topological quasiparticle of quantum dimension greater than 1 can be used for non-local secret sharing, while only a special class of topological phases can do secret communication in the scenario we study. Secret sharing uses only the fusion properties of the quasiparticles, while secret communication exploits both the fusion structure and the braiding/exchange statistics. Consequently, secret sharing capability cannot distinguish between different topological phases whose fusion structure of point particles are equivalent, but have inequivalent braiding/exchange statistics~\footnote{A specific example in 3+1D can be constructed by deconfined gauge theories based on isocategorical groups. Specifically, Ref. gives an example of two groups $G_1$ and $G_2$~(both of order 64), such that the SFCs $\Rep(G_1)$ and $\Rep(G_2)$ are equivalent as fusion categories, but inequivalent as SFCs. }, which can be distinguished by our secret communication games. 

We also note that while there are prior proposals of playing non-local games~\cite{brassardQuantumPseudoTelepathy2005} in topological phases~\cite{Burnell2023TCnonlocalgame,Nandkishore2025TCnonlocalgame}, our secret communication challenge games fundamentally differ in goal and design. The  prior non-local games~\cite{Burnell2023TCnonlocalgame,Nandkishore2025TCnonlocalgame} generalize the parity game~\cite{Mermin1990Extreme,brassard2005recasting} from the GHZ %
state~\cite{greenbergerGoingBellsTheorem1989} to the 2D toric code, %
mainly to demonstrate noise-robust quantum advantage; while our goal %
is  to distinguish between different types of emergent quasiparticle statistics and topological order. %
The design and strategy of our secret communication games are not based on generalizing prior non-local games in few-body systems; in particular, our games do not appear to have a few body analog, as emergent quasiparticle statistics is an intrinsically many-body phenomenon. %

We end our discussion by mentioning some potential generalizations and future directions. 
First, it may be interesting to extend the challenge game protocols to the relativistic case, to allow the use of elementary particles in relativist quantum field theories. 
A key open question speculated in Ref.~\onlinecite{wang2023para} concerns whether $R$-paraparticles may exist in the universe as elementary particles. %
In view of the aforementioned operational definition of $R$-paraparticles, %
this translates into the following question: does there exist a 3+1-dimensional relativistic quantum field theory 
that can win the relativistic version of the secret communication challenge? 
Of course, in the relativistic case, one needs to formulate the games carefully to respect all the fundamental symmetries of relativistic quantum field theories and to prevent any possibility of cheating~(e.g., prevent the players from leaving any trace information behind that is locally accessible by a third party).
A positive answer to this question would imply the theoretical possibility of elementary $R$-paraparticles in the universe, in a way compatible with locality and Lorentz invariance.   

Another interesting future direction is to generalize our results to $R$-paraparticles in gapless phases of matter, including the axioms, %
the game protocols, the winning strategies, and the proof that only emergent $R$-paraparticles can win. 
Generalizing the axioms to gapless phases requires more than simply removing the spectral gap assumption on the Hamiltonian $\hat{H}$. For example, the exponentially small correction in Eq.~\eqref{eq:TPdegenerate} may need to be replaced by a power-law correction. The localization properties of the local unitary operators and observable in Axioms~\ref{Axiom3},  \ref{Axiom5}, and \ref{Axiom6} may also need to be relaxed to accommodate gapless phases in general. A primary example of a gapless system hosting emergent $R$-paraparticles is the 2D exactly solvable model constructed in Ref.~\onlinecite{wang2023para}, when the parameters $J,\mu$ are chosen such that the single particle spectrum is gapless~(e.g., a conducting band), and the axioms should be generalized to accommodate this example.
To generalize the game protocols to gapless phases, a potential difficulty %
is that it may be challenging for the Referees to efficiently verify the condition that no excitations exist beyond the circle areas. 
A related discussion of this can be found in App.~\ref{sec:relax_FF}.
To generalize the winning strategies to gapless phases, a technical challenge is how to confine the paraparticles inside the circle areas, as quasiparticles in gapless phases typically have power-law decaying tails. 
To prove that only $R$-paraparticles can win the game in the gapless case, we need a complete description of the universal properties of quasiparticles in gapless  phases, which is still under development~\cite{kongMathematicalTheoryGapless2020,kongMathematicalTheoryGapless2021,kongCategoriesQuantumLiquids2022,bhardwajClubSandwichGapless2025}.  

Last but not least,  it is interesting to generalize the secret communication games to mixed state topological order~\cite{Grusdt2017Topological,zini2021mixed,bao2023mixed,Grover2024Separability,Vishwanath2024Diagnostics,Hsieh2025Markov,Cheng2025Premodular, Wang2025Intrinsic,Sohal2025Noisy,Li2025Replica,yang2025topological,lessa2025higher,sang2025mixed,ogata2025mixed}, in particular,
topological phases  at finite temperature~\cite{Ortiz2008autocorrelations,TCFiniteT2008,alickiThermalStabilityTopological2010,Hastings2011finiteTTPO,Bombin_2013,Hsieh2020finiteTTPO,Cheng2025finiteTTPO}. %
This can be experimentally relevant as all quantum matter in reality interact with the environment and are practically described by mixed states. The generalization of the game protocol to mixed states is briefly discussed in App.~\ref{app:relaxFF2DM}. Winning the challenge at finite temperature in 3D may be possible as it is demonstrated recently that there exists a 3D topological phase that exhibits long-range quantum entanglement~\cite{Cheng2025finiteTTPO} at finite temperature. %

\acknowledgments
We thank Alexei Kitaev, Xiao-Liang Qi, J. Ignacio Cirac, Norbert Schuch, Meng Cheng, Kaden Hazzard,  Timothy Hsieh,  Hao Song, Yingfei Gu, Jingyuan Chen,  Qingrui Wang, Liang Kong, Zhihao Zhang, Bowen Shi, and Xiaoqi Sun for discussions. Our research at MPQ is supported by 
the Munich Center for Quantum Science and Technology~(MCQST), funded by the Deutsche Forschungsgemeinschaft~(DFG) under Germany’s Excellence Strategy~(EXC2111-390814868).  
Research at Perimeter Institute is supported in part by the Government of Canada through the Department of
Innovation, Science and Economic Development and by the Province of Ontario through the Ministry of Colleges and Universities.

\appendix

\section{Technical requirements on the Hamiltonian $\hat{H}$}\label{sec:technical_requirement_H}
In this section we discuss the various technical conditions on the Hamiltonian $\hat{H}$ for any physical system the players may use to win the challenge game. In App.~\ref{app:condition_defect} we detail what kind of defects are allowed to be present in the system, and in App.~\ref{sec:relax_FF} we discuss several possible ways to relax the frustration-free condition on $\hat{H}$. 

\subsection{Technical conditions on defects}\label{app:condition_defect}
In formulating the requirements on the physical system in Sec.~\ref{sec:original_version}, we do not require the Hamiltonian $\hat{H}$ to be translationally invariant, and allow defects to be present, as long as $\hat{H}$ has a unique, gapped, and frustration-free ground state, as we have seen that having point-like defects at the special points $\oA$ and $\oB$ is important for the winning strategies. In this section we formulate precisely what kind of defect configuration are allowed to be present in the system. 

Let us begin with some generalities. We write the Hamiltonian in the following form
\begin{equation}
	\hat{H}=\hat{H}_{\text{bulk}}+\hat{H}_{\text{defect}},
\end{equation}
where $\hat{H}_{\text{bulk}}=\sum_{i}\hat{h}_i$ is the bulk Hamiltonian, and $i$ runs through lattice points in the bulk, away from the defects. We write  $\hat{H}_{\text{defect}}$ as
\begin{equation}
	\hat{H}_{\text{defect}}=\sum_{K}\sum_{j\in K} \hat{h}'_j,
\end{equation}
where $K$ runs through all defects in the system, and the second sum runs over points within the defect $K$. Importantly, the ground state is required to be frustration-free everywhere, including in the defect region, i.e.,  $\hat{h}'_j\ket{G}=0$ on any defect. Furthermore, we require that $\hat{h}'_{j_1}$ and $\hat{h}'_{j_2}$ are related by translation if $j_1$ and $j_2$ lie on the same type of defect. We also require $\hat{h}_{i_1}$ and $\hat{h}_{i_2}$ are related by translation for any two points $i_1$ and $i_2$ in the bulk. 

We now specify various different defect configurations we allow. %
In the simplest configuration, we only allow 2 point-like defects in the system, so that $\hat{H}_{\text{defect}}=\hat{h}'_\oA+\hat{h}'_\oB$, and the whole system is defined on a topologically trivial manifold. Here we do not require $\hat{h}'_\oA$ and $\hat{h}'_\oB$ to be related by a translation. Note that in this case, it only makes sense for the players to choose the special points $\oA$ and $\oB$ to be the positions of the two defects. %
With this configuration, the two defects $\sigma$ and $\bar{\sigma}$ can only be quasiparticles of the bulk topological phase. The winning strategy is classified by the categorical analysis in Sec.~\ref{sec:categorical_analysis}. Only a very special class of SFCs can pass this version of the challenge, %
examples are given in  Sec.~\ref{app:GCT}. 

In order to allow more general winning strategies in which $\sigma$ and $\bar{\sigma}$ are genuine point-like defects~(beyond quasiparticles) described by the module category analysis in Sec.~\ref{sec:ModCatDefect}, we should allow higher dimensional topological defects to be present in the system. %
A simple example is a line-like defect in 2D or 3D with two end points at $\oA$ and $\oB$, as shown in Fig.~\ref{fig:LineDefect}. In this case, $\hat{H}_{\text{defect}}$ is required to have the form $\hat{H}_{\text{defect}}=\sum_{j\in L} \hat{h}'_j+\hat{h}'_\oA+\hat{h}'_\oB$, where $L$ is the set of interior points~(excluding end points $\oA,\oB$) of the line segment $L$, and $\hat{h}'_\oA$ is not required to be related to any $\hat{h}'_j$ by a simple translation, and similarly for $\hat{h}'_\oB$.
We also allow the line-like defect to lie on the boundary in the 3D case, as shown in Fig.~\ref{fig:3DBC-1}. In 2D, we allow the hybrid boundary condition shown in Fig.~\ref{fig:BoundaryDefect}, where $\oA$ and $\oB$ are the two intersection points. In 3D, we allow the hybrid boundary condition shown in Fig.~\ref{fig:3DBC-2}, where $\oA$ and $\oB$ are the two intersection points of the three different boundaries. In all these configurations, the two point-like defects $\oA$ and $\oB$ are described by module categories $\calM$ and $\calM^\op$ over $\calC$~(where $\calC$ is the tensor category of bulk particles) that are dual to each other, and the analysis in Sec.~\ref{sec:ModCatDefect} applies. 
Note that in the last case, there are defects of different dimensions in the system: 0D defects at $\oA,\oB$, 1D defects at the intersection between two different types of boundaries, and the boundaries are 2D defects. The most natural language to describe the interplay between defects of different dimensions is using higher-category theory~\cite{kong2014braided,kong2015boundary,douglas2018fusion,LanKongWen3DAB,LanWen3DEF,gaiotto2019condensations,Kong2020Algebraichigher,Kong2020Classification,johnson20203+,Johnson-Freyd2022}, which provides a unified mathematical framework to study all topological particles and defects in a topological phase. We do not discuss this higher-category approach to winning strategies in this work, but hope to adopt it in a future paper. 

\begin{figure}
	\begin{subfigure}[t]{.4\linewidth}
		\includegraphics[width=\linewidth]{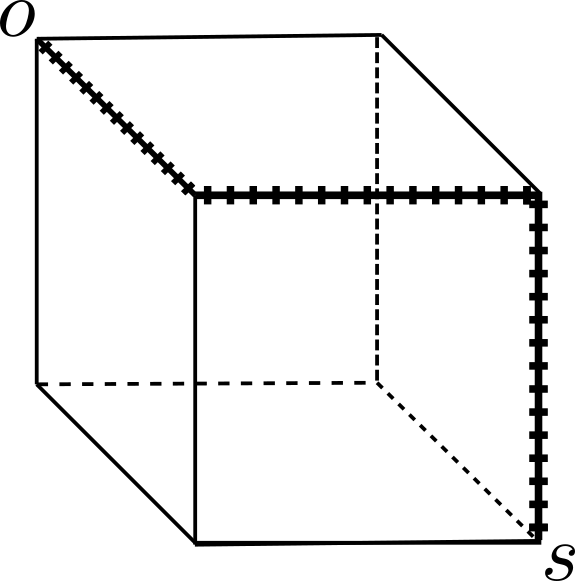}
		\caption{\label{fig:3DBC-1}}
	\end{subfigure}
	\begin{subfigure}[t]{.4\linewidth}
		\includegraphics[width=\linewidth]{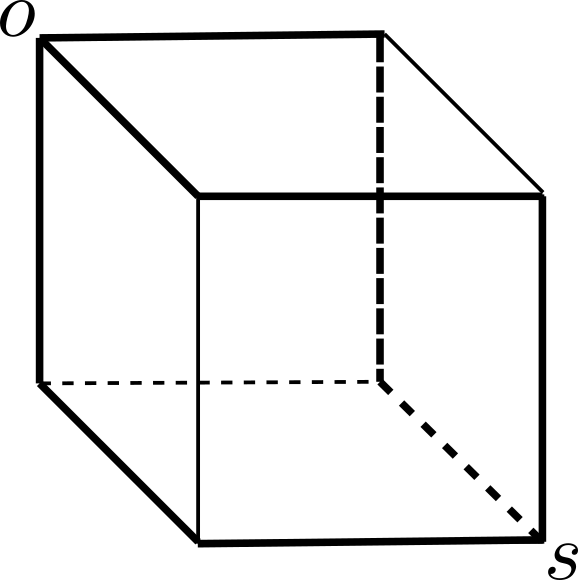} 
		\caption{\label{fig:3DBC-2}}
	\end{subfigure}
	\caption{\label{fig:3DBC} Examples of defects and  boundary conditions allowed in 3D. (a) a line defect on a gapped boundary, where $\oA$ and $\oB$ are the two endpoints (b) three different types of gapped boundaries intersecting at $\oA$ and $\oB$.
	}
\end{figure}

\subsection{Relaxing the frustration free condition}\label{sec:relax_FF}
Recall that in Sec.~\ref{sec:parachallenge_variants_twists}, we require that the Hamiltonian $\hat{H}$ to be frustration-free, meaning that $\hat{H}$ can be written in the form $\hat{H}=\sum_{i}\hat{h}_i$ such that (1). $\hat{h}_i$ is a local Hamiltonian acting a neighborhood of site $i$; (2) $\hat{h}_i \ge 0$; and (3) $\hat{h}_i \ket{G} = 0$. Many interesting topological phases are known to be realizable in frustration-free Hamiltonians. 
For example, in 2D, it is known that any nonchiral topological phases can be realized by a string-net model~\cite{levin2005string}; in 3D, it is conjectured that all topological phases can be realized by  commuting projector Hamiltonians~\cite{Zhu3DTPOModels,LanKongWen3DAB,LanWen3DEF}. Furthermore, although realizing 2D chiral topological order in strictly local frustration-free Hamiltonians on a lattice with finite local Hilbert space is known to be hard~\cite{DubailRead2015,kapustinLocalCommutingProjector2020,Kapustin2020Thermal,Levin2022Vanishing}, 2D chiral topological states can still be frustration-free in a relaxed sense, for example,  if one allows $\hat{H}$ to be defined in continuum~\cite{Wang2017Number,Wang2018}, or allows $\hat{H}$ to contain longer range %
interactions~\cite{Greiter2007CSLParentH,Cirac2013chiralPEPS,Yang2015chiralPEPS}.  %
Nevertheless, there is still a strong motivation to relax this unnecessarily restrictive frustration-free condition.
Our ultimate goal is to formulate the game in such a way that a physical system can win the game \textit{if and only if} it hosts emergent $R$-paraparticles. Emergent parastatistics is defined as a universal property of the underlying phase of matter at long distance, %
but being frustration-free is not a universal property of the phase, and it excludes many interesting physical systems that host emergent parapartcles that are not frustration-free. For example, the 2D exactly solvable model constructed in Ref.~\onlinecite{wang2023para} is shown to host emergent $R$-paraparticles, but it is not always frustration-free depending on the model parameters. 

The main reason why we assumed the frustration-free condition in the main text is for the following two reasons:\\
(1). Before the game, the Referees can efficiently verify that the state $\ket{G}$ prepared by the players is indeed the ground state of the Hamiltonian $\hat{H}$, by verifying that $\hat{h}_i\ket{G}=0$ everywhere;\\
(2). During the game, the Referees can efficiently monitor the game to ensure that no extra excitations are left outside the circle areas, again by verifying $\hat{h}_i\ket{\Psi(t)}=0$ for $i$ outside the two circles.

In the following we discuss two possible ways to relax this frustration-free condition. A third option is to use a different version of the game based on adiabatic evolution, which we present in App.~\ref{app:adiabatic}. 

\subsubsection{Option 1: using approximate local ground state annihilators for gapped ground states}\label{app:relaxFF1Hastings}
By a theorem of Hastings~\cite{hastings2006solving}~(see also Proposition D.1 of Ref.~\onlinecite{kitaev2006anyons}), if a locally interacting Hamiltonian $\hat{H}=\sum_i \hat{h}_i$ has a unique, gapped ground state $\ket{G}$, then one can rewrite $\hat{H}$ as $\hat{H}=E_0+\sum_i \hat{h}'_i$, where $E_0$ is the ground-state energy and $\hat{h}'_i$ is localized near $i$ with subexponentially decaying tails, satisfying $\hat{h}'_i\ket{G}=0$ for all $i$. 
In this way,  the Referees can still efficiently verify conditions (1) and (2) above by measuring $\hat{h}'_i$ instead. The main technical difficulty in this approach is that it is generally hard in practice to explicitly write down the exact form of the local ground state annihilators $\hat{h}'_i$, and even more difficult to measure them in experiment, although we know that they theoretically exist. 

\subsubsection{Option 2: an alternative formulation using local reduced density matrices}\label{app:relaxFF2DM}
Another way to relax the frustration-free condition is to formulate the above two conditions (1) and (2) in an alternative way, in terms of local reduced density matrices. %
This formulation is based on the following theorem that follows from the main theorem of Ref.~\onlinecite{hastings2006solving}: 
\begin{theorem}\label{thm:local_verifiable}
	Let $\ket{G}$ be the unique, gapped ground state of a local Hamiltonian  $\hat{H}$, and let $\rho=\ket{G}\bra{G}$. Then there exists a finite length scale $\xi$ independent of the system size, such that the local reduced density matrices on all disk regions of radius $\xi$ uniquely determines the ground state $\ket{G}$. More precisely, if there exists another (potentially mixed) state $\rho'$ such that $\rho'_D=\rho_D$ for any disk region $D$ of radius no larger than $\xi$, then $\rho'=\rho$.  
\end{theorem}
[We also expect that there exists an approximate version of this result.] 
With this, we can reformulate the conditions (1) and (2) as follow. When the Referees receive the state $\rho'$ from the players, they simply verify $\rho'_D=\rho_D$ for any disk region $D$ of radius no larger than $\xi$. 
At any time $t$ during the game, for any disk region $D$ of radius no larger than $\xi$ that lie outside the circles of the two players~\footnote{Here it is important that we only demand $\rho_{D}(t)=\rho_D$ for simply connected local regions $D$ that do not intersect circle A and B. All the winning strategies in this paper satisfy this condition, however, with a topological quasiparticle in circle A, $\rho_{K}(t)\neq\rho_K$ if $K$ is taken to be an annulus region encircling circle A. It is unlikely that any pure state topological phases can win this game if we demand $\rho_{K}(t)=\rho_K$ on multiply connected regions as well, due to the principle of remote detectability of topological excitations~\cite{LanKongWen3DAB,LanWen3DEF}. }, the reduced density matrix $\rho_D(t)$ of the physical system is required to be equal to the initial ground state value, i.e., $\rho_{D}(t)=\rho_D$. %

Although it may be practically challenging for the Referees to efficiently verify $\rho_{D}(t)=\rho_D$, 
this alternative formulation can still be conceptually useful, as with this formulation one can argue that any gapped topological phase hosting emergent $R$-paraparticles can pass this formulation of the game. Furthermore, this formulation can be generalized to mixed states in a straightforward way, since 
Thm.~\ref{thm:local_verifiable} %
generalize  to mixed states with finite Markov length~\cite{yang2025topological,Hsieh2025Markov,sang2025mixed}. %

\section{Remarks on noise-robustness}\label{app:noise_robust}
In Sec.~\ref{sec:robust_noise_eaves} we stated that a realistic winning strategy is required to be robust against local noise, and we have argued that paraparticle based winning strategies are robust against local noise in the circles, as such noise can always be corrected by the players themselves using local operations in the circles, while the emergent fermion based strategy given in Sec.~\ref{sec:emergent_fermion} is not robust. Here we discuss the robustness against local noise that happen outside of the circle areas.  Such local noise cannot destroy the information stored in the fusion space of the topological quasiparticles 
due to topological protection, Eq.~\eqref{eq:TPdegenerate}.  However, local noise outside the circles can break the rule of the game by  accidentally creating local excitations beyond the circle areas, and the players cannot clean them up since they are not allowed to do any operation beyond the circles.

Ref.~\onlinecite{wang2024parastatistics} briefly discussed one possible way to resolve this issue, by introducing a third player, called the corrector, Carol, who plays in the same team as Alice and Bob, but is not allowed classical communication with other players. During the game, whenever the Referees detects an excitation outside the circle areas, they first ask Carol to eliminate the excitation using local operations. If Carol can successfully eliminate the excitation, then the game proceeds; otherwise the challenge fails. This modified game protocol allows the paraparticle based strategies to win the game in the presence of local noise, provided that the noise rate is below a certain ``error correction threshold''~\cite{Dennis2002TQM,bombinIntroductionTopologicalQuantum2013}, and still prevents cheating~(e.g., if the players intentionally leave local excitations at $\oA$ or $\oB$, then the Referees will quickly detect such excitations and immediately ask Carol to clean them up).%

The above proposal applies ideas of quantum error correction in topological quantum codes~\cite{Dennis2002TQM,bombinIntroductionTopologicalQuantum2013}. An alternative option is to use a different version of the game based on adiabatic evolution presented in the next section, where robustness against noise is guaranteed by the spectral gap and the stability of topological phase against local perturbations. 
\section{A version based on adiabatic evolution}\label{app:adiabatic}
In our original formulation of the game  presented in Sec.~\ref{sec:original_version}, 
quasiparticles are moved using local unitary operations. Here we briefly mention an alternative formulation, 
in which quasiparticles are moved using adiabatic evolution, using a time dependent Hamiltonian with slowly moving trapping potentials. This version gives an alternative way to guarantee noise-robustness for paraparticle-based strategies, and maybe more suitable for realization in condensed matter materials. It also allows the game to be defined in continuum, and does not require the Hamiltonian $\hat{H}$ to be frustration-free.

Below we list the main difference from the original design:\\ %
(1). Before the game, in addition to the Hamiltonian $\hat{H}$, the players are also required to submit a local operator $\hat{V}$ with finite support, called the trapping potential;\\
(2). During the game, the players are only allowed to perform local operations inside their assigned circles at $t=0$ and $t=T$, and they have no control over the system for $0<t<T$;\\
(3). For $0<t<T$, the system evolves according to the time dependent Hamiltonian
\begin{equation}
	\hat{H}(t)=\hat{H}+\hat{V}_A(t)+\hat{V}_B(t),
\end{equation}
where $\hat{V}_A(t)$ is the trapping potential acting inside the circle A, and similarly $\hat{V}_B(t)$. Note that the dependence of $\hat{V}_A(t)$ on time $t$ is solely through the position of the circle $A$, which changes slowly and continuously in time, as controlled by the Referees. It is required that the trapping potential $\hat{V}$ can successfully trap all the excitations the players create at $t=0$~(near $\oA$ and $\oB$); otherwise, shortly after the circles A and B move away from $\oA$ and $\oB$, if the Referees detect any extra excitations at $\oA$ or $\oB$, then the challenge fails. 

With this formulation, topological phases hosting emergent $R$-paraparticles can also win the game, the winning strategy is essentially the same as described in the main text, the only new thing here is to construct the trapping potential $\hat{V}$ for the paraparticles. Such trapping potential always exist for any mobile topological quasiparticle in any topological phase. This winning strategy is robust against local noise, based on a similar argument on the noise-robustness of topological quantum computation. The categorical analysis presented in Sec.~\ref{sec:categorical_analysis} still applies in this case, up to some minor modifications in technical details, and the main conclusions are the same. 

\section{Deriving the YBE for the $R$-matrix constructed from SFC}\label{app:deriveYBE}
Let $\calC$ be a SFC and let $\calM$ be a module category over $\calC$ with a fusion rule of the form in Eq.~\eqref{eq:sigmapsifusion-0}, where $\sigma\in\calM$ and $\psi\in\calC$~(Sec.~\ref{sec:win1pt-SFC} corresponds to the special case $\calM=\calC$). 
In this section we show that $R^{b'a'}_{ab}=\!\!
\begin{tikzpicture}[baseline={([yshift=-.6ex]current bounding box.center)}, scale=0.45]
	\Rmatrix{0}{0}{R}
	\node  at (-1.5*\AL,\AL) {\footnotesize $b'$};
	\node  at (1.7*\AL,\AL) {\footnotesize $a'$};
	\node  at (-1.5*\AL,-\AL) {\footnotesize $a$};
	\node  at (1.5*\AL,-\AL) {\footnotesize $b$};
\end{tikzpicture}$ defined in Eq.~\eqref{def:RfromSFC} satisfies the Yang-Baxter equation~\eqref{eq:YBE}, and is therefore consistent with the second quantization formulation of $R$-parastatistics proposed in Ref.~\onlinecite{wang2023para}. %
We first prove $R^2=\mathds{1}$.
Since $\calC$ is symmetric, we have the following diagrammatic relation
\begin{equation}\label{eq:SFCinvol}
	\adjincludegraphics[height=22ex,valign=c]{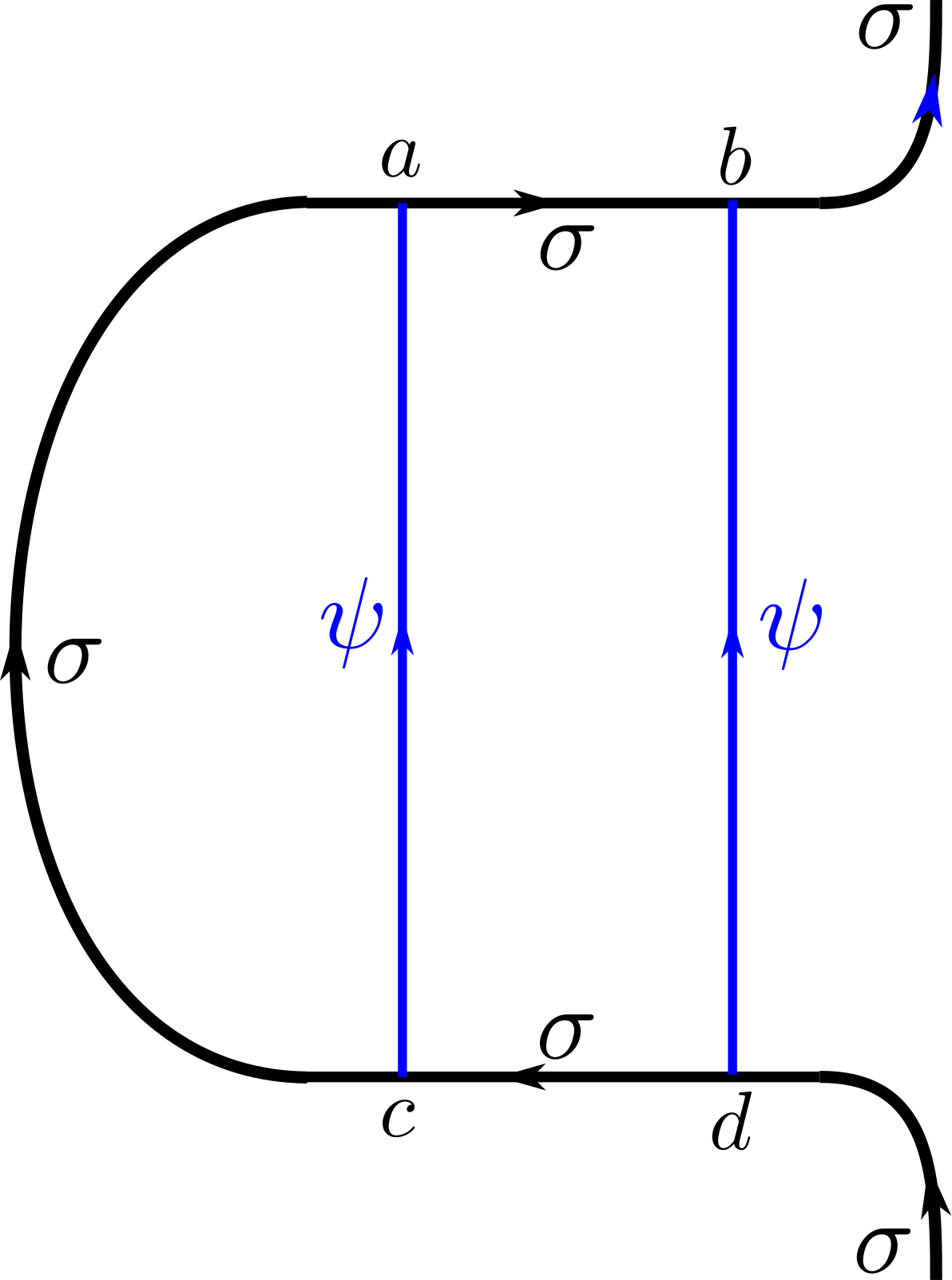}=
	~\adjincludegraphics[height=22ex,valign=c]{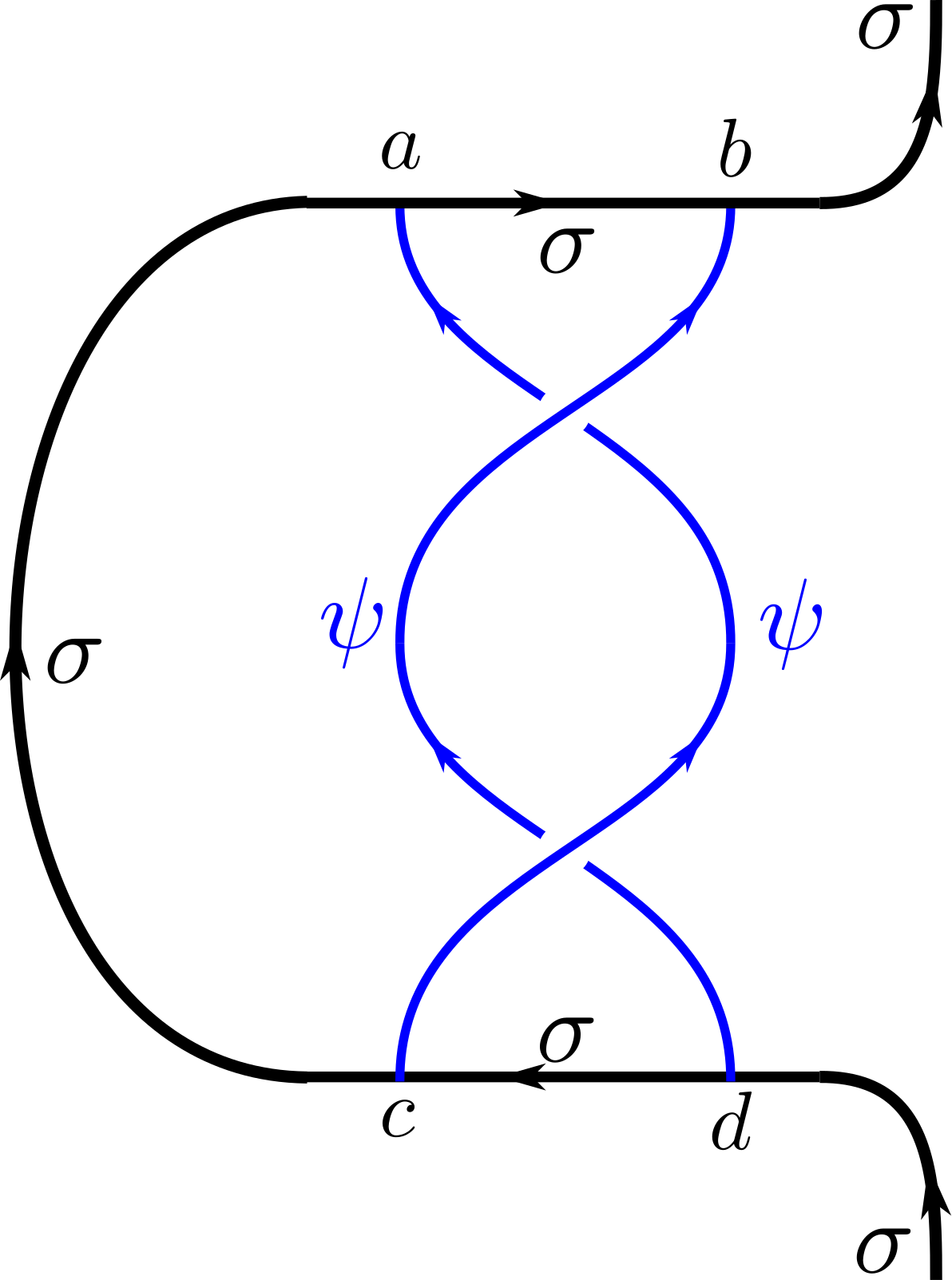}.
\end{equation}
Now we apply the following orthonormality and completeness relation satisfied by the fusion vertex 
$\adjincludegraphics[height=5ex,valign=c]{Figures/SFC/fusionvertex.png}$:
\begin{equation}\label{eq:orthocomplete}
	\adjincludegraphics[height=12ex,valign=c]{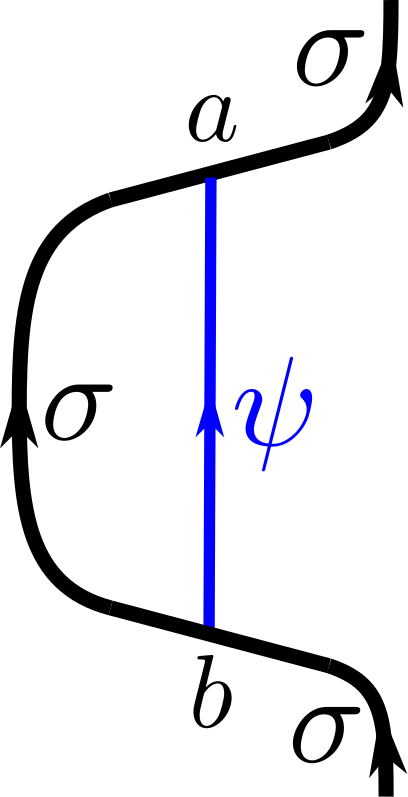}~=
	~\delta_{ab}~\adjincludegraphics[height=12ex,valign=c]{Figures/SFC/sigmaline.png}~,\quad
	\adjincludegraphics[height=12ex,valign=c]{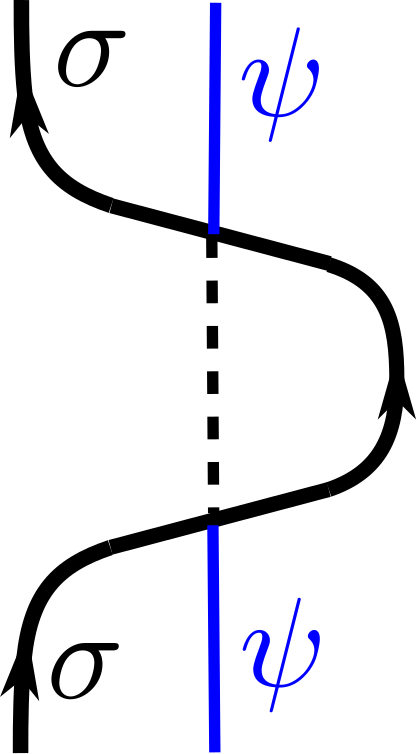}~=
	~\adjincludegraphics[height=12ex,valign=c]{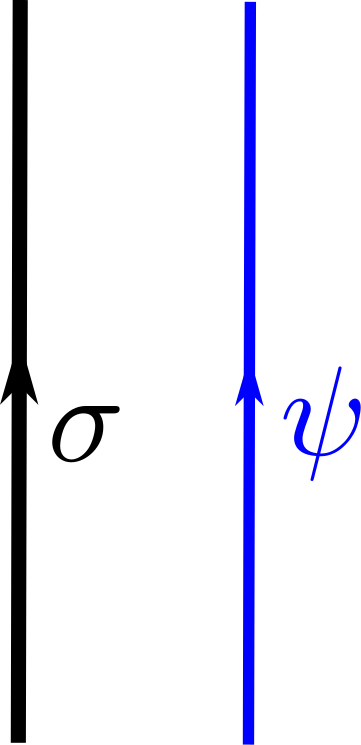}~,
\end{equation}
where a dashed line indicates a contraction of indices. With Eq.~\eqref{eq:orthocomplete},  Eq.~\eqref{eq:SFCinvol} is equivalent to 
\begin{equation}
	\delta^{a}_{c}\delta^{b}_{d}\quad\adjincludegraphics[height=12ex,valign=c]{Figures/SFC/sigmaline.png}=
	~\adjincludegraphics[height=22ex,valign=c]{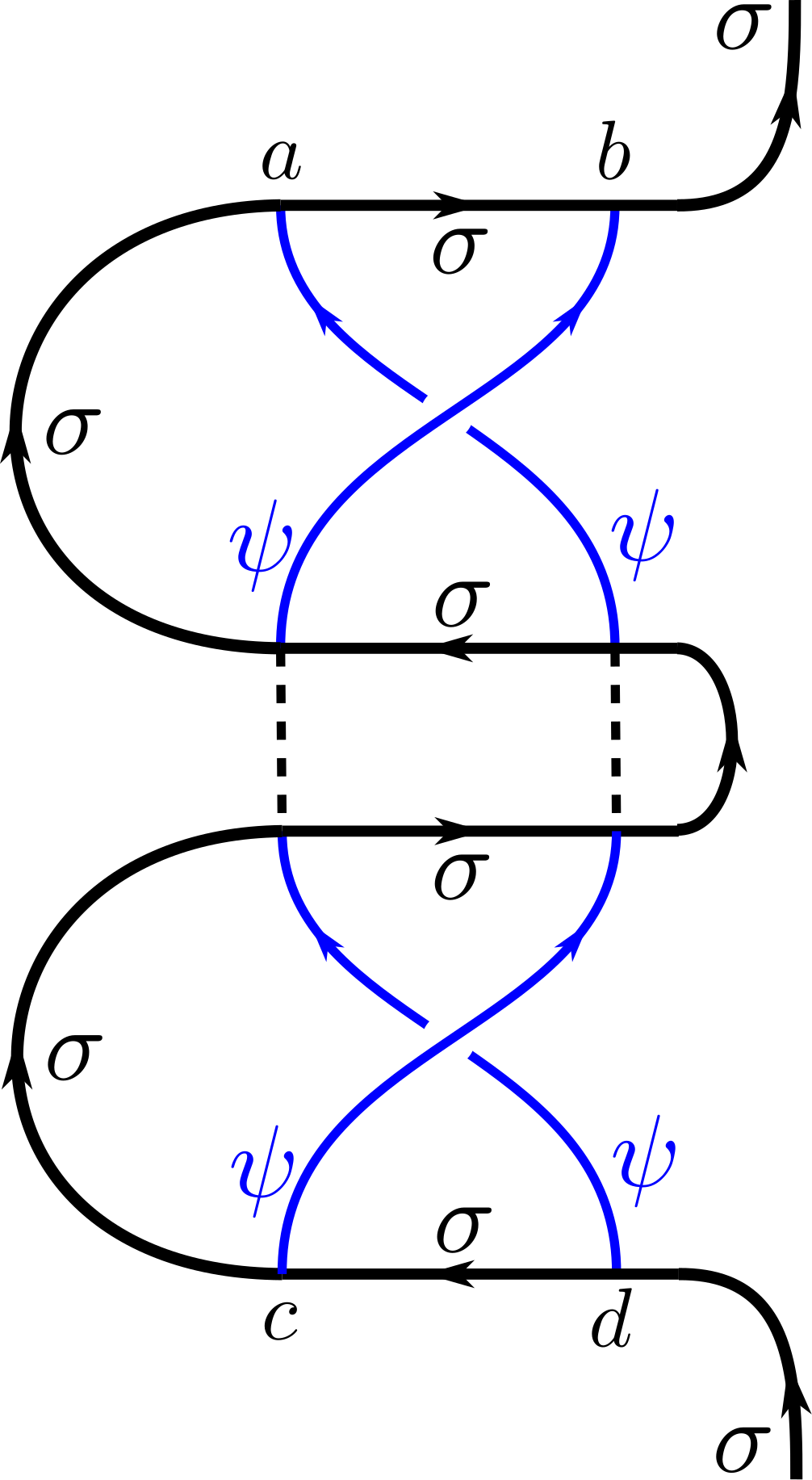}=\begin{tikzpicture}[baseline={([yshift=-.8ex]current bounding box.center)}, scale=0.5]
		\Rmatrix{0}{\AL}{R}
		\Rmatrix{0}{-\AL}{R}
		\node  at (-\AL,2.5*\AL) {\footnotesize $a$};
		\node  at (\AL,2.5*\AL) {\footnotesize $b$};
		\node  at (-\AL,-2.5*\AL) {\footnotesize $c$};
		\node  at (\AL,-2.5*\AL) {\footnotesize $d$};
	\end{tikzpicture}\quad\adjincludegraphics[height=12ex,valign=c]{Figures/SFC/sigmaline.png}.
\end{equation}
This proves the the first relation in Eq.~\eqref{eq:YBE}. 
The second relation in Eq.~\eqref{eq:YBE} is proved in a similar way using the following diagrammatic relation 
\begin{equation}\label{eq:SFCRYBELR}
	\adjincludegraphics[height=30ex,valign=c]{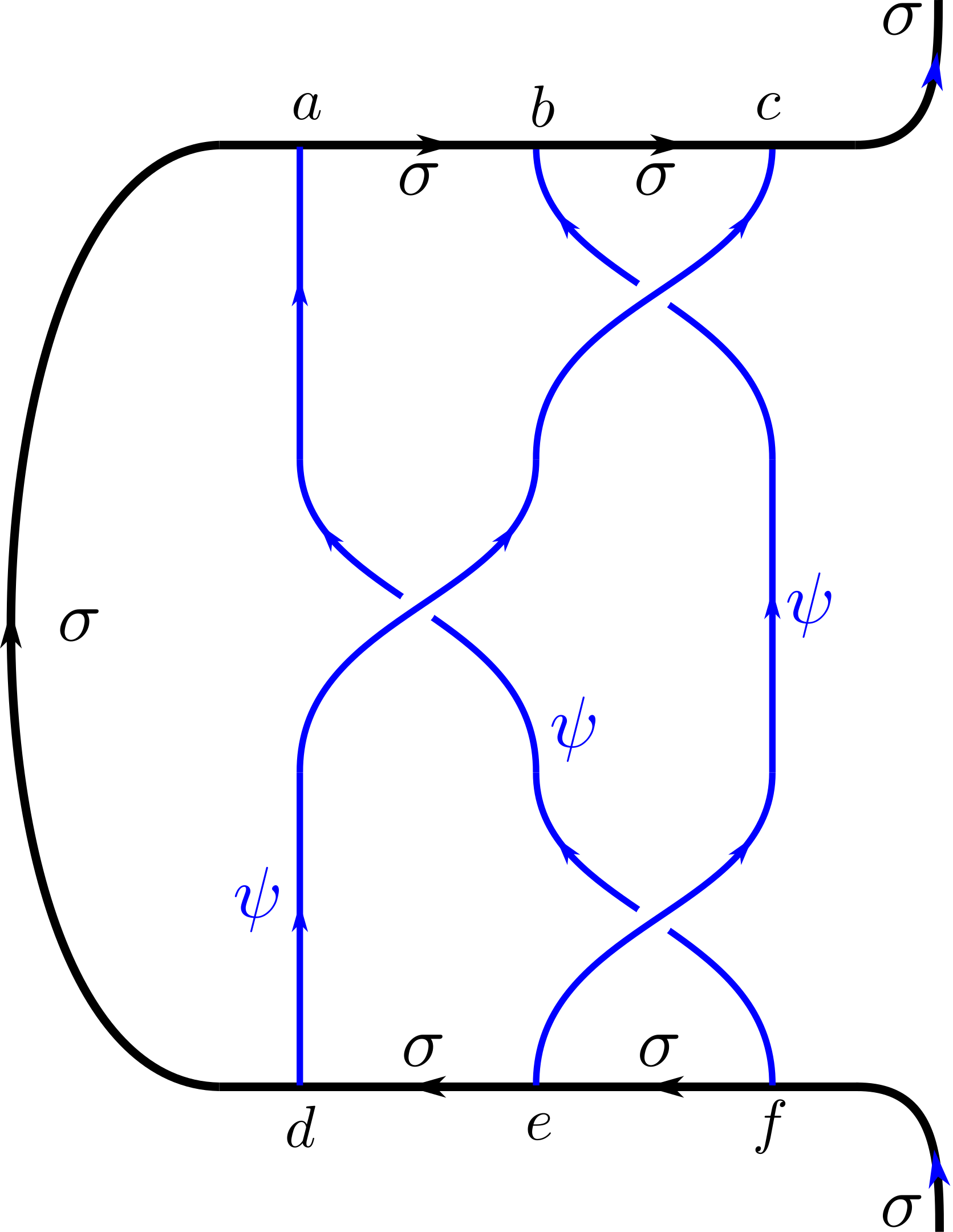}=
	~\adjincludegraphics[height=30ex,valign=c]{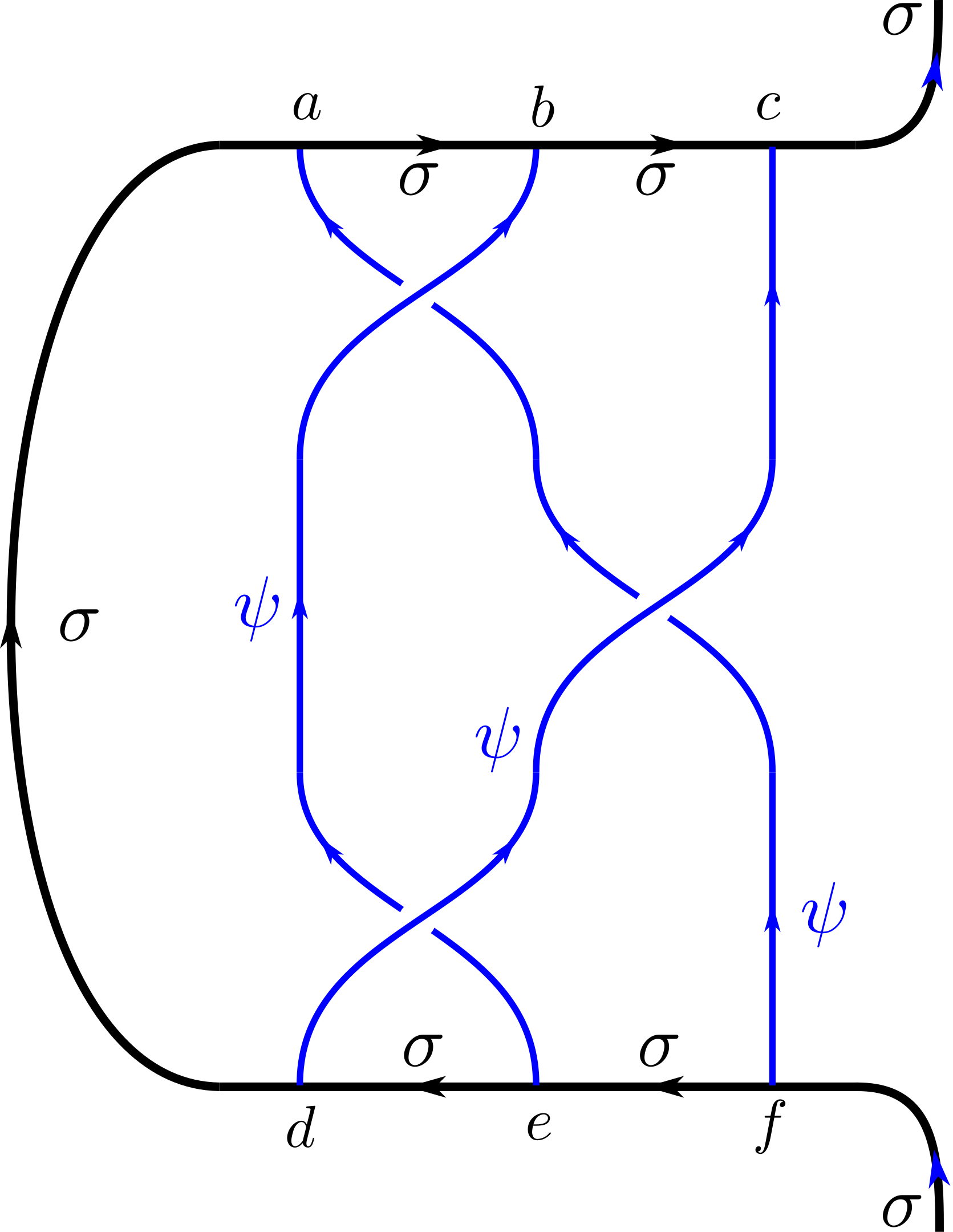},
\end{equation}
which holds more generally in any braided fusion category that has a fusion rule of the form~\eqref{eq:sigmapsifusion-0}. 

We finally remark that if $\calC$ is a braided fusion category describing anyons in 2D, Eq.~\eqref{eq:SFCRYBELR} still holds, but Eq.~\eqref{eq:SFCinvol} is no longer true. Consequently, the $R$-matrix defined in Eq.~\eqref{def:RfromSFC} still satisfies the second relation in Eq.~\eqref{eq:YBE}, but not the first one~(i.e., $R^2\neq \mathds{1}$ for anyons).

\section{Proof: $\sigma \times \psi  =m~\sigma\Rightarrow \psi\times \bar{\sigma}  =m~\bar{\sigma}$}\label{app:proofdualFrule}
Let $\calC$ be a unitary fusion category, and let $\sigma,\psi\in\calC$ be simple particle types. Below we prove a crucial fact we used in Sec.~\ref{sec:win2pt-SFC} that if $\calC$  has a fusion rule $\sigma \times \psi  =m~\sigma$ for some positive integer $m$, then we must have $\psi\times \bar{\sigma}  =m~\bar{\sigma}$, where $\bar{\sigma}$ is the antiparticle of $\sigma$. We first show that
$N_{\sigma\psi}^\sigma=N_{\psi\bar{\sigma}}^{\bar{\sigma}}$. Using the associativity of the fusion product, we have
\begin{alignat}{3}\label{eq:multiplicity_deriv}
	&&(\sigma \times \psi)\times \bar{\sigma}&=\sigma \times (\psi\times \bar{\sigma})\\
	\Leftrightarrow&& \sum_{\beta\in\calC}N_{\sigma\psi}^\beta \beta\times \bar{\sigma} &=\sum_{\beta\in\calC}N_{\psi\bar{\sigma}}^\beta \sigma\times \beta\nonumber\\
	\Leftrightarrow&& \sum_{\beta\in\calC}N_{\sigma\psi}^\beta N_{\beta\bar{\sigma}}^\nu  &=\sum_{\beta\in\calC}N_{\sigma\beta}^\nu N_{\psi\bar{\sigma}}^\beta ,\forall \nu\in\calC.  \nonumber
\end{alignat}
Taking $\nu=I$ in Eq.~\eqref{eq:multiplicity_deriv}, and using the fact that $N_{\sigma\beta}^I=\delta_{\bar{\sigma}\beta}$ for any simple types $\sigma,\beta\in\calC$, we obtain $N_{\sigma\psi}^\sigma=N_{\psi\bar{\sigma}}^{\bar{\sigma}}$. Now suppose we have $\sigma \times \psi  =m~\sigma$, i.e., $N_{\sigma\psi}^\beta=m\delta_{\sigma\beta}$. Inserting into the  last line of Eq.~\eqref{eq:multiplicity_deriv}, we obtain 
\begin{equation}
	m N_{\sigma\bar{\sigma}}^\nu=\sum_{\beta\in\calC}N_{\sigma\beta}^\nu N_{\psi\bar{\sigma}}^\beta=m N_{\sigma\bar{\sigma}}^\nu+\sum_{\beta\neq \bar{\sigma}}N_{\sigma\beta}^\nu N_{\psi\bar{\sigma}}^\beta,
\end{equation}
leading to 
\begin{equation}
	\sum_{\beta\neq \bar{\sigma}}(\sigma\times\beta) N_{\psi\bar{\sigma}}^\beta=0,
\end{equation}
which is impossible unless $N_{\psi\bar{\sigma}}^\beta=0$ for any $\beta\neq \bar{\sigma}$. Combining with $N_{\psi\bar{\sigma}}^{\bar{\sigma}}=N_{\sigma\psi}^\sigma=m$, we obtain $\psi\times \bar{\sigma}  =m~\bar{\sigma}$. 

We now generalize the above result to the module category case where $\sigma,\bar{\sigma}$ are point-like defects. Let $\calM$ be a right module category over the unitary fusion category $\calC$, and let $\calM^{\mathrm{op}}$ be its opposite category. It is useful to make the following identifications 
\begin{eqnarray}
	\calC&\cong& \Fun_{\calC}(\calC,\calC),\nonumber\\
	\calM&\cong&\Fun_{\calC}(\calC,\calM),\nonumber\\
	\calM^{\op}&\cong&\Fun_{\calC}(\calM, \calC),
\end{eqnarray}
as we did in Sec.~\ref{sec:blackdefectQD}. This makes it clear that $\calM^{\op}$ is a left $\calC$-module category, where the module action is given by functor composition. Then we can show that if we have a fusion rule of the form $\sigma \times \psi  =m~\sigma$ for some $\sigma\in\calM$ and $\psi\in\calC$, then we must have  $\psi\times \bar{\sigma}  =m~\bar{\sigma}$, where $\bar{\sigma}\in\calM^\op$ is the dual of $\sigma$. Indeed, the proof above still formally applies here without change. For example, in the module category case, the associativity condition used in Eq.~\eqref{eq:multiplicity_deriv} follows from the associativity of functor composition, and both sides of  %
Eq.~\eqref{eq:multiplicity_deriv} are objects of the fusion category $\Fun_{\calC}(\calM,\calM)$. 
Other steps also generalize to the module category case in a straightforward way. 

\section{A winning strategy using $D(S_3)$ anyon that is robust against the simple anti-anyon twist}~\label{app:winD(S3)}
In this section we present a winning strategy using the non-Abelian anyon in the quantum double model $D(S_3)$ with the $R$-matrix given in Eq.~\eqref{eq:DS3anyon-7}. We will see that this strategy is robust against the simple anti-anyon twist in Sec.~\ref{sec:simpletwist} but still gets blocked by the full-braid scrambler twist in Sec.~\ref{sec:immunity}.

The strategy is as follows. %
During the pregame discussion, the players
agree the following: Alice uses $a\in\{1,2\}$ to encode the bit of information that she will receive from the Referee~(here they choose $m_0=2$ for simplicity), while Bob uses $b\in\{1,3\}$. In this way, after Alice measures $a'$ when the game ends at $t=T$, according to the $R$-matrix in Eq.~\eqref{eq:DS3anyon-7}, if she has $a=1$ and $a'=2$ or $a'=3$, then even without knowing $n$ she immediately knows that $b=3$. This already means that a nonzero amount of information can be transferred between the two players. 
\begin{figure}
	\includegraphics[width=.5\linewidth]{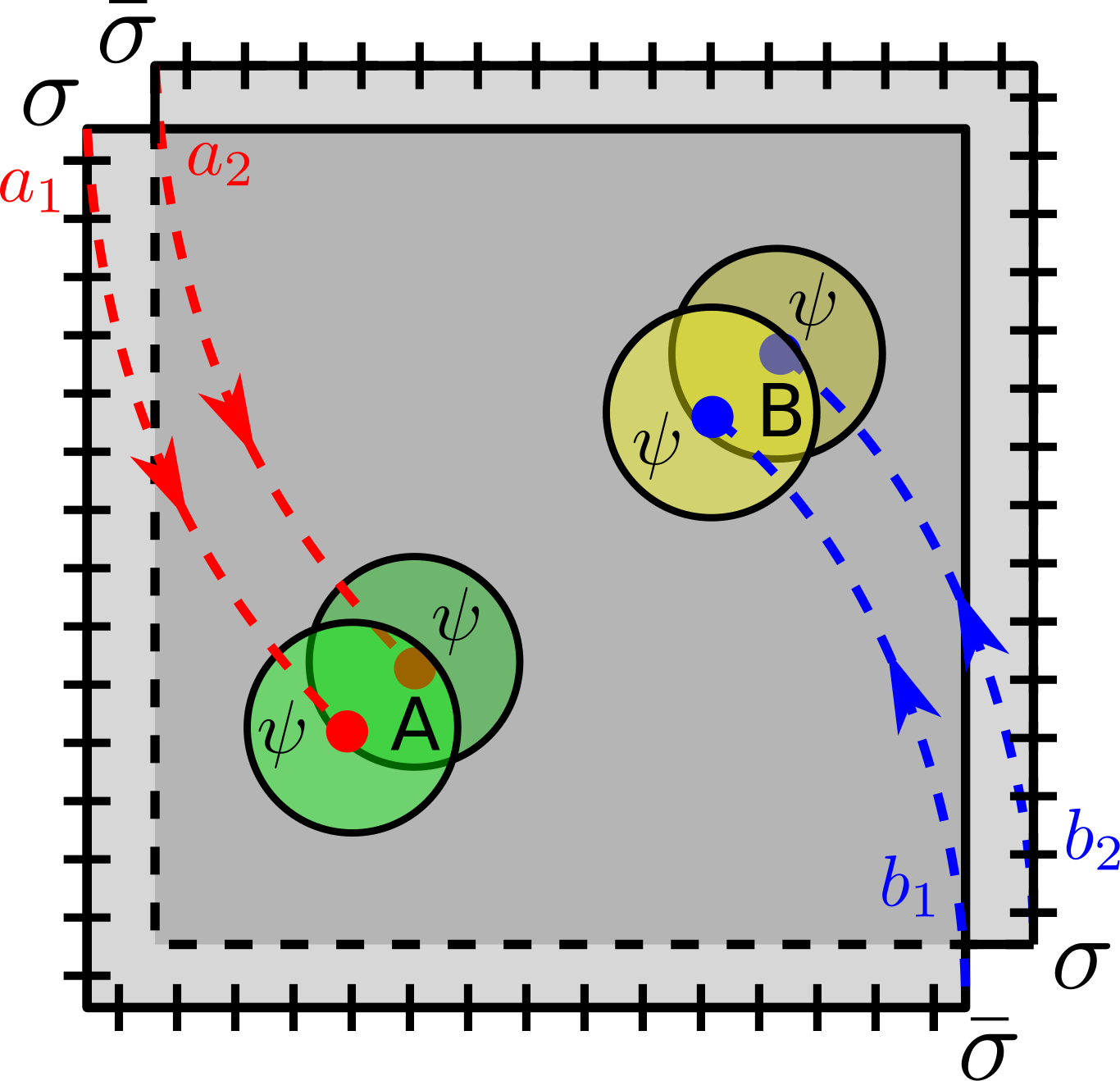}
	\caption{\label{fig:bilayerDS3} Winning strategy for the $D(S_3)$ anyon with $R$-matrix in Eq.~\eqref{eq:DS3anyon-7}. This double layer system~[where each layer is a quantum double model $D(S_3)$] can win the game with 100\% success rate, in a way robust to the simple anti-anyon twist. 
	}
\end{figure}
Indeed, using the multilayer trick introduced in Sec.~\ref{sec:win_2pt}, 
the players can have a winning strategy with a 100\% success rate. Here we stack two layers of this system and use the boundary condition shown in Fig~\ref{fig:bilayerDS3}. Notice that the second layer is obtained by rotating the first layer by 180 degrees, which swaps the positions of the boundary defects $\sigma$ and $\bar{\sigma}$. 
Therefore, with respect to the defect at $\oA$, the $R$-matrix for $\psi$ in the second layer is $R_2=XR_1X$~(this can be obtained directly from the axioms in Sec.~\ref{sec:axioms_emergent_para} by swapping the roles of $\oA$ and $\oB$), where $R_1$ is given in Eq.~\eqref{eq:DS3anyon-7}. In the winning strategy each player uses one $\psi$ in each layer, and Alice uses $a=(a_1,a_2)\in\{(1,1),(2,2)\}$ to encode the number she receives from the Referees, while Bob uses $b=(b_1,b_2)\in\{(1,1),(3,3)\}$. Here $a_1$ and $a_2$ denotes the internal state of the anyon in the first and second layer at $t=0$, respectively, and similarly for $(b_1,b_2)$.   At $t=T$, Alice measures the internal states of $\psi$ in each layer, and obtain $(a'_1,a'_2)$, and similarly for Bob. 
In Tab.~\ref{tab:DS3-7solution} we compute the measurement outcomes $(a'_1 a'_2,b'_1 b'_2)$ for all possible values of $n~(\mathrm{mod}~3)$ and all possible input values $(a_1a_2,b_1b_2)$ with $a_1=a_2, b_1=b_2$. From here we can see that in any case, knowing $a_1,a'_1,a'_2$ allows Alice to uniquely determine $b_1$, and similarly for Bob. 
For example, according to Tab.~\ref{tab:DS3-7solution}, if $(a_1,a'_1,a'_2)\in\{111,223,211,232\}$, then Alice knows $b=1$, and if $(a_1,a'_1,a'_2)\in\{112,133,121,221,233,212\}$, then $b=3$~(importantly, note that these two sets have no overlap, so this algorithm is unambiguously defined). 
Therefore, this strategy works and is robust against the simple twist.   

\begin{table}[t]
	\centering
	\begin{tabular}{c|c|cccc}
		&$(a_1a_2,b_1b_2)$ & (11,11) & (11,33) & (22,11) & (22,33) \\
		\hline
		$n=0$& $(a'_1a'_2,b'_1b'_2)$& (11,11) & (12,23) & (23,31) &(21,13) \\
		$n=1$& $(a'_1a'_2,b'_1b'_2)$ & (11,11) & (33,11) &(11,22) & (33,22) \\
		$n=2$& $(a'_1a'_2,b'_1b'_2)$ & (11,11) & (21,32) &(32,13) & (12,31) \\
	\end{tabular}
	\caption{\label{tab:DS3-7solution} Winning strategy for the $R$-matrix in Eq.~\eqref{eq:DS3anyon-7} that is robust against the simple anti-anyon twist. Here we exhaust all possible measurement outcomes $(a'_1a'_2,b'_1b'_2)$, for $n=0,1,2~(\mathrm{mod}~3)$, and all possible input values $(a_1a_2,b_1b_2)$. }
\end{table}

\section{Computing $R$-matrices from central type factor groups}\label{sec:RfromCentralType}
In this section we show explicitly how to compute $R$-matrices from central type factor groups. In principle, one can always do this computation by directly  evaluating Eq.~\eqref{eq:RmatfromFRmove}, however, to do this in practice, one first need to know the $F$-symbols, which is not easy to compute directly from group-theoretical data. Below in App.~\ref{app:computeRfrommorphisms} we present an alternative approach to compute the $R$-matrix, by obtaining the tensor network representation of the LHS of Eq.~\eqref{def:RfromSFC}. We then present several concrete examples in Apps.~\ref{app:AbCTFG}-\ref{app:G64}. Finally in App.~\ref{app:GCT} we present a special class of finite groups $G$ in which $\Rep(G)$ itself has a fusion rule of the form in Eq.~\eqref{eq:sigmapsifusion-0}.
\subsection{Extracting $R$ from fusion morphisms}\label{app:computeRfrommorphisms}
Let us recall the situation in Sec.~\ref{sec:blackdefectRep(G)}. 
Let $G$ be a finite group, $\calC=\Rep(G)$ be the category of finite dimensional representations of $G$. Let $H$ be a subgroup of $G$ and let $\omega\in \mathcal{H}^2(H,U(1))$ be a non-degenerate 2-cocycle of $H$, i.e., $H$ is a central type factor group, and $\calM=\Rep^\omega(H)$ is a module category over $\Rep(G)$ describing a black defect. In Sec.~\ref{sec:blackdefectRep(G)} we have argued that for any $\psi\in\calC$, we must have a fusion rule $\sigma \otimes \psi = m~\sigma$, where $m$ is the dimension of $\psi$.  
This fusion rule indicates an isomorphism between the corresponding projective representations of $H$, %
meaning that there exists a $d_\sigma d_\psi$-dimensional invertible matrix $V$~(indeed, $V$ can always be chosen to be unitary, which follows from the general fact that if two unitary representations of a group are equivalent, then they are unitarily equivalent) satisfying
\begin{equation}\label{eq:RepInterwiner}
	[\sigma(g) \otimes \psi(g)] \cdot V=V\cdot[I_m \otimes \sigma(g)],\quad \forall ~g\in H, %
\end{equation} 
where $I_m$ is the $m\times m$ identity matrix, so that $I_m \otimes \sigma(g)$ is equivalent to a direct sum of $m$ copies of the projective representation $\sigma$.  
We can rewrite Eq.~\eqref{eq:RepInterwiner} in tensor graphical form as %
\begin{equation}\label{eq:RepInterwinerTN}
	\begin{tikzpicture}[baseline={([yshift=-1.4ex]current bounding box.center)}, scale=.8]
		\umatrix{0}{0}{V}{a}
		\indexflowout{-1.4}{0}{-1}{}
		\indexflowoutV{0}{1.4}{1}{}
		\indexflowoutV{0}{0}{1}{}
		\indexflowin{0}{0}{1}{}
		\indexflowin{-1.4}{0}{1}{}
		\circleTensor{-1.4 }{0}{\small \sigma(g)}
		\circleTensor{0}{1.4}{\small \psi(g)}
	\end{tikzpicture}=
	\begin{tikzpicture}[baseline={([yshift=-.6ex]current bounding box.center)}, scale=.8]
		\umatrix{0}{0}{V}{a}
		\indexflowout{0}{0}{-1}{}
		\indexflowoutV{0}{0}{1}{}
		\indexflowin{0}{0}{1}{}
		\indexflowin{1.4}{0}{1}{}
		\circleTensor{1.4}{0}{\small \sigma(g)}
	\end{tikzpicture}
\end{equation} 
where the arrows indicate the flow of indices in matrix multiplications. In fact, the four-index tensor $\begin{tikzpicture}[baseline={([yshift=.4ex]current bounding box.center)}, scale=.6]
	\umatrix{0}{0}{V}{a}
\end{tikzpicture}$ is the explicit matrix representation of the fusion vertex $\adjincludegraphics[height=5ex,valign=c]{Figures/SFC/fusionvertex.png}$. 
Such a unitary tensor $V$ satisfying Eq.~\eqref{eq:RepInterwinerTN} is unique up to a basis transformation in the fusion space $\Hom(\sigma,\sigma\times \psi)$, i.e., up to a unitary transformation in the index $a$. 
The unitarity of $V$ is expressed by the following tensor graphical equations
\begin{eqnarray}\label{eq:Vunitarity}
	\begin{tikzpicture}[baseline={([yshift=.4ex]current bounding box.center)}, scale=.8]
		\umatrix{0}{1}{V^\dagger}{}
		\umatrix{0}{0}{V}{}
		\indexflowout{0}{1}{1}{}
		\indexflowin{0}{0}{1}{}
		\deltatensorRup{-0.5}{1}
	\end{tikzpicture}~=~ 
	\begin{tikzpicture}[baseline={([yshift=-.4ex]current bounding box.center)}, scale=.8]
		\deltavert{0}{0}
		\deltatensorRup{0.5}{0.5}
	\end{tikzpicture}~~,\quad
	\begin{tikzpicture}[baseline={([yshift=.4ex]current bounding box.center)}, scale=.8]
		\umatrix{0}{1}{V}{}
		\umatrix{0}{0}{V^\dagger}{}
		\indexflowout{0}{1}{-1}{}
		\indexflowin{0}{0}{-1}{}
		\deltatensorup{0.5}{1}
	\end{tikzpicture}~=~ 
	\begin{tikzpicture}[baseline={([yshift=-.4ex]current bounding box.center)}, scale=.8]
		\deltavert{0}{0}
		\deltatensorup{-0.5}{0.5}
	\end{tikzpicture}~,
\end{eqnarray}
which is reminiscent of the orthonormality and completeness relation~\eqref{eq:orthocomplete} satisfied by the fusion vertex. By replacing $\adjincludegraphics[height=5ex,valign=c]{Figures/SFC/fusionvertex.png}$ with $\begin{tikzpicture}[baseline={([yshift=.4ex]current bounding box.center)}, scale=.6]
	\umatrix{0}{0}{V}{a}
\end{tikzpicture}$, we obtain the tensor network representation of Eq.~\eqref{def:RfromSFC}:
\begin{equation}\label{eq:RmatfromGrp}
		\begin{tikzpicture}[baseline={([yshift=-.4ex]current bounding box.center)}, scale=.8]
			\umatrix{0}{0}{V}{a}
			\umatrix{1}{0}{V}{b}
			\umatrix{0}{2}{V^\dagger}{}
			\umatrix{1}{2}{V^\dagger}{}
			\pimatrix{0.5}{1}
			\indexflowin{1}{0}{1}{s}
			\indexflowout{1}{2}{1}{r}
			\deltatensorRlongup{-0.5}{2}{2}
			\hindices{{b',a'}}{0}{2.7}
		\end{tikzpicture}=R^{b'a'}_{ab}\delta_{r,s}.
	\end{equation} 

We give a few remarks before presenting examples. First, notice that Eq.~\eqref{eq:RmatfromGrp} can be rewritten in the following simpler form
\begin{equation}\label{eq:RXweakequiv}
\begin{tikzpicture}[baseline={([yshift=.4ex]current bounding box.center)}, scale=.8]
	\umatrix{0}{0}{V}{}
	\umatrix{1}{0}{V}{}
	\Rmatrix{0.5}{-1}{R}
\end{tikzpicture}= \begin{tikzpicture}[baseline={([yshift=.4ex]current bounding box.center)}, scale=.8]
	\pimatrix{0.5}{1}
	\umatrix{0}{0}{V}{}
	\umatrix{1}{0}{V}{}
\end{tikzpicture}. 
\end{equation}
This gives an important criterion on which type of $R$-matrices can be produced by an SFC: there exists a unitary~[in the sense of Eq.~\eqref{eq:Vunitarity}] four-index tensor $V$ that satisfies Eq.~\eqref{eq:RXweakequiv}. Not every $R$-matrix satisfies this condition, for example, for  $R=-\mathds{1}_{m\times m}$~\cite{wang2023para}, there does not exist a finite dimensional unitary tensor $V$ that satisfies Eq.~\eqref{eq:RXweakequiv}, and therefore $R=-\mathds{1}_{m\times m}$ cannot be produced by any SFC. Note the close similarity between Eq.~\eqref{eq:RXweakequiv} and Eq.~(1) in Ref.~\onlinecite{wang2024hopf}, which is a consequence of the connection between SFC and triangular Hopf algebras~\cite{TenCat_EGNO}. 
Second, the above method only gives the $R$-matrix for $\calC=\Rep(G)$; if one instead considers the SFC $\mathrm{sRep}(G,z)$, where $z\in G$ is a central element of order 2, then one should instead use the braiding of $\mathrm{sRep}(G,z)$ in  the LHS of Eq.~\eqref{def:RfromSFC}. The resulting $R$-matrix is equal to $\psi(z)R^{b'a'}_{ab}$, where $R^{b'a'}_{ab}$ is computed from Eq.~\eqref{eq:RmatfromGrp}. 
Finally, %
one can straightforwardly generalize Eq.~\eqref{def:RfromSFC} to give the mutual $R$-matrix between two particles $\psi,\varphi\in\calC$:
	\begin{equation}\label{eq:RmatfromGrp-mutual}
			\begin{tikzpicture}[baseline={([yshift=-.4ex]current bounding box.center)}, scale=.8]
				\umatrix{0}{0}{V_\psi}{a}
				\umatrix{1}{0}{V_\varphi}{b}
				\umatrix{0}{2}{V^\dagger_\varphi}{}
				\umatrix{1}{2}{V^\dagger_\psi}{}
				\pimatrix{0.5}{1}
				\indexflowin{1}{0}{1}{s}
				\indexflowout{1}{2}{1}{r}
				\deltatensorRlongup{-0.5}{2}{2}
				\hindices{{b',a'}}{0}{2.7}
			\end{tikzpicture}=[R^{(\psi\varphi)}]^{b'a'}_{ab}\delta_{r,s},
		\end{equation} 
		where $V_\psi$ is the solution to Eq.~\eqref{eq:RepInterwinerTN} while $V_\varphi$ is the solution to Eq.~\eqref{eq:RepInterwinerTN} with $\psi(g)$ replaced by $\varphi(g)$. 
		Note that for the case $\calC=\mathrm{sRep}(G,z)$, the mutual $R$-matrix has an extra $(-1)$ sign factor if $\psi(z)=\varphi(z)=-1$. 
		In the following we give concrete examples of the group $G$ and use Eq.~\eqref{eq:RmatfromGrp} to compute the $R$-matrix. 
		
		\subsection{Example: Abelian CTFGs}\label{app:AbCTFG}
		We begin by considering examples of finite groups $G$ that has an Abelian CTFG subgroup $H$. We first prove a general fact
		\begin{fact}\label{fact:AbCTFGSWAP}
		If the CTFG $H$ is Abelian, then the $R$-matrix produced by Eq.~\eqref{eq:RmatfromGrp} 
		 must be a swap-type $R$-matrix of the form  $R^{b'a'}_{ab}=\delta_{aa'}\delta_{bb'}\theta_{ab}$, up to a internal space basis transformation in Eq.~\eqref{eq:basistransformRmat}. 
		\end{fact}
		\begin{proof}
		Since $H$ is an Abelian subgroup of $G$ and $\psi\in\Rep(G)$, we have $\psi(g)\psi(h)=\psi(h)\psi(g)$ for any $g,h\in H$. Therefore, $\{\psi(g)|g\in H\}$ can be simultaneously diagonalized, and without loss of generality, we can assume that all of them are already diagonal, by choosing a suitable basis for the representation space of $\psi$. Let $[\psi(g)]_{ab}=\delta_{ab}\psi_b(g)$, where $\psi_b(g)$ is a U$(1)$ phase factor.
		Then Eq.~\eqref{eq:RepInterwinerTN} becomes
		\begin{equation}\label{eq:sgVcommu}
		\psi_b(g)\sigma(g)V^b_a=V^b_a\sigma(g),\quad\forall ~g\in H,
		\end{equation}
		where for each pair $(a,b)$, $V^b_a$ is a $d_\sigma\times d_\sigma$ matrix, defined as $[V^b_a]_{ij}=
		\begin{tikzpicture}[baseline={([yshift=.4ex]current bounding box.center)}, scale=.65]
			\umatrix{0}{0}{V}{}
			\quantumindices{0}{0}{b}{a}
			\paraindices{0}{0}{i}{j}
		\end{tikzpicture}$. Taking Hermitian conjugate on Eq.~\eqref{eq:sgVcommu}, we obtain
		\begin{equation}\label{eq:sgVcommudagger}
			\psi^*_b(g)\sigma(g)(V^b_a)^\dagger=(V^b_a)^\dagger \sigma(g),\quad\forall ~g\in H.
		\end{equation}
		It then follows from Eqs.~(\ref{eq:sgVcommu},\ref{eq:sgVcommudagger}) that $(V^b_a)^\dagger V^b_a$ commutes with $\sigma(g)$ for all $g\in H$. Since $\sigma$ is an irreducible projective representation of $H$, Schur's lemma implies that $(V^b_a)^\dagger V^b_a=\lambda^b_a \mathds{1}$ for some $\lambda^b_a\geq 0$. It is clear that for each $b$, $\lambda^b_a$ cannot be identically zero for all $a\in\{1,2,\ldots,m\}$--otherwise, $V$ cannot satisfy Eq.~\eqref{eq:Vunitarity}. For each $b$, choose $a(b)\in\{1,2,\ldots,m\}$ such that $\lambda^b_{a(b)}\neq 0$. This allows us to construct another unitary solution to Eq.~\eqref{eq:RepInterwinerTN}
		\begin{equation}\label{Vprimedef}
		\begin{tikzpicture}[baseline={([yshift=.4ex]current bounding box.center)}, scale=.75]
			\umatrix{0}{0}{V'}{}
			\quantumindices{0}{0}{b}{c}
			\paraindices{0}{0}{i}{j}
		\end{tikzpicture}=\delta_{bc}[V^b_{a(b)}]_{ij}/\sqrt{\lambda^b_a}.
		\end{equation}
It is clear from Eq.~\eqref{eq:RmatfromGrp} that the $R$-matrix constructed from $V'$ must be of the swap-type due to the $\delta_{bc}$ factor in the RHS of Eq.~\eqref{Vprimedef}. 
As we remarked below Eq.~\eqref{eq:RepInterwinerTN}, the two unitary solutions $V$ and  $V'$ to Eq.~\eqref{eq:RepInterwinerTN} must be related by a unitary transformation in the index $a$, and from  Eq.~\eqref{eq:RmatfromGrp} we see that the two $R$-matrices produced by $V$ and  $V'$ must be related by a internal space basis transformation in Eq.~\eqref{eq:basistransformRmat}, which proves our claim. 
		\end{proof}
		\subsubsection{The group $G=D_8$~(level 3)}\label{app:AbCTFGD8}
		The dihedral group $D_8$ is the group of symmetries of a square, and has the following presentation
		\begin{equation}\label{def:DihedralGrp}
			D_8=\left\langle r, s \mid r^4=s^2=1, \quad s r s^{-1}=r^{-1}\right\rangle.
		\end{equation}
		We also have $D_8\cong Z_2\ltimes (Z_2\times Z_2)$, with presentation
		\begin{eqnarray}\label{def:DihedralGrp-alt}
			D_8=\langle x,z_1,z_2 \mid&& x^2=z_1^2=z_2^2=1,\nonumber\\
			&& x z_1=z_2x, z_1 z_2=z_2z_1\rangle.
		\end{eqnarray}
		The relation between the two presentations is given by $(x,z_1,z_2)=(sr,s,sr^2)$. 
		$D_8$ has a 2-dimensional representation $\psi\in\Rep(D_8)$ defined by
		\begin{equation}\label{rep:D8rep2}
			\psi(r)=\begin{pmatrix}
				i & 0 \\
				0 & -i
			\end{pmatrix},\quad \psi(s)=\begin{pmatrix}
				0 & 1 \\
				1 & 0
			\end{pmatrix}.
		\end{equation}
		$D_8$ has a subgroup $H=Z_2\times Z_2$ generated by $z_1$ and $z_2$, which is a CTFG, with a 2-dimensional irreducible projective representation $\sigma\in\Rep^\omega(H)$, defined by  
		\begin{equation}\label{eq:Z2Z2projRep}
			\sigma(z_1)=\sigma^x, \quad \sigma(z_2)=\sigma^y.
		\end{equation}
		In this case, we have $\sigma\otimes\psi=2\sigma$, and the solution to Eq.~\eqref{eq:RepInterwinerTN} is expressed as
		\begin{equation}
			\begin{tikzpicture}[baseline={([yshift=.4ex]current bounding box.center)}, scale=.65]
				\umatrix{0}{0}{V}{}
				\quantumindices{0}{0}{1}{1}
				\paraindices{0}{0}{i}{j}
			\end{tikzpicture}=[\sigma^x]_{ij},\quad
			\begin{tikzpicture}[baseline={([yshift=.4ex]current bounding box.center)}, scale=.65]
				\umatrix{0}{0}{V}{}
				\quantumindices{0}{0}{2}{2}
				\paraindices{0}{0}{i}{j}
			\end{tikzpicture}=[\sigma^y]_{ij}.
		\end{equation}
		One can then use Eq.~\eqref{eq:RXweakequiv} to conveniently extract the $R$-matrix, and the result is $R^{b'a'}_{ab}=(-1)^{a+b} \delta_{aa'}\delta_{bb'}$ 
		with quantum dimension $m=2$.

\subsubsection{The group $G=A_4$~(level 4)}\label{app:CTFGA4}
The alternating group $A_4$ is defined as the group of even permutations of the set $\{1,2,3,4\}$, and has order $|A_4|=12$. It can alternatively be  described as a semidirect product $Z_3\ltimes (Z_2\times Z_2)$ with the following presentation
\begin{eqnarray}\label{def:A4Grp}
	A_4=\langle x,z_1,z_2 \mid&& x^3=z_1^2=z_2^2=1,~z_1z_2=z_2z_1,\nonumber\\
	&&xz_1=z_2x,~xz_2=z_1z_2 x \rangle,
\end{eqnarray}
where the permutation representation of the generators are
\begin{equation}
	x=(123),\quad z_1=(12)(34),\quad z_2=(14)(23).	
\end{equation}
For later convenience, we define the following $3\times 3$ matrices
\begin{alignat}{3}
	p_1&=\begin{pmatrix}
		1 & 0 & 0 \\
		0 & -1 & 0\\
		0 & 0 & -1
	\end{pmatrix},&~~
	p_2&=\begin{pmatrix}
		-1 & 0 & 0 \\
		0 & 1 & 0\\
		0 & 0 & -1
	\end{pmatrix},\nonumber\\
	s&=\begin{pmatrix}
		1 & 0 & 0 \\
		0 & \omega & 0\\
		0 & 0 & \omega^2
	\end{pmatrix},&~~
	t&=\begin{pmatrix}
		0 & 0 & 1 \\
		1 & 0 & 0\\
		0 & 1 & 0
	\end{pmatrix},\label{def:p1p2st}
\end{alignat}
where $\omega=e^{2\pi i/3}$, %
and $s,t$ are constructed to satisfy
\begin{equation}\label{eq:stalgebra}
	s^3=t^3=1,\quad st=\omega ts.
\end{equation}
Then $A_4$ has the following 3-dimensional irreducible representation
\begin{eqnarray}\label{eq:A4_3dimrep}
	\psi(z_1)=p_1,~\psi(z_2)=p_2,~\psi(x)=t.
			\end{eqnarray}
			$A_4$ also has a CTFG subgroup $H=Z_2\times Z_2$ generated by $z_1$ and $z_2$, and we consider the same irreducible projective representation $\sigma$ defined by  
			Eq.~\eqref{eq:Z2Z2projRep}.
			In this case, we have $\sigma\otimes\psi=3\sigma$, and the solution to Eq.~\eqref{eq:RepInterwinerTN} is 
			\begin{equation}
				\begin{tikzpicture}[baseline={([yshift=.4ex]current bounding box.center)}, scale=.65]
					\umatrix{0}{0}{V}{}
					\quantumindices{0}{0}{1}{1}
					\paraindices{0}{0}{i}{j}
				\end{tikzpicture}=[\sigma^x]_{ij},~
				\begin{tikzpicture}[baseline={([yshift=.4ex]current bounding box.center)}, scale=.65]
					\umatrix{0}{0}{V}{}
					\quantumindices{0}{0}{2}{2}
					\paraindices{0}{0}{i}{j}
				\end{tikzpicture}=[\sigma^y]_{ij},~
				\begin{tikzpicture}[baseline={([yshift=.4ex]current bounding box.center)}, scale=.65]
					\umatrix{0}{0}{V}{}
					\quantumindices{0}{0}{3}{3}
					\paraindices{0}{0}{i}{j}
				\end{tikzpicture}=[\sigma^z]_{ij}.
			\end{equation}
			The $R$ matrix computed from Eq.~\eqref{eq:RXweakequiv} is $R^{b'a'}_{ab}=(-1)^{\delta_{ab}+1}\delta_{aa'}\delta_{bb'}$ with quantum dimension $m=3$.

						\subsection{Example: the CTFG $G=A_4\times Z_3$~(level 5)}\label{sec:A4Z3}
					We now give an example of a non-Abelian CTFG. Consider the group $G=A_4\times Z_3$, the direct product of $A_4$ and $Z_3$, with order $|G|=36$. We use $z$ to denote the generator of $Z_3$. Then $G$ is the group generated by $z_1,z_2,x,z$ where $z_1,z_2,x$ satisfies all the relations in Eq.~\eqref{def:A4Grp}, and $z$ satisfies $z^3=1$ and commutes with $z_1,z_2,x$. 
					
					Let $\psi$ be the 3-dimensional irreducible representation of $G$ where $\psi(z_1),\psi(z_2),\psi(x)$ are given in Eq.~\eqref{eq:A4_3dimrep} and $\psi(z)=I_3$. Let $\sigma$ be the following projective representation of $G$:
					\begin{alignat}{3}
						\sigma(z_1)&=i\sigma^x\otimes I_3,&\quad \sigma(x)&=q\otimes t,\nonumber\\
						\sigma(z_2)&=i\sigma^y\otimes I_3,& \sigma(z)&=I_2\otimes s,
					\end{alignat}
					where $$q=\exp\left(\frac{2\pi i}{3}\frac{\sigma^x+\sigma^y+\sigma^z}{\sqrt{3}}\right),$$
					and $s,t$ are defined in Eq.~\eqref{def:p1p2st}. 
					Notice that $\sigma$ is irreducible and has dimension $d_\sigma=6=\sqrt{|G|}$. Therefore $G$ is a central-type factor group we are looking for. 
					
					To construct the $R$-matrix, we find the matrix $V$ satisfying Eq.~\eqref{eq:RepInterwiner}, 
					and insert it into  Eq.~\eqref{eq:RmatfromGrp} to obtain
					\begin{equation}\label{eq:RA4Z3}
						R\ket{a,b}=\varphi(a,b)\ket{b+1,a-1},
					\end{equation}
					where $a,b=1,2,3$ are understood modulo 3, and $\varphi(a,b)=-(-1)^{\delta_{a,b+1}}$. 
					It is straightforward to verify that this $R$-matrix can win the game in a way robust against noise and eavesdropping, and it can also win the who-entered-first challenge. Therefore, this $R$-matrix describes a ``full-fledged paraparticle''~(level 5).
					\subsection{Example: the CTFG $G=D_8\ltimes Z_2^{\times 3}$~(level 5)}\label{app:G64}
					We now define a non-Abelian CTFG that leads to the $R$-matrix in Eq.~\eqref{eqApp:seth-R}. Here we 
					present the Abelian group $Z_2^{\times 3}$ by four mutually commuting elements $z_1,z_2,z_3,z_4$, satisfying $z_j^2=z_1z_2z_3z_4=1$. The action of $D_8$ on $Z_2^{\times 3}$ is defined as follow
					\begin{equation}\label{eq:D8actiononA8}
						r z_j r^{-1}=z_{j+1},\quad s z_j s^{-1}=z_{\bar{j}}, %
					\end{equation}
					where $(\bar{1},\bar{2},\bar{3},\bar{4})=(2,1,4,3)$, and the subscript $j$ is understood modulo 4. This action can be depicted in the following graph:
					\begin{equation}
						\begin{tikzpicture}[baseline={([yshift=.4ex]current bounding box.center)}, scale=.8]
							\node (r) at (1,1) {$r?r^{-1}$};
							\node (z4) at (0,0) {$z_4$};
							\node (z1) at (0,2) {$z_1$};
							\node (z2) at (2,2) {$z_2$};
							\node (z3) at (2,0) {$z_3$};
							\draw[very thick,->-=.7]  (z1) to (z2);
							\draw[very thick,->-=.7]  (z2) to (z3);
							\draw[very thick,->-=.7]  (z3) to (z4);
							\draw[very thick,->-=.7]  (z4) to (z1);
						\end{tikzpicture},\quad
						\begin{tikzpicture}[baseline={([yshift=.4ex]current bounding box.center)}, scale=.8]
							\node (r) at (1,1) {$s?s^{-1}$};
							\node (z4) at (0,0) {$z_4$};
							\node (z1) at (0,2) {$z_1$};
							\node (z2) at (2,2) {$z_2$};
							\node (z3) at (2,0) {$z_3$};
							\draw[very thick,->-=1.]  (z1) to (z2);
							\draw[very thick,->-=1.]  (z2) to (z1);
							\draw[very thick,->-=1.]  (z3) to (z4);
							\draw[very thick,->-=1.]  (z4) to (z3);
						\end{tikzpicture}.
					\end{equation}
					Therefore, $D_8$ acts on the  generators $z_1,z_2,z_3,z_4$ by its defining action on a square. Using this action, we construct a semidirect product group  $G=D_8\ltimes Z_2^{\times 3}$, which has order $|G|=64$. Explicitly, $G$ is generated by $\{z_1,z_2,z_3,z_4,s,t\}$ subject to the relations in Eqs.~(\ref{def:DihedralGrp},\ref{eq:D8actiononA8}) along with $z_j^2=z_1z_2z_3z_4=1$. $G$ has the following two irreducible 4-dimensional representations $\psi_{\pm}$:
					\begin{eqnarray}\label{eq:defpsipmG64}
						\psi_{\pm}(z_1)&=&\sigma_1^z,~ \psi_{\pm}(z_2)=\sigma_2^z,~\psi_{\pm}(z_3)=-\sigma_1^z,%
						\nonumber\\
						\psi_{\pm}(s)&=&\pm X_{12},\quad \psi_{\pm}(r)=X_{12}\sigma_2^x.
					\end{eqnarray}
					We now construct an 8-dimensional irreducible projective representation of $G$. Let $\sigma_2$ be the 2-dimensional projective representation defined as
					\begin{equation}\label{rep:G64sigma2}
						\sigma_2(r)=\begin{pmatrix}
							1 & 0 \\
							0 & i
						\end{pmatrix},\quad \sigma_2(s)=\begin{pmatrix}
							0 & 1 \\
							1 & 0
						\end{pmatrix},
					\end{equation}
					and $\sigma_2(z_j)=\mathds{1}_2$. 
					Now we construct a 4-dimensional irreducible projective representation $\sigma_4$ of $G$ as follow
					\begin{eqnarray}
						\sigma_4(z_j)&=&i\gamma_j\gamma_{j+1},\nonumber\\
						\sigma_4(r)&=&\gamma_2 \gamma_3 \gamma_4 B_{1} B_{2} B_{3},\nonumber\\
						\sigma_4(s)&=&\frac{\gamma_1+\gamma_3}{\sqrt{2}} i \gamma_2 \gamma_4,
					\end{eqnarray}
					where $\{\gamma_j\}_{j=1}^4$ are Dirac matrices~(i.e. generators of Clifford algebra, a.k.a. Majorana fermion operators) satisfying
					\begin{equation}
						\{\gamma_i,\gamma_j\}=2\delta_{i,j},\quad 1\leq i,j\leq 4,
					\end{equation}
					and 
					\begin{equation}
						B_{j}=\frac{1-\gamma_{j} \gamma_{j+1}}{\sqrt{2}}, \quad j=1,2,3,
					\end{equation}
					are braid matrices. It is straightforwardly verified that
					\begin{eqnarray}
						r \gamma_j r^{-1}&=&\gamma_{j+1},\quad s \gamma_1 s^{-1}=\gamma_{3},\nonumber\\
						s \gamma_2 s^{-1}&=&-\gamma_{2},\quad s \gamma_4 s^{-1}=-\gamma_{4}, %
					\end{eqnarray}
					and $\sigma_4$ is an irreducible projective representation of $G$. 
					
					We now define $\sigma=\sigma_4\otimes\sigma_2$. Then it is straightforward to verify that $\sigma$ is also irreducible, and has dimension $\sigma=8=\sqrt{|G|}$.  Therefore $G$ is a central-type factor group. Note that $z=z_1z_3$ is a central element of order 2, and we have $\psi_\pm(z)=-1$. Below we consider the SFC $\calC=\mathrm{sRep}(G,z)$,
					and use Eq.~(\ref{eq:RmatfromGrp}) to compute the $R$-matrices of $\psi_{\pm}$ with respect to the black defect $\sigma$. The $R$-matrix for $\psi_-$ is given in Eq.~\eqref{eqApp:seth-R}. The $R$-matrix for $\psi_+$ is similar, but with 
					\begin{eqnarray}\label{eq:seth-R-G64-p}
						(b',a')%
						=\left(
						\begin{array}{cccc}
							11 & 44 & 32 & 23 \\
							33 & 22 & 14 & 41 \\
							42 & 13 & 21 & 34 \\
							24 & 31 & 43 & 12 \\
						\end{array}
						\right)_{ab}.
					\end{eqnarray}
					The mutual $R$-matrix between $\psi_+$ and $\psi_-$ is computed using Eq.~(\ref{eq:RmatfromGrp}),
					where $\psi=\psi_+$ and $\varphi=\psi_-$, and the result is 
					\begin{equation}\label{eqApp:seth-R-G64mutual}
						[R^{+-}]^{b'a'}_{ab}=1, \text{ if }	(b',a')=
						\left(
						\begin{array}{cccc}
							12 & 34 & 23 & 41 \\
							43 & 21 & 32 & 14 \\
							24 & 42 & 11 & 33 \\
							31 & 13 & 44 & 22 \\
						\end{array}
						\right)_{ab}.%
					\end{equation}

\subsection{Groups of central type: the case where $\sigma\in \Rep(\GCT)$ is a quasiparticle}\label{app:GCT}
In all examples above, $\sigma$ is an object of a $\calC$-module category $\calM=\Rep^\omega(H)$ describing a point-like defect, rather than a quasiparticle in $\calC$ itself. In this section, we present a special family of finite groups $\GCT$ such that $\calC=\Rep(\GCT)$ contain paraparticles and in addition,  $\sigma$ is itself a quasiparticle in $\calC$, i.e., $\calC$ itself has a fusion rule of the form in Eq.~\eqref{eq:sigmapsifusion-0}. 
This special family of groups are called groups of central type~\cite{iwahori1964several,Liebler1979,Howlett1982}, and we recall their definition here
\begin{definition}
	Let $\GCT$ be a finite group, and let $Z(\GCT)=\{c\in \GCT|cg=gc,\forall g\in \GCT\}$ be the center of $\GCT$. %
	We say that $\GCT$ is a  group of central type~\cite{iwahori1964several,Liebler1979,Howlett1982} if $\GCT$ has an irreducible representation $\sigma$ with dimension $d_\sigma$ satisfying
	\begin{equation}\label{def:centraltype}
		d_\sigma^2=|\GCT|/|Z(\GCT)|. 
	\end{equation}
\end{definition}
Now let $\GCT$ be a group of central type and let $\sigma$ be a representation satisfying Eq.~\eqref{def:centraltype}. Since $\sigma$ is irreducible, Schur's lemma implies that $\sigma(z)\propto\mathds{1}$ for any $z\in Z(\GCT)$. Therefore $\sigma$ is an irreducible projective representation of the quotient group $\CTFG=\GCT/Z(\GCT)$. 
Since $d_\sigma^2=|\CTFG|$~[which follows from Eq.~\eqref{def:centraltype}], the density theorem in representation theory~\cite{etingof2011introduction} implies that $\sigma: \C_\omega[\CTFG]\to M_{d_\sigma}(\C)$ defines an isomorphism of algebras, and the representation theory of the $d_\sigma\times d_\sigma$ matrix algebra $M_{d_\sigma}(\C)$ implies that $\sigma$ is the only irreducible representation of $\C_\omega[\CTFG]$~(up to isomorphism). 
Here $\C_\omega[\CTFG]$ is the group algebra of $\CTFG$ twisted by the cocycle $\omega$, which is a $|\CTFG|$-dimensional algebra spanned by basis elements $\{e_g|g\in \CTFG\}$ with multiplication $e_g\cdot e_h=\omega(g,h)e_{gh}$. 
Therefore, $\CTFG$ is a CTFG. 
Furthermore, any irrep $\rho$ of $\GCT$ that satisfies 
$\rho(z)\sim \sigma(z)~\forall z\in Z(\GCT)$ must be isomorphic to $\sigma$, where $\rho(z)\sim \sigma(z)$ means that $\rho(z)=c(z)\mathds{1}_{d_\rho}, \sigma(z)=c(z)\mathds{1}_{d_\sigma}$ for some $c(z)\in \C$.  %
Now let $\psi\in \Rep(\GCT)$ be an irrep with dimension $d_\psi=m>1$, such that $\psi(z)=\mathds{1},~\forall z\in Z(\GCT)$, so that $\psi$ defines a linear irrep of the quotient group $\CTFG$~(note that the existence of such $\psi$ requires the CTFG $\CTFG$ to be non-Abelian).  
The tensor product representation $\rho=\sigma\otimes \psi\in\Rep(\GCT)$ satisfies $\rho(z)\sim \sigma(z),~\forall z\in Z(\GCT)$, therefore $\sigma\otimes \psi$ must be isomorphic to a direct sum of $m$ copies of $\sigma$, leading to the fusion rule in Eq.~\eqref{eq:sigmapsifusion-0}. 

Therefore, if $\GCT$ is a group of central type such that the quotient group $\CTFG=G/Z(\GCT)$ is non-Abelian, then $\Rep(\GCT)$ has a fusion rule of the form in Eq.~\eqref{eq:sigmapsifusion-0} with $m>1$. 
The $R$-matrix is computed using the same Eq.~\eqref{eq:RmatfromGrp}, where $V$ is the isomorphism between the two projective representations of $\CTFG$, $\sigma\otimes\psi$ and $m\sigma$.
It is useful to note the close relationship between CTFGs and groups of central type: if $\GCT$ is a group of central type, then $G/Z(\GCT)$ is a CTFG~\footnote{A warning about terminology: in some math literature people use ``group of central type'' to refer to what we call a CTFG here. }; conversely, as we mentioned at the end of Sec.~\ref{sec:blackdefectRep(G)}, given any CTFG $\CTFG$, one can always construct a group of central type $\GCT$ as a central extension of $\CTFG$ by lifting the projective representation $\sigma$ to a linear representation, such that $G/Z(\GCT)=\CTFG$.

For example, if we take $\CTFG=A_4\times Z_3$ as defined in App.~\ref{sec:A4Z3}, the corresponding central extension $\GCT$ has order $|\GCT|=216$, with the following presentation
\begin{eqnarray}\label{def:A4Z3CentExt}
	G=\langle x,z,a,b,c \mid&& x^3=z^3=c^6=1,~a^2=b^2=c^3,\nonumber\\
	&& ab=c^3ba,~xa=bx,~xb=ba x,\nonumber\\
	&& zx=c^2 xz,~za=az,~zb=bz,\nonumber\\
	&& c \text{ commutes with }x,z,a,b\rangle.
\end{eqnarray}
One can check that the group center $Z(\GCT)$ is the cyclic group generated by $c$, and $G/Z(\GCT)$ is exactly $\CTFG=A_4\times Z_3$. 
The irreducible projective representation $\sigma\in\Rep^\omega(\CTFG)$ is lifted to a linear representation of $\GCT$ with $\sigma(c)=e^{2\pi i/6}\mathds{1}$. The representation $\psi\in\Rep(\CTFG)$ defined in App.~\ref{sec:A4Z3} is lifted to a representation of $\GCT$ with $\psi(c)=\mathds{1}$. The $R$-matrix of the paraparticle $\psi$ with respect to the quasiparticle $\sigma$ %
is the same as that given in Eq.~\eqref{eq:RA4Z3}. 

In a similar way, for the CTFG $\CTFG=D_8\ltimes Z_2^{\times 3}$ defined in App.~\ref{app:G64}, the corresponding central extension $\GCT$ has order $|\GCT|=128$, which is generated by 
$r, s, z_1,z_2,z_3,z_4$ along with a central element $c$ subject to the relations
\begin{eqnarray}\label{def:G128}
	&& r^4=s^2=c^2=1,~ s r s^{-1}=r^{-1},\nonumber\\
	&& z_j^2=1,~	r z_j r^{-1}=z_{j+1},~ s z_j s^{-1}=c z_{\bar{j}}, \nonumber\\
	&& z_i z_j=c^{i-j} z_j z_i,~\text{ for }1\leq i,j\leq 4,
\end{eqnarray}
where $(\bar{1},\bar{2},\bar{3},\bar{4})=(2,1,4,3)$, and the subscript $j$ is understood modulo 4. The group center $Z(\GCT)$ is generated by $c$, and $G/Z(\GCT)$ is exactly $\CTFG=D_8\ltimes Z_2^{\times 3}$. 
The projective representation $\sigma\in\Rep^\omega(\CTFG)$ is lifted to a linear representation of $\GCT$ with $\sigma(c)=-\mathds{1}$, and the representations $\psi_\pm\in\Rep(\CTFG)$ defined in Eq.~\eqref{eq:defpsipmG64} are lifted to linear representations of $\GCT$ with $\psi_\pm(c)=\mathds{1}$. The $R$-matrices of $\psi_\pm$ with respect to the quasiparticle $\sigma$ are the same as that given in  Eqs.~(\ref{eqApp:seth-R},\ref{eq:seth-R-G64-p},\ref{eqApp:seth-R-G64mutual}). %

\section{Black defects in quantum double phases}\label{sec:blackdefectQD}
In this section we give examples of black defects in 2D topological phases, which is crucial for winning the 2D version of the challenge game. 
For simplicity, we restrict to quantum double topological orders~\cite{kitaev2003fault}, described by modular tensor categories of the form $\calC=\Rep[D(H)]$, where $H$ is a finite group or more generally a finite dimensional $\C^*$-Hopf algebra~\cite{Buerschaper2013HATC}, and $D(H)$ is its Drinfeld double~\cite{drinfeld1986quantum}. Point-like defects in such 2D topological phases can be realized at the end points of 1D boundaries and domain walls, as shown in 
Fig.~\ref{fig:DefectRealization}. %

Let us first briefly review the categorical description of boundaries and domain walls in Levin-Wen string-net models~\cite{levin2005string}~(which include Kitaev's quantum double models~\cite{kitaev2003fault} as special cases) given in Ref.~\onlinecite{Kitaev2012gappedboundary}. 
Given an arbitrary input unitary fusion category $\calC_0$, the string-net construction produces a 2D topological order described by the Drinfeld center $\calC=Z(\calC_0)$. Gapped boundaries of this model are in one-to-one correspondence to indecomposable module categories over $\calC_0$. Given a gapped boundary defined by the module category $\calM_0$~(also called a type-$\calM_0$ boundary), its boundary excitations form a fusion category $\B=\mathrm{Fun}_{\calC_0}(\calM_0,\calM_0)$, in which objects are $\calC_0$-module functors from $\calM_0$ to $\calM_0$, and fusion of objects corresponds to functor composition. 
A quasiparticle in the bulk $\psi\in Z(\calC_0)$ can be moved to the boundary and becomes a boundary excitation $\mathfrak{F}(\psi)\in \B$, %
where $\mathfrak{F}: Z(\calC_0) \to \B$ is a central functor  called the bulk-to-boundary map~\cite{kong2015boundary,kongBoundarybulkRelationTopological2017}. 

\subsection{Black defects on gapped boundaries}\label{sec:blackdefectgappedboundary}
Consider the configuration shown in Fig.~\ref{fig:BoundaryDefect}, where the string-net model is surrounded by two different types of gapped boundaries labeled by $\calC_0$-module categories $\calM_0$ and $\calN_0$, respectively. We have
\begin{eqnarray}\label{eq:MoritaContext}
	\B&=&\mathrm{Fun}_{\calC_0}(\calM_0,\calM_0),\nonumber\\
	\calD&=&\Fun_{\calC_0}(\calN_0,\calN_0),\nonumber\\
	\calM&=&\Fun_{\calC_0}(\calM_0,\calN_0),\nonumber\\
	\calM^{\op}&=&\Fun_{\calC_0}(\calN_0,\calM_0),
\end{eqnarray}
where $\calM$ describes the point-like defect at the upper-left corner, and a state of this defect corresponds to a $\calC_0$-module functor from $\calM_0$ to $\calN_0$, and similarly for $\calM^{\op}$. $\calM$ is naturally a module category over $\calC$, where the fusion of a bulk excitation $\psi\in \calC$ with $\sigma\in \calM$ is given by functor composition $\mathfrak{F}(\psi)\circ\sigma\in \calM$. 
In the final step shown in Fig.~\ref{fig:BoundaryDefect}, the fusion between $\sigma$ and $\bar{\sigma}$ is also defined by functor composition
$\sigma\otimes \bar{\sigma}:=\sigma\circ\bar{\sigma}\in  \Fun_{\calC_0}(\calN_0,\calN_0)$, therefore, the fusion produces a boundary excitation in $q\in\calD$. %

We now specialize to concrete examples where $\calM$ in Eq.~\eqref{eq:MoritaContext} describes a black defect. To explicitly compute $\calM$ and $\calM^\op$,  the following general result will be useful: for any module category $\calM_0$ over a fusion category $\calC_0$, we have
\begin{equation}\label{eq:ModFunGeneralResult}
	\Fun_{\calC_0}(\calC_0,\calM_0)\cong\calM_0,
\end{equation}
where $\cong$  means equivalent as right-$\calC_0$ module categories. 
The equivalence is established as follows. For any $\sigma\in \calM_0$, we define a right $\calC_0$-module functor $\mathfrak{F}_\sigma: \calC_0\to\calM_0$ as $\mathfrak{F}_\sigma(\psi):=\sigma\otimes \psi$ for any $\psi\in \calC_0$. It is straightforward to show that all elements of  $\Fun_{\calC_0}(\calC_0,\calM_0)$ must be of this form, and the map $\sigma\to \mathfrak{F}_\sigma$ defines a categorical equivalence. 

Our first example realizes $R$-paraparticles in $\Rep(G)$ and $\mathrm{sRep}(G,z)$ in 2D.  
Let $G$ be a central type factor group, and let $\omega$ be a non-degenerate 2-cocycle of $G$. 
In Eq.~\eqref{eq:MoritaContext}, take
\begin{equation}\label{eq:hybridboundaryDG}
	\calC_0=\Rep(G), \quad \calM_0=\Rep^\omega(G), \quad \calN_0=\calC_0.
\end{equation}
We have seen before that $\calM_0$ has only one simple object given that $\omega$ is non-degenerate. Applying Eq.~\eqref{eq:ModFunGeneralResult} to  Eq.~\eqref{eq:MoritaContext}, we get $\calM^\op=\calM_0$, which describes a black defect at $\oB$. Its opposite category  $\calM=\calM_0^\op$  describes a black defect at the upper-right corner. Given any quasiparticle in the bulk $\psi\in \calC=\Rep[D(G)]$, we can use Eq.~\eqref{def:RfromSFC} 
to compute its $R$-matrix. In particular, the $R$-paraparticle in Eq.~\eqref{eqApp:seth-R} can be realized in this way with %
$G_{64}=D_8\ltimes Z_2^{\times 3}$, a central type factor group of order 64, whose definition is given in App.~\ref{sec:RfromCentralType}.

As a second example, let $H$ be a finite dimensional $\C^*$-Hopf algebra, and consider Kitaev's quantum double model based on $H$, which realizes the same topological phase as a string-net model with $\calC_0=\Rep(H)$. Let $\calM_0=\mathrm{Vec}$ and $\calN_0=\Rep(H)$ corresponding to the rough and smooth boundaries of this model, respectively. Here $\mathrm{Vec}$ is the category of finite dimensional vector spaces, which becomes a module category over $\Rep(H)$ via the forgetful functor $\mathrm{Forg}: \Rep(H)\to \mathrm{Vec}$~(which forgets about the $H$-action). In this case, we have
\begin{eqnarray}\label{eq:HAQDMoritaContext}
	\calC&=&Z[\Rep(H)]\cong\Rep[D(H)],\nonumber\\
	\B&=&\mathrm{Fun}_{\Rep(H)}(\mathrm{Vec},\mathrm{Vec})\cong\Rep(H^*),\nonumber\\
	\calD&=&\Fun_{\calC_0}(\calC_0,\calC_0)\cong\calC_0=\Rep(H),\nonumber\\
	\calM&=&\Fun_{\calC_0}(\calM_0,\calC_0)\cong\calM_0^\op\cong\Vect,\nonumber\\
	\calM^{\op}&=&\Fun_{\calC_0}(\calC_0,\calM_0)\cong\calM_0=\Vect,
\end{eqnarray}
where $H^*$ is the dual Hopf algebra of $H$. Again, $\calM$ has only one simple object $\sigma$. Therefore, the upper right and lower left corners in this configuration are black defects, described by $\calM$ and $\calM^\op$. 

For any simple type of quasiparticle~(which can be a non-Abelian anyon in general) in the bulk $\psi\in \Rep[D(H)]$, Eq.~\eqref{def:RfromSFC} defines the $R$-matrix for $\psi$ with respect to the point-like defect $\sigma$. One can certainly use Eq.~\eqref{eq:RmatfromFRmove} to compute $R$, but there is actually a neat way to compute $R$ directly from Hopf algebra data~\cite{Majid1990}:
\begin{equation}\label{eq:RfromUniversalR}
	R=X[(\psi\otimes\psi)\mathcal{R}],
\end{equation}
where $\mathcal{R}\in D(H)\otimes D(H)$ is the universal $\mathcal{R}$-matrix of $D(H)$. 

The above analysis is valid for any finite dimensional $\C^*$-Hopf algebra $H$. If $H$ is a minimal triangular Hopf algebra~\cite{Radford1993MQHA} obtained by twisting a central type factor group $G$ using the method in Ref.~\onlinecite{etingof1998THAconstruction}, then the quantum double model $D(H)$ with the aforementioned rough-smooth boundary condition $\calM_0=\mathrm{Vec},~\calN_0=\Rep(H)$ is actually equivalent to the quantum double model $D(G)$ with the hybrid boundary condition in  Eq.~\eqref{eq:hybridboundaryDG}. 
For example, in Ref.~\onlinecite{wang2024parastatistics} we demonstrated the winning strategy through a specific 2D exactly solvable lattice model that host emergent paraparticles.  %
The non-chiral phase of this model~(e.g. when the paraparticle tunneling constants are all zero) is a quantum double phase described by $\calC=\Rep[D(H_{64})]$, where $H_{64}$ is a 64-dimensional minimal triangular Hopf algebra that were obtained as a Drinfeld twist of $G_{64}$. 
According to our discussion above, this Hopf algebra quantum double model provides an alternative way to realize the $R$-matrix in Eq.~\eqref{eqApp:seth-R}. 
\subsection{Black defects on domain walls}
We can directly generalize the above to construct black defects on domain walls in string-net models, which are classified by bimodule categories~\cite{Kitaev2012gappedboundary}.
Instead of presenting a similar derivation, here we mention a more general result claiming that any indecomposable module category $\calM$ over a modular tensor category $\calC$~(describing some 2D topological phase) can always be realized at an end point of some gapped domain wall between $\calC$ and itself, as shown in Fig.~\ref{fig:LineDefect},
which follows from the condensation completion principle~\cite{carqueville2016orbifold,douglas2018fusion,gaiotto2019condensations,Kong2020Classification,Johnson-Freyd2022}. Therefore, all black defects in $\calC$ can be physically realized. 

\subsection{Example: a particle $\psi$ can have different $R$-matrices with respect to different defect $\sigma$}\label{app:ex-diffR}
Here we give an explicit example that the same type of particle $\psi$ can have different $R$-matrices for different choices of  $\sigma$.  Consider again the quantum double model $\calC=D(G_{64})$ based on the CTFG $G_{64}=D_8\ltimes Z_2^{\times 3}$~(this example actually works for an arbitrary CTFG), with the hybrid boundary condition shown in Fig.~\ref{fig:BoundaryDefect}. This time, instead of using Eq.~\eqref{eq:hybridboundaryDG}, we choose $\calM_0=\Vect$ as a module category over $\Rep(H)$ via the trivial forgetful functor $\mathrm{Forg}: \Rep(G_{64})\to \Vect$, and choose $\calN_0=\calC_0=\Rep(G_{64})$ as before. [This is equivalent to simply taking the special case $H=G_{64}$ in Eq.~\eqref{eq:HAQDMoritaContext}.] With this choice, the upper-left and lower right corners in Fig.~\ref{fig:BoundaryDefect} are still black defects. However, %
due to a different choice of $\calM_0$,  the black defect $\calM$ will have a different module category structure over $\calC$, and in this case Eq.~\eqref{def:RfromSFC} gives a trivial $R$-matrix $R=-X$ instead of Eq.~\eqref{eqApp:seth-R} for exactly the same type of paraparticle $\psi\in\Rep[D(G_{64})]$ with $m=4$. This means that the same type of paraparticle can have different $R$-matrices depending on the choice of $\sigma$. %

\bibliography{/home/lagrenge/Dropbox/zoteroLibrary/library1}

\begin{thebibliography}{173}%
\makeatletter
\providecommand \@ifxundefined [1]{%
 \@ifx{#1\undefined}
}%
\providecommand \@ifnum [1]{%
 \ifnum #1\expandafter \@firstoftwo
 \else \expandafter \@secondoftwo
 \fi
}%
\providecommand \@ifx [1]{%
 \ifx #1\expandafter \@firstoftwo
 \else \expandafter \@secondoftwo
 \fi
}%
\providecommand \natexlab [1]{#1}%
\providecommand \enquote  [1]{``#1''}%
\providecommand \bibnamefont  [1]{#1}%
\providecommand \bibfnamefont [1]{#1}%
\providecommand \citenamefont [1]{#1}%
\providecommand \href@noop [0]{\@secondoftwo}%
\providecommand \href [0]{\begingroup \@sanitize@url \@href}%
\providecommand \@href[1]{\@@startlink{#1}\@@href}%
\providecommand \@@href[1]{\endgroup#1\@@endlink}%
\providecommand \@sanitize@url [0]{\catcode `\\12\catcode `\$12\catcode
  `\&12\catcode `\#12\catcode `\^12\catcode `\_12\catcode `\%12\relax}%
\providecommand \@@startlink[1]{}%
\providecommand \@@endlink[0]{}%
\providecommand \url  [0]{\begingroup\@sanitize@url \@url }%
\providecommand \@url [1]{\endgroup\@href {#1}{\urlprefix }}%
\providecommand \urlprefix  [0]{URL }%
\providecommand \Eprint [0]{\href }%
\providecommand \doibase [0]{http://dx.doi.org/}%
\providecommand \selectlanguage [0]{\@gobble}%
\providecommand \bibinfo  [0]{\@secondoftwo}%
\providecommand \bibfield  [0]{\@secondoftwo}%
\providecommand \translation [1]{[#1]}%
\providecommand \BibitemOpen [0]{}%
\providecommand \bibitemStop [0]{}%
\providecommand \bibitemNoStop [0]{.\EOS\space}%
\providecommand \EOS [0]{\spacefactor3000\relax}%
\providecommand \BibitemShut  [1]{\csname bibitem#1\endcsname}%
\let\auto@bib@innerbib\@empty
%</preamble>
\bibitem [{\citenamefont {Wang}(2024)}]{wang2024parastatistics}%
  \BibitemOpen
  \bibfield  {author} {\bibinfo {author} {\bibfnamefont {Zhiyuan}\ \bibnamefont
  {Wang}},\ }\bibfield  {title} {\enquote {\bibinfo {title} {Parastatistics and
  a secret communication challenge},}\ }\href@noop {} {\bibfield  {journal}
  {\bibinfo  {journal} {arXiv:2412.13360}\ } (\bibinfo {year} {2024})},\
  \Eprint {http://arxiv.org/abs/2412.13360} {arXiv:2412.13360} \BibitemShut
  {NoStop}%
\bibitem [{\citenamefont {Wang}\ and\ \citenamefont
  {Hazzard}(2025)}]{wang2023para}%
  \BibitemOpen
  \bibfield  {author} {\bibinfo {author} {\bibfnamefont {Zhiyuan}\ \bibnamefont
  {Wang}}\ and\ \bibinfo {author} {\bibfnamefont {Kaden R~A}\ \bibnamefont
  {Hazzard}},\ }\bibfield  {title} {\enquote {\bibinfo {title} {Particle
  exchange statistics beyond fermions and bosons},}\ }\href {\doibase
  10.1038/s41586-024-08262-7} {\bibfield  {journal} {\bibinfo  {journal}
  {Nature}\ }\textbf {\bibinfo {volume} {637}},\ \bibinfo {pages} {314--318}
  (\bibinfo {year} {2025})}\BibitemShut {NoStop}%
\bibitem [{\citenamefont {Doplicher}\ \emph {et~al.}(1971)\citenamefont
  {Doplicher}, \citenamefont {Haag},\ and\ \citenamefont
  {Roberts}}]{doplicher1971local}%
  \BibitemOpen
  \bibfield  {author} {\bibinfo {author} {\bibfnamefont {Sergio}\ \bibnamefont
  {Doplicher}}, \bibinfo {author} {\bibfnamefont {Rudolf}\ \bibnamefont
  {Haag}}, \ and\ \bibinfo {author} {\bibfnamefont {John~E}\ \bibnamefont
  {Roberts}},\ }\bibfield  {title} {\enquote {\bibinfo {title} {Local
  observables and particle statistics {{I}}},}\ }\href {\doibase
  10.1007/BF01877742} {\bibfield  {journal} {\bibinfo  {journal} {Commun. Math.
  Phys.}\ }\textbf {\bibinfo {volume} {23}},\ \bibinfo {pages} {199--230}
  (\bibinfo {year} {1971})}\BibitemShut {NoStop}%
\bibitem [{\citenamefont {Doplicher}\ \emph {et~al.}(1974)\citenamefont
  {Doplicher}, \citenamefont {Haag},\ and\ \citenamefont
  {Roberts}}]{doplicher1974local}%
  \BibitemOpen
  \bibfield  {author} {\bibinfo {author} {\bibfnamefont {Sergio}\ \bibnamefont
  {Doplicher}}, \bibinfo {author} {\bibfnamefont {Rudolf}\ \bibnamefont
  {Haag}}, \ and\ \bibinfo {author} {\bibfnamefont {John~E}\ \bibnamefont
  {Roberts}},\ }\bibfield  {title} {\enquote {\bibinfo {title} {Local
  observables and particle statistics {{II}}},}\ }\href@noop {} {\bibfield
  {journal} {\bibinfo  {journal} {Commun. Math. Phys.}\ }\textbf {\bibinfo
  {volume} {35}},\ \bibinfo {pages} {49--85} (\bibinfo {year}
  {1974})}\BibitemShut {NoStop}%
\bibitem [{\citenamefont {Lan}\ \emph {et~al.}(2018)\citenamefont {Lan},
  \citenamefont {Kong},\ and\ \citenamefont {Wen}}]{LanKongWen3DAB}%
  \BibitemOpen
  \bibfield  {author} {\bibinfo {author} {\bibfnamefont {Tian}\ \bibnamefont
  {Lan}}, \bibinfo {author} {\bibfnamefont {Liang}\ \bibnamefont {Kong}}, \
  and\ \bibinfo {author} {\bibfnamefont {Xiao-Gang}\ \bibnamefont {Wen}},\
  }\bibfield  {title} {\enquote {\bibinfo {title} {Classification of (3+1){{D}}
  bosonic topological orders: {{The}} case when pointlike excitations are all
  bosons},}\ }\href {\doibase 10.1103/PhysRevX.8.021074} {\bibfield  {journal}
  {\bibinfo  {journal} {Phys. Rev. X}\ }\textbf {\bibinfo {volume} {8}},\
  \bibinfo {pages} {21074} (\bibinfo {year} {2018})}\BibitemShut {NoStop}%
\bibitem [{\citenamefont {Lan}\ and\ \citenamefont {Wen}(2019)}]{LanWen3DEF}%
  \BibitemOpen
  \bibfield  {author} {\bibinfo {author} {\bibfnamefont {Tian}\ \bibnamefont
  {Lan}}\ and\ \bibinfo {author} {\bibfnamefont {Xiao-Gang}\ \bibnamefont
  {Wen}},\ }\bibfield  {title} {\enquote {\bibinfo {title} {Classification of
  3+{{1D}} bosonic topological orders ({{II}}): {{The}} case when some
  pointlike excitations are fermions},}\ }\href {\doibase
  10.1103/PhysRevX.9.021005} {\bibfield  {journal} {\bibinfo  {journal} {Phys.
  Rev. X}\ }\textbf {\bibinfo {volume} {9}},\ \bibinfo {pages} {21005}
  (\bibinfo {year} {2019})}\BibitemShut {NoStop}%
\bibitem [{\citenamefont {Howlett}\ and\ \citenamefont
  {Isaacs}(1982)}]{Howlett1982}%
  \BibitemOpen
  \bibfield  {author} {\bibinfo {author} {\bibfnamefont {Robert~B}\
  \bibnamefont {Howlett}}\ and\ \bibinfo {author} {\bibfnamefont {I~Martin}\
  \bibnamefont {Isaacs}},\ }\bibfield  {title} {\enquote {\bibinfo {title} {On
  groups of central type},}\ }\href {\doibase 10.1007/BF01215066} {\bibfield
  {journal} {\bibinfo  {journal} {Math Z.}\ }\textbf {\bibinfo {volume}
  {179}},\ \bibinfo {pages} {555--569} (\bibinfo {year} {1982})}\BibitemShut
  {NoStop}%
\bibitem [{\citenamefont {Leinaas}\ and\ \citenamefont
  {Myrheim}(1977)}]{Leinaas1977}%
  \BibitemOpen
  \bibfield  {author} {\bibinfo {author} {\bibfnamefont {J~M}\ \bibnamefont
  {Leinaas}}\ and\ \bibinfo {author} {\bibfnamefont {J}~\bibnamefont
  {Myrheim}},\ }\bibfield  {title} {\enquote {\bibinfo {title} {On the theory
  of identical particles},}\ }\href {\doibase 10.1007/BF02727953} {\bibfield
  {journal} {\bibinfo  {journal} {Nuovo Cim. B}\ }\textbf {\bibinfo {volume}
  {37}},\ \bibinfo {pages} {1--23} (\bibinfo {year} {1977})}\BibitemShut
  {NoStop}%
\bibitem [{\citenamefont {Wilczek}(1982{\natexlab{a}})}]{Wilczek1982Magnetic}%
  \BibitemOpen
  \bibfield  {author} {\bibinfo {author} {\bibfnamefont {Frank}\ \bibnamefont
  {Wilczek}},\ }\bibfield  {title} {\enquote {\bibinfo {title} {Magnetic flux,
  angular momentum, and statistics},}\ }\href {\doibase
  10.1103/PhysRevLett.48.1144} {\bibfield  {journal} {\bibinfo  {journal}
  {Phys. Rev. Lett.}\ }\textbf {\bibinfo {volume} {48}},\ \bibinfo {pages}
  {1144--1146} (\bibinfo {year} {1982}{\natexlab{a}})}\BibitemShut {NoStop}%
\bibitem [{\citenamefont {Wilczek}(1982{\natexlab{b}})}]{Wilczek1982Quantum}%
  \BibitemOpen
  \bibfield  {author} {\bibinfo {author} {\bibfnamefont {Frank}\ \bibnamefont
  {Wilczek}},\ }\bibfield  {title} {\enquote {\bibinfo {title} {Quantum
  mechanics of fractional-spin particles},}\ }\href {\doibase
  10.1103/PhysRevLett.49.957} {\bibfield  {journal} {\bibinfo  {journal} {Phys.
  Rev. Lett.}\ }\textbf {\bibinfo {volume} {49}},\ \bibinfo {pages} {957--959}
  (\bibinfo {year} {1982}{\natexlab{b}})}\BibitemShut {NoStop}%
\bibitem [{\citenamefont {Wilczek}(1990)}]{Wilczek1990book}%
  \BibitemOpen
  \bibfield  {author} {\bibinfo {author} {\bibfnamefont {Frank}\ \bibnamefont
  {Wilczek}},\ }\href {\doibase 10.1142/0961} {\emph {\bibinfo {title}
  {Fractional Statistics and Anyon Superconductivity}}}\ (\bibinfo  {publisher}
  {WORLD SCIENTIFIC},\ \bibinfo {address} {Singapore},\ \bibinfo {year}
  {1990})\BibitemShut {NoStop}%
\bibitem [{\citenamefont {Nayak}\ \emph {et~al.}(2008)\citenamefont {Nayak},
  \citenamefont {Simon}, \citenamefont {Stern}, \citenamefont {Freedman},\ and\
  \citenamefont {Das~Sarma}}]{Nayak2008NAAnyons}%
  \BibitemOpen
  \bibfield  {author} {\bibinfo {author} {\bibfnamefont {Chetan}\ \bibnamefont
  {Nayak}}, \bibinfo {author} {\bibfnamefont {Steven~H}\ \bibnamefont {Simon}},
  \bibinfo {author} {\bibfnamefont {Ady}\ \bibnamefont {Stern}}, \bibinfo
  {author} {\bibfnamefont {Michael}\ \bibnamefont {Freedman}}, \ and\ \bibinfo
  {author} {\bibfnamefont {Sankar}\ \bibnamefont {Das~Sarma}},\ }\bibfield
  {title} {\enquote {\bibinfo {title} {Non-{{Abelian}} anyons and topological
  quantum computation},}\ }\href {\doibase 10.1103/RevModPhys.80.1083}
  {\bibfield  {journal} {\bibinfo  {journal} {Rev. Mod. Phys.}\ }\textbf
  {\bibinfo {volume} {80}},\ \bibinfo {pages} {1083--1159} (\bibinfo {year}
  {2008})}\BibitemShut {NoStop}%
\bibitem [{\citenamefont {Stern}(2008)}]{STERN2008204}%
  \BibitemOpen
  \bibfield  {author} {\bibinfo {author} {\bibfnamefont {Ady}\ \bibnamefont
  {Stern}},\ }\bibfield  {title} {\enquote {\bibinfo {title} {Anyons and the
  quantum {{Hall}} effect--{{A}} pedagogical review},}\ }\href {\doibase
  10.1016/j.aop.2007.10.008} {\bibfield  {journal} {\bibinfo  {journal} {Ann.
  Phys.}\ }\textbf {\bibinfo {volume} {323}},\ \bibinfo {pages} {204--249}
  (\bibinfo {year} {2008})}\BibitemShut {NoStop}%
\bibitem [{\citenamefont {Wen}(1990)}]{TPorder1}%
  \BibitemOpen
  \bibfield  {author} {\bibinfo {author} {\bibfnamefont {Xiao-Gang}\
  \bibnamefont {Wen}},\ }\bibfield  {title} {\enquote {\bibinfo {title}
  {Topological orders in rigid states},}\ }\href@noop {} {\bibfield  {journal}
  {\bibinfo  {journal} {Int. J. Mod. Phys. B}\ }\textbf {\bibinfo {volume}
  {4}},\ \bibinfo {pages} {239--271} (\bibinfo {year} {1990})}\BibitemShut
  {NoStop}%
\bibitem [{\citenamefont {Wen}(1995)}]{wenTopologicalOrdersEdge1995}%
  \BibitemOpen
  \bibfield  {author} {\bibinfo {author} {\bibfnamefont {Xiao-Gang}\
  \bibnamefont {Wen}},\ }\bibfield  {title} {\enquote {\bibinfo {title}
  {Topological orders and edge excitations in fractional quantum {{Hall}}
  states},}\ }\href {\doibase 10.1080/00018739500101566} {\bibfield  {journal}
  {\bibinfo  {journal} {Advances in Physics}\ }\textbf {\bibinfo {volume}
  {44}},\ \bibinfo {pages} {405--473} (\bibinfo {year} {1995})}\BibitemShut
  {NoStop}%
\bibitem [{\citenamefont {Kitaev}(2003)}]{kitaev2003fault}%
  \BibitemOpen
  \bibfield  {author} {\bibinfo {author} {\bibfnamefont {{\relax
  A.Yu}.}~\bibnamefont {Kitaev}},\ }\bibfield  {title} {\enquote {\bibinfo
  {title} {Fault-tolerant quantum computation by anyons},}\ }\href {\doibase
  10.1016/S0003-4916(02)00018-0} {\bibfield  {journal} {\bibinfo  {journal}
  {Ann. Phys.}\ }\textbf {\bibinfo {volume} {303}},\ \bibinfo {pages} {2--30}
  (\bibinfo {year} {2003})}\BibitemShut {NoStop}%
\bibitem [{\citenamefont {Levin}\ and\ \citenamefont
  {Wen}(2005)}]{levin2005string}%
  \BibitemOpen
  \bibfield  {author} {\bibinfo {author} {\bibfnamefont {Michael~A}\
  \bibnamefont {Levin}}\ and\ \bibinfo {author} {\bibfnamefont {Xiao-Gang}\
  \bibnamefont {Wen}},\ }\bibfield  {title} {\enquote {\bibinfo {title}
  {String-net condensation: {{A}} physical mechanism for topological phases},}\
  }\href@noop {} {\bibfield  {journal} {\bibinfo  {journal} {Phys. Rev. B}\
  }\textbf {\bibinfo {volume} {71}},\ \bibinfo {pages} {45110} (\bibinfo {year}
  {2005})}\BibitemShut {NoStop}%
\bibitem [{\citenamefont {Kitaev}(2006)}]{kitaev2006anyons}%
  \BibitemOpen
  \bibfield  {author} {\bibinfo {author} {\bibfnamefont {Alexei}\ \bibnamefont
  {Kitaev}},\ }\bibfield  {title} {\enquote {\bibinfo {title} {Anyons in an
  exactly solved model and beyond},}\ }\href {\doibase
  10.1016/j.aop.2005.10.005} {\bibfield  {journal} {\bibinfo  {journal} {Ann.
  Phys.}\ }\bibinfo {series} {January {{Special Issue}}},\ \textbf {\bibinfo
  {volume} {321}},\ \bibinfo {pages} {2--111} (\bibinfo {year}
  {2006})}\BibitemShut {NoStop}%
\bibitem [{\citenamefont {Chen}\ \emph {et~al.}(2010)\citenamefont {Chen},
  \citenamefont {Gu},\ and\ \citenamefont {Wen}}]{Chen2010LUT}%
  \BibitemOpen
  \bibfield  {author} {\bibinfo {author} {\bibfnamefont {Xie}\ \bibnamefont
  {Chen}}, \bibinfo {author} {\bibfnamefont {Zheng-Cheng}\ \bibnamefont {Gu}},
  \ and\ \bibinfo {author} {\bibfnamefont {Xiao-Gang}\ \bibnamefont {Wen}},\
  }\bibfield  {title} {\enquote {\bibinfo {title} {Local unitary
  transformation, long-range quantum entanglement, wave function
  renormalization, and topological order},}\ }\href@noop {} {\bibfield
  {journal} {\bibinfo  {journal} {Phys. Rev. B}\ }\textbf {\bibinfo {volume}
  {82}},\ \bibinfo {pages} {155138} (\bibinfo {year} {2010})}\BibitemShut
  {NoStop}%
\bibitem [{\citenamefont {Wen}(2017)}]{Wen2017Zoo}%
  \BibitemOpen
  \bibfield  {author} {\bibinfo {author} {\bibfnamefont {Xiao~Gang}\
  \bibnamefont {Wen}},\ }\bibfield  {title} {\enquote {\bibinfo {title}
  {Colloquium: {{Zoo}} of quantum-topological phases of matter},}\ }\href
  {\doibase 10.1103/RevModPhys.89.041004} {\bibfield  {journal} {\bibinfo
  {journal} {Rev. Mod. Phys.}\ }\textbf {\bibinfo {volume} {89}} (\bibinfo
  {year} {2017}),\ 10.1103/RevModPhys.89.041004}\BibitemShut {NoStop}%
\bibitem [{\citenamefont {Barkeshli}\ \emph {et~al.}(2019)\citenamefont
  {Barkeshli}, \citenamefont {Bonderson}, \citenamefont {Cheng},\ and\
  \citenamefont {Wang}}]{Barkeshli2019SETclassification}%
  \BibitemOpen
  \bibfield  {author} {\bibinfo {author} {\bibfnamefont {Maissam}\ \bibnamefont
  {Barkeshli}}, \bibinfo {author} {\bibfnamefont {Parsa}\ \bibnamefont
  {Bonderson}}, \bibinfo {author} {\bibfnamefont {Meng}\ \bibnamefont {Cheng}},
  \ and\ \bibinfo {author} {\bibfnamefont {Zhenghan}\ \bibnamefont {Wang}},\
  }\bibfield  {title} {\enquote {\bibinfo {title} {Symmetry fractionalization,
  defects, and gauging of topological phases},}\ }\href {\doibase
  10.1103/PhysRevB.100.115147} {\bibfield  {journal} {\bibinfo  {journal}
  {Phys. Rev. B}\ }\textbf {\bibinfo {volume} {100}},\ \bibinfo {pages}
  {115147} (\bibinfo {year} {2019})}\BibitemShut {NoStop}%
\bibitem [{\citenamefont {Freedman}\ \emph {et~al.}(2003)\citenamefont
  {Freedman}, \citenamefont {Kitaev}, \citenamefont {Larsen},\ and\
  \citenamefont {Wang}}]{freedman2003topological}%
  \BibitemOpen
  \bibfield  {author} {\bibinfo {author} {\bibfnamefont {Michael}\ \bibnamefont
  {Freedman}}, \bibinfo {author} {\bibfnamefont {Alexei}\ \bibnamefont
  {Kitaev}}, \bibinfo {author} {\bibfnamefont {Michael}\ \bibnamefont
  {Larsen}}, \ and\ \bibinfo {author} {\bibfnamefont {Zhenghan}\ \bibnamefont
  {Wang}},\ }\bibfield  {title} {\enquote {\bibinfo {title} {Topological
  quantum computation},}\ }\href@noop {} {\bibfield  {journal} {\bibinfo
  {journal} {Bull. Am. Math. Soc.}\ }\textbf {\bibinfo {volume} {40}},\
  \bibinfo {pages} {31--38} (\bibinfo {year} {2003})}\BibitemShut {NoStop}%
\bibitem [{\citenamefont {Dennis}\ \emph {et~al.}(2002)\citenamefont {Dennis},
  \citenamefont {Kitaev}, \citenamefont {Landahl},\ and\ \citenamefont
  {Preskill}}]{Dennis2002TQM}%
  \BibitemOpen
  \bibfield  {author} {\bibinfo {author} {\bibfnamefont {Eric}\ \bibnamefont
  {Dennis}}, \bibinfo {author} {\bibfnamefont {Alexei}\ \bibnamefont {Kitaev}},
  \bibinfo {author} {\bibfnamefont {Andrew}\ \bibnamefont {Landahl}}, \ and\
  \bibinfo {author} {\bibfnamefont {John}\ \bibnamefont {Preskill}},\
  }\bibfield  {title} {\enquote {\bibinfo {title} {Topological quantum
  memory},}\ }\href {\doibase 10.1063/1.1499754} {\bibfield  {journal}
  {\bibinfo  {journal} {J. Math. Phys.}\ }\textbf {\bibinfo {volume} {43}},\
  \bibinfo {pages} {4452--4505} (\bibinfo {year} {2002})}\BibitemShut {NoStop}%
\bibitem [{\citenamefont {Wang}(2010)}]{wang2010topological}%
  \BibitemOpen
  \bibfield  {author} {\bibinfo {author} {\bibfnamefont {Zhenghan}\
  \bibnamefont {Wang}},\ }\href@noop {} {\emph {\bibinfo {title} {Topological
  Quantum Computation}}},\ \bibinfo {number} {112}\ (\bibinfo  {publisher}
  {American Mathematical Soc.},\ \bibinfo {year} {2010})\BibitemShut {NoStop}%
\bibitem [{\citenamefont {Stern}\ and\ \citenamefont
  {Lindner}(2013)}]{sternTopologicalQuantumComputation2013}%
  \BibitemOpen
  \bibfield  {author} {\bibinfo {author} {\bibfnamefont {Ady}\ \bibnamefont
  {Stern}}\ and\ \bibinfo {author} {\bibfnamefont {Netanel~H.}\ \bibnamefont
  {Lindner}},\ }\bibfield  {title} {\enquote {\bibinfo {title} {Topological
  {{Quantum Computation}}---{{From Basic Concepts}} to {{First
  Experiments}}},}\ }\href {\doibase 10.1126/science.1231473} {\bibfield
  {journal} {\bibinfo  {journal} {Science}\ }\textbf {\bibinfo {volume}
  {339}},\ \bibinfo {pages} {1179--1184} (\bibinfo {year} {2013})}\BibitemShut
  {NoStop}%
\bibitem [{\citenamefont {Lahtinen}\ and\ \citenamefont
  {Pachos}(2017)}]{lahtinenShortIntroductionTopological2017}%
  \BibitemOpen
  \bibfield  {author} {\bibinfo {author} {\bibfnamefont {Ville}\ \bibnamefont
  {Lahtinen}}\ and\ \bibinfo {author} {\bibfnamefont {Jiannis}\ \bibnamefont
  {Pachos}},\ }\bibfield  {title} {\enquote {\bibinfo {title} {A {{Short
  Introduction}} to {{Topological Quantum Computation}}},}\ }\href {\doibase
  10.21468/SciPostPhys.3.3.021} {\bibfield  {journal} {\bibinfo  {journal}
  {SciPost Phys.}\ }\textbf {\bibinfo {volume} {3}},\ \bibinfo {pages} {021}
  (\bibinfo {year} {2017})}\BibitemShut {NoStop}%
\bibitem [{\citenamefont {Green}(1953)}]{Green1952}%
  \BibitemOpen
  \bibfield  {author} {\bibinfo {author} {\bibfnamefont {H~S}\ \bibnamefont
  {Green}},\ }\bibfield  {title} {\enquote {\bibinfo {title} {A generalized
  method of field quantization},}\ }\href {\doibase 10.1103/PhysRev.90.270}
  {\bibfield  {journal} {\bibinfo  {journal} {Phys. Rev.}\ }\textbf {\bibinfo
  {volume} {90}},\ \bibinfo {pages} {270--273} (\bibinfo {year}
  {1953})}\BibitemShut {NoStop}%
\bibitem [{\citenamefont {Araki}(1961)}]{Araki1961}%
  \BibitemOpen
  \bibfield  {author} {\bibinfo {author} {\bibfnamefont {Huzihiro}\
  \bibnamefont {Araki}},\ }\bibfield  {title} {\enquote {\bibinfo {title} {On
  the connection of spin and commutation relations between different fields},}\
  }\href {\doibase 10.1063/1.1703710} {\bibfield  {journal} {\bibinfo
  {journal} {J. Math. Phys.}\ }\textbf {\bibinfo {volume} {2}},\ \bibinfo
  {pages} {267--270} (\bibinfo {year} {1961})}\BibitemShut {NoStop}%
\bibitem [{\citenamefont
  {Greenberg}(1964)}]{greenbergSpinUnitarySpinIndependence1964}%
  \BibitemOpen
  \bibfield  {author} {\bibinfo {author} {\bibfnamefont {O.~W.}\ \bibnamefont
  {Greenberg}},\ }\bibfield  {title} {\enquote {\bibinfo {title} {Spin and
  {{Unitary-Spin Independence}} in a {{Paraquark Model}} of {{Baryons}} and
  {{Mesons}}},}\ }\href {\doibase 10.1103/PhysRevLett.13.598} {\bibfield
  {journal} {\bibinfo  {journal} {Phys. Rev. Lett.}\ }\textbf {\bibinfo
  {volume} {13}},\ \bibinfo {pages} {598--602} (\bibinfo {year}
  {1964})}\BibitemShut {NoStop}%
\bibitem [{\citenamefont {Greenberg}\ and\ \citenamefont
  {Messiah}(1965)}]{Greenberg1965}%
  \BibitemOpen
  \bibfield  {author} {\bibinfo {author} {\bibfnamefont {O.~W.}\ \bibnamefont
  {Greenberg}}\ and\ \bibinfo {author} {\bibfnamefont {A.~M.~L.}\ \bibnamefont
  {Messiah}},\ }\bibfield  {title} {\enquote {\bibinfo {title} {Selection rules
  for parafields and the absence of para particles in nature},}\ }\href
  {\doibase 10.1103/PhysRev.138.B1155} {\bibfield  {journal} {\bibinfo
  {journal} {Phys. Rev.}\ }\textbf {\bibinfo {volume} {138}},\ \bibinfo {pages}
  {B1155--B1167} (\bibinfo {year} {1965})}\BibitemShut {NoStop}%
\bibitem [{\citenamefont {Landshoff}\ and\ \citenamefont
  {Stapp}(1967)}]{LANDSHOFF196772}%
  \BibitemOpen
  \bibfield  {author} {\bibinfo {author} {\bibfnamefont {P~V}\ \bibnamefont
  {Landshoff}}\ and\ \bibinfo {author} {\bibfnamefont {Henry~P}\ \bibnamefont
  {Stapp}},\ }\bibfield  {title} {\enquote {\bibinfo {title} {Parastatistics
  and a unified theory of identical particles},}\ }\href {\doibase
  10.1016/0003-4916(67)90317-X} {\bibfield  {journal} {\bibinfo  {journal}
  {Ann. Phys.}\ }\textbf {\bibinfo {volume} {45}},\ \bibinfo {pages} {72--92}
  (\bibinfo {year} {1967})}\BibitemShut {NoStop}%
\bibitem [{\citenamefont {Dr{\"u}hl}\ \emph {et~al.}(1970)\citenamefont
  {Dr{\"u}hl}, \citenamefont {Haag},\ and\ \citenamefont
  {Roberts}}]{druhl1970parastatistics}%
  \BibitemOpen
  \bibfield  {author} {\bibinfo {author} {\bibfnamefont {K}~\bibnamefont
  {Dr{\"u}hl}}, \bibinfo {author} {\bibfnamefont {R}~\bibnamefont {Haag}}, \
  and\ \bibinfo {author} {\bibfnamefont {J~E}\ \bibnamefont {Roberts}},\
  }\bibfield  {title} {\enquote {\bibinfo {title} {On parastatistics},}\ }\href
  {\doibase 10.1007/BF01649433} {\bibfield  {journal} {\bibinfo  {journal}
  {Commun. Math. Phys.}\ }\textbf {\bibinfo {volume} {18}},\ \bibinfo {pages}
  {204--226} (\bibinfo {year} {1970})}\BibitemShut {NoStop}%
\bibitem [{\citenamefont {Stolt}\ and\ \citenamefont
  {Taylor}(1970)}]{Taylor1970b}%
  \BibitemOpen
  \bibfield  {author} {\bibinfo {author} {\bibfnamefont {Robert~H.}\
  \bibnamefont {Stolt}}\ and\ \bibinfo {author} {\bibfnamefont {John~R.}\
  \bibnamefont {Taylor}},\ }\bibfield  {title} {\enquote {\bibinfo {title}
  {Correspondence between the first- and second-quantized theories of
  paraparticles},}\ }\href {\doibase 10.1016/0550-3213(70)90024-6} {\bibfield
  {journal} {\bibinfo  {journal} {Nucl. Phys. B}\ }\textbf {\bibinfo {volume}
  {19}},\ \bibinfo {pages} {1--19} (\bibinfo {year} {1970})}\BibitemShut
  {NoStop}%
\bibitem [{\citenamefont {Doplicher}\ and\ \citenamefont
  {Roberts}(1972)}]{doplicherFieldsStatisticsNonabelian1972}%
  \BibitemOpen
  \bibfield  {author} {\bibinfo {author} {\bibfnamefont {Sergio}\ \bibnamefont
  {Doplicher}}\ and\ \bibinfo {author} {\bibfnamefont {John~E.}\ \bibnamefont
  {Roberts}},\ }\bibfield  {title} {\enquote {\bibinfo {title} {Fields,
  statistics and non-abelian gauge groups},}\ }\href {\doibase cmp/1103858447}
  {\bibfield  {journal} {\bibinfo  {journal} {Commun. Math. Phys.}\ }\textbf
  {\bibinfo {volume} {28}},\ \bibinfo {pages} {331--348} (\bibinfo {year}
  {1972})}\BibitemShut {NoStop}%
\bibitem [{\citenamefont {Doplicher}\ and\ \citenamefont
  {Roberts}(1990)}]{doplicher1990}%
  \BibitemOpen
  \bibfield  {author} {\bibinfo {author} {\bibfnamefont {Sergio}\ \bibnamefont
  {Doplicher}}\ and\ \bibinfo {author} {\bibfnamefont {John~E}\ \bibnamefont
  {Roberts}},\ }\bibfield  {title} {\enquote {\bibinfo {title} {Why there is a
  field algebra with a compact gauge group describing the superselection
  structure in particle physics},}\ }\href {\doibase cmp/1104200703} {\bibfield
   {journal} {\bibinfo  {journal} {Commun. Math. Phys.}\ }\textbf {\bibinfo
  {volume} {131}},\ \bibinfo {pages} {51--107} (\bibinfo {year}
  {1990})}\BibitemShut {NoStop}%
\bibitem [{\citenamefont {Tolstoy}(2014)}]{tolstoyOnceMoreParastatistics2014}%
  \BibitemOpen
  \bibfield  {author} {\bibinfo {author} {\bibfnamefont {V.~N.}\ \bibnamefont
  {Tolstoy}},\ }\bibfield  {title} {\enquote {\bibinfo {title} {Once more on
  parastatistics},}\ }\href {\doibase 10.1134/S1547477114070449} {\bibfield
  {journal} {\bibinfo  {journal} {Physics of Particles and Nuclei Letters}\
  }\textbf {\bibinfo {volume} {11}},\ \bibinfo {pages} {933--937} (\bibinfo
  {year} {2014})}\BibitemShut {NoStop}%
\bibitem [{\citenamefont {Stoilova}\ and\ \citenamefont {{Van der
  Jeugt}}(2020)}]{Stoilova2020}%
  \BibitemOpen
  \bibfield  {author} {\bibinfo {author} {\bibfnamefont {N~I}\ \bibnamefont
  {Stoilova}}\ and\ \bibinfo {author} {\bibfnamefont {J}~\bibnamefont {{Van der
  Jeugt}}},\ }\bibfield  {title} {\enquote {\bibinfo {title} {Partition
  functions and thermodynamic properties of paraboson and parafermion
  systems},}\ }\href {\doibase 10.1016/j.physleta.2020.126421} {\bibfield
  {journal} {\bibinfo  {journal} {Phys. Lett. A}\ }\textbf {\bibinfo {volume}
  {384}},\ \bibinfo {pages} {126421} (\bibinfo {year} {2020})}\BibitemShut
  {NoStop}%
\bibitem [{\citenamefont {Baker}\ \emph {et~al.}(2015)\citenamefont {Baker},
  \citenamefont {Halvorson},\ and\ \citenamefont
  {Swanson}}]{Baker2015Conventionality}%
  \BibitemOpen
  \bibfield  {author} {\bibinfo {author} {\bibfnamefont {David~John}\
  \bibnamefont {Baker}}, \bibinfo {author} {\bibfnamefont {Hans}\ \bibnamefont
  {Halvorson}}, \ and\ \bibinfo {author} {\bibfnamefont {Noel}\ \bibnamefont
  {Swanson}},\ }\bibfield  {title} {\enquote {\bibinfo {title} {The
  conventionality of parastatistics},}\ }\href {\doibase 10.1093/bjps/axu018}
  {\bibfield  {journal} {\bibinfo  {journal} {Br. J. Philos. Sci.}\ }\textbf
  {\bibinfo {volume} {66}},\ \bibinfo {pages} {929--976} (\bibinfo {year}
  {2015})}\BibitemShut {NoStop}%
\bibitem [{\citenamefont {Freedman}\ \emph {et~al.}(2011)\citenamefont
  {Freedman}, \citenamefont {Hastings}, \citenamefont {Nayak}, \citenamefont
  {Qi}, \citenamefont {Walker},\ and\ \citenamefont
  {Wang}}]{freedmanProjectiveRibbonPermutation2011}%
  \BibitemOpen
  \bibfield  {author} {\bibinfo {author} {\bibfnamefont {Michael}\ \bibnamefont
  {Freedman}}, \bibinfo {author} {\bibfnamefont {Matthew~B.}\ \bibnamefont
  {Hastings}}, \bibinfo {author} {\bibfnamefont {Chetan}\ \bibnamefont
  {Nayak}}, \bibinfo {author} {\bibfnamefont {Xiao-Liang}\ \bibnamefont {Qi}},
  \bibinfo {author} {\bibfnamefont {Kevin}\ \bibnamefont {Walker}}, \ and\
  \bibinfo {author} {\bibfnamefont {Zhenghan}\ \bibnamefont {Wang}},\
  }\bibfield  {title} {\enquote {\bibinfo {title} {Projective ribbon
  permutation statistics: {{A}} remnant of non-{{Abelian}} braiding in higher
  dimensions},}\ }\href {\doibase 10.1103/PhysRevB.83.115132} {\bibfield
  {journal} {\bibinfo  {journal} {Phys. Rev. B}\ }\textbf {\bibinfo {volume}
  {83}},\ \bibinfo {pages} {115132} (\bibinfo {year} {2011})}\BibitemShut
  {NoStop}%
\bibitem [{\citenamefont {Simon}(2023)}]{simon2023topological}%
  \BibitemOpen
  \bibfield  {author} {\bibinfo {author} {\bibfnamefont {Steven~H}\
  \bibnamefont {Simon}},\ }\href@noop {} {\emph {\bibinfo {title} {Topological
  Quantum}}}\ (\bibinfo  {publisher} {Oxford University Press},\ \bibinfo
  {year} {2023})\BibitemShut {NoStop}%
\bibitem [{\citenamefont {Buchholz}\ and\ \citenamefont
  {Fredenhagen}(1982)}]{Buchholz1982}%
  \BibitemOpen
  \bibfield  {author} {\bibinfo {author} {\bibfnamefont {Detlev}\ \bibnamefont
  {Buchholz}}\ and\ \bibinfo {author} {\bibfnamefont {Klaus}\ \bibnamefont
  {Fredenhagen}},\ }\bibfield  {title} {\enquote {\bibinfo {title} {Locality
  and the structure of particle states},}\ }\href {\doibase 10.1007/BF01208370}
  {\bibfield  {journal} {\bibinfo  {journal} {Commun. Math. Phys.}\ }\textbf
  {\bibinfo {volume} {84}},\ \bibinfo {pages} {1--54} (\bibinfo {year}
  {1982})}\BibitemShut {NoStop}%
\bibitem [{\citenamefont {Halvorson}\ and\ \citenamefont
  {M{\"u}ger}(2006)}]{halvorson2006algebraic}%
  \BibitemOpen
  \bibfield  {author} {\bibinfo {author} {\bibfnamefont {Hans}\ \bibnamefont
  {Halvorson}}\ and\ \bibinfo {author} {\bibfnamefont {Michael}\ \bibnamefont
  {M{\"u}ger}},\ }\bibfield  {title} {\enquote {\bibinfo {title} {Algebraic
  quantum field theory},}\ }\href@noop {} {\bibfield  {journal} {\bibinfo
  {journal} {ArXiv Prepr Math-Ph0602036}\ } (\bibinfo {year} {2006})},\ \Eprint
  {http://arxiv.org/abs/math-ph/0602036} {arXiv:math-ph/0602036} \BibitemShut
  {NoStop}%
\bibitem [{\citenamefont {Haag}(1996)}]{haagLocalQuantumPhysics1996}%
  \BibitemOpen
  \bibfield  {author} {\bibinfo {author} {\bibfnamefont {Rudolf}\ \bibnamefont
  {Haag}},\ }\href {\doibase 10.1007/978-3-642-61458-3} {\emph {\bibinfo
  {title} {Local Quantum Physics: {{Fields}}, Particles, Algebras}}}\ (\bibinfo
   {publisher} {Springer-Verlag, Berlin, Heidelberg},\ \bibinfo {year}
  {1996})\BibitemShut {NoStop}%
\bibitem [{\citenamefont {Deligne}(2002)}]{Deligne2002CategoriesTensorielles}%
  \BibitemOpen
  \bibfield  {author} {\bibinfo {author} {\bibfnamefont {Pierre}\ \bibnamefont
  {Deligne}},\ }\bibfield  {title} {\enquote {\bibinfo {title} {Cat\'egories
  tensorielles},}\ }\href {\doibase 10.17323/1609-4514-2002-2-2-227-248}
  {\bibfield  {journal} {\bibinfo  {journal} {Mosc. Math. J.}\ }\textbf
  {\bibinfo {volume} {2}},\ \bibinfo {pages} {227--248} (\bibinfo {year}
  {2002})}\BibitemShut {NoStop}%
\bibitem [{\citenamefont {Etingof}\ \emph {et~al.}(2016)\citenamefont
  {Etingof}, \citenamefont {Gelaki}, \citenamefont {Nikshych},\ and\
  \citenamefont {Ostrik}}]{TenCat_EGNO}%
  \BibitemOpen
  \bibfield  {author} {\bibinfo {author} {\bibfnamefont {Pavel}\ \bibnamefont
  {Etingof}}, \bibinfo {author} {\bibfnamefont {Shlomo}\ \bibnamefont
  {Gelaki}}, \bibinfo {author} {\bibfnamefont {Dmitri}\ \bibnamefont
  {Nikshych}}, \ and\ \bibinfo {author} {\bibfnamefont {Victor}\ \bibnamefont
  {Ostrik}},\ }\href@noop {} {\emph {\bibinfo {title} {Tensor Categories}}},\
  Vol.\ \bibinfo {volume} {205}\ (\bibinfo  {publisher} {American Mathematical
  Soc.},\ \bibinfo {year} {2016})\BibitemShut {NoStop}%
\bibitem [{Note1()}]{Note1}%
  \BibitemOpen
  \bibinfo {note} {Since this paper is intended for physicists, we do not
  pursue the highest standard of mathematical rigor here. The main gap from
  being mathematically rigorous is discussed at the beginning of Sec.~\ref
  {sec:categorical_analysis} and Sec.~\ref
  {sec:categorical_framework}.}\BibitemShut {Stop}%
\bibitem [{Note2()}]{Note2}%
  \BibitemOpen
  \bibinfo {note} {In some game protocols to be presented later, we introduce
  additional players who play against Alice and Bob, such as the scramblers in
  the anti-anyon twists. In these cases we assume that these opposing players
  are promised not to leave any excitations behind. Alternatively, we can say
  that if these opposing players violate any rule, it is their team that lose
  the game, rather than the team of Alice and Bob}\BibitemShut {NoStop}%
\bibitem [{\citenamefont {Mekonnen}\ \emph {et~al.}(2025)\citenamefont
  {Mekonnen}, \citenamefont {Galley},\ and\ \citenamefont
  {Mueller}}]{mekonnen2025invariance}%
  \BibitemOpen
  \bibfield  {author} {\bibinfo {author} {\bibfnamefont {Manuel}\ \bibnamefont
  {Mekonnen}}, \bibinfo {author} {\bibfnamefont {Thomas~D}\ \bibnamefont
  {Galley}}, \ and\ \bibinfo {author} {\bibfnamefont {Markus~P}\ \bibnamefont
  {Mueller}},\ }\bibfield  {title} {\enquote {\bibinfo {title} {Invariance
  under quantum permutations rules out parastatistics},}\ }\href@noop {}
  {\bibfield  {journal} {\bibinfo  {journal} {arXiv Prepr. arXiv2502.17576}\ }
  (\bibinfo {year} {2025})}\BibitemShut {NoStop}%
\bibitem [{\citenamefont {Wang}\ and\ \citenamefont
  {Hazzard}()}]{wangOnRparaI}%
  \BibitemOpen
  \bibfield  {author} {\bibinfo {author} {\bibfnamefont {Zhiyuan}\ \bibnamefont
  {Wang}}\ and\ \bibinfo {author} {\bibfnamefont {Kaden}\ \bibnamefont
  {Hazzard}},\ }\bibfield  {title} {\enquote {\bibinfo {title} {On the
  {{R-matrix}} theory of parastatistics {{I}}: {{Foundation}}},}\ }\href@noop
  {} {\bibinfo  {journal} {unpublished}\ }\BibitemShut {NoStop}%
\bibitem [{Note3()}]{Note3}%
  \BibitemOpen
\bibfield  {journal} {  }\bibinfo {note} {Since these axioms aim to capture
  only the universal properties of paraparticles at long distance, they do not
  address what happens when two particles get close, which depends on the
  non-universal microscopic details of $\protect \hat {H}$.}\BibitemShut
  {Stop}%
\bibitem [{\citenamefont {Turaev}(1988)}]{Turaev1988}%
  \BibitemOpen
  \bibfield  {author} {\bibinfo {author} {\bibfnamefont {V.~G.}\ \bibnamefont
  {Turaev}},\ }\bibfield  {title} {\enquote {\bibinfo {title} {The
  {{Yang-Baxter}} equation and invariants of links},}\ }\href {\doibase
  10.1007/BF01393746} {\bibfield  {journal} {\bibinfo  {journal} {Invent.
  Math.}\ }\textbf {\bibinfo {volume} {92}},\ \bibinfo {pages} {527--553}
  (\bibinfo {year} {1988})}\BibitemShut {NoStop}%
\bibitem [{\citenamefont {Majid}(1990)}]{Majid1990}%
  \BibitemOpen
  \bibfield  {author} {\bibinfo {author} {\bibfnamefont {Shahn}\ \bibnamefont
  {Majid}},\ }\bibfield  {title} {\enquote {\bibinfo {title} {Quasitriangular
  {{Hopf Algebras}} and {{Yang-Baxter}} equations},}\ }\href {\doibase
  10.1142/S0217751X90000027} {\bibfield  {journal} {\bibinfo  {journal} {Int.
  J. Mod. Phys. A}\ }\textbf {\bibinfo {volume} {05}},\ \bibinfo {pages}
  {1--91} (\bibinfo {year} {1990})}\BibitemShut {NoStop}%
\bibitem [{\citenamefont {Etingof}\ \emph {et~al.}(1999)\citenamefont
  {Etingof}, \citenamefont {Schedler},\ and\ \citenamefont
  {Soloviev}}]{etingof1999set}%
  \BibitemOpen
  \bibfield  {author} {\bibinfo {author} {\bibfnamefont {Pavel}\ \bibnamefont
  {Etingof}}, \bibinfo {author} {\bibfnamefont {Travis}\ \bibnamefont
  {Schedler}}, \ and\ \bibinfo {author} {\bibfnamefont {Alexandre}\
  \bibnamefont {Soloviev}},\ }\bibfield  {title} {\enquote {\bibinfo {title}
  {Set-theoretical solutions to the quantum {{Yang-Baxter}} equation},}\ }\href
  {\doibase 10.1215/S0012-7094-99-10007-X} {\bibfield  {journal} {\bibinfo
  {journal} {Duke Math. J.}\ }\textbf {\bibinfo {volume} {100}},\ \bibinfo
  {pages} {169--209} (\bibinfo {year} {1999})}\BibitemShut {NoStop}%
\bibitem [{\citenamefont {Kassel}\ and\ \citenamefont
  {Turaev}(2008)}]{kassel2008braid}%
  \BibitemOpen
  \bibfield  {author} {\bibinfo {author} {\bibfnamefont {Christian}\
  \bibnamefont {Kassel}}\ and\ \bibinfo {author} {\bibfnamefont {Vladimir}\
  \bibnamefont {Turaev}},\ }\href {\doibase 10.1007/978-0-387-68548-9} {\emph
  {\bibinfo {title} {Braid Groups}}},\ \bibinfo {series} {Graduate Texts in
  Mathematics}, Vol.\ \bibinfo {volume} {247}\ (\bibinfo  {publisher} {Springer
  New York},\ \bibinfo {year} {2008})\BibitemShut {NoStop}%
\bibitem [{\citenamefont {Hastings}\ and\ \citenamefont
  {Koma}(2006)}]{hastings2006}%
  \BibitemOpen
  \bibfield  {author} {\bibinfo {author} {\bibfnamefont {Matthew~B}\
  \bibnamefont {Hastings}}\ and\ \bibinfo {author} {\bibfnamefont {Tohru}\
  \bibnamefont {Koma}},\ }\bibfield  {title} {\enquote {\bibinfo {title}
  {Spectral gap and exponential decay of correlations},}\ }\href {\doibase
  10.1007/s00220-006-0030-4} {\bibfield  {journal} {\bibinfo  {journal}
  {Commun. Math. Phys.}\ }\textbf {\bibinfo {volume} {265}},\ \bibinfo {pages}
  {781--804} (\bibinfo {year} {2006})}\BibitemShut {NoStop}%
\bibitem [{\citenamefont {Nachtergaele}\ and\ \citenamefont
  {Sims}(2006)}]{nachtergaele2006}%
  \BibitemOpen
  \bibfield  {author} {\bibinfo {author} {\bibfnamefont {Bruno}\ \bibnamefont
  {Nachtergaele}}\ and\ \bibinfo {author} {\bibfnamefont {Robert}\ \bibnamefont
  {Sims}},\ }\bibfield  {title} {\enquote {\bibinfo {title} {Lieb-{{Robinson}}
  bounds and the exponential clustering theorem},}\ }\href@noop {} {\bibfield
  {journal} {\bibinfo  {journal} {Commun. Math. Phys.}\ }\textbf {\bibinfo
  {volume} {265}},\ \bibinfo {pages} {119--130} (\bibinfo {year}
  {2006})}\BibitemShut {NoStop}%
\bibitem [{\citenamefont {Haegeman}\ \emph {et~al.}(2013)\citenamefont
  {Haegeman}, \citenamefont {Michalakis}, \citenamefont {Nachtergaele},
  \citenamefont {Osborne}, \citenamefont {Schuch},\ and\ \citenamefont
  {Verstraete}}]{haegeman2013elementary}%
  \BibitemOpen
  \bibfield  {author} {\bibinfo {author} {\bibfnamefont {Jutho}\ \bibnamefont
  {Haegeman}}, \bibinfo {author} {\bibfnamefont {Spyridon}\ \bibnamefont
  {Michalakis}}, \bibinfo {author} {\bibfnamefont {Bruno}\ \bibnamefont
  {Nachtergaele}}, \bibinfo {author} {\bibfnamefont {Tobias~J}\ \bibnamefont
  {Osborne}}, \bibinfo {author} {\bibfnamefont {Norbert}\ \bibnamefont
  {Schuch}}, \ and\ \bibinfo {author} {\bibfnamefont {Frank}\ \bibnamefont
  {Verstraete}},\ }\bibfield  {title} {\enquote {\bibinfo {title} {Elementary
  excitations in gapped quantum spin systems},}\ }\href@noop {} {\bibfield
  {journal} {\bibinfo  {journal} {Phys. Rev. Lett.}\ }\textbf {\bibinfo
  {volume} {111}},\ \bibinfo {pages} {80401} (\bibinfo {year}
  {2013})}\BibitemShut {NoStop}%
\bibitem [{Note4()}]{Note4}%
  \BibitemOpen
  \bibinfo {note} {The $R$-matrix here is a perfect tensor: grouping any two
  indices as inputs and the remaining two as outputs yields a unitary map—an
  invertible classical gate in our setting. Thus each player can recover the
  partner’s number by inverting this map.}\BibitemShut {Stop}%
\bibitem [{\citenamefont {Levin}\ and\ \citenamefont
  {Wen}(2003)}]{Levin2003Fermions}%
  \BibitemOpen
  \bibfield  {author} {\bibinfo {author} {\bibfnamefont {Michael}\ \bibnamefont
  {Levin}}\ and\ \bibinfo {author} {\bibfnamefont {Xiao-Gang}\ \bibnamefont
  {Wen}},\ }\bibfield  {title} {\enquote {\bibinfo {title} {Fermions, strings,
  and gauge fields in lattice spin models},}\ }\href {\doibase
  10.1103/PhysRevB.67.245316} {\bibfield  {journal} {\bibinfo  {journal} {Phys.
  Rev. B}\ }\textbf {\bibinfo {volume} {67}},\ \bibinfo {pages} {245316}
  (\bibinfo {year} {2003})}\BibitemShut {NoStop}%
\bibitem [{\citenamefont {Kitaev}\ and\ \citenamefont
  {Kong}(2012)}]{Kitaev2012gappedboundary}%
  \BibitemOpen
  \bibfield  {author} {\bibinfo {author} {\bibfnamefont {Alexei}\ \bibnamefont
  {Kitaev}}\ and\ \bibinfo {author} {\bibfnamefont {Liang}\ \bibnamefont
  {Kong}},\ }\bibfield  {title} {\enquote {\bibinfo {title} {Models for gapped
  boundaries and domain walls},}\ }\href {\doibase 10.1007/s00220-012-1500-5}
  {\bibfield  {journal} {\bibinfo  {journal} {Commun. Math. Phys.}\ }\textbf
  {\bibinfo {volume} {313}},\ \bibinfo {pages} {351--373} (\bibinfo {year}
  {2012})}\BibitemShut {NoStop}%
\bibitem [{\citenamefont {Kong}\ and\ \citenamefont
  {Wen}(2014)}]{kong2014braided}%
  \BibitemOpen
  \bibfield  {author} {\bibinfo {author} {\bibfnamefont {Liang}\ \bibnamefont
  {Kong}}\ and\ \bibinfo {author} {\bibfnamefont {Xiao-Gang}\ \bibnamefont
  {Wen}},\ }\bibfield  {title} {\enquote {\bibinfo {title} {Braided fusion
  categories, gravitational anomalies, and the mathematical framework for
  topological orders in any dimensions},}\ }\href@noop {} {\bibfield  {journal}
  {\bibinfo  {journal} {arXiv Prepr. arXiv1405.5858}\ } (\bibinfo {year}
  {2014})}\BibitemShut {NoStop}%
\bibitem [{\citenamefont {{Johnson-Freyd}}(2022)}]{Johnson-Freyd2022}%
  \BibitemOpen
  \bibfield  {author} {\bibinfo {author} {\bibfnamefont {Theo}\ \bibnamefont
  {{Johnson-Freyd}}},\ }\bibfield  {title} {\enquote {\bibinfo {title} {On the
  classification of topological orders},}\ }\href {\doibase
  10.1007/s00220-022-04380-3} {\bibfield  {journal} {\bibinfo  {journal}
  {Commun. Math. Phys.}\ }\textbf {\bibinfo {volume} {393}},\ \bibinfo {pages}
  {989--1033} (\bibinfo {year} {2022})}\BibitemShut {NoStop}%
\bibitem [{\citenamefont {Kong}\ and\ \citenamefont
  {Zhang}(2022)}]{kong2022invitation}%
  \BibitemOpen
  \bibfield  {author} {\bibinfo {author} {\bibfnamefont {Liang}\ \bibnamefont
  {Kong}}\ and\ \bibinfo {author} {\bibfnamefont {Zhi-Hao}\ \bibnamefont
  {Zhang}},\ }\bibfield  {title} {\enquote {\bibinfo {title} {An invitation to
  topological orders and category theory},}\ }\href@noop {} {\bibfield
  {journal} {\bibinfo  {journal} {ArXiv Prepr. ArXiv220505565}\ } (\bibinfo
  {year} {2022})},\ \Eprint {http://arxiv.org/abs/2205.05565}
  {arXiv:2205.05565} \BibitemShut {NoStop}%
\bibitem [{Note5()}]{Note5}%
  \BibitemOpen
  \bibinfo {note} {Note: for the 2D case, topological quasiparticles in a
  bosonic~(fermionic) topological phase are described by (super-)~modular
  tensor categories~\cite {lanModularExtensionsUnitary2017}, a special class of
  braided fusion categories satisfying the braiding non-degeneracy
  condition~[(super-)~modularity]. This (super-)~modularity condition is not
  important for our discussions in this paper, so we do not emphasize it
  here.}\BibitemShut {Stop}%
\bibitem [{Note6()}]{Note6}%
  \BibitemOpen
  \bibinfo {note} {For simplicity, here we only consider the subclass of finite
  depth LUTs. But in principle we can relax the frustration-free condition in
  the game requirement~(see App.~\ref {sec:relax_FF}) and generalize these
  arguments to the full class of LUTs.}\BibitemShut {Stop}%
\bibitem [{Note7()}]{Note7}%
  \BibitemOpen
  \bibinfo {note} {A subtlety here is that these transformed operations are
  generally supported on a larger area due to the evolution by $\protect \hat
  {U}$, so the players may need to choose a larger $r_0$ for the radius of the
  two circles~(or spheres in 3D) in order to accommodate these local
  operations. It is for this reason that in the game protocol we allow the
  players to choose the radius $r_0$.}\BibitemShut {Stop}%
\bibitem [{\citenamefont {Hastings}\ and\ \citenamefont
  {Wen}(2005)}]{hastings2005quasiadiabatic1}%
  \BibitemOpen
  \bibfield  {author} {\bibinfo {author} {\bibfnamefont {Matthew~B}\
  \bibnamefont {Hastings}}\ and\ \bibinfo {author} {\bibfnamefont {Xiao-Gang}\
  \bibnamefont {Wen}},\ }\bibfield  {title} {\enquote {\bibinfo {title}
  {Quasiadiabatic continuation of quantum states: {{The}} stability of
  topological ground-state degeneracy and emergent gauge invariance},}\ }\href
  {\doibase 10.1103/PhysRevB.72.045141} {\bibfield  {journal} {\bibinfo
  {journal} {Phys. Rev. B}\ }\textbf {\bibinfo {volume} {72}},\ \bibinfo
  {pages} {045141} (\bibinfo {year} {2005})}\BibitemShut {NoStop}%
\bibitem [{\citenamefont {Bachmann}\ \emph {et~al.}(2012)\citenamefont
  {Bachmann}, \citenamefont {Michalakis}, \citenamefont {Nachtergaele},\ and\
  \citenamefont {Sims}}]{LPPL}%
  \BibitemOpen
  \bibfield  {author} {\bibinfo {author} {\bibfnamefont {Sven}\ \bibnamefont
  {Bachmann}}, \bibinfo {author} {\bibfnamefont {Spyridon}\ \bibnamefont
  {Michalakis}}, \bibinfo {author} {\bibfnamefont {Bruno}\ \bibnamefont
  {Nachtergaele}}, \ and\ \bibinfo {author} {\bibfnamefont {Robert}\
  \bibnamefont {Sims}},\ }\bibfield  {title} {\enquote {\bibinfo {title}
  {Automorphic equivalence within gapped phases of quantum lattice systems},}\
  }\href {\doibase 10.1007/s00220-011-1380-0} {\bibfield  {journal} {\bibinfo
  {journal} {Commun. Math. Phys.}\ }\textbf {\bibinfo {volume} {309}},\
  \bibinfo {pages} {835--871} (\bibinfo {year} {2012})}\BibitemShut {NoStop}%
\bibitem [{\citenamefont {Naaijkens}(2017)}]{naaijkens2017quantum}%
  \BibitemOpen
  \bibfield  {author} {\bibinfo {author} {\bibfnamefont {Pieter}\ \bibnamefont
  {Naaijkens}},\ }\href {\doibase 10.1007/978-3-319-51458-1} {\emph {\bibinfo
  {title} {Quantum Spin Systems on Infinite Lattices. {{A}} Concise
  Introduction}}},\ Lecture Notes in Physics\ (\bibinfo  {publisher} {Springer
  Cham},\ \bibinfo {year} {2017})\BibitemShut {NoStop}%
\bibitem [{\citenamefont {Shi}\ \emph {et~al.}(2020)\citenamefont {Shi},
  \citenamefont {Kato},\ and\ \citenamefont {Kim}}]{Shi2020}%
  \BibitemOpen
  \bibfield  {author} {\bibinfo {author} {\bibfnamefont {Bowen}\ \bibnamefont
  {Shi}}, \bibinfo {author} {\bibfnamefont {Kohtaro}\ \bibnamefont {Kato}}, \
  and\ \bibinfo {author} {\bibfnamefont {Isaac~H}\ \bibnamefont {Kim}},\
  }\bibfield  {title} {\enquote {\bibinfo {title} {Fusion rules from
  entanglement},}\ }\href {\doibase 10.1016/j.aop.2020.168164} {\bibfield
  {journal} {\bibinfo  {journal} {Ann. Phys.}\ }\textbf {\bibinfo {volume}
  {418}},\ \bibinfo {pages} {168164} (\bibinfo {year} {2020})}\BibitemShut
  {NoStop}%
\bibitem [{\citenamefont {Shi}(2020)}]{Shi2020Verlinde}%
  \BibitemOpen
  \bibfield  {author} {\bibinfo {author} {\bibfnamefont {Bowen}\ \bibnamefont
  {Shi}},\ }\bibfield  {title} {\enquote {\bibinfo {title} {Verlinde formula
  from entanglement},}\ }\href {\doibase 10.1103/PhysRevResearch.2.023132}
  {\bibfield  {journal} {\bibinfo  {journal} {Phys. Rev. Res.}\ }\textbf
  {\bibinfo {volume} {2}},\ \bibinfo {pages} {023132} (\bibinfo {year}
  {2020})}\BibitemShut {NoStop}%
\bibitem [{\citenamefont {Shi}\ and\ \citenamefont
  {Kim}(2021)}]{Kim2021DomainWall}%
  \BibitemOpen
  \bibfield  {author} {\bibinfo {author} {\bibfnamefont {Bowen}\ \bibnamefont
  {Shi}}\ and\ \bibinfo {author} {\bibfnamefont {Isaac~H.}\ \bibnamefont
  {Kim}},\ }\bibfield  {title} {\enquote {\bibinfo {title} {Entanglement
  bootstrap approach for gapped domain walls},}\ }\href {\doibase
  10.1103/PhysRevB.103.115150} {\bibfield  {journal} {\bibinfo  {journal}
  {Phys. Rev. B}\ }\textbf {\bibinfo {volume} {103}},\ \bibinfo {pages}
  {115150} (\bibinfo {year} {2021})}\BibitemShut {NoStop}%
\bibitem [{\citenamefont {Yang}\ \emph {et~al.}(2025)\citenamefont {Yang},
  \citenamefont {Shi},\ and\ \citenamefont {Lee}}]{yang2025topological}%
  \BibitemOpen
  \bibfield  {author} {\bibinfo {author} {\bibfnamefont {Tai-Hsuan}\
  \bibnamefont {Yang}}, \bibinfo {author} {\bibfnamefont {Bowen}\ \bibnamefont
  {Shi}}, \ and\ \bibinfo {author} {\bibfnamefont {Jong~Yeon}\ \bibnamefont
  {Lee}},\ }\bibfield  {title} {\enquote {\bibinfo {title} {Topological mixed
  states: {{Axiomatic}} approaches and phases of matter},}\ }\href@noop {}
  {\bibfield  {journal} {\bibinfo  {journal} {ArXiv Prepr. ArXiv250604221}\ }
  (\bibinfo {year} {2025})},\ \Eprint {http://arxiv.org/abs/2506.04221}
  {arXiv:2506.04221} \BibitemShut {NoStop}%
\bibitem [{\citenamefont {Kitaev}(2024)}]{kitaev2024almost}%
  \BibitemOpen
  \bibfield  {author} {\bibinfo {author} {\bibfnamefont {Alexei}\ \bibnamefont
  {Kitaev}},\ }\bibfield  {title} {\enquote {\bibinfo {title}
  {Almost-idempotent quantum channels and approximate {{C}}*-algebras},}\
  }\href@noop {} {\bibfield  {journal} {\bibinfo  {journal} {arXiv Prepr.
  arXiv2405.02434}\ } (\bibinfo {year} {2024})}\BibitemShut {NoStop}%
\bibitem [{\citenamefont {Preskill}(1999)}]{preskill1999lecture}%
  \BibitemOpen
  \bibfield  {author} {\bibinfo {author} {\bibfnamefont {John}\ \bibnamefont
  {Preskill}},\ }\bibfield  {title} {\enquote {\bibinfo {title} {Lecture notes
  for physics 219: {{Quantum}} computation},}\ }\href@noop {} {\bibfield
  {journal} {\bibinfo  {journal} {Caltech Lect. Notes}\ }\textbf {\bibinfo
  {volume} {7}} (\bibinfo {year} {1999})}\BibitemShut {NoStop}%
\bibitem [{\citenamefont {Wen}(2007)}]{wenQuantumFieldTheory2007}%
  \BibitemOpen
  \bibfield  {author} {\bibinfo {author} {\bibfnamefont {Xiao-Gang}\
  \bibnamefont {Wen}},\ }\href {\doibase
  10.1093/acprof:oso/9780199227259.001.0001} {\emph {\bibinfo {title} {Quantum
  {{Field Theory}} of {{Many-Body Systems}}: {{From}} the {{Origin}} of
  {{Sound}} to an {{Origin}} of {{Light}} and {{Electrons}}}}}\ (\bibinfo
  {publisher} {Oxford University Press},\ \bibinfo {year} {2007})\BibitemShut
  {NoStop}%
\bibitem [{Note8()}]{Note8}%
  \BibitemOpen
  \bibinfo {note} {In our challenge game, it is clear that whatever excitations
  the players have in their respective circles must be movable by local unitary
  operations, since the players are always obliged to move them to follow the
  circle movements. However, the special topological excitations $\sigma $ and
  $\protect \bar {\sigma }$ sitting at $o$ and $s$ do not need to be mobile,
  since they remain at the same position throughout the game. In 3D, it is
  known that there are phases of matter hosting quasiparticles with restricted
  mobility, called fractons~\cite
  {haah2011local,vijay2016fracton1,nandkishore2019fractons}. We do not consider
  fracton phases in this work, but we expect that even in a 3D fracton phase,
  the category of all point-like topological quasiparticles with unrestricted
  mobility still forms an SFC $\protect \mathcal {C}$, and the fusion between
  fractons and mobile quasiparticles is still described by a $\protect \mathcal
  {C}$-module category $\protect \mathcal {M}$, so the more general module
  category analysis of the winning strategy presented in Sec.~\ref
  {sec:ModCatDefect} should still be valid.}\BibitemShut {Stop}%
\bibitem [{Note9()}]{Note9}%
  \BibitemOpen
  \bibinfo {note} {Although Ref.~\protect \rev@citealpnum {Shi2020} only
  defines information convex set in regions with no excitations~(or more
  precisely, for a reference state that satisfies their axioms A0 and A1
  exactly everywhere), it is straightforward to extend their definitions to our
  case.}\BibitemShut {Stop}%
\bibitem [{\citenamefont {Moore}\ and\ \citenamefont
  {Seiberg}(1989)}]{MooreSeiberg1989}%
  \BibitemOpen
  \bibfield  {author} {\bibinfo {author} {\bibfnamefont {Gregory}\ \bibnamefont
  {Moore}}\ and\ \bibinfo {author} {\bibfnamefont {Nathan}\ \bibnamefont
  {Seiberg}},\ }\bibfield  {title} {\enquote {\bibinfo {title} {Classical and
  quantum conformal field theory},}\ }\href {\doibase 10.1007/BF01238857}
  {\bibfield  {journal} {\bibinfo  {journal} {Communications in Mathematical
  Physics}\ }\textbf {\bibinfo {volume} {123}},\ \bibinfo {pages} {177--254}
  (\bibinfo {year} {1989})}\BibitemShut {NoStop}%
\bibitem [{Note10()}]{Note10}%
  \BibitemOpen
  \bibinfo {note} {Here $\mathinner {|{G}\rangle }$ is considered as an excited
  state of $\protect \hat {H}_0$ above the ground state $\mathinner
  {|{G_0}\rangle }$.}\BibitemShut {Stop}%
\bibitem [{Note11()}]{Note11}%
  \BibitemOpen
  \bibinfo {note} {It may be tempting to say that $a$ must be stored in the
  region $K_1$ given that $\{U_A(t)\}_{0\leq t\leq t_1}$ only acts on $K_1$.
  However, this is not true in general if the initial state has strong
  entanglement between $K_1$ and its complement. For example, one can encode a
  classical bit in the GHZ state using a local phase gate, but such information
  cannot be locally decoded from the reduced density matrix of any proper
  subsystem.}\BibitemShut {Stop}%
\bibitem [{\citenamefont {PETZ}(2003)}]{petz2003recovery}%
  \BibitemOpen
  \bibfield  {author} {\bibinfo {author} {\bibfnamefont {{\relax
  D{\'E}NES}}~\bibnamefont {PETZ}},\ }\bibfield  {title} {\enquote {\bibinfo
  {title} {{{MONOTONICITY OF QUANTUM RELATIVE ENTROPY REVISITED}}},}\ }\href
  {\doibase 10.1142/S0129055X03001576} {\bibfield  {journal} {\bibinfo
  {journal} {Rev. Math. Phys.}\ }\textbf {\bibinfo {volume} {15}},\ \bibinfo
  {pages} {79--91} (\bibinfo {year} {2003})}\BibitemShut {NoStop}%
\bibitem [{\citenamefont {Nielsen}\ and\ \citenamefont
  {Chuang}(2010)}]{nielsenQuantumComputationQuantum2010}%
  \BibitemOpen
  \bibfield  {author} {\bibinfo {author} {\bibfnamefont {Michael~A.}\
  \bibnamefont {Nielsen}}\ and\ \bibinfo {author} {\bibfnamefont {Isaac~L.}\
  \bibnamefont {Chuang}},\ }\href {\doibase 10.1017/CBO9780511976667} {\emph
  {\bibinfo {title} {Quantum {{Computation}} and {{Quantum Information}}: 10th
  {{Anniversary Edition}}}}}\ (\bibinfo  {publisher} {Cambridge University
  Press},\ \bibinfo {address} {Cambridge},\ \bibinfo {year} {2010})\BibitemShut
  {NoStop}%
\bibitem [{Note12()}]{Note12}%
  \BibitemOpen
  \bibinfo {note} {Here ``the local reduced density matrix of the system in the
  fusion space $V_{\psi \protect \bar {\sigma }}^{\protect \bar {\sigma }}$''
  means $\Phi [\rho _K(t_2)]$, where $\Phi $ is the isomorphism from the
  information convex set in the region $K$ to the convex set $\protect \mathcal
  {S}(V_{\psi \protect \bar {\sigma }}^{\protect \bar {\sigma
  }})$.}\BibitemShut {Stop}%
\bibitem [{\citenamefont {Kong}(2014)}]{kongAnyonCondensationTensor2014}%
  \BibitemOpen
  \bibfield  {author} {\bibinfo {author} {\bibfnamefont {Liang}\ \bibnamefont
  {Kong}},\ }\bibfield  {title} {\enquote {\bibinfo {title} {Anyon condensation
  and tensor categories},}\ }\href {\doibase 10.1016/j.nuclphysb.2014.07.003}
  {\bibfield  {journal} {\bibinfo  {journal} {Nucl. Phys. B}\ }\textbf
  {\bibinfo {volume} {886}},\ \bibinfo {pages} {436--482} (\bibinfo {year}
  {2014})}\BibitemShut {NoStop}%
\bibitem [{Note13()}]{Note13}%
  \BibitemOpen
  \bibinfo {note} {Technically, the module category defined in Sec.~\ref
  {sec:ModCatrecap} is more precisely known as a right module category in the
  math literature~\cite {TenCat_EGNO}. In general, if $\protect \mathcal {M}$
  is a right module category over the fusion category $\protect \mathcal {C}$,
  then the dual $\protect \bar {\protect \mathcal {M}}$ is naturally a left
  $\protect \mathcal {C}$-module category. A left $\protect \mathcal
  {C}$-module category is analogous to a right $\protect \mathcal {C}$-module
  category, but in which an object $\psi \in \protect \mathcal {C}$ fuse with
  $\protect \bar {\sigma }\in \protect \mathcal {M}$ from the left, denoted by
  $\psi \times \protect \bar {\sigma }$, and the fusion rule and the $\protect
  \tilde {F}$-move are defined in a similar way. If $\protect \mathcal {C}$ is
  symmetric~(or braided in the 2D case), then one can use the braiding
  structure of $\protect \mathcal {C}$ to turn a right $\protect \mathcal
  {C}$-module category into a left $\protect \mathcal {C}$-module
  category~\cite {kong2022invitation,TenCat_EGNO} and vice versa. Since in this
  paper, $\protect \mathcal {C}$ is always symmetric~(or braided in the 2D
  case), we do not need to distinguish left and right $\protect \mathcal
  {C}$-module categories.}\BibitemShut {Stop}%
\bibitem [{\citenamefont {Ostrik}(2003)}]{Ostrik2003}%
  \BibitemOpen
  \bibfield  {author} {\bibinfo {author} {\bibfnamefont {Victor}\ \bibnamefont
  {Ostrik}},\ }\bibfield  {title} {\enquote {\bibinfo {title} {Module
  categories, weak {{Hopf}} algebras and modular invariants},}\ }\href
  {\doibase 10.1007/s00031-003-0515-6} {\bibfield  {journal} {\bibinfo
  {journal} {Transform. Groups}\ }\textbf {\bibinfo {volume} {8}},\ \bibinfo
  {pages} {177--206} (\bibinfo {year} {2003})}\BibitemShut {NoStop}%
\bibitem [{\citenamefont {Shahriari}(1991)}]{Shahriari1991CTFG}%
  \BibitemOpen
  \bibfield  {author} {\bibinfo {author} {\bibfnamefont {Shahriar}\
  \bibnamefont {Shahriari}},\ }\bibfield  {title} {\enquote {\bibinfo {title}
  {On central type factor groups.}}\ }\href@noop {} {\bibfield  {journal}
  {\bibinfo  {journal} {Pac. J. Math.}\ }\textbf {\bibinfo {volume} {151}},\
  \bibinfo {pages} {151--178} (\bibinfo {year} {1991})}\BibitemShut {NoStop}%
\bibitem [{\citenamefont {Etingof}\ and\ \citenamefont
  {Gelaki}(2000)}]{etingof2000semisimpleTHAclassification}%
  \BibitemOpen
  \bibfield  {author} {\bibinfo {author} {\bibfnamefont {Pavel}\ \bibnamefont
  {Etingof}}\ and\ \bibinfo {author} {\bibfnamefont {Shlomo}\ \bibnamefont
  {Gelaki}},\ }\bibfield  {title} {\enquote {\bibinfo {title} {The
  classification of triangular semisimple and cosemisimple {{Hopf}} algebras
  over an algebraically closed field},}\ }\href@noop {} {\bibfield  {journal}
  {\bibinfo  {journal} {Int. Math. Res. Not.}\ }\textbf {\bibinfo {volume}
  {2000}},\ \bibinfo {pages} {223--234} (\bibinfo {year} {2000})}\BibitemShut
  {NoStop}%
\bibitem [{Note14()}]{Note14}%
  \BibitemOpen
  \bibinfo {note} {One may worry that Alice may be able to calculate $n$ by
  carefully analyzing her circle movements in this process. To prevent this, we
  can add another round of the twist in which the role of Alice and Bob are
  exchanged, i.e., this time the circle B braids around circle A in
  counterclockwise direction for $k$ times. In this way the two circles
  complete $n+k$ full braids. But Alice does not know $k$, while Bob does not
  know $n$, so neither of them know $n+k$.}\BibitemShut {Stop}%
\bibitem [{Mma()}]{MmaCode-SFC}%
  \BibitemOpen
  \href@noop {} {}\BibitemShut {NoStop}%
\bibitem [{Note15()}]{Note15}%
  \BibitemOpen
  \bibinfo {note} {Note that we do not call this a projective representation
  since the associated 2-cocycle is trivial.}\BibitemShut {Stop}%
\bibitem [{\citenamefont {{de C. Chamon}}\ \emph {et~al.}(1997)\citenamefont
  {{de C. Chamon}}, \citenamefont {Freed}, \citenamefont {Kivelson},
  \citenamefont {Sondhi},\ and\ \citenamefont
  {Wen}}]{Wen1997FQHinterferometer}%
  \BibitemOpen
  \bibfield  {author} {\bibinfo {author} {\bibfnamefont {C.}~\bibnamefont {{de
  C. Chamon}}}, \bibinfo {author} {\bibfnamefont {D.~E.}\ \bibnamefont
  {Freed}}, \bibinfo {author} {\bibfnamefont {S.~A.}\ \bibnamefont {Kivelson}},
  \bibinfo {author} {\bibfnamefont {S.~L.}\ \bibnamefont {Sondhi}}, \ and\
  \bibinfo {author} {\bibfnamefont {X.~G.}\ \bibnamefont {Wen}},\ }\bibfield
  {title} {\enquote {\bibinfo {title} {Two point-contact interferometer for
  quantum {{Hall}} systems},}\ }\href {\doibase 10.1103/PhysRevB.55.2331}
  {\bibfield  {journal} {\bibinfo  {journal} {Phys. Rev. B}\ }\textbf {\bibinfo
  {volume} {55}},\ \bibinfo {pages} {2331--2343} (\bibinfo {year}
  {1997})}\BibitemShut {NoStop}%
\bibitem [{\citenamefont {Camino}\ \emph {et~al.}(2005)\citenamefont {Camino},
  \citenamefont {Zhou},\ and\ \citenamefont {Goldman}}]{Camino2005Realization}%
  \BibitemOpen
  \bibfield  {author} {\bibinfo {author} {\bibfnamefont {F.~E.}\ \bibnamefont
  {Camino}}, \bibinfo {author} {\bibfnamefont {Wei}\ \bibnamefont {Zhou}}, \
  and\ \bibinfo {author} {\bibfnamefont {V.~J.}\ \bibnamefont {Goldman}},\
  }\bibfield  {title} {\enquote {\bibinfo {title} {Realization of a
  {{Laughlin}} quasiparticle interferometer: {{Observation}} of fractional
  statistics},}\ }\href {\doibase 10.1103/PhysRevB.72.075342} {\bibfield
  {journal} {\bibinfo  {journal} {Phys. Rev. B}\ }\textbf {\bibinfo {volume}
  {72}},\ \bibinfo {pages} {075342} (\bibinfo {year} {2005})}\BibitemShut
  {NoStop}%
\bibitem [{\citenamefont {Stern}\ and\ \citenamefont
  {Halperin}(2006)}]{Stern2006Proposed}%
  \BibitemOpen
  \bibfield  {author} {\bibinfo {author} {\bibfnamefont {Ady}\ \bibnamefont
  {Stern}}\ and\ \bibinfo {author} {\bibfnamefont {Bertrand~I.}\ \bibnamefont
  {Halperin}},\ }\bibfield  {title} {\enquote {\bibinfo {title} {Proposed
  experiments to probe the non-abelian {{$\nu$=5/2}} quantum hall state},}\
  }\href {\doibase 10.1103/PhysRevLett.96.016802} {\bibfield  {journal}
  {\bibinfo  {journal} {Phys. Rev. Lett.}\ }\textbf {\bibinfo {volume} {96}},\
  \bibinfo {pages} {016802} (\bibinfo {year} {2006})}\BibitemShut {NoStop}%
\bibitem [{\citenamefont {Bonderson}\ \emph {et~al.}(2006)\citenamefont
  {Bonderson}, \citenamefont {Kitaev},\ and\ \citenamefont
  {Shtengel}}]{BondersonKitaev2006}%
  \BibitemOpen
  \bibfield  {author} {\bibinfo {author} {\bibfnamefont {Parsa}\ \bibnamefont
  {Bonderson}}, \bibinfo {author} {\bibfnamefont {Alexei}\ \bibnamefont
  {Kitaev}}, \ and\ \bibinfo {author} {\bibfnamefont {Kirill}\ \bibnamefont
  {Shtengel}},\ }\bibfield  {title} {\enquote {\bibinfo {title} {Detecting
  non-abelian statistics in the {{$\nu$=5/2}} fractional quantum hall state},}\
  }\href {\doibase 10.1103/PhysRevLett.96.016803} {\bibfield  {journal}
  {\bibinfo  {journal} {Phys. Rev. Lett.}\ }\textbf {\bibinfo {volume} {96}},\
  \bibinfo {pages} {016803} (\bibinfo {year} {2006})}\BibitemShut {NoStop}%
\bibitem [{\citenamefont {Bonderson}\ \emph {et~al.}(2008)\citenamefont
  {Bonderson}, \citenamefont {Shtengel},\ and\ \citenamefont
  {Slingerland}}]{bondersonInterferometryNonAbelianAnyons2008}%
  \BibitemOpen
  \bibfield  {author} {\bibinfo {author} {\bibfnamefont {Parsa}\ \bibnamefont
  {Bonderson}}, \bibinfo {author} {\bibfnamefont {Kirill}\ \bibnamefont
  {Shtengel}}, \ and\ \bibinfo {author} {\bibfnamefont {J.~K.}\ \bibnamefont
  {Slingerland}},\ }\bibfield  {title} {\enquote {\bibinfo {title}
  {Interferometry of non-{{Abelian}} anyons},}\ }\href {\doibase
  10.1016/j.aop.2008.01.012} {\bibfield  {journal} {\bibinfo  {journal} {Annals
  of Physics}\ }\textbf {\bibinfo {volume} {323}},\ \bibinfo {pages}
  {2709--2755} (\bibinfo {year} {2008})}\BibitemShut {NoStop}%
\bibitem [{\citenamefont {Bishara}\ \emph {et~al.}(2009)\citenamefont
  {Bishara}, \citenamefont {Bonderson}, \citenamefont {Nayak}, \citenamefont
  {Shtengel},\ and\ \citenamefont {Slingerland}}]{Nayak2009Interferometric}%
  \BibitemOpen
  \bibfield  {author} {\bibinfo {author} {\bibfnamefont {Waheb}\ \bibnamefont
  {Bishara}}, \bibinfo {author} {\bibfnamefont {Parsa}\ \bibnamefont
  {Bonderson}}, \bibinfo {author} {\bibfnamefont {Chetan}\ \bibnamefont
  {Nayak}}, \bibinfo {author} {\bibfnamefont {Kirill}\ \bibnamefont
  {Shtengel}}, \ and\ \bibinfo {author} {\bibfnamefont {J.~K.}\ \bibnamefont
  {Slingerland}},\ }\bibfield  {title} {\enquote {\bibinfo {title}
  {Interferometric signature of non-{{Abelian}} anyons},}\ }\href {\doibase
  10.1103/PhysRevB.80.155303} {\bibfield  {journal} {\bibinfo  {journal} {Phys.
  Rev. B}\ }\textbf {\bibinfo {volume} {80}},\ \bibinfo {pages} {155303}
  (\bibinfo {year} {2009})}\BibitemShut {NoStop}%
\bibitem [{\citenamefont {Feldman}\ and\ \citenamefont
  {Halperin}(2021)}]{Feldman_2021}%
  \BibitemOpen
  \bibfield  {author} {\bibinfo {author} {\bibfnamefont {D~E}\ \bibnamefont
  {Feldman}}\ and\ \bibinfo {author} {\bibfnamefont {Bertrand~I}\ \bibnamefont
  {Halperin}},\ }\bibfield  {title} {\enquote {\bibinfo {title} {Fractional
  charge and fractional statistics in the quantum {{Hall}} effects},}\ }\href
  {\doibase 10.1088/1361-6633/ac03aa} {\bibfield  {journal} {\bibinfo
  {journal} {Rep. Prog. Phys.}\ }\textbf {\bibinfo {volume} {84}},\ \bibinfo
  {pages} {076501} (\bibinfo {year} {2021})}\BibitemShut {NoStop}%
\bibitem [{Note16()}]{Note16}%
  \BibitemOpen
  \bibinfo {note} {A small exception is that these beyond-SFC $R$-paraparticles
  can be realized in a family of 1D solvable quantum spin chains~\cite
  {BSundar2019,wang2023para}, which can be either gapped or gapless depending
  on the model parameters, and they satisfy all our axioms except Axiom.~\ref
  {Axiom2}~(because there is no topological order in 1D).}\BibitemShut {Stop}%
\bibitem [{Note17()}]{Note17}%
  \BibitemOpen
  \bibinfo {note} {Here we mention a possibility that is reasonably achievable
  using current experimental technology. The quantum double model based on the
  solvable groups $D_8$ and $A_4$ can win at least levels 3 and 4 in Tab.~\ref
  {tab:protocols_strategies}, respectively, provided that one carefully design
  a hybrid boundary condition that we describe in App.~\ref
  {sec:blackdefectgappedboundary}. Ground state of these quantum double models
  can be efficiently prepared using the recently developed topological quantum
  state preparation protocol involving measurements and feedforward~\cite
  {HierarchyTPO_LOCC,Iqbal2024}~(indeed, $D_8$ is already realized in
  Ref.~\protect \rev@citealpnum {Iqbal2024}, which is called $D_4$ in there due
  to a different naming convention. However, the system realized in
  Ref.~\protect \rev@citealpnum {Iqbal2024} cannot win our challenge due to the
  uniform boundary condition used there). It would be interesting to actually
  play the challenge games on these platforms.}\BibitemShut {Stop}%
\bibitem [{\citenamefont {Jordan}(2010)}]{Jordan2010PQC}%
  \BibitemOpen
  \bibfield  {author} {\bibinfo {author} {\bibfnamefont {Stephen~P}\
  \bibnamefont {Jordan}},\ }\bibfield  {title} {\enquote {\bibinfo {title}
  {Permutational quantum computing},}\ }\href@noop {} {\bibfield  {journal}
  {\bibinfo  {journal} {Quantum Info. Comput.}\ }\textbf {\bibinfo {volume}
  {10}},\ \bibinfo {pages} {470--497} (\bibinfo {year} {2010})}\BibitemShut
  {NoStop}%
\bibitem [{Note18()}]{Note18}%
  \BibitemOpen
  \bibinfo {note} {Note that in the original PQC scheme proposed in
  Ref.~\protect \rev@citealpnum {Jordan2010PQC}, PQC was illustrated using
  ordinary spin-1/2 particles, such as ordinary electrons, where quantum
  information is stored in the spin degree of freedom, and consequently there
  is no topological protection. However, it is also possible to perform PQC
  using~(indeed, this was already mentioned in Ref.~\protect \rev@citealpnum
  {Jordan2010PQC}) topological quasiparticles in 3+1D topological phases,
  analogous to topological quantum computation~\cite
  {kitaev2003fault,Nayak2008NAAnyons} in 2D using non-Abelian anyons, where
  quantum information stored in the fusion space is topologically
  protected.}\BibitemShut {Stop}%
\bibitem [{\citenamefont {Kato}\ \emph {et~al.}(2016)\citenamefont {Kato},
  \citenamefont {Furrer},\ and\ \citenamefont {Murao}}]{Kato2016Information}%
  \BibitemOpen
  \bibfield  {author} {\bibinfo {author} {\bibfnamefont {Kohtaro}\ \bibnamefont
  {Kato}}, \bibinfo {author} {\bibfnamefont {Fabian}\ \bibnamefont {Furrer}}, \
  and\ \bibinfo {author} {\bibfnamefont {Mio}\ \bibnamefont {Murao}},\
  }\bibfield  {title} {\enquote {\bibinfo {title} {Information-theoretical
  analysis of topological entanglement entropy and multipartite
  correlations},}\ }\href {\doibase 10.1103/PhysRevA.93.022317} {\bibfield
  {journal} {\bibinfo  {journal} {Phys. Rev. At. Mol. Opt. Phys.}\ }\textbf
  {\bibinfo {volume} {93}},\ \bibinfo {pages} {022317} (\bibinfo {year}
  {2016})}\BibitemShut {NoStop}%
\bibitem [{\citenamefont {Fiedler}\ \emph {et~al.}(2017)\citenamefont
  {Fiedler}, \citenamefont {Naaijkens},\ and\ \citenamefont
  {Osborne}}]{fiedlerJonesIndexSecret2017}%
  \BibitemOpen
  \bibfield  {author} {\bibinfo {author} {\bibfnamefont {Leander}\ \bibnamefont
  {Fiedler}}, \bibinfo {author} {\bibfnamefont {Pieter}\ \bibnamefont
  {Naaijkens}}, \ and\ \bibinfo {author} {\bibfnamefont {Tobias~J}\
  \bibnamefont {Osborne}},\ }\bibfield  {title} {\enquote {\bibinfo {title}
  {Jones index, secret sharing and total quantum dimension},}\ }\href {\doibase
  10.1088/1367-2630/aa5c0c} {\bibfield  {journal} {\bibinfo  {journal} {New J.
  Phys.}\ }\textbf {\bibinfo {volume} {19}},\ \bibinfo {pages} {023039}
  (\bibinfo {year} {2017})}\BibitemShut {NoStop}%
\bibitem [{Note19()}]{Note19}%
  \BibitemOpen
  \bibinfo {note} {A specific example in 3+1D can be constructed by deconfined
  gauge theories based on isocategorical groups. Specifically, Ref. gives an
  example of two groups $G_1$ and $G_2$~(both of order 64), such that the SFCs
  $\protect \mathrm {Rep}(G_1)$ and $\protect \mathrm {Rep}(G_2)$ are
  equivalent as fusion categories, but inequivalent as SFCs.}\BibitemShut
  {Stop}%
\bibitem [{\citenamefont {Brassard}\ \emph
  {et~al.}(2005{\natexlab{a}})\citenamefont {Brassard}, \citenamefont
  {Broadbent},\ and\ \citenamefont
  {Tapp}}]{brassardQuantumPseudoTelepathy2005}%
  \BibitemOpen
  \bibfield  {author} {\bibinfo {author} {\bibfnamefont {Gilles}\ \bibnamefont
  {Brassard}}, \bibinfo {author} {\bibfnamefont {Anne}\ \bibnamefont
  {Broadbent}}, \ and\ \bibinfo {author} {\bibfnamefont {Alain}\ \bibnamefont
  {Tapp}},\ }\bibfield  {title} {\enquote {\bibinfo {title} {Quantum
  {{Pseudo-Telepathy}}},}\ }\href {\doibase 10.1007/s10701-005-7353-4}
  {\bibfield  {journal} {\bibinfo  {journal} {Found. Phys.}\ }\textbf {\bibinfo
  {volume} {35}},\ \bibinfo {pages} {1877--1907} (\bibinfo {year}
  {2005}{\natexlab{a}})}\BibitemShut {NoStop}%
\bibitem [{\citenamefont {Bulchandani}\ \emph {et~al.}(2023)\citenamefont
  {Bulchandani}, \citenamefont {Burnell},\ and\ \citenamefont
  {Sondhi}}]{Burnell2023TCnonlocalgame}%
  \BibitemOpen
  \bibfield  {author} {\bibinfo {author} {\bibfnamefont {Vir~B.}\ \bibnamefont
  {Bulchandani}}, \bibinfo {author} {\bibfnamefont {Fiona~J.}\ \bibnamefont
  {Burnell}}, \ and\ \bibinfo {author} {\bibfnamefont {S.~L.}\ \bibnamefont
  {Sondhi}},\ }\bibfield  {title} {\enquote {\bibinfo {title} {A multiplayer
  multiteam nonlocal game for the toric code},}\ }\href {\doibase
  10.1103/PhysRevB.107.035409} {\bibfield  {journal} {\bibinfo  {journal}
  {Phys. Rev. B}\ }\textbf {\bibinfo {volume} {107}},\ \bibinfo {pages}
  {035409} (\bibinfo {year} {2023})}\BibitemShut {NoStop}%
\bibitem [{\citenamefont {Hart}\ \emph {et~al.}(2025)\citenamefont {Hart},
  \citenamefont {Stephen}, \citenamefont {Williamson}, \citenamefont
  {{Foss-Feig}},\ and\ \citenamefont
  {Nandkishore}}]{Nandkishore2025TCnonlocalgame}%
  \BibitemOpen
  \bibfield  {author} {\bibinfo {author} {\bibfnamefont {Oliver}\ \bibnamefont
  {Hart}}, \bibinfo {author} {\bibfnamefont {David~T.}\ \bibnamefont
  {Stephen}}, \bibinfo {author} {\bibfnamefont {Dominic~J.}\ \bibnamefont
  {Williamson}}, \bibinfo {author} {\bibfnamefont {Michael}\ \bibnamefont
  {{Foss-Feig}}}, \ and\ \bibinfo {author} {\bibfnamefont {Rahul}\ \bibnamefont
  {Nandkishore}},\ }\bibfield  {title} {\enquote {\bibinfo {title} {Playing
  nonlocal games across a topological phase transition on a quantum
  computer},}\ }\href {\doibase 10.1103/PhysRevLett.134.130602} {\bibfield
  {journal} {\bibinfo  {journal} {Phys. Rev. Lett.}\ }\textbf {\bibinfo
  {volume} {134}},\ \bibinfo {pages} {130602} (\bibinfo {year}
  {2025})}\BibitemShut {NoStop}%
\bibitem [{\citenamefont {Mermin}(1990)}]{Mermin1990Extreme}%
  \BibitemOpen
  \bibfield  {author} {\bibinfo {author} {\bibfnamefont {N.~David}\
  \bibnamefont {Mermin}},\ }\bibfield  {title} {\enquote {\bibinfo {title}
  {Extreme quantum entanglement in a superposition of macroscopically distinct
  states},}\ }\href {\doibase 10.1103/PhysRevLett.65.1838} {\bibfield
  {journal} {\bibinfo  {journal} {Phys. Rev. Lett.}\ }\textbf {\bibinfo
  {volume} {65}},\ \bibinfo {pages} {1838--1840} (\bibinfo {year}
  {1990})}\BibitemShut {NoStop}%
\bibitem [{\citenamefont {Brassard}\ \emph
  {et~al.}(2005{\natexlab{b}})\citenamefont {Brassard}, \citenamefont
  {Broadbent},\ and\ \citenamefont {Tapp}}]{brassard2005recasting}%
  \BibitemOpen
  \bibfield  {author} {\bibinfo {author} {\bibfnamefont {Gilles}\ \bibnamefont
  {Brassard}}, \bibinfo {author} {\bibfnamefont {Anne}\ \bibnamefont
  {Broadbent}}, \ and\ \bibinfo {author} {\bibfnamefont {Alain}\ \bibnamefont
  {Tapp}},\ }\bibfield  {title} {\enquote {\bibinfo {title} {Recasting mermin's
  multi-player game into the framework of pseudo-telepathy},}\ }\href@noop {}
  {\bibfield  {journal} {\bibinfo  {journal} {Quantum Inf. Comput.}\ }\textbf
  {\bibinfo {volume} {5}},\ \bibinfo {pages} {538--550} (\bibinfo {year}
  {2005}{\natexlab{b}})}\BibitemShut {NoStop}%
\bibitem [{\citenamefont {Greenberger}\ \emph {et~al.}(1989)\citenamefont
  {Greenberger}, \citenamefont {Horne},\ and\ \citenamefont
  {Zeilinger}}]{greenbergerGoingBellsTheorem1989}%
  \BibitemOpen
  \bibfield  {author} {\bibinfo {author} {\bibfnamefont {Daniel~M.}\
  \bibnamefont {Greenberger}}, \bibinfo {author} {\bibfnamefont {Michael~A.}\
  \bibnamefont {Horne}}, \ and\ \bibinfo {author} {\bibfnamefont {Anton}\
  \bibnamefont {Zeilinger}},\ }\bibfield  {title} {\enquote {\bibinfo {title}
  {Going {{Beyond Bell}}'s {{Theorem}}},}\ }in\ \href {\doibase
  10.1007/978-94-017-0849-4_10} {\emph {\bibinfo {booktitle} {Bell's
  {{Theorem}}, {{Quantum Theory}} and {{Conceptions}} of the {{Universe}}}}},\
  \bibinfo {editor} {edited by\ \bibinfo {editor} {\bibfnamefont {Menas}\
  \bibnamefont {Kafatos}}}\ (\bibinfo  {publisher} {Springer Netherlands},\
  \bibinfo {address} {Dordrecht},\ \bibinfo {year} {1989})\ pp.\ \bibinfo
  {pages} {69--72}\BibitemShut {NoStop}%
\bibitem [{\citenamefont {Kong}\ and\ \citenamefont
  {Zheng}(2020)}]{kongMathematicalTheoryGapless2020}%
  \BibitemOpen
  \bibfield  {author} {\bibinfo {author} {\bibfnamefont {Liang}\ \bibnamefont
  {Kong}}\ and\ \bibinfo {author} {\bibfnamefont {Hao}\ \bibnamefont {Zheng}},\
  }\bibfield  {title} {\enquote {\bibinfo {title} {A mathematical theory of
  gapless edges of 2d topological orders. {{Part I}}},}\ }\href {\doibase
  10.1007/JHEP02(2020)150} {\bibfield  {journal} {\bibinfo  {journal} {J. High
  Energy Phys.}\ }\textbf {\bibinfo {volume} {2020}},\ \bibinfo {pages} {150}
  (\bibinfo {year} {2020})}\BibitemShut {NoStop}%
\bibitem [{\citenamefont {Kong}\ and\ \citenamefont
  {Zheng}(2021)}]{kongMathematicalTheoryGapless2021}%
  \BibitemOpen
  \bibfield  {author} {\bibinfo {author} {\bibfnamefont {Liang}\ \bibnamefont
  {Kong}}\ and\ \bibinfo {author} {\bibfnamefont {Hao}\ \bibnamefont {Zheng}},\
  }\bibfield  {title} {\enquote {\bibinfo {title} {A mathematical theory of
  gapless edges of 2d topological orders. {{Part II}}},}\ }\href {\doibase
  10.1016/j.nuclphysb.2021.115384} {\bibfield  {journal} {\bibinfo  {journal}
  {Nucl. Phys. B}\ }\textbf {\bibinfo {volume} {966}},\ \bibinfo {pages}
  {115384} (\bibinfo {year} {2021})}\BibitemShut {NoStop}%
\bibitem [{\citenamefont {Kong}\ and\ \citenamefont
  {Zheng}(2022)}]{kongCategoriesQuantumLiquids2022}%
  \BibitemOpen
  \bibfield  {author} {\bibinfo {author} {\bibfnamefont {Liang}\ \bibnamefont
  {Kong}}\ and\ \bibinfo {author} {\bibfnamefont {Hao}\ \bibnamefont {Zheng}},\
  }\bibfield  {title} {\enquote {\bibinfo {title} {Categories of quantum
  liquids {{I}}},}\ }\href {\doibase 10.1007/JHEP08(2022)070} {\bibfield
  {journal} {\bibinfo  {journal} {J. High Energy Phys.}\ }\textbf {\bibinfo
  {volume} {2022}},\ \bibinfo {pages} {70} (\bibinfo {year}
  {2022})}\BibitemShut {NoStop}%
\bibitem [{\citenamefont {Bhardwaj}\ \emph {et~al.}(2025)\citenamefont
  {Bhardwaj}, \citenamefont {Bottini}, \citenamefont {Pajer},\ and\
  \citenamefont {{Sch{\"a}fer-Nameki}}}]{bhardwajClubSandwichGapless2025}%
  \BibitemOpen
  \bibfield  {author} {\bibinfo {author} {\bibfnamefont {Lakshya}\ \bibnamefont
  {Bhardwaj}}, \bibinfo {author} {\bibfnamefont {Lea~E.}\ \bibnamefont
  {Bottini}}, \bibinfo {author} {\bibfnamefont {Daniel}\ \bibnamefont {Pajer}},
  \ and\ \bibinfo {author} {\bibfnamefont {Sakura}\ \bibnamefont
  {{Sch{\"a}fer-Nameki}}},\ }\bibfield  {title} {\enquote {\bibinfo {title}
  {The club sandwich: {{Gapless}} phases and phase transitions with
  non-invertible symmetries},}\ }\href {\doibase 10.21468/SciPostPhys.18.5.156}
  {\bibfield  {journal} {\bibinfo  {journal} {SciPost Phys.}\ }\textbf
  {\bibinfo {volume} {18}},\ \bibinfo {pages} {156} (\bibinfo {year}
  {2025})}\BibitemShut {NoStop}%
\bibitem [{\citenamefont {Grusdt}(2017)}]{Grusdt2017Topological}%
  \BibitemOpen
  \bibfield  {author} {\bibinfo {author} {\bibfnamefont {F.}~\bibnamefont
  {Grusdt}},\ }\bibfield  {title} {\enquote {\bibinfo {title} {Topological
  order of mixed states in correlated quantum many-body systems},}\ }\href
  {\doibase 10.1103/PhysRevB.95.075106} {\bibfield  {journal} {\bibinfo
  {journal} {Phys. Rev. B}\ }\textbf {\bibinfo {volume} {95}},\ \bibinfo
  {pages} {075106} (\bibinfo {year} {2017})}\BibitemShut {NoStop}%
\bibitem [{\citenamefont {Zini}\ and\ \citenamefont
  {Wang}(2021)}]{zini2021mixed}%
  \BibitemOpen
  \bibfield  {author} {\bibinfo {author} {\bibfnamefont {Modjtaba~Shokrian}\
  \bibnamefont {Zini}}\ and\ \bibinfo {author} {\bibfnamefont {Zhenghan}\
  \bibnamefont {Wang}},\ }\bibfield  {title} {\enquote {\bibinfo {title}
  {Mixed-state tqfts},}\ }\href@noop {} {\bibfield  {journal} {\bibinfo
  {journal} {ArXiv Prepr. ArXiv211013946}\ } (\bibinfo {year} {2021})},\
  \Eprint {http://arxiv.org/abs/2110.13946} {arXiv:2110.13946} \BibitemShut
  {NoStop}%
\bibitem [{\citenamefont {Bao}\ \emph {et~al.}(2023)\citenamefont {Bao},
  \citenamefont {Fan}, \citenamefont {Vishwanath},\ and\ \citenamefont
  {Altman}}]{bao2023mixed}%
  \BibitemOpen
  \bibfield  {author} {\bibinfo {author} {\bibfnamefont {Yimu}\ \bibnamefont
  {Bao}}, \bibinfo {author} {\bibfnamefont {Ruihua}\ \bibnamefont {Fan}},
  \bibinfo {author} {\bibfnamefont {Ashvin}\ \bibnamefont {Vishwanath}}, \ and\
  \bibinfo {author} {\bibfnamefont {Ehud}\ \bibnamefont {Altman}},\ }\bibfield
  {title} {\enquote {\bibinfo {title} {Mixed-state topological order and the
  errorfield double formulation of decoherence-induced transitions},}\
  }\href@noop {} {\bibfield  {journal} {\bibinfo  {journal} {ArXiv Prepr.
  ArXiv230105687}\ } (\bibinfo {year} {2023})},\ \Eprint
  {http://arxiv.org/abs/2301.05687} {arXiv:2301.05687} \BibitemShut {NoStop}%
\bibitem [{\citenamefont {Chen}\ and\ \citenamefont
  {Grover}(2024)}]{Grover2024Separability}%
  \BibitemOpen
  \bibfield  {author} {\bibinfo {author} {\bibfnamefont {Yu-Hsueh}\
  \bibnamefont {Chen}}\ and\ \bibinfo {author} {\bibfnamefont {Tarun}\
  \bibnamefont {Grover}},\ }\bibfield  {title} {\enquote {\bibinfo {title}
  {Separability transitions in topological states induced by local
  decoherence},}\ }\href {\doibase 10.1103/PhysRevLett.132.170602} {\bibfield
  {journal} {\bibinfo  {journal} {Phys. Rev. Lett.}\ }\textbf {\bibinfo
  {volume} {132}},\ \bibinfo {pages} {170602} (\bibinfo {year}
  {2024})}\BibitemShut {NoStop}%
\bibitem [{\citenamefont {Fan}\ \emph {et~al.}(2024)\citenamefont {Fan},
  \citenamefont {Bao}, \citenamefont {Altman},\ and\ \citenamefont
  {Vishwanath}}]{Vishwanath2024Diagnostics}%
  \BibitemOpen
  \bibfield  {author} {\bibinfo {author} {\bibfnamefont {Ruihua}\ \bibnamefont
  {Fan}}, \bibinfo {author} {\bibfnamefont {Yimu}\ \bibnamefont {Bao}},
  \bibinfo {author} {\bibfnamefont {Ehud}\ \bibnamefont {Altman}}, \ and\
  \bibinfo {author} {\bibfnamefont {Ashvin}\ \bibnamefont {Vishwanath}},\
  }\bibfield  {title} {\enquote {\bibinfo {title} {Diagnostics of mixed-state
  topological order and breakdown of quantum memory},}\ }\href {\doibase
  10.1103/PRXQuantum.5.020343} {\bibfield  {journal} {\bibinfo  {journal} {PRX
  Quantum}\ }\textbf {\bibinfo {volume} {5}},\ \bibinfo {pages} {020343}
  (\bibinfo {year} {2024})}\BibitemShut {NoStop}%
\bibitem [{\citenamefont {Sang}\ and\ \citenamefont
  {Hsieh}(2025)}]{Hsieh2025Markov}%
  \BibitemOpen
  \bibfield  {author} {\bibinfo {author} {\bibfnamefont {Shengqi}\ \bibnamefont
  {Sang}}\ and\ \bibinfo {author} {\bibfnamefont {Timothy~H.}\ \bibnamefont
  {Hsieh}},\ }\bibfield  {title} {\enquote {\bibinfo {title} {Stability of
  mixed-state quantum phases via finite markov length},}\ }\href {\doibase
  10.1103/PhysRevLett.134.070403} {\bibfield  {journal} {\bibinfo  {journal}
  {Phys. Rev. Lett.}\ }\textbf {\bibinfo {volume} {134}},\ \bibinfo {pages}
  {070403} (\bibinfo {year} {2025})}\BibitemShut {NoStop}%
\bibitem [{\citenamefont {Ellison}\ and\ \citenamefont
  {Cheng}(2025)}]{Cheng2025Premodular}%
  \BibitemOpen
  \bibfield  {author} {\bibinfo {author} {\bibfnamefont {Tyler~D.}\
  \bibnamefont {Ellison}}\ and\ \bibinfo {author} {\bibfnamefont {Meng}\
  \bibnamefont {Cheng}},\ }\bibfield  {title} {\enquote {\bibinfo {title}
  {Toward a classification of mixed-state topological orders in two
  dimensions},}\ }\href {\doibase 10.1103/PRXQuantum.6.010315} {\bibfield
  {journal} {\bibinfo  {journal} {PRX Quantum}\ }\textbf {\bibinfo {volume}
  {6}},\ \bibinfo {pages} {010315} (\bibinfo {year} {2025})}\BibitemShut
  {NoStop}%
\bibitem [{\citenamefont {Wang}\ \emph {et~al.}(2025)\citenamefont {Wang},
  \citenamefont {Wu},\ and\ \citenamefont {Wang}}]{Wang2025Intrinsic}%
  \BibitemOpen
  \bibfield  {author} {\bibinfo {author} {\bibfnamefont {Zijian}\ \bibnamefont
  {Wang}}, \bibinfo {author} {\bibfnamefont {Zhengzhi}\ \bibnamefont {Wu}}, \
  and\ \bibinfo {author} {\bibfnamefont {Zhong}\ \bibnamefont {Wang}},\
  }\bibfield  {title} {\enquote {\bibinfo {title} {Intrinsic mixed-state
  topological order},}\ }\href {\doibase 10.1103/PRXQuantum.6.010314}
  {\bibfield  {journal} {\bibinfo  {journal} {PRX Quantum}\ }\textbf {\bibinfo
  {volume} {6}},\ \bibinfo {pages} {010314} (\bibinfo {year}
  {2025})}\BibitemShut {NoStop}%
\bibitem [{\citenamefont {Sohal}\ and\ \citenamefont
  {Prem}(2025)}]{Sohal2025Noisy}%
  \BibitemOpen
  \bibfield  {author} {\bibinfo {author} {\bibfnamefont {Ramanjit}\
  \bibnamefont {Sohal}}\ and\ \bibinfo {author} {\bibfnamefont {Abhinav}\
  \bibnamefont {Prem}},\ }\bibfield  {title} {\enquote {\bibinfo {title} {Noisy
  approach to intrinsically mixed-state topological order},}\ }\href {\doibase
  10.1103/PRXQuantum.6.010313} {\bibfield  {journal} {\bibinfo  {journal} {PRX
  Quantum}\ }\textbf {\bibinfo {volume} {6}},\ \bibinfo {pages} {010313}
  (\bibinfo {year} {2025})}\BibitemShut {NoStop}%
\bibitem [{\citenamefont {Li}\ and\ \citenamefont
  {Mong}(2025)}]{Li2025Replica}%
  \BibitemOpen
  \bibfield  {author} {\bibinfo {author} {\bibfnamefont {Zhuan}\ \bibnamefont
  {Li}}\ and\ \bibinfo {author} {\bibfnamefont {Roger S.~K.}\ \bibnamefont
  {Mong}},\ }\bibfield  {title} {\enquote {\bibinfo {title} {Replica
  topological order in quantum mixed states and quantum error correction},}\
  }\href {\doibase 10.1103/PhysRevB.111.125106} {\bibfield  {journal} {\bibinfo
   {journal} {Phys. Rev. B}\ }\textbf {\bibinfo {volume} {111}},\ \bibinfo
  {pages} {125106} (\bibinfo {year} {2025})}\BibitemShut {NoStop}%
\bibitem [{\citenamefont {Lessa}\ \emph {et~al.}(2025)\citenamefont {Lessa},
  \citenamefont {Sang}, \citenamefont {Lu}, \citenamefont {Hsieh},\ and\
  \citenamefont {Wang}}]{lessa2025higher}%
  \BibitemOpen
  \bibfield  {author} {\bibinfo {author} {\bibfnamefont {Leonardo~A}\
  \bibnamefont {Lessa}}, \bibinfo {author} {\bibfnamefont {Shengqi}\
  \bibnamefont {Sang}}, \bibinfo {author} {\bibfnamefont {Tsung-Cheng}\
  \bibnamefont {Lu}}, \bibinfo {author} {\bibfnamefont {Timothy~H}\
  \bibnamefont {Hsieh}}, \ and\ \bibinfo {author} {\bibfnamefont {Chong}\
  \bibnamefont {Wang}},\ }\bibfield  {title} {\enquote {\bibinfo {title}
  {Higher-form anomaly and long-range entanglement of mixed states},}\
  }\href@noop {} {\bibfield  {journal} {\bibinfo  {journal} {ArXiv Prepr.
  ArXiv250312792}\ } (\bibinfo {year} {2025})},\ \Eprint
  {http://arxiv.org/abs/2503.12792} {arXiv:2503.12792} \BibitemShut {NoStop}%
\bibitem [{\citenamefont {Sang}\ \emph {et~al.}(2025)\citenamefont {Sang},
  \citenamefont {Lessa}, \citenamefont {Mong}, \citenamefont {Grover},
  \citenamefont {Wang},\ and\ \citenamefont {Hsieh}}]{sang2025mixed}%
  \BibitemOpen
  \bibfield  {author} {\bibinfo {author} {\bibfnamefont {Shengqi}\ \bibnamefont
  {Sang}}, \bibinfo {author} {\bibfnamefont {Leonardo~A}\ \bibnamefont
  {Lessa}}, \bibinfo {author} {\bibfnamefont {Roger~SK}\ \bibnamefont {Mong}},
  \bibinfo {author} {\bibfnamefont {Tarun}\ \bibnamefont {Grover}}, \bibinfo
  {author} {\bibfnamefont {Chong}\ \bibnamefont {Wang}}, \ and\ \bibinfo
  {author} {\bibfnamefont {Timothy~H}\ \bibnamefont {Hsieh}},\ }\bibfield
  {title} {\enquote {\bibinfo {title} {Mixed-state phases from local
  reversibility},}\ }\href@noop {} {\bibfield  {journal} {\bibinfo  {journal}
  {ArXiv Prepr. ArXiv250702292}\ } (\bibinfo {year} {2025})},\ \Eprint
  {http://arxiv.org/abs/2507.02292} {arXiv:2507.02292} \BibitemShut {NoStop}%
\bibitem [{\citenamefont {Ogata}(2025)}]{ogata2025mixed}%
  \BibitemOpen
  \bibfield  {author} {\bibinfo {author} {\bibfnamefont {Yoshiko}\ \bibnamefont
  {Ogata}},\ }\bibfield  {title} {\enquote {\bibinfo {title} {Mixed state
  topological order: Operator algebraic approach},}\ }\href@noop {} {\bibfield
  {journal} {\bibinfo  {journal} {ArXiv Prepr. ArXiv250102398}\ } (\bibinfo
  {year} {2025})},\ \Eprint {http://arxiv.org/abs/2501.02398}
  {arXiv:2501.02398} \BibitemShut {NoStop}%
\bibitem [{\citenamefont {Nussinov}\ and\ \citenamefont
  {Ortiz}(2008)}]{Ortiz2008autocorrelations}%
  \BibitemOpen
  \bibfield  {author} {\bibinfo {author} {\bibfnamefont {Zohar}\ \bibnamefont
  {Nussinov}}\ and\ \bibinfo {author} {\bibfnamefont {Gerardo}\ \bibnamefont
  {Ortiz}},\ }\bibfield  {title} {\enquote {\bibinfo {title} {Autocorrelations
  and thermal fragility of anyonic loops in topologically quantum ordered
  systems},}\ }\href {\doibase 10.1103/PhysRevB.77.064302} {\bibfield
  {journal} {\bibinfo  {journal} {Phys. Rev. B}\ }\textbf {\bibinfo {volume}
  {77}},\ \bibinfo {pages} {064302} (\bibinfo {year} {2008})}\BibitemShut
  {NoStop}%
\bibitem [{\citenamefont {Castelnovo}\ and\ \citenamefont
  {Chamon}(2008)}]{TCFiniteT2008}%
  \BibitemOpen
  \bibfield  {author} {\bibinfo {author} {\bibfnamefont {Claudio}\ \bibnamefont
  {Castelnovo}}\ and\ \bibinfo {author} {\bibfnamefont {Claudio}\ \bibnamefont
  {Chamon}},\ }\bibfield  {title} {\enquote {\bibinfo {title} {Topological
  order in a three-dimensional toric code at finite temperature},}\ }\href
  {\doibase 10.1103/PhysRevB.78.155120} {\bibfield  {journal} {\bibinfo
  {journal} {Phys. Rev. B}\ }\textbf {\bibinfo {volume} {78}},\ \bibinfo
  {pages} {155120} (\bibinfo {year} {2008})}\BibitemShut {NoStop}%
\bibitem [{\citenamefont {Alicki}\ \emph {et~al.}(2010)\citenamefont {Alicki},
  \citenamefont {Horodecki}, \citenamefont {Horodecki},\ and\ \citenamefont
  {Horodecki}}]{alickiThermalStabilityTopological2010}%
  \BibitemOpen
  \bibfield  {author} {\bibinfo {author} {\bibfnamefont {R.}~\bibnamefont
  {Alicki}}, \bibinfo {author} {\bibfnamefont {M.}~\bibnamefont {Horodecki}},
  \bibinfo {author} {\bibfnamefont {P.}~\bibnamefont {Horodecki}}, \ and\
  \bibinfo {author} {\bibfnamefont {R.}~\bibnamefont {Horodecki}},\ }\bibfield
  {title} {\enquote {\bibinfo {title} {On {{Thermal Stability}} of
  {{Topological Qubit}} in {{Kitaev}}'s {{4D Model}}},}\ }\href {\doibase
  10.1142/S1230161210000023} {\bibfield  {journal} {\bibinfo  {journal} {Open
  Syst. Inf. Dyn.}\ }\textbf {\bibinfo {volume} {17}},\ \bibinfo {pages}
  {1--20} (\bibinfo {year} {2010})}\BibitemShut {NoStop}%
\bibitem [{\citenamefont {Hastings}(2011)}]{Hastings2011finiteTTPO}%
  \BibitemOpen
  \bibfield  {author} {\bibinfo {author} {\bibfnamefont {Matthew~B.}\
  \bibnamefont {Hastings}},\ }\bibfield  {title} {\enquote {\bibinfo {title}
  {Topological order at nonzero temperature},}\ }\href {\doibase
  10.1103/PhysRevLett.107.210501} {\bibfield  {journal} {\bibinfo  {journal}
  {Phys. Rev. Lett.}\ }\textbf {\bibinfo {volume} {107}},\ \bibinfo {pages}
  {210501} (\bibinfo {year} {2011})}\BibitemShut {NoStop}%
\bibitem [{\citenamefont {Bombin}\ \emph {et~al.}(2013)\citenamefont {Bombin},
  \citenamefont {Chhajlany}, \citenamefont {Horodecki},\ and\ \citenamefont
  {{Martin-Delgado}}}]{Bombin_2013}%
  \BibitemOpen
  \bibfield  {author} {\bibinfo {author} {\bibfnamefont {H}~\bibnamefont
  {Bombin}}, \bibinfo {author} {\bibfnamefont {R~W}\ \bibnamefont {Chhajlany}},
  \bibinfo {author} {\bibfnamefont {M}~\bibnamefont {Horodecki}}, \ and\
  \bibinfo {author} {\bibfnamefont {M~A}\ \bibnamefont {{Martin-Delgado}}},\
  }\bibfield  {title} {\enquote {\bibinfo {title} {Self-correcting quantum
  computers},}\ }\href {\doibase 10.1088/1367-2630/15/5/055023} {\bibfield
  {journal} {\bibinfo  {journal} {New J. Phys.}\ }\textbf {\bibinfo {volume}
  {15}},\ \bibinfo {pages} {055023} (\bibinfo {year} {2013})}\BibitemShut
  {NoStop}%
\bibitem [{\citenamefont {Lu}\ \emph {et~al.}(2020)\citenamefont {Lu},
  \citenamefont {Hsieh},\ and\ \citenamefont {Grover}}]{Hsieh2020finiteTTPO}%
  \BibitemOpen
  \bibfield  {author} {\bibinfo {author} {\bibfnamefont {Tsung-Cheng}\
  \bibnamefont {Lu}}, \bibinfo {author} {\bibfnamefont {Timothy~H.}\
  \bibnamefont {Hsieh}}, \ and\ \bibinfo {author} {\bibfnamefont {Tarun}\
  \bibnamefont {Grover}},\ }\bibfield  {title} {\enquote {\bibinfo {title}
  {Detecting topological order at finite temperature using entanglement
  negativity},}\ }\href {\doibase 10.1103/PhysRevLett.125.116801} {\bibfield
  {journal} {\bibinfo  {journal} {Phys. Rev. Lett.}\ }\textbf {\bibinfo
  {volume} {125}},\ \bibinfo {pages} {116801} (\bibinfo {year}
  {2020})}\BibitemShut {NoStop}%
\bibitem [{\citenamefont {Zhou}\ \emph {et~al.}(2025)\citenamefont {Zhou},
  \citenamefont {Cheng}, \citenamefont {Rakovszky}, \citenamefont {{von
  Keyserlingk}},\ and\ \citenamefont {Ellison}}]{Cheng2025finiteTTPO}%
  \BibitemOpen
  \bibfield  {author} {\bibinfo {author} {\bibfnamefont {Shu-Tong}\
  \bibnamefont {Zhou}}, \bibinfo {author} {\bibfnamefont {Meng}\ \bibnamefont
  {Cheng}}, \bibinfo {author} {\bibfnamefont {Tibor}\ \bibnamefont
  {Rakovszky}}, \bibinfo {author} {\bibfnamefont {Curt}\ \bibnamefont {{von
  Keyserlingk}}}, \ and\ \bibinfo {author} {\bibfnamefont {Tyler~D.}\
  \bibnamefont {Ellison}},\ }\bibfield  {title} {\enquote {\bibinfo {title}
  {Finite-temperature quantum topological order in three dimensions},}\ }\href
  {\doibase 10.1103/n9sq-8cxw} {\bibfield  {journal} {\bibinfo  {journal}
  {Phys. Rev. Lett.}\ }\textbf {\bibinfo {volume} {135}},\ \bibinfo {pages}
  {040402} (\bibinfo {year} {2025})}\BibitemShut {NoStop}%
\bibitem [{\citenamefont {Kong}\ \emph {et~al.}(2015)\citenamefont {Kong},
  \citenamefont {Wen},\ and\ \citenamefont {Zheng}}]{kong2015boundary}%
  \BibitemOpen
  \bibfield  {author} {\bibinfo {author} {\bibfnamefont {Liang}\ \bibnamefont
  {Kong}}, \bibinfo {author} {\bibfnamefont {Xiao-Gang}\ \bibnamefont {Wen}}, \
  and\ \bibinfo {author} {\bibfnamefont {Hao}\ \bibnamefont {Zheng}},\
  }\bibfield  {title} {\enquote {\bibinfo {title} {Boundary-bulk relation for
  topological orders as the functor mapping higher categories to their
  centers},}\ }\href@noop {} {\bibfield  {journal} {\bibinfo  {journal} {ArXiv
  Prepr. ArXiv150201690}\ } (\bibinfo {year} {2015})},\ \Eprint
  {http://arxiv.org/abs/1502.01690} {arXiv:1502.01690} \BibitemShut {NoStop}%
\bibitem [{\citenamefont {Douglas}\ and\ \citenamefont
  {Reutter}(2018)}]{douglas2018fusion}%
  \BibitemOpen
  \bibfield  {author} {\bibinfo {author} {\bibfnamefont {Christopher~L}\
  \bibnamefont {Douglas}}\ and\ \bibinfo {author} {\bibfnamefont {David~J}\
  \bibnamefont {Reutter}},\ }\bibfield  {title} {\enquote {\bibinfo {title}
  {Fusion 2-categories and a state-sum invariant for 4-manifolds},}\
  }\href@noop {} {\bibfield  {journal} {\bibinfo  {journal} {ArXiv Prepr.
  ArXiv181211933}\ } (\bibinfo {year} {2018})},\ \Eprint
  {http://arxiv.org/abs/1812.11933} {arXiv:1812.11933} \BibitemShut {NoStop}%
\bibitem [{\citenamefont {Gaiotto}\ and\ \citenamefont
  {{Johnson-Freyd}}(2019)}]{gaiotto2019condensations}%
  \BibitemOpen
  \bibfield  {author} {\bibinfo {author} {\bibfnamefont {Davide}\ \bibnamefont
  {Gaiotto}}\ and\ \bibinfo {author} {\bibfnamefont {Theo}\ \bibnamefont
  {{Johnson-Freyd}}},\ }\bibfield  {title} {\enquote {\bibinfo {title}
  {Condensations in higher categories},}\ }\href@noop {} {\bibfield  {journal}
  {\bibinfo  {journal} {arXiv Prepr. arXiv1905.09566}\ } (\bibinfo {year}
  {2019})}\BibitemShut {NoStop}%
\bibitem [{\citenamefont {Kong}\ \emph
  {et~al.}(2020{\natexlab{a}})\citenamefont {Kong}, \citenamefont {Lan},
  \citenamefont {Wen}, \citenamefont {Zhang},\ and\ \citenamefont
  {Zheng}}]{Kong2020Algebraichigher}%
  \BibitemOpen
  \bibfield  {author} {\bibinfo {author} {\bibfnamefont {Liang}\ \bibnamefont
  {Kong}}, \bibinfo {author} {\bibfnamefont {Tian}\ \bibnamefont {Lan}},
  \bibinfo {author} {\bibfnamefont {Xiao-Gang}\ \bibnamefont {Wen}}, \bibinfo
  {author} {\bibfnamefont {Zhi-Hao}\ \bibnamefont {Zhang}}, \ and\ \bibinfo
  {author} {\bibfnamefont {Hao}\ \bibnamefont {Zheng}},\ }\bibfield  {title}
  {\enquote {\bibinfo {title} {Algebraic higher symmetry and categorical
  symmetry: {{A}} holographic and entanglement view of symmetry},}\ }\href
  {\doibase 10.1103/PhysRevResearch.2.043086} {\bibfield  {journal} {\bibinfo
  {journal} {Phys. Rev. Res.}\ }\textbf {\bibinfo {volume} {2}},\ \bibinfo
  {pages} {043086} (\bibinfo {year} {2020}{\natexlab{a}})}\BibitemShut
  {NoStop}%
\bibitem [{\citenamefont {Kong}\ \emph
  {et~al.}(2020{\natexlab{b}})\citenamefont {Kong}, \citenamefont {Lan},
  \citenamefont {Wen}, \citenamefont {Zhang},\ and\ \citenamefont
  {Zheng}}]{Kong2020Classification}%
  \BibitemOpen
  \bibfield  {author} {\bibinfo {author} {\bibfnamefont {Liang}\ \bibnamefont
  {Kong}}, \bibinfo {author} {\bibfnamefont {Tian}\ \bibnamefont {Lan}},
  \bibinfo {author} {\bibfnamefont {Xiao-Gang}\ \bibnamefont {Wen}}, \bibinfo
  {author} {\bibfnamefont {Zhi-Hao}\ \bibnamefont {Zhang}}, \ and\ \bibinfo
  {author} {\bibfnamefont {Hao}\ \bibnamefont {Zheng}},\ }\bibfield  {title}
  {\enquote {\bibinfo {title} {Classification of topological phases with finite
  internal symmetries in all dimensions},}\ }\href {\doibase
  10.1007/JHEP09(2020)093} {\bibfield  {journal} {\bibinfo  {journal} {J. High
  Energy Phys.}\ }\textbf {\bibinfo {volume} {2020}},\ \bibinfo {pages} {93}
  (\bibinfo {year} {2020}{\natexlab{b}})}\BibitemShut {NoStop}%
\bibitem [{\citenamefont {{Johnson-Freyd}}(2020)}]{johnson20203+}%
  \BibitemOpen
  \bibfield  {author} {\bibinfo {author} {\bibfnamefont {Theo}\ \bibnamefont
  {{Johnson-Freyd}}},\ }\bibfield  {title} {\enquote {\bibinfo {title} {(3+ 1)
  {{D}} topological orders with only a \$\textbackslash mathbb
  \textbraceleft{{Z}}\textbraceright{} \_2 \$-charged particle},}\ }\href@noop
  {} {\bibfield  {journal} {\bibinfo  {journal} {ArXiv Prepr. ArXiv201111165}\
  } (\bibinfo {year} {2020})},\ \Eprint {http://arxiv.org/abs/2011.11165}
  {arXiv:2011.11165} \BibitemShut {NoStop}%
\bibitem [{\citenamefont {Zhu}\ \emph {et~al.}(2019)\citenamefont {Zhu},
  \citenamefont {Lan},\ and\ \citenamefont {Wen}}]{Zhu3DTPOModels}%
  \BibitemOpen
  \bibfield  {author} {\bibinfo {author} {\bibfnamefont {Chenchang}\
  \bibnamefont {Zhu}}, \bibinfo {author} {\bibfnamefont {Tian}\ \bibnamefont
  {Lan}}, \ and\ \bibinfo {author} {\bibfnamefont {Xiao-Gang}\ \bibnamefont
  {Wen}},\ }\bibfield  {title} {\enquote {\bibinfo {title} {Topological
  nonlinear {{$\sigma$}}-model, higher gauge theory, and a systematic
  construction of 3+{{1D}} topological orders for boson systems},}\ }\href
  {\doibase 10.1103/PhysRevB.100.045105} {\bibfield  {journal} {\bibinfo
  {journal} {Phys. Rev. B}\ }\textbf {\bibinfo {volume} {100}},\ \bibinfo
  {pages} {45105} (\bibinfo {year} {2019})}\BibitemShut {NoStop}%
\bibitem [{\citenamefont {Dubail}\ and\ \citenamefont
  {Read}(2015)}]{DubailRead2015}%
  \BibitemOpen
  \bibfield  {author} {\bibinfo {author} {\bibfnamefont {J}~\bibnamefont
  {Dubail}}\ and\ \bibinfo {author} {\bibfnamefont {N}~\bibnamefont {Read}},\
  }\bibfield  {title} {\enquote {\bibinfo {title} {Tensor network trial states
  for chiral topological phases in two dimensions and a no-go theorem in any
  dimension},}\ }\href {\doibase 10.1103/PhysRevB.92.205307} {\bibfield
  {journal} {\bibinfo  {journal} {Phys. Rev. B}\ }\textbf {\bibinfo {volume}
  {92}},\ \bibinfo {pages} {205307} (\bibinfo {year} {2015})}\BibitemShut
  {NoStop}%
\bibitem [{\citenamefont {Kapustin}\ and\ \citenamefont
  {Fidkowski}(2020)}]{kapustinLocalCommutingProjector2020}%
  \BibitemOpen
  \bibfield  {author} {\bibinfo {author} {\bibfnamefont {Anton}\ \bibnamefont
  {Kapustin}}\ and\ \bibinfo {author} {\bibfnamefont {Lukasz}\ \bibnamefont
  {Fidkowski}},\ }\bibfield  {title} {\enquote {\bibinfo {title} {Local
  {{Commuting Projector Hamiltonians}} and the {{Quantum Hall Effect}}},}\
  }\href {\doibase 10.1007/s00220-019-03444-1} {\bibfield  {journal} {\bibinfo
  {journal} {Communications in Mathematical Physics}\ }\textbf {\bibinfo
  {volume} {373}},\ \bibinfo {pages} {763--769} (\bibinfo {year}
  {2020})}\BibitemShut {NoStop}%
\bibitem [{\citenamefont {Kapustin}\ and\ \citenamefont
  {Spodyneiko}(2020)}]{Kapustin2020Thermal}%
  \BibitemOpen
  \bibfield  {author} {\bibinfo {author} {\bibfnamefont {Anton}\ \bibnamefont
  {Kapustin}}\ and\ \bibinfo {author} {\bibfnamefont {Lev}\ \bibnamefont
  {Spodyneiko}},\ }\bibfield  {title} {\enquote {\bibinfo {title} {Thermal
  {{Hall}} conductance and a relative topological invariant of gapped
  two-dimensional systems},}\ }\href {\doibase 10.1103/PhysRevB.101.045137}
  {\bibfield  {journal} {\bibinfo  {journal} {Phys. Rev. B}\ }\textbf {\bibinfo
  {volume} {101}},\ \bibinfo {pages} {045137} (\bibinfo {year}
  {2020})}\BibitemShut {NoStop}%
\bibitem [{\citenamefont {Zhang}\ \emph {et~al.}(2022)\citenamefont {Zhang},
  \citenamefont {Levin},\ and\ \citenamefont {Bachmann}}]{Levin2022Vanishing}%
  \BibitemOpen
  \bibfield  {author} {\bibinfo {author} {\bibfnamefont {Carolyn}\ \bibnamefont
  {Zhang}}, \bibinfo {author} {\bibfnamefont {Michael}\ \bibnamefont {Levin}},
  \ and\ \bibinfo {author} {\bibfnamefont {Sven}\ \bibnamefont {Bachmann}},\
  }\bibfield  {title} {\enquote {\bibinfo {title} {Vanishing hall conductance
  for commuting hamiltonians},}\ }\href {\doibase 10.1103/PhysRevB.105.L081103}
  {\bibfield  {journal} {\bibinfo  {journal} {Phys. Rev. B}\ }\textbf {\bibinfo
  {volume} {105}},\ \bibinfo {pages} {L081103} (\bibinfo {year}
  {2022})}\BibitemShut {NoStop}%
\bibitem [{\citenamefont {Wang}\ \emph {et~al.}(2017)\citenamefont {Wang},
  \citenamefont {Xu}, \citenamefont {Pu},\ and\ \citenamefont
  {Hazzard}}]{Wang2017Number}%
  \BibitemOpen
  \bibfield  {author} {\bibinfo {author} {\bibfnamefont {Zhiyuan}\ \bibnamefont
  {Wang}}, \bibinfo {author} {\bibfnamefont {Youjiang}\ \bibnamefont {Xu}},
  \bibinfo {author} {\bibfnamefont {Han}\ \bibnamefont {Pu}}, \ and\ \bibinfo
  {author} {\bibfnamefont {Kaden R.A.~A}\ \bibnamefont {Hazzard}},\ }\bibfield
  {title} {\enquote {\bibinfo {title} {Number-conserving interacting fermion
  models with exact topological superconducting ground states},}\ }\href
  {\doibase 10.1103/PhysRevB.96.115110} {\bibfield  {journal} {\bibinfo
  {journal} {Phys. Rev. B}\ }\textbf {\bibinfo {volume} {96}},\ \bibinfo
  {pages} {115110} (\bibinfo {year} {2017})},\ \Eprint
  {http://arxiv.org/abs/1703.01249} {arXiv:1703.01249} \BibitemShut {NoStop}%
\bibitem [{\citenamefont {Wang}\ and\ \citenamefont
  {Hazzard}(2018)}]{Wang2018}%
  \BibitemOpen
  \bibfield  {author} {\bibinfo {author} {\bibfnamefont {Zhiyuan}\ \bibnamefont
  {Wang}}\ and\ \bibinfo {author} {\bibfnamefont {Kaden~R.A.}\ \bibnamefont
  {Hazzard}},\ }\bibfield  {title} {\enquote {\bibinfo {title} {Analytic ground
  state wave functions of mean-field px+ipy superconductors with vortices and
  boundaries},}\ }\href {\doibase 10.1103/PhysRevB.97.104501} {\bibfield
  {journal} {\bibinfo  {journal} {Phys. Rev. B}\ }\textbf {\bibinfo {volume}
  {97}} (\bibinfo {year} {2018}),\ 10.1103/PhysRevB.97.104501},\ \Eprint
  {http://arxiv.org/abs/1712.09904} {arXiv:1712.09904} \BibitemShut {NoStop}%
\bibitem [{\citenamefont {Schroeter}\ \emph {et~al.}(2007)\citenamefont
  {Schroeter}, \citenamefont {Kapit}, \citenamefont {Thomale},\ and\
  \citenamefont {Greiter}}]{Greiter2007CSLParentH}%
  \BibitemOpen
  \bibfield  {author} {\bibinfo {author} {\bibfnamefont {Darrell~F.}\
  \bibnamefont {Schroeter}}, \bibinfo {author} {\bibfnamefont {Eliot}\
  \bibnamefont {Kapit}}, \bibinfo {author} {\bibfnamefont {Ronny}\ \bibnamefont
  {Thomale}}, \ and\ \bibinfo {author} {\bibfnamefont {Martin}\ \bibnamefont
  {Greiter}},\ }\bibfield  {title} {\enquote {\bibinfo {title} {Spin
  hamiltonian for which the chiral spin liquid is the exact ground state},}\
  }\href {\doibase 10.1103/PhysRevLett.99.097202} {\bibfield  {journal}
  {\bibinfo  {journal} {Phys. Rev. Lett.}\ }\textbf {\bibinfo {volume} {99}},\
  \bibinfo {pages} {097202} (\bibinfo {year} {2007})}\BibitemShut {NoStop}%
\bibitem [{\citenamefont {Wahl}\ \emph {et~al.}(2013)\citenamefont {Wahl},
  \citenamefont {Tu}, \citenamefont {Schuch},\ and\ \citenamefont
  {Cirac}}]{Cirac2013chiralPEPS}%
  \BibitemOpen
  \bibfield  {author} {\bibinfo {author} {\bibfnamefont {T.~B.}\ \bibnamefont
  {Wahl}}, \bibinfo {author} {\bibfnamefont {H.-H.}\ \bibnamefont {Tu}},
  \bibinfo {author} {\bibfnamefont {N.}~\bibnamefont {Schuch}}, \ and\ \bibinfo
  {author} {\bibfnamefont {J.~I.}\ \bibnamefont {Cirac}},\ }\bibfield  {title}
  {\enquote {\bibinfo {title} {Projected entangled-pair states can describe
  chiral topological states},}\ }\href {\doibase
  10.1103/PhysRevLett.111.236805} {\bibfield  {journal} {\bibinfo  {journal}
  {Phys. Rev. Lett.}\ }\textbf {\bibinfo {volume} {111}},\ \bibinfo {pages}
  {236805} (\bibinfo {year} {2013})}\BibitemShut {NoStop}%
\bibitem [{\citenamefont {Yang}\ \emph {et~al.}(2015)\citenamefont {Yang},
  \citenamefont {Wahl}, \citenamefont {Tu}, \citenamefont {Schuch},\ and\
  \citenamefont {Cirac}}]{Yang2015chiralPEPS}%
  \BibitemOpen
  \bibfield  {author} {\bibinfo {author} {\bibfnamefont {Shuo}\ \bibnamefont
  {Yang}}, \bibinfo {author} {\bibfnamefont {Thorsten~B.}\ \bibnamefont
  {Wahl}}, \bibinfo {author} {\bibfnamefont {Hong-Hao}\ \bibnamefont {Tu}},
  \bibinfo {author} {\bibfnamefont {Norbert}\ \bibnamefont {Schuch}}, \ and\
  \bibinfo {author} {\bibfnamefont {J.~Ignacio}\ \bibnamefont {Cirac}},\
  }\bibfield  {title} {\enquote {\bibinfo {title} {Chiral projected
  entangled-pair state with topological order},}\ }\href {\doibase
  10.1103/PhysRevLett.114.106803} {\bibfield  {journal} {\bibinfo  {journal}
  {Phys. Rev. Lett.}\ }\textbf {\bibinfo {volume} {114}},\ \bibinfo {pages}
  {106803} (\bibinfo {year} {2015})}\BibitemShut {NoStop}%
\bibitem [{\citenamefont {Hastings}(2006)}]{hastings2006solving}%
  \BibitemOpen
  \bibfield  {author} {\bibinfo {author} {\bibfnamefont {Matthew~B}\
  \bibnamefont {Hastings}},\ }\bibfield  {title} {\enquote {\bibinfo {title}
  {Solving gapped {{Hamiltonians}} locally},}\ }\href@noop {} {\bibfield
  {journal} {\bibinfo  {journal} {Phys. Rev. B}\ }\textbf {\bibinfo {volume}
  {73}},\ \bibinfo {pages} {85115} (\bibinfo {year} {2006})}\BibitemShut
  {NoStop}%
\bibitem [{Note20()}]{Note20}%
  \BibitemOpen
  \bibinfo {note} {Here it is important that we only demand $\rho _{D}(t)=\rho
  _D$ for simply connected local regions $D$ that do not intersect circle A and
  B. All the winning strategies in this paper satisfy this condition, however,
  with a topological quasiparticle in circle A, $\rho _{K}(t)\protect \neq \rho
  _K$ if $K$ is taken to be an annulus region encircling circle A. It is
  unlikely that any pure state topological phases can win this game if we
  demand $\rho _{K}(t)=\rho _K$ on multiply connected regions as well, due to
  the principle of remote detectability of topological excitations~\cite
  {LanKongWen3DAB,LanWen3DEF}.}\BibitemShut {Stop}%
\bibitem [{\citenamefont
  {Bombin}(2013)}]{bombinIntroductionTopologicalQuantum2013}%
  \BibitemOpen
  \bibfield  {author} {\bibinfo {author} {\bibfnamefont {H.}~\bibnamefont
  {Bombin}},\ }\bibfield  {title} {\enquote {\bibinfo {title} {An introduction
  to topological quantum codes},}\ }\href {\doibase 10.48550/arXiv.1311.0277}
  {\bibfield  {journal} {\bibinfo  {journal} {arXiv e-prints}\ ,\ \bibinfo
  {pages} {arXiv:1311.0277}} (\bibinfo {year} {2013})}\BibitemShut {NoStop}%
\bibitem [{\citenamefont {Wang}(2025)}]{wang2024hopf}%
  \BibitemOpen
  \bibfield  {author} {\bibinfo {author} {\bibfnamefont {Zhiyuan}\ \bibnamefont
  {Wang}},\ }\bibfield  {title} {\enquote {\bibinfo {title} {Hopf algebras and
  solvable unitary circuits},}\ }\href {\doibase 10.1103/PhysRevB.111.104315}
  {\bibfield  {journal} {\bibinfo  {journal} {Phys. Rev. B}\ }\textbf {\bibinfo
  {volume} {111}},\ \bibinfo {pages} {104315} (\bibinfo {year}
  {2025})}\BibitemShut {NoStop}%
\bibitem [{\citenamefont {Iwahori}\ and\ \citenamefont
  {Matsumoto}(1964)}]{iwahori1964several}%
  \BibitemOpen
  \bibfield  {author} {\bibinfo {author} {\bibfnamefont {Nagayoshi}\
  \bibnamefont {Iwahori}}\ and\ \bibinfo {author} {\bibfnamefont {Hideya}\
  \bibnamefont {Matsumoto}},\ }\bibfield  {title} {\enquote {\bibinfo {title}
  {Several remarks on projective representations of finite groups},}\
  }\href@noop {} {\bibfield  {journal} {\bibinfo  {journal} {J. Fac. Sci. Univ.
  Tokyo Sect. I}\ }\textbf {\bibinfo {volume} {10}},\ \bibinfo {pages}
  {129--146} (\bibinfo {year} {1964})}\BibitemShut {NoStop}%
\bibitem [{\citenamefont {Liebler}\ and\ \citenamefont
  {Yellen}(1979)}]{Liebler1979}%
  \BibitemOpen
  \bibfield  {author} {\bibinfo {author} {\bibfnamefont {Robert~A}\
  \bibnamefont {Liebler}}\ and\ \bibinfo {author} {\bibfnamefont {Jay~E}\
  \bibnamefont {Yellen}},\ }\bibfield  {title} {\enquote {\bibinfo {title} {In
  search of nonsolvable groups of central type.}}\ }\href@noop {} {\bibfield
  {journal} {\bibinfo  {journal} {Pac. J. Math.}\ }\textbf {\bibinfo {volume}
  {82}},\ \bibinfo {pages} {485--492} (\bibinfo {year} {1979})}\BibitemShut
  {NoStop}%
\bibitem [{\citenamefont {Etingof}\ \emph {et~al.}(2011)\citenamefont
  {Etingof}, \citenamefont {Golberg}, \citenamefont {Hensel}, \citenamefont
  {Liu}, \citenamefont {Schwendner}, \citenamefont {Vaintrob},\ and\
  \citenamefont {Yudovina}}]{etingof2011introduction}%
  \BibitemOpen
  \bibfield  {author} {\bibinfo {author} {\bibfnamefont {Pavel~I}\ \bibnamefont
  {Etingof}}, \bibinfo {author} {\bibfnamefont {Oleg}\ \bibnamefont {Golberg}},
  \bibinfo {author} {\bibfnamefont {Sebastian}\ \bibnamefont {Hensel}},
  \bibinfo {author} {\bibfnamefont {Tiankai}\ \bibnamefont {Liu}}, \bibinfo
  {author} {\bibfnamefont {Alex}\ \bibnamefont {Schwendner}}, \bibinfo {author}
  {\bibfnamefont {Dmitry}\ \bibnamefont {Vaintrob}}, \ and\ \bibinfo {author}
  {\bibfnamefont {Elena}\ \bibnamefont {Yudovina}},\ }\href {\doibase
  10.1090/stml/059} {\emph {\bibinfo {title} {Introduction to Representation
  Theory}}},\ Vol.~\bibinfo {volume} {59}\ (\bibinfo  {publisher} {American
  Mathematical Society},\ \bibinfo {year} {2011})\BibitemShut {NoStop}%
\bibitem [{Note21()}]{Note21}%
  \BibitemOpen
  \bibinfo {note} {A warning about terminology: in some math literature people
  use ``group of central type'' to refer to what we call a CTFG
  here.}\BibitemShut {Stop}%
\bibitem [{\citenamefont {Buerschaper}\ \emph {et~al.}(2013)\citenamefont
  {Buerschaper}, \citenamefont {Mombelli}, \citenamefont {Christandl},\ and\
  \citenamefont {Aguado}}]{Buerschaper2013HATC}%
  \BibitemOpen
  \bibfield  {author} {\bibinfo {author} {\bibfnamefont {Oliver}\ \bibnamefont
  {Buerschaper}}, \bibinfo {author} {\bibfnamefont {Juan~Mart{\'i}n}\
  \bibnamefont {Mombelli}}, \bibinfo {author} {\bibfnamefont {Matthias}\
  \bibnamefont {Christandl}}, \ and\ \bibinfo {author} {\bibfnamefont {Miguel}\
  \bibnamefont {Aguado}},\ }\bibfield  {title} {\enquote {\bibinfo {title} {A
  hierarchy of topological tensor network states},}\ }\href {\doibase
  10.1063/1.4773316} {\bibfield  {journal} {\bibinfo  {journal} {J. Math.
  Phys.}\ }\textbf {\bibinfo {volume} {54}},\ \bibinfo {pages} {12201}
  (\bibinfo {year} {2013})}\BibitemShut {NoStop}%
\bibitem [{\citenamefont {Drinfeld}(1986)}]{drinfeld1986quantum}%
  \BibitemOpen
  \bibfield  {author} {\bibinfo {author} {\bibfnamefont
  {Vladimir~Gershonovich}\ \bibnamefont {Drinfeld}},\ }\bibfield  {title}
  {\enquote {\bibinfo {title} {Quantum groups},}\ }\href@noop {} {\bibfield
  {journal} {\bibinfo  {journal} {Proc Int Congr Math Berkeley 1986}\ }\textbf
  {\bibinfo {volume} {155}},\ \bibinfo {pages} {18--49} (\bibinfo {year}
  {1986})}\BibitemShut {NoStop}%
\bibitem [{\citenamefont {Kong}\ \emph {et~al.}(2017)\citenamefont {Kong},
  \citenamefont {Wen},\ and\ \citenamefont
  {Zheng}}]{kongBoundarybulkRelationTopological2017}%
  \BibitemOpen
  \bibfield  {author} {\bibinfo {author} {\bibfnamefont {Liang}\ \bibnamefont
  {Kong}}, \bibinfo {author} {\bibfnamefont {Xiao-Gang}\ \bibnamefont {Wen}}, \
  and\ \bibinfo {author} {\bibfnamefont {Hao}\ \bibnamefont {Zheng}},\
  }\bibfield  {title} {\enquote {\bibinfo {title} {Boundary-bulk relation in
  topological orders},}\ }\href {\doibase 10.1016/j.nuclphysb.2017.06.023}
  {\bibfield  {journal} {\bibinfo  {journal} {Nucl. Phys. B}\ }\textbf
  {\bibinfo {volume} {922}},\ \bibinfo {pages} {62--76} (\bibinfo {year}
  {2017})}\BibitemShut {NoStop}%
\bibitem [{\citenamefont {Radford}(1993)}]{Radford1993MQHA}%
  \BibitemOpen
  \bibfield  {author} {\bibinfo {author} {\bibfnamefont {D~E}\ \bibnamefont
  {Radford}},\ }\bibfield  {title} {\enquote {\bibinfo {title} {Minimal
  quasitriangular hopf algebras},}\ }\href {\doibase 10.1006/jabr.1993.1102}
  {\bibfield  {journal} {\bibinfo  {journal} {J. Algebra}\ }\textbf {\bibinfo
  {volume} {157}},\ \bibinfo {pages} {285--315} (\bibinfo {year}
  {1993})}\BibitemShut {NoStop}%
\bibitem [{\citenamefont {Etingof}\ and\ \citenamefont
  {Gelaki}(1998)}]{etingof1998THAconstruction}%
  \BibitemOpen
  \bibfield  {author} {\bibinfo {author} {\bibfnamefont {Pavel}\ \bibnamefont
  {Etingof}}\ and\ \bibinfo {author} {\bibfnamefont {Shlomo}\ \bibnamefont
  {Gelaki}},\ }\bibfield  {title} {\enquote {\bibinfo {title} {A method of
  construction of finite-dimensional triangular semisimple {{Hopf}}
  algebras},}\ }\href@noop {} {\bibfield  {journal} {\bibinfo  {journal} {Math.
  Res. Lett.}\ }\textbf {\bibinfo {volume} {5}},\ \bibinfo {pages} {551--561}
  (\bibinfo {year} {1998})}\BibitemShut {NoStop}%
\bibitem [{\citenamefont {Carqueville}\ and\ \citenamefont
  {Runkel}(2016)}]{carqueville2016orbifold}%
  \BibitemOpen
  \bibfield  {author} {\bibinfo {author} {\bibfnamefont {Nils}\ \bibnamefont
  {Carqueville}}\ and\ \bibinfo {author} {\bibfnamefont {Ingo}\ \bibnamefont
  {Runkel}},\ }\bibfield  {title} {\enquote {\bibinfo {title} {Orbifold
  completion of defect bicategories},}\ }\href@noop {} {\bibfield  {journal}
  {\bibinfo  {journal} {Quantum Topol.}\ }\textbf {\bibinfo {volume} {7}},\
  \bibinfo {pages} {203--279} (\bibinfo {year} {2016})}\BibitemShut {NoStop}%
\bibitem [{\citenamefont {Lan}\ \emph {et~al.}(2017)\citenamefont {Lan},
  \citenamefont {Kong},\ and\ \citenamefont
  {Wen}}]{lanModularExtensionsUnitary2017}%
  \BibitemOpen
  \bibfield  {author} {\bibinfo {author} {\bibfnamefont {Tian}\ \bibnamefont
  {Lan}}, \bibinfo {author} {\bibfnamefont {Liang}\ \bibnamefont {Kong}}, \
  and\ \bibinfo {author} {\bibfnamefont {Xiao-Gang}\ \bibnamefont {Wen}},\
  }\bibfield  {title} {\enquote {\bibinfo {title} {Modular {{Extensions}} of
  {{Unitary Braided Fusion Categories}} and 2+{{1D Topological}}/{{SPT Orders}}
  with {{Symmetries}}},}\ }\href {\doibase 10.1007/s00220-016-2748-y}
  {\bibfield  {journal} {\bibinfo  {journal} {Communications in Mathematical
  Physics}\ }\textbf {\bibinfo {volume} {351}},\ \bibinfo {pages} {709--739}
  (\bibinfo {year} {2017})}\BibitemShut {NoStop}%
\bibitem [{\citenamefont {Haah}(2011)}]{haah2011local}%
  \BibitemOpen
  \bibfield  {author} {\bibinfo {author} {\bibfnamefont {Jeongwan}\
  \bibnamefont {Haah}},\ }\bibfield  {title} {\enquote {\bibinfo {title} {Local
  stabilizer codes in three dimensions without string logical operators},}\
  }\href@noop {} {\bibfield  {journal} {\bibinfo  {journal} {Phys. Rev. At.
  Mol. Opt. Phys.}\ }\textbf {\bibinfo {volume} {83}},\ \bibinfo {pages}
  {42330} (\bibinfo {year} {2011})}\BibitemShut {NoStop}%
\bibitem [{\citenamefont {Vijay}\ \emph {et~al.}(2016)\citenamefont {Vijay},
  \citenamefont {Haah},\ and\ \citenamefont {Fu}}]{vijay2016fracton1}%
  \BibitemOpen
  \bibfield  {author} {\bibinfo {author} {\bibfnamefont {Sagar}\ \bibnamefont
  {Vijay}}, \bibinfo {author} {\bibfnamefont {Jeongwan}\ \bibnamefont {Haah}},
  \ and\ \bibinfo {author} {\bibfnamefont {Liang}\ \bibnamefont {Fu}},\
  }\bibfield  {title} {\enquote {\bibinfo {title} {Fracton topological order,
  generalized lattice gauge theory, and duality},}\ }\href@noop {} {\bibfield
  {journal} {\bibinfo  {journal} {Phys. Rev. B}\ }\textbf {\bibinfo {volume}
  {94}},\ \bibinfo {pages} {235157} (\bibinfo {year} {2016})}\BibitemShut
  {NoStop}%
\bibitem [{\citenamefont {Nandkishore}\ and\ \citenamefont
  {Hermele}(2019)}]{nandkishore2019fractons}%
  \BibitemOpen
  \bibfield  {author} {\bibinfo {author} {\bibfnamefont {Rahul~M}\ \bibnamefont
  {Nandkishore}}\ and\ \bibinfo {author} {\bibfnamefont {Michael}\ \bibnamefont
  {Hermele}},\ }\bibfield  {title} {\enquote {\bibinfo {title} {Fractons},}\
  }\href@noop {} {\bibfield  {journal} {\bibinfo  {journal} {Annu. Rev.
  Condens. Matter Phys.}\ }\textbf {\bibinfo {volume} {10}},\ \bibinfo {pages}
  {295--313} (\bibinfo {year} {2019})}\BibitemShut {NoStop}%
\bibitem [{\citenamefont {Sundar}\ \emph {et~al.}(2019)\citenamefont {Sundar},
  \citenamefont {Thibodeau}, \citenamefont {Wang}, \citenamefont {Gadway},\
  and\ \citenamefont {Hazzard}}]{BSundar2019}%
  \BibitemOpen
  \bibfield  {author} {\bibinfo {author} {\bibfnamefont {Bhuvanesh}\
  \bibnamefont {Sundar}}, \bibinfo {author} {\bibfnamefont {Matthew}\
  \bibnamefont {Thibodeau}}, \bibinfo {author} {\bibfnamefont {Zhiyuan}\
  \bibnamefont {Wang}}, \bibinfo {author} {\bibfnamefont {Bryce}\ \bibnamefont
  {Gadway}}, \ and\ \bibinfo {author} {\bibfnamefont {Kaden R.A.~A}\
  \bibnamefont {Hazzard}},\ }\bibfield  {title} {\enquote {\bibinfo {title}
  {Strings of ultracold molecules in a synthetic dimension},}\ }\href {\doibase
  10.1103/PhysRevA.99.013624} {\bibfield  {journal} {\bibinfo  {journal} {Phys.
  Rev. At. Mol. Opt. Phys.}\ }\textbf {\bibinfo {volume} {99}},\ \bibinfo
  {pages} {13624} (\bibinfo {year} {2019})},\ \Eprint
  {http://arxiv.org/abs/1812.02229} {arXiv:1812.02229} \BibitemShut {NoStop}%
\bibitem [{\citenamefont {Tantivasadakarn}\ \emph {et~al.}(2023)\citenamefont
  {Tantivasadakarn}, \citenamefont {Vishwanath},\ and\ \citenamefont
  {Verresen}}]{HierarchyTPO_LOCC}%
  \BibitemOpen
  \bibfield  {author} {\bibinfo {author} {\bibfnamefont {Nathanan}\
  \bibnamefont {Tantivasadakarn}}, \bibinfo {author} {\bibfnamefont {Ashvin}\
  \bibnamefont {Vishwanath}}, \ and\ \bibinfo {author} {\bibfnamefont {Ruben}\
  \bibnamefont {Verresen}},\ }\bibfield  {title} {\enquote {\bibinfo {title}
  {Hierarchy of topological order from finite-depth unitaries, measurement, and
  feedforward},}\ }\href {\doibase 10.1103/PRXQuantum.4.020339} {\bibfield
  {journal} {\bibinfo  {journal} {PRX Quantum}\ }\textbf {\bibinfo {volume}
  {4}},\ \bibinfo {pages} {20339} (\bibinfo {year} {2023})}\BibitemShut
  {NoStop}%
\bibitem [{\citenamefont {Iqbal}\ \emph {et~al.}(2024)\citenamefont {Iqbal},
  \citenamefont {Tantivasadakarn}, \citenamefont {Verresen}, \citenamefont
  {Campbell}, \citenamefont {Dreiling}, \citenamefont {Figgatt}, \citenamefont
  {Gaebler}, \citenamefont {Johansen}, \citenamefont {Mills}, \citenamefont
  {Moses}, \citenamefont {Pino}, \citenamefont {Ransford}, \citenamefont
  {Rowe}, \citenamefont {Siegfried}, \citenamefont {Stutz}, \citenamefont
  {{Foss-Feig}}, \citenamefont {Vishwanath},\ and\ \citenamefont
  {Dreyer}}]{Iqbal2024}%
  \BibitemOpen
  \bibfield  {author} {\bibinfo {author} {\bibfnamefont {Mohsin}\ \bibnamefont
  {Iqbal}}, \bibinfo {author} {\bibfnamefont {Nathanan}\ \bibnamefont
  {Tantivasadakarn}}, \bibinfo {author} {\bibfnamefont {Ruben}\ \bibnamefont
  {Verresen}}, \bibinfo {author} {\bibfnamefont {Sara~L}\ \bibnamefont
  {Campbell}}, \bibinfo {author} {\bibfnamefont {Joan~M}\ \bibnamefont
  {Dreiling}}, \bibinfo {author} {\bibfnamefont {Caroline}\ \bibnamefont
  {Figgatt}}, \bibinfo {author} {\bibfnamefont {John~P}\ \bibnamefont
  {Gaebler}}, \bibinfo {author} {\bibfnamefont {Jacob}\ \bibnamefont
  {Johansen}}, \bibinfo {author} {\bibfnamefont {Michael}\ \bibnamefont
  {Mills}}, \bibinfo {author} {\bibfnamefont {Steven~A}\ \bibnamefont {Moses}},
  \bibinfo {author} {\bibfnamefont {Juan~M}\ \bibnamefont {Pino}}, \bibinfo
  {author} {\bibfnamefont {Anthony}\ \bibnamefont {Ransford}}, \bibinfo
  {author} {\bibfnamefont {Mary}\ \bibnamefont {Rowe}}, \bibinfo {author}
  {\bibfnamefont {Peter}\ \bibnamefont {Siegfried}}, \bibinfo {author}
  {\bibfnamefont {Russell~P}\ \bibnamefont {Stutz}}, \bibinfo {author}
  {\bibfnamefont {Michael}\ \bibnamefont {{Foss-Feig}}}, \bibinfo {author}
  {\bibfnamefont {Ashvin}\ \bibnamefont {Vishwanath}}, \ and\ \bibinfo {author}
  {\bibfnamefont {Henrik}\ \bibnamefont {Dreyer}},\ }\bibfield  {title}
  {\enquote {\bibinfo {title} {Non-{{Abelian}} topological order and anyons on
  a trapped-ion processor},}\ }\href {\doibase 10.1038/s41586-023-06934-4}
  {\bibfield  {journal} {\bibinfo  {journal} {Nature}\ }\textbf {\bibinfo
  {volume} {626}},\ \bibinfo {pages} {505--511} (\bibinfo {year}
  {2024})}\BibitemShut {NoStop}%
\end{thebibliography}%
\end{document}